\documentclass[10pt,letterpaper]{report} 

\usepackage[mathscr]{eucal}
\usepackage{epsfig,epsf,psfrag}
\usepackage{amssymb,amsmath,amsthm,amsfonts,latexsym}
\usepackage{amsmath,graphicx,bm,xcolor,url}
\usepackage[caption=false]{subfig} 
\usepackage{fixltx2e}
\usepackage{array}
\usepackage{verbatim}
\usepackage{bm}
\usepackage{algorithmic}
\usepackage{algorithm}
\usepackage{verbatim}
\usepackage{textcomp}
\usepackage{mathrsfs, bbm,overpic}
\usepackage{epstopdf}

\usepackage{stackengine}

\usepackage{tikz}
\usepackage{float}
\usepackage{tabularx}
\usepackage{multirow}
\usetikzlibrary{patterns}

\catcode`~=11 \def\UrlSpecials{\do\~{\kern -.15em\lower .7ex\hbox{~}\kern .04em}} \catcode`~=13 

\allowdisplaybreaks[1]
 
\newcommand{\nn}{\nonumber}

\newcommand{\calA}{\mathcal{A}}
\newcommand{\calB}{\mathcal{B}}
\newcommand{\calC}{\mathcal{C}}
\newcommand{\calD}{\mathcal{D}}
\newcommand{\calE}{\mathcal{E}}
\newcommand{\calF}{\mathcal{F}}
\newcommand{\calG}{\mathcal{G}}

\newcommand{\calL}{\mathcal{L}}
\newcommand{\calM}{\mathcal{M}}
\newcommand{\calN}{\mathcal{N}}

\newcommand{\calP}{\mathcal{P}}
\newcommand{\calQ}{\mathcal{Q}}
\newcommand{\calR}{\mathcal{R}}
\newcommand{\calS}{\mathcal{S}}
\newcommand{\calT}{\mathcal{T}}
\newcommand{\calU}{\mathcal{U}}
\newcommand{\calV}{\mathcal{V}}
\newcommand{\calW}{\mathcal{W}}
\newcommand{\calX}{\mathcal{X}}
\newcommand{\calY}{\mathcal{Y}}
\newcommand{\calZ}{\mathcal{Z}}

\newcommand{\hatcalX}{\hat{\calX}}

\newcommand{\tilrmV}{\tilde{\rmV}}

\newcommand{\bA}{\mathbf{A}}
\newcommand{\bb}{\mathbf{b}}

\newcommand{\bs}{\mathbf{s}}
\newcommand{\bS}{\mathbf{S}}
\newcommand{\bt}{\mathbf{t}}

\newcommand{\bV}{\mathbf{V}}

\newcommand{\bx}{\mathbf{x}}
\newcommand{\bX}{\mathbf{X}}

\newcommand{\bZ}{\mathbf{Z}}

\newcommand{\rmb}{\mathrm{b}}

\newcommand{\rmc}{\mathrm{c}}

\newcommand{\rmd}{\mathrm{d}}

\newcommand{\rme}{\mathrm{e}}

\newcommand{\rmG}{\mathrm{G}}

\newcommand{\rmH}{\mathrm{H}}

\newcommand{\rmP}{\mathrm{P}}

\newcommand{\rmQ}{\mathrm{Q}}

\newcommand{\rmR}{\mathrm{R}}
\newcommand{\rms}{\mathrm{s}}

\newcommand{\rmT}{\mathrm{T}}

\newcommand{\rmV}{\mathrm{V}}

\newcommand{\Var}{\mathrm{Var}}

\newcommand{\bbE}{\mathsf{E}}

\newcommand{\bbN}{\mathbb{N}}

\newcommand{\bbP}{\mathbb{P}}

\newcommand{\bbR}{\mathbb{R}}

\newcommand{\bbo}{\mathbbm{1}}

\DeclareMathAlphabet{\mathbsf}{OT1}{cmss}{bx}{n}
\DeclareMathAlphabet{\mathssf}{OT1}{cmss}{m}{sl}

\newcommand{\rvR}{\mathsf{R}}

\DeclareSymbolFont{bsfletters}{OT1}{cmss}{bx}{n}  
\DeclareSymbolFont{ssfletters}{OT1}{cmss}{m}{n}
\DeclareMathSymbol{\bsfGamma}{0}{bsfletters}{'000}
\DeclareMathSymbol{\ssfGamma}{0}{ssfletters}{'000}
\DeclareMathSymbol{\bsfDelta}{0}{bsfletters}{'001}
\DeclareMathSymbol{\ssfDelta}{0}{ssfletters}{'001}
\DeclareMathSymbol{\bsfTheta}{0}{bsfletters}{'002}
\DeclareMathSymbol{\ssfTheta}{0}{ssfletters}{'002}
\DeclareMathSymbol{\bsfLambda}{0}{bsfletters}{'003}
\DeclareMathSymbol{\ssfLambda}{0}{ssfletters}{'003}
\DeclareMathSymbol{\bsfXi}{0}{bsfletters}{'004}
\DeclareMathSymbol{\ssfXi}{0}{ssfletters}{'004}
\DeclareMathSymbol{\bsfPi}{0}{bsfletters}{'005}
\DeclareMathSymbol{\ssfPi}{0}{ssfletters}{'005}
\DeclareMathSymbol{\bsfSigma}{0}{bsfletters}{'006}
\DeclareMathSymbol{\ssfSigma}{0}{ssfletters}{'006}
\DeclareMathSymbol{\bsfUpsilon}{0}{bsfletters}{'007}
\DeclareMathSymbol{\ssfUpsilon}{0}{ssfletters}{'007}
\DeclareMathSymbol{\bsfPhi}{0}{bsfletters}{'010}
\DeclareMathSymbol{\ssfPhi}{0}{ssfletters}{'010}
\DeclareMathSymbol{\bsfPsi}{0}{bsfletters}{'011}
\DeclareMathSymbol{\ssfPsi}{0}{ssfletters}{'011}
\DeclareMathSymbol{\bsfOmega}{0}{bsfletters}{'012}
\DeclareMathSymbol{\ssfOmega}{0}{ssfletters}{'012}

\newcommand{\tilc}{\tilde{c}}

\newcommand{\hatm}{\hat{m}}

\newcommand{\hatT}{\hat{T}}

\newcommand{\hatx}{\hat{x}}
\newcommand{\hatX}{\hat{X}}

\newcommand{\hatbx}{\hat{\bx}}
\newcommand{\hatbX}{\hat{\bX}}

\newcommand{\haty}{\hat{y}}
\newcommand{\hatY}{\hat{Y}}

\newcommand{\barb}{\bar{b}}

\newcommand{\bard}{\bar{d}}

\newcommand{\barP}{\bar{P}}

\newcommand{\barR}{\bar{R}}

\newcommand{\bmu}{\bm{\mu}}

\newcommand{\bSigma }{\bm{\Sigma}}

\DeclareMathOperator*{\argmax}{arg\,max}
\DeclareMathOperator*{\argmin}{arg\,min}
\DeclareMathOperator*{\argsup}{arg\,sup}

\DeclareMathOperator{\supp}{supp}

\newcommand{\bzero}{\mathbf{0}}
\newcommand{\bone}{\mathbf{1}}

\newtheorem{theorem}{Theorem} 
\newtheorem{lemma}{Lemma}

\newtheorem{proposition}{Proposition}
\newtheorem{corollary}{Corollary}
\newtheorem{definition}{Definition} 

\usepackage[utf8]{inputenc}
\usepackage{mathtools}
\graphicspath{{./figures/}}
\usepackage[utf8]{inputenc}
\usepackage[ colorlinks = true,
             linkcolor = blue,
             urlcolor  = blue,
             citecolor = red,
             anchorcolor = green,]{hyperref}

\usepackage{fancyhdr}
\fancypagestyle{plain}{}
\title{On Finite Blocklength Lossy Source Coding}

\author{
Lin Zhou \\
School of Cyber Science and Technology\\
Beihang University \\
lzhou@buaa.edu.cn\\
\and
Mehul Motani\\
Department of Electrical and Computer Engineering\\
National University of Singapore\\
motani@nus.edu.sg
}

\begin{document}
\maketitle
\tableofcontents

\chapter*{Abstract}
\addcontentsline{toc}{chapter}{Abstract}
Shannon propounded a theoretical framework (collectively called information theory) that uses mathematical tools to understand, model and analyze modern mobile wireless communication systems. A key component of such a system is source coding, which compresses the data to be transmitted by eliminating redundancy and allows reliable recovery of the information from the compressed version. In modern 5G networks and beyond, finite blocklength lossy source coding is essential to provide ultra-reliable and low-latency communications. The analysis of point-to-point and multiterminal settings from the perspective of finite blocklength lossy source coding is therefore of great interest to 5G system designers and is also related to other long-standing problems in information theory.
 
In this monograph, we review recent advances in second-order asymptotics for lossy source coding, which provides approximations to the finite blocklength performance of optimal codes. The monograph is divided into three parts. In part I, we motivate the monograph, present basic definitions, introduce mathematical tools and illustrate the motivation of non-asymptotic and second-order asymptotics via the example of lossless source coding. In part II, we first present existing results for the rate-distortion problem with proof sketches. Subsequently, we present five generations of the rate-distortion problem to tackle various aspects of practical quantization tasks: noisy source, noisy channel, mismatched code, Gauss-Markov source and fixed-to-variable length compression. By presenting theoretical bounds for these settings, we illustrate the effect of noisy observation of the source, the influence of noisy transmission of the compressed information, the effect of using a fixed coding scheme for an arbitrary source and the roles of source memory and  variable rate. In part III, we present four multiterminal generalizations of the rate-distortion problem to consider multiple encoders, decoders or source sequences: the Kaspi problem, the successive refinement problem, the Fu-Yeung problem and the Gray-Wyner problem. By presenting theoretical bounds for these multiterminal problems, we illustrate the role of side information, the optimality of stop and transmit, the effect of simultaneous lossless and lossy compression, and the tradeoff between encoders' rates in compressing correlated sources.  Finally, we conclude the monograph, mention related results and discuss future directions.

\clearpage

\part{Basics}

\chapter{Introduction}
\label{chap:intro}

\section{Motivation}

Shannon~\cite{shannon1948mathematical} developed a theoretical framework (collectively called information theory) that uses mathematical tools to understand, model and analyze digital communication systems over noisy channels. A basic digital communication system includes blocks for source and channel encoding at the transmitter and blocks for source and channel decoding at the receiver. Source coding, also known as data compression, aims to remove the redundancy of information and allows reliable recovery of the information from its compressed version. In contrast, channel coding aims to counter the noise in the transmission channel between the transmitter and the receiver and allows reliable recovery of a message. 

For a discrete memoryless source (DMS), Shannon showed that the asymptotic minimal compression rate that ensures accurate recovery with vanishing error probability is the entropy of the source, provided that the blocklength of the source sequence to be compressed tends to infinity. However, lossless source coding does not apply to continuous sources since it requires an infinite number of bits to describe a real number. Furthermore, practical image and video compression systems usually tolerate some imperfection. To resolve these issues, Shannon studied the lossy source coding problem~\cite{shannon1959coding} (also known as the rate-distortion problem) and derived the asymptotic minimal achievable rate.  

For a discrete memoryless channel (DMC), Shannon showed that the maximal asymptotic message rate to ensure reliable recovery with vanishing error probability at the receiver is the capacity of the noisy channel, provided that the blocklength (the number of channel uses) tends to infinity. In other words, Shannon showed that, at rates below the channel capacity, there exist good channel coding strategies with arbitrarily low probability of error. The above results for source coding and channel coding are collectively known as Shannon's coding theorems~\cite{cover2012elements}. These results are very insightful and set benchmarks for practical code design in the last seventy years.

In practical communication systems, especially in 5G and beyond, low-latency is desired and dictates the use of short blocklength codes.  However, Shannon's coding theorems cannot provide exact theoretical benchmarks for low-latency communication since these theorems hold under the assumption that the blocklength tends to infinity, which leads to undesired arbitrarily large latency. To tackle this problem, information theorists developed the theory of finite blocklength analysis and second-order asymptotic approximation, starting with the seminal work of Strassen~\cite{strassen1962asymptotische} in 1962. The finite blocklength analysis for channel coding has been revived by Hayashi~\cite{hayashi2009information} and by Polyanskiy, Poor and Verd\'u~\cite{polyanskiy2010finite}. In particular, the authors of \cite{polyanskiy2010finite} derived upper and lower bounds for any finite blocklength and showed that the bounds match the dispersion type Gaussian approximation for blocklength of hundreds for various types of point-to-point channels. The Gaussian approximation is coined second-order asymptotics by Hayashi~\cite{hayashi2009information}. The results of \cite{hayashi2009information} and \cite{polyanskiy2010finite} have been generalized to various channel models. Readers can refer to \cite{TanBook} for a systematic review of such advances.

Finite blocklength analyses and second-order asymptotics have also been derived for source coding. The simplest such example is the lossless source coding problem. In this problem, one aims to recover a random source sequence $X^n$ exactly from its compressed version that takes values in a finite set of $M$ elements. The performance metric is the error probability in reproducing the source sequence and the rate is defined as $R_n:=\frac{\log M}{n}$, where the unit is bits per source symbol when the logarithm is base $2$. In second-order asymptotics, one is interested in characterizing the backoff of the non-asymptotic coding rate $R_n$ from the minimum achievable rate --- the entropy of the source $H(P_X)$, while tolerating a non-vanishing error probability. Such a result was first shown by Yushkevich for sources with Markovian memory \cite{yushkevich1953limit}. Strassen~\cite{strassen1962asymptotische}, and later Hayashi~\cite{hayashi2008source}, showed that the backoff is in the order of the reciprocal of the square root of the blocklength. Such a result is simple and elegant and parallels the finite blocklength results of channel coding. 

As noted by Shannon~\cite{shannon1959coding}, lossless source coding is not possible for continuous sources and lossy source coding with imperfect recovery is thus important. Shannon's rate-distortion theory~\cite{shannon1959coding} forms a core part of modern quantization theory and is usually known as vector quantization. For a complete survey of various aspects of quantization, readers may refer to the seminal paper by Gray and Neuhoff~\cite{gray1998tit}. For the rate-distortion problem that deals with point-to-point lossy data compression, the second-order asymptotics for a DMS were derived by Ingber and Kochman~\cite{ingber2011dispersion}, and both finite blocklength bounds and second-order asymptotics were derived by Kostina and Verd\'u~\cite{kostina2012fixed} for a DMS and a Gaussian memoryless source (GMS). The results in \cite{ingber2011dispersion,kostina2012fixed} were further generalized to various scenarios in the point-to-point case~\cite{kostina2016noisy,kostinajscc,wang2011dispersion,zhou2017refined,tian2019tit,kostina2015tit} and to problems in network information theory~\cite{zhou2017non,zhou2017kaspishort,zhou2016second,no2016,zhou2017fy,zhou2015second}, usually for a DMS.

However, despite the undeniable importance of lossy source coding and its diverse applications beyond low-latency communications in various domains including privacy utility tradeoff~\cite{sankar2013utility}, machine learning~\cite{gao2019rate} and image/video compression~\cite{sullivan1998rate,ortega1998rate,habibian2019video}, there is no single source that systematically summarizes recent advances for finite blocklength analyses and second-order asymptotics of lossy source coding problems, especially the multiterminal cases. One might argue that \cite{TanBook} covers these topics. Specifically, \cite[Chapter 3]{TanBook} focuses on the point-to-point setting by presenting non-asymptotic and refined asymptotics bounds for both lossless and lossy source coding problems, \cite[Chapter 4.5]{TanBook} briefly presents the results for joint source-channel coding without proof sketches while \cite[Chapter 6]{TanBook} studies a lossless multiterminal source coding problem named the Slepian-Wolf problem~\cite{slepian1973noiseless}. It is important to note that recent advances of lossy source coding (e.g., \cite{kostina2016noisy,kostina2015tit,zhou2017refined,zhou2015second}) and the multiterminal cases~\cite{zhou2017non,zhou2017kaspishort,zhou2016second,no2016,zhou2017fy,zhou2015second} are not included in \cite{TanBook}. Our monograph aims to fill the missing piece of finite blocklength analyses by summarizing recent theoretical advances for finite blocklength lossy source coding problems. Furthermore, for point-to-point lossless and lossy source coding problems, we present proof techniques different from those covered in \cite[Chapter 3]{TanBook}.

\section{Organization}
The rest of this monograph is organized as follows. In the rest of this chapter, we present the notation used throughout the monograph and recall critical mathematical theorems on sums of i.i.d. random variables including the Berry-Esseen theorem~\cite{berry1941accuracy,esseen1942liapounoff}. In Chapter \ref{chap:lossless}, we illustrate the meaning of finite blocklength analysis, first-order asymptotics, and second-order asymptotics via the example of lossless source coding. We also recall other refined asymptotics including large and moderate deviations and explain why we focus on second-order asymptotics.

Part II of this monograph is devoted to the rate-distortion problem and its five generalizations to consider various aspects of practical quantization tasks. In Chapter \ref{chap:rd}, we review existing results on the rate-distortion problem. Specifically, we formulate the problem of finite blocklength analysis of the rate-distortion problem, define the distortion-tilted information density, present non-asymptotic and second-order asymptotic theorems, and finally provide detailed proof sketches. This chapter is mainly based on \cite{kostina2012fixed,ingber2011dispersion}.

In Chapter \ref{chap:noisy}, we present results for the noisy lossy source coding problem, where the encoder can only access a noisy version of the source sequence. This problem is also known as quantizing noisy sources and is motivated by practical compression of speech signals distorted by environmental noise or images corrupted by camera imperfections. The non-asymptotic and second-order asymptotic results for this problem reveal the role of noisy observations in the finite blocklength regime, which is not apparent in asymptotic analyses~\cite{Dobrushin1962,sakrison1968,Witsenhausen1980tit}. This chapter is based on~\cite{kostina2016noisy}.

In Chapter \ref{chap:jscc}, we present results for the lossy joint source-channel coding problem, where the output of the encoder is passed though a noisy channel and then provided to the decoder. This problem is also known as quantization for a noisy channel. The classical separation theorem of Shannon establishes that it is asymptotically optimal to separate lossy source coding and channel coding. However, non-asymptotic and second-order asymptotic results suggest that, at finite blocklengths, separate source-channel coding is strictly suboptimal. This chapter is based on~\cite{kostinajscc,wang2011dispersion}.

In Chapter \ref{chap:mismatch}, we present results for the mismatched code of Lapidoth~\cite[Theorem 3]{lapidoth1997}, where a fixed code with an i.i.d. Gaussian codebook and minimum Euclidean distance encoding is used to compress an arbitrary memoryless source. This problem is motivated by the fact that the distribution of the source to be compressed is usually unknown and thus the matched coding scheme where the source distribution is assumed perfectly known is impractical. Theoretical results demonstrate that both i.i.d. Gaussian and spherical codebooks achieve the same finite blocklength performance. This chapter is based on \cite{zhou2017refined}.

In Chapter \ref{chap:gm}, we present results for the Gauss-Markov source, where the source sequence forms a first-order Markov chain and thus has memory. This problem is motivated by practical applications where the source sequence, such as sensor data, is usually not memoryless. The non-asymptotic and second-order results for the Gauss-Markov source is the first for a source with memory and
reveal the role of memory on the finite blocklength performance of optimal codes. This chapter is based on~\cite{tian2019tit}.

In Chapter \ref{chap:variable}, we present results for fixed-to-variable length compression, where the encoder's output to each source sequence is a binary string with potentially different lengths. The motivation is to further reduce the average coding rate based on the intuition that more frequent symbols should be assigned codewords with fewer bits, an idea captured in the Huffman code. The theoretical results reveal the role of flexible rates on the finite blocklength performance and demonstrate a stark difference with the fixed-length counterpart. This chapter is based on~\cite{kostina2015tit}.

Part III deals with four multiterminal extensions of the rate-distortion problem with increasing complexity and also includes a conclusion chapter. In Chapter \ref{chap:kaspi}, we present results for the Kaspi problem~\cite{kaspi1994}, which is a lossy source coding problem with one encoder and two decoders. This problem generalizes the rate-distortion problem by providing side information at the encoder and adding one additional decoder that accesses the same side information. Both decoders share the same compressed information of the source sequence and the decoder with side information is required to produce a finer estimate of the source sequence. Through the lens of this problem, we reveal the impact of side information on the finite blocklength performance of optimal codes. This chapter is mainly based on the first part of \cite{zhou2017non}.

In Chapter \ref{chap:sr}, we present results for the successive refinement problem~\cite{rimoldi1994}. This problem generalizes the rate-distortion problem by having one additional encoder and decoder pair. The additional encoder further compresses the source sequence and the additional decoder uses compressed information from both encoders to produce a finer estimate of the source sequence than the other decoder that only has access to the original encoder. We present results under two performance criteria: the joint excess-distortion probability (JEP) and the separate excess-distortion probabilities (SEP). Under JEP, we reveal the tradeoff between the coding rate of the two encoders and, under SEP, we revisit the successively refinability property, from a second-order asymptotic perspective. A key message from this chapter is that considering a joint excess-distortion probability enables us to characterize the tradeoff of rates of different encoders in second-order asymptotics. This chapter is mainly based on \cite{no2016,zhou2016second}.

In Chapter \ref{chap:fu-yeung}, we present results for the multiple description problem with one deterministic decoder~\cite{fu2002rate}. In this problem, two encoders compress the source sequence and three decoders aim to recover the source sequence with different criteria: two decoders aim to recover the source sequence in a lossy manner with different distortion levels and the other decoder aims to perfectly reproduce a function of the source sequence. This problem generalizes the successive refinement problem by having one additional lossless decoder. Under the joint excess-distortion and error probability criterion, we reveal the tradeoff among encoders and decoders in simultaneous lossless and lossy compression in second-order asymptotics. This chapter is mainly based on the second part of \cite{zhou2017non}.

In Chapter \ref{chap:gw}, we present results for the lossy Gray-Wyner problem~\cite{gray1974source}. In this problem, three encoders compress two correlated source sequences and each of the two decoders aims to recover one source sequence. This is a fully multiterminal lossy compression problem with multiple encoders, multiple decoders and multiple correlated source sequences. It significantly generalizes the rate-distortion problem by having one more source sequence, two more encoders and one more decoder. Under the joint excess-distortion probability criterion, we reveal the tradeoff among the coding rates of the three encoders in second-order asymptotics. This chapter is mainly based on \cite{zhou2015second}.

Finally, in Chapter \ref{chap:future}, we conclude the monograph and discuss future research directions. The relationship among chapters of this monograph is illustrated in Fig. \ref{illus:logic}.
\begin{figure}[t]
\centering
\includegraphics[width=\columnwidth]{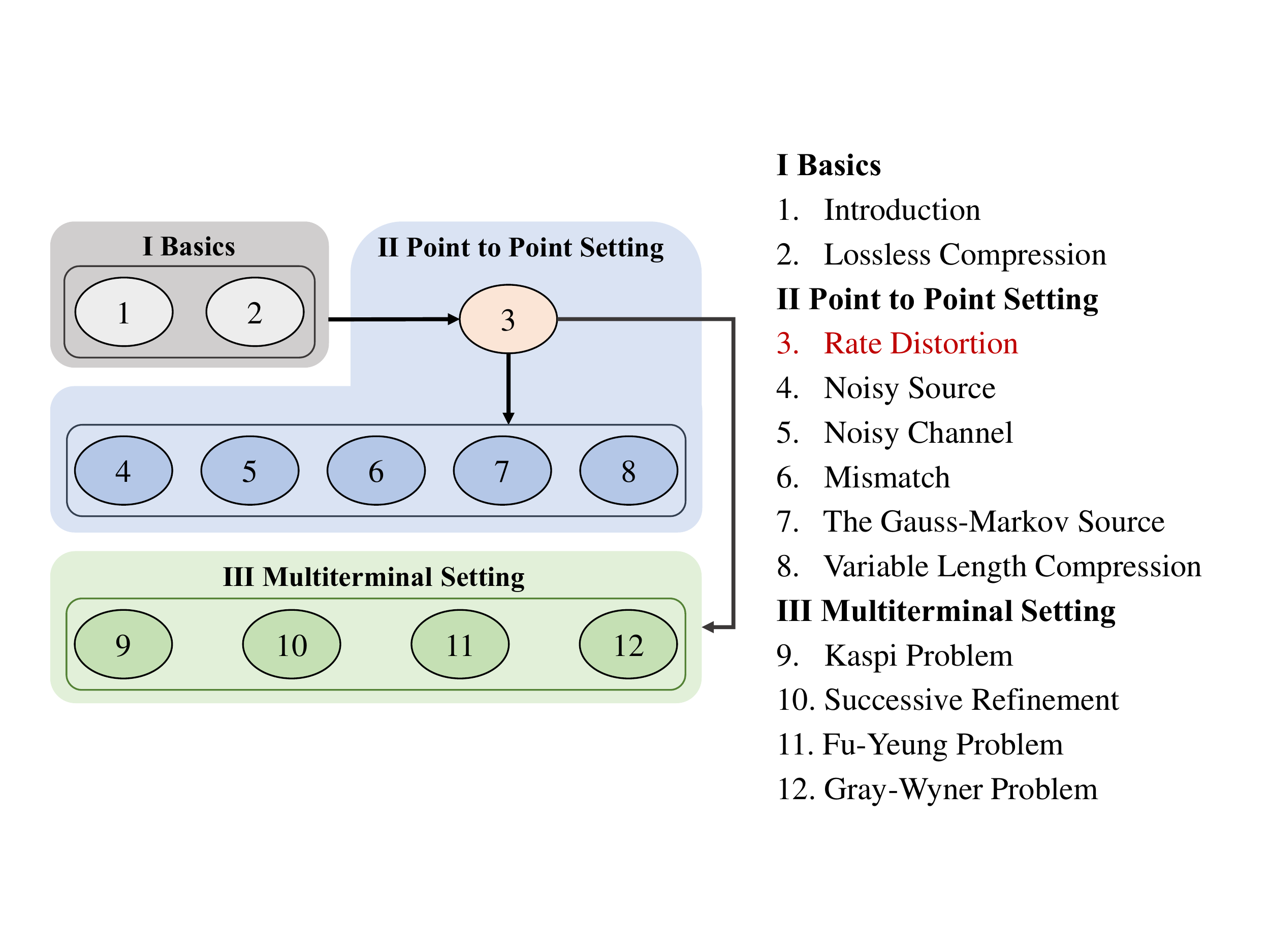}
\caption{Relationship among chapters of this monograph.}
\label{illus:logic}
\end{figure}

\section{Preliminaries}
\label{sec:prelim}
In this section, we set up the mathematical notation used throughout the monograph and review definitions of basic information theoretical quantities, key properties in method of types and mathematical theorems central to our analyses.

\subsection{Notation}

The set of real numbers, non-negative real numbers, and natural numbers are denoted by $\bbR$, $\bbR_+$, and $\bbN$, respectively. For any two natural numbers $(a,b)\in\bbN^2$, we use $[a:b]$ to denote the set of all natural numbers between $a$ and $b$ (inclusive) and use $[a]$ to denote $[1:a]$. For any $(m_1,m_2)\in\bbN^2$, we use $\bzero_{m_1}$ to denote the length-$m_1$ vector of all zeroes and use $\bone_{m_1,m_2}$ to denote the $m_1\times m_2$ matrix of all ones. For any real number $a\in\bbR$, we use $|a|^+$ to denote $\max\{a,0\}$.

Random variables and their realizations are in capital (e.g.,\ $X$) and lower case (e.g.,\ $x$) respectively. All sets (e.g., alphabets of random variables) are denoted in calligraphic font (e.g.,\ $\mathcal{X}$). We use $\calX^{\mathrm{c}}$ to denote the complement of $\calX$. Let $X^n:=(X_1,\ldots,X_n)$ be a random vector of length-$n$ and $x^n=(x_1,\ldots,x_n)$ be a particular realization. We use $\|x^n\|=\sqrt{\sum_{i\in[n]} x_i^2}$ to denote the $\ell_2$ norm of a vector $x^n\in\bbR^n$.  Given two sequences $x^n$ and $y^n$, the quadratic distortion measure (squared Euclidean norm) is defined as $d(x^n,y^n):=\frac{1}{n}\|x^n-y^n\|^2=\frac{1}{n}\sum_{i\in[n]}(x_i-y_i)^2$. 

The set of all probability distributions on an alphabet $\calX$ is denoted by $\calP(\calX)$ and the set of all conditional probability distribution from $\calX$ to $\calY$ is denoted by $\calP(\calY|\calX)$. Given $P\in\calP(\calX)$, we use $\supp(P)$ to denote the support of distribution $P$, i.e., $\supp(P)=\{x\in\calX:P(x)>0\}$. Given a conditional distribution $P_{Y|X}\in\calP(\calY|\calX)$ and $x\in\calX$, we use $P_{Y|x}$ to denote the conditional distribution $P_{Y|X}(\cdot|x)$. Given $P\in\calP(\calX)$ and $V\in\calP(\calY|\calX)$, we use $P\times V$ to denote the joint distribution induced by $P$ and $V$.  Given a joint probability distribution $P_{XY}\in\calP(\calX\times\calY)$, let $m=|\supp(P_{XY})|$ and let $\Gamma(P_{XY})$ be the sorted distribution such that for each $i\in[m]$, $\Gamma_i(P_{XY})=P_{XY}(x_i,y_i)$ is the $i$-th largest value of $\{P_{XY}(x,y):~(x,y)\in\calX\times\calY\}$.

We use standard asymptotic notations such as $\Theta(\cdot)$, $O(\cdot)$ and $o(\cdot)$~(cf.~\cite{Cor03}). We use $\bbo(\cdot)$ as the indicator function and we use $\log(\cdot)$ with base $e$ unless otherwise stated. We let $\rmQ(t) := \int_{t}^\infty\frac{1}{\sqrt{2\pi}}e^{-u^2/2}\, \rmd u$ be the complementary cumulative distribution function of the standard Gaussian. Let $\mathrm{Q}^{-1}$ be the inverse of $\rmQ$. We use $\Psi_k(x_1,\ldots,x_k;\bmu,\mathbf{\Sigma})$ to denote the multivariate generalization of the Gaussian cumulative distribution function (cdf), i.e.,
$\Psi(x_1,\ldots,x_k;\bmu,\mathbf{\Sigma})=\int_{-\infty}^{x_1}\ldots\int_{-\infty}^{x_k}\calN(\bx; \bmu;\bSigma)\, \rmd \bx$, where $\calN(\bx; \bmu;\bSigma)$ is the probability density function (PDF) of a $k$-variate Gaussian with mean vector $\bmu$ and covariance matrix $\bSigma$.

\subsection{Basic Definitions}
To smoothly present the results in this monograph, we recall necessary information theoretical definitions. Given any distribution $P_X\in\calP(\calX)$ defined on a finite alphabet $\calX$, the entropy is defined as
\begin{align}
H(X)=H(P_X)&:=\sum_{x\in\supp(P_X)}-P_X(x)\log P_X(x).
\end{align}
Note that the notation $H(X)$ is used in classical textbooks as \cite{cover2012elements} and the notation $H(P_X)$ that clarifies the dependence of the entropy on the distribution is used in \cite{csiszar2011information}. We use both notations for the entropy and other information theoretical quantities interchangeably. Specifically, when we need to specify the distribution of a random variable, we use the distribution dependence version $H(P_X)$; when the distribution of the random variable is clear, we use $H(X)$ for its simplicity. Analogously, given a joint probability mass function (PMF) $P_{XY}\in\calP(\calX\times\calY)$ defined on a finite alphabet $\calX\times\calY$, the joint entropy is defined as
\begin{align}
H(X,Y)=H(P_{XY})=\sum_{(x,y)\in\supp(P_{XY})}-P_{XY}(x,y)\log P_{XY}(x,y),
\end{align}
and the conditional entropy of $Y$ given $X$ is defined as
\begin{align}
H(Y|X)=H(P_{Y|X}|P_X)=\sum_{(x,y)\in\supp(P_{XY})}-P_{XY}(x,y)\log P_{Y|X}(x,y),
\end{align}
where $(P_{Y|X},P_X)$ are the induced conditional and marginal distributions of $P_{XY}$. The conditional entropy $H(P_{Y|X}|P_X)$ of $X$ given $Y$ is defined similarly. 

Furthermore, the mutual information that measures dependence of two random variables $(X,Y)$ with distribution $P_{XY}$ is defined as
\begin{align}
I(X;Y)
&=I(P_X,P_{X|Y})=H(P_X)-H(P_{X|Y}|P_Y),
\end{align}
where $P_{X|Y}$ is also induced by $P_{XY}$. Note that mutual information $I(X;Y)$ is symmetric so that $I(X;Y)=I(Y;X)$. Similar to the definition of entropy, we use $I(X;Y)$ and the distribution dependence version $I(P_X,P_{X|Y})$ interchangeably. Analogously, given the joint distribution $P_{XYZ}$ of three random variables $(X,Y,Z)$ defined on a finite alphabet $\calX\times\calY\times\calZ$, define the conditional mutual information $I(X;Y|Z)$ as
\begin{align}
I(X;Y|Z)
=I(P_{X|Z},P_{X|YZ}|P_Z)=H(P_{X|Z}|P_Z)-H(P_{X|YZ}|P_{YZ}),
\end{align}
where all distributions are induced by the joint distribution $P_{XYZ}$.

Another critical quantity that we use frequency is the Kullback-Leiber (KL) divergence, also known as the relative entropy. Given any two distributions $(P_X,Q_X)$ defined on the finite alphabet $\calX$, the KL divergence $D(P_X\|Q_X)$ is defined as
\begin{align}
D(P_X\|Q_X)=\sum_{x\in\supp(P_X)}P_X(x)\log\frac{P_X(x)}{Q_X(x)}.
\end{align}
Note that $D(P_X\|Q_X)$ measures closeness of two distributions $P_X$ and $Q_X$ and equals zero if and only if $P_X=Q_X$. For any two distributions $P_{XY}$ and $Q_{XY}$ defined on a finite alphabet $\calX\times\calY$, the KL divergence $D(P_{XY}\|Q_{XY})$ is defined similarly; when the marginal distributions $P_X=Q_X$, the conditional KL divergence is defined as
\begin{align}
D(P_{Y|X}\|Q_{Y|X}|P_X)=\sum_{x\in\supp(P_X)}P_X(x)D(P_{Y|X}(\cdot|x)\|Q_{Y|X}(\cdot|x)).
\end{align}

\subsection{The Method of Types}
Since we focus on DMSes, the method of types plays a critical role in our analyses. Thus, we also recall definitions and results in this domain~\cite{csiszar1998mt} (see also~\cite[Chapter 11]{cover2012elements} and \cite[Chapter 2]{csiszar2011information}). Given a length-$n$ discrete sequence $x^n\in\calX^n$, the empirical distribution $\hatT_{x^n}$ is defined as
\begin{align}
\hatT_{x^n}(a)=\frac{1}{n}\sum_{i\in[n]}1\{x_i=a\},~\forall~a\in\calX.
\end{align}
The set of types formed from length-$n$ sequences in $\calX$ is denoted by $\calP_{n}(\calX)$. Given a type $P_X\in\calP_{n}(\calX)$, the set of all sequences of length-$n$ with type $P_X$ is the type class denoted by $\calT_{P_X}$. For any $n\in\bbN$, the number of types satisfies
\begin{align}
|\calP_n(\calX)|\leq (n+1)^{|\calX|}\label{number:types}.
\end{align}
For any type $P_X\in\calP_n(\calX)$, the size of type class $\calT_P^n$ satisfies
\begin{align}
(n+1)^{-|\calX|}\exp(nH(P_X))\leq |\calT_{P_X}^n|\leq \exp(nH(P_X)).
\end{align}
For any sequence $x^n$ that is generated i.i.d. from a distribution $P_X\in\calP(\calX)$, its probability satisfies
\begin{align}
P_X^n(x^n)=\exp(-n (D(\hatT_{x^n}\|P_X)+H(\hatT_{x^n}))).
\end{align}
Thus, for any type $Q_X\in\calP_n(\calX)$, the probability of the type class $\calT_{Q_X}^n$ satisfies
\begin{align}
(n+1)^{-|\calX|}\leq \frac{P_X^n(\calT_{Q_X}^n)}{\exp(-nD(Q_X\|P_X))}\leq 1.
\end{align}

Given any two sequences $(x^n,y^n)\in\calX^n\times\calY^n$, the joint empirical distribution $\hatT_{x^ny^n}$ is defined as
\begin{align}
\hatT_{x^ny^n}(a,b)=\frac{1}{n}\sum_{i\in[n]}1\{(x_,y_i)=(a,b)\}.
\end{align}
Given any $x^n\in\calX^n$ and conditional distribution $V_{Y|X}\in\calP(\calY|\calX)$, the set of all sequences $y^n\in\calY^n$ such that $\hatT_{x^ny^n}=\calT_{x^n}\times V_{Y|X}$ is the conditional type class denoted by $\calT_{V_{Y|X}}(x^n)$. For any $x^n\in\calT_{P_X}^n$, the set of all conditional distributions $V_{Y|X}\in\calP(\calY|\calX)$ such that the conditional type class $\calT_{V_{Y|X}}(x^n)$ is not empty is the set of conditional types given the marginal type $P_X$ and is denoted by $\calV^n(\calY;P_X)$.

The following results hold. For any $P_X\in\calP_n(\calX)$, the number of conditional types is upper bounded by
\begin{align}
|\calV^n(\calY;P_X)|\leq (n+1)^{|\calX||\calY|}.
\end{align}
For any $x^n\in\calT_{P_X}^n$, the size of the conditional type class $\calT_{V_{Y|X}}(x^n)$ satisfies
\begin{align}
(n+1)^{-|\calX||\calY|}\exp(nH(V_{Y|X})|P_X)\leq |\calT_{V_{Y|X}}(x^n)|\leq \exp(nH(V_{Y|X})|P_X)\label{size:condtype}.
\end{align}
Given any $x^n\in\calT_{P_X}^n$, $W_{Y|X}\in\calP(\calY|\calX)$ and $V_{Y|X}\in\calV^n(\calY;P_X)$, for any $y^n\in\calT_{V_{Y|X}}(x^n)$, 
\begin{align}
W_{Y|X}^n(y^n|x^n)=\exp(-n(H(V_{Y|X})+D(V_{Y|X}\|W_{Y|X})|P_X))\label{prob:condseq}.
\end{align}
Thus, it follows from \eqref{size:condtype} and \eqref{prob:condseq} that
\begin{align}
(n+1)^{-|\calX||\calY|}\leq \frac{W_{Y|X}^n(\calT_{V_{Y|X}}(x^n))}{\exp(-nD(V_{Y|X}\|W_{Y|X})|P_X))}&\leq 1
\end{align}

\subsection{Mathematical Tools}
In this section, we present the mathematical tools used to prove second-order asymptotics, which are essentially the generalization of central limit theorems. Let $X^n=(X_1,\ldots,X_n)$ be a collection of $n$ i.i.d. random variables with zero mean and variance $\sigma^2$ and let the normalized sum of these $n$ random variables be
\begin{align}
S_n&:=\frac{1}{n}\sum_{i\in[n]} X_i.
\end{align}

We first recall the weak law of large numbers~\cite{feller1971law}, which states that the normalized sum $S_n$ converges in probability to its mean.
\begin{theorem}[The Weak Law of Large Numbers]
\label{wlln}
For any positive real number $\delta\in\bbR_+$, 
\begin{align}
\lim_{n\to\infty}\Pr\{S_n>\delta\}=0.
\end{align}
\end{theorem}

In the proofs of many theorems, the Markov's inequality is used.
\begin{theorem}[The Markov's Inequality]
\label{markovineq}
For any non-negative real number $\theta\in\bbR_+$ and any positive real number $t$,
\begin{align}
\Pr\{S_n>t\}
&\leq \frac{\bbE[\exp(\theta S_n)]}{\exp(t\theta)}.
\end{align}
\end{theorem}

The Berry-Esseen Theorem for i.i.d. random variables~\cite{berry1941accuracy,esseen1942liapounoff} is critical in deriving second-order asymptotics.
\begin{theorem}[The Berry-Esseen Theorem]
\label{berrytheorem}
Assume that the third absolute moment of $X_1$ is finite, i.e., $T:=\mathbb{E}{|X_1|^3}<\infty$. For each $n\in\mathbb{N}$,
\begin{align}
\sup_{t\in\bbR}\left|\Pr\left\{S_n\geq t\sqrt{\frac{\sigma^2}{n}}\right\}-\rmQ(t)\right|\leq \frac{T}{\sigma^3\sqrt{n}}.
\end{align}
\end{theorem}
The Berry-Esseen theorem states that the probability that the normalized sum $S_n$ deviates from its mean by a sequence which scales as $\Theta\Big(\frac{1}{\sqrt{n}}\Big)$ is well approximated by the same probability for a standard normal variable, with the difference in the order of $O\Big(\frac{1}{\sqrt{n}}\Big)$ that depends on the variance $\sigma^2$ and the third absolute moment $T$. The assumption that $T$ is finite is satisfied by any DMS. It is the mathematical theorem that one applies in the analysis of second-order asymptotics for source and channel coding problems that involve a single encoder.

To tackle certain problems, we need to consider independent but not identically distributed (i.n.i.d.) random variables. Let $X^n=(X_1,\ldots,X_n)$ be a sequence of random variables, where each random variable $X_i$ has zero mean, variance $\sigma_i^2:=\bbE[X_i^2]>0$ and finite third-absolute moment $T_i:=\bbE[|X_i|^3]$. Define the average variance and third-absolute moment as follows:
\begin{align}
\sigma^2&:=\frac{1}{n}\sum_{i\in[n]}\sigma_i^2,\\
T&:=\frac{1}{n}\sum_{i\in[n]}T_i^.
\end{align}
The Berry-Esseen theorem for i.n.i.d. random variables states as follows.
\begin{theorem}
\label{berrytheorem4general}
For each $n\in\mathbb{N}$,
\begin{align}
\sup_{t\in\bbR}\left|\Pr\left\{S_n\geq t\sqrt{\frac{\sigma^2}{n}}\right\}-\rmQ(t)\right|\leq \frac{6T}{\sigma^3\sqrt{n}}.
\end{align}
\end{theorem}

To derive results for multiterminal lossy source coding problems with multiple encoders, we need the following multivariate generalization of the Berry-Esseen theorem~\cite{gotze1991rate}. Given $d\in\bbN$, for each $i\in[n]$, let $\bX_i=(X_{i,1},\ldots,X_{i,k})$ be a $k$-dimensional random vector with zero mean vector and covariance matrix $\bSigma$. Let the normalized sum vector be $\bS_n:=\frac{1}{\sqrt{n}}\bX_i$.
\begin{theorem}[Vector Version of the Berry-Esseen Theorem]
\label{vectorBerry}
Let the third absolute moment of $\bX_1$ be $T:=\mathbb{E}[\|\bX_1\|^3]$. For each $n\in\bbN$, we have
\begin{align}
\sup_{(t_1,\ldots,t_d)\in\bbR^d}\Big|\Pr\{\bS_n\leq \bt\}-\Psi_k(t_1,\ldots,t_k;\bzero_k,\bSigma)\Big|\leq \frac{K(d)T}{\sqrt{n}},
\end{align}
where $>$ refers to elementwise comparison and $K(d)$ is a constant depending on the dimension $d$ only (see~\cite{Ben03,raivc2019multivariate} for explicit bounds).
\end{theorem}


\chapter{Lossless Compression}
\label{chap:lossless}

This chapter focuses on lossless source coding, the notably simplest problem in vector quantization. In his seminal 1948 paper~\cite{shannon1948mathematical}, Shannon showed that the minimal compression rate for reliable lossless source coding is the entropy of the discrete memoryless source, assuming that the blocklength of the source to be compressed tends to infinity. Inspired by the low-latency requirement of practical communications systems, one wonders what the performance degradation is if one operates at a finite blocklength. This question was answered by Yushkevich~\cite{yushkevich1953limit} and by Strassen~\cite{strassen1962asymptotische} who derived the second-order asymptotic approximation to the finite blocklength performance, revived by Hayashi~\cite{hayashi2009information} who rediscovered the result using the information spectrum method and further refined by Kontoyiannis and Verd\'u~\cite{verdu14} and by Chen, Effros and Kostina~\cite{chen2020lossless} who improved the previous bounds. 

In this chapter, we present finite blocklength and second-order asymptotic bounds for lossless source coding, demonstrate the tightness of the second-order asymptotics and discuss the relationship of second-order asymptotics and other refined asymptotic analyses. This chapter is largely based on~\cite{hayashi2009information,strassen1962asymptotische}.

\section{Problem Formulation and Shannon's Result}
Consider any length-$n$ source sequence $X^n$ that is generated i.i.d. from a probability mass function (PMF) $P_X\in\calP(\calX)$. In lossless source coding, one is interested in perfectly recovering the source sequence $X^n$ from its compressed version. Formally, a code is defined as follows.
\begin{definition}
Given any $(n,M)\in\bbN^2$, an $(n,M)$-code for source coding consists of 
\begin{itemize}
\item an encoder $f:\calX^n\to\calM:=[1:M]$,
\item a decoder $\phi:\calM\to\calX^n$.
\end{itemize}
\end{definition}
For simplicity, we use $\hatX^n$ to denote the reproduced source sequence at the decoder, i.e., $\hatX^n=\phi(f(X^n))$. The performance metric for lossless source coding is the error probability, i.e.,
\begin{align}
\rmP_{\rme,n}
&:=\Pr\{\phi(f(X^n))\neq X^n\}\\
&=\Pr\{\hatX^n\neq X^n\}.
\end{align}
In the above definition, $n$ is the blocklength of the source sequence and $M$ is the number of codewords that encoder can use. 

To achieve zero error, $M$ should be chosen such that $M\geq |\calX|^n$ to allow one to one mapping. However, this means no compression is done. Thus, to compress the source, we need to tolerate a non-zero error probability. For efficient compression, one hopes $M$ is as small as possible given any blocklength $n$ and error probability $\rmP_{\rme,n}$. To capture the fundamental limit of lossless source coding, for any $n\in\bbN$, let $M^*(n,\varepsilon)$ be the minimum number of codewords such that there exists an $(n,M)$-code satisfying $\rmP_{\rme,n}\leq \varepsilon$, i.e.,
\begin{align}
M^*(n,\varepsilon):=\inf\big\{M:~\exists\mathrm{~an~}(n,M)\mathrm{-code~s.t.~}\rmP_{\rme,n}\leq \varepsilon\big\}\label{def:M*4lossless}.
\end{align}
Ideally, one would like to exactly characterize $M^*(n,\varepsilon)$ for each finite $n\in\bbN$ and any tolerable error probability $\varepsilon\in(0,1)$. But this is very challenging and information theorists instead derived approximations to $M^*(n,\varepsilon)$.

The most famous such approximation for lossless source coding was provided by Shannon~\cite{shannon1948mathematical}, which states that
\begin{align}
\lim_{\varepsilon\to 0}\lim_{n\to\infty}\frac{1}{n}\log M^*(n,\varepsilon)=H(P_X)\label{shannon:source}.
\end{align}
The above result means that to achieve vanishing error probability with respect to the blocklength $n$, the average minimal number of bits that one should use to compress a source symbol equals the entropy $H(P_X)$ of the source. The above result is also known as the first-order asymptotics since it characterizes the first dominant term in the expansion of the non-asymptotic rate $R(n,\varepsilon):=\frac{1}{n}\log M^*(n,\varepsilon)$ of an optimal code when $\varepsilon\to 0$. In fact, the above result holds for any $\varepsilon\in(0,1)$, which is known as strong converse and implied by second-order asymptotics.

\section{Non-Asymptotic Bounds}
Second-order asymptotics provides approximation to the finite blocklength performance $M^*(n,\varepsilon)$, which demonstrates a deeper understanding for the interplay among the blocklength, the error probability and the coding rate. Usually, to obtain second-order asymptotics, one first derives non-asymptotic achievability and converse bounds for any finite blocklength $n$ and next apply the Berry-Esseen theorem to the derived bounds appropriately. 

In~\cite[Sections 3.1-3.2]{TanBook}, the non-asymptotic and second-order asymptotic bounds by Strassen~\cite{strassen1962asymptotische} were presented and in \cite[Section 3.3]{TanBook}, an alternative proof of second-order asymptotic using the method of types~\cite{csiszar1998mt,csiszar2011information} was given. In this section, we present the non-asymptotic bounds of Han~\cite{han2003information} based on the information spectrum method and provide an alternative proof of second-order asymptotics using Han's results.

For ease of notation, given any $x\in\calX$, define the entropy density $\imath(x|P_X)$ as
\begin{align}
\imath(x|P_X):=-\log P_X(x)\label{def:entropydensity}.
\end{align}
We first recall a finite blocklength achievability bound~\cite[Lemma 1.3.1]{han2003information}.
\begin{theorem}
\label{ach:lossless}
For any $(n,M)\in\bbN^2$, there exists an $(n,M)$-code whose error probability is upper bounded by
\begin{align}
\rmP_{\rme,n}\leq \Pr\Big\{\sum_{i\in[n]}\imath(X_i|P_X)\geq \log M\Big\}\label{ach:pen}.
\end{align}
\end{theorem}
The proof of Theorem \ref{ach:lossless} is simple and elegant. For completeness, we present the proof here.
\begin{proof}
For any $n\in\bbN$, define a set
\begin{align}
\calA_n:=\Big\{x^n\in\calX^n:~\sum_{i\in[n]}\imath(x_i|P_X)<\log M\Big\}\label{def:calAn}.
\end{align}
Note that if $x^n\in\calA^n$, we have
\begin{align}
P_X^n(x^n)
&=\prod_{i\in[n]}P_X(x_i)\\
&=\prod_{i\in[n]}\exp(-\imath(x_i|P_X))\\
&=\exp\big(-\sum_{i\in[n]}\imath(x_i|P_X)\big)\\
&>\frac{1}{M}\label{lower:pxn}.
\end{align}
It follows that
\begin{align}
1&\geq \sum_{x^n\in\calA_n}P_X^n(x^n)\\
&\geq \sum_{x^n\in\calA^n}\frac{1}{M}\\
&=\frac{|\calA_n|}{M}.
\end{align}
Thus, $|\calA_n|\leq M$. Then we can construct an $(n,M)$-code where the encoder $f$ encodes each element of $\calA_n$ to a unique number in $[|\calA_n|]$ and declares an error otherwise. This way, the number of codewords required is $|\calA_n|\leq M$ and the error probability satisfies \eqref{ach:pen}.
\end{proof}

We next recall the finite blocklength converse bound~\cite[Lemma 1.3.2]{han2003information}, which presents a lower bound on the error probability of any $(n,M)$-code.
\begin{theorem}
\label{lossless:converse}
For any $(n,M)\in\bbN^2$ and $\gamma\in\bbR_+$, any $(n,M)$-code satisfies
\begin{align}
\rmP_{\rme,n}&\geq \Pr\Big\{\sum_{i\in[n]}\imath(X_i|P_X)\geq \log M+n\gamma\Big\}-\exp(-n\gamma).
\end{align}
\end{theorem}
The proof of Theorem \ref{lossless:converse} is similar to that of Theorem \ref{ach:lossless} and is also recalled here.
\begin{proof}
Analogously to $\calA^n$ in \eqref{def:calAn}, for any $\gamma\in\bbR$, define a set
\begin{align}
\calB_n(\gamma):=\Big\{x^n\in\calX^n:~\sum_{i\in[n]}\imath(x_i|P_X)\geq \log M+n\gamma\Big\}.
\end{align}
Furthermore, define the set of correctly decoded source sequences as
\begin{align}
\calC_n:=\big\{x^n\in\calX^n:~\phi(f(x^n))=x^n\big\}.
\end{align}
Then,
\begin{align}
\nn&\Pr\{X^n\in\calB_n(\gamma)\}\\*
&=\Pr\{X^n\in(\calB_n(\gamma)\cap\calC_n^\rmc)\}+\Pr\{X^n\in(\calB_n(\gamma)\cap\calC_n)\}\\
&\leq \Pr\{X^n\in\calC_n^\rmc\}+\Pr\{X^n\in(\calB_n(\gamma)\cap\calC_n)\}\\
&\leq \rmP_{\rme,n}+\Pr\{X^n\in(\calB_n(\gamma)\cap\calC_n)\label{epdefine},
\end{align}
where \eqref{epdefine} follows from the definition of the error probability $\rmP_{\rme,n}$. Similarly to \eqref{lower:pxn}, if $x^n\in\calB_n(\gamma)$, 
\begin{align}
P_X(x^n)
&=\exp\Big(-\sum_{i\in[n]}\imath(x_i|P_X)\Big)\\
&\leq \frac{\exp(-n\gamma)}{M}.
\end{align}
It follows that
\begin{align}
\Pr\{X^n\in(\calB_n(\gamma)\cap\calC_n)\}
&=\sum_{x^n\in(\calB_n(\gamma)\cap\calC_n)}P_X^n(x^n)\\
&\leq \sum_{x^n\in(\calB_n(\gamma)\cap\calC_n)}\frac{\exp(-n\gamma)}{M}\\
&\leq \frac{|\calC_n|\exp(-n\gamma)}{M}\\
&\leq \exp(-n\gamma)\label{sizeCn},
\end{align}
where \eqref{sizeCn} follows since for any $(n,M)$-code, the number of corrected decoded source sequences is no greater than $M$.
\end{proof}
In subsequent chapters, the proofs of Theorems \ref{ach:lossless} and \ref{lossless:converse} are generalized to obtain finite blocklength bounds for lossy source coding problems, which are also known as lossy vector quantization~\cite{gray1998tit}. 

\section{Second-Order Asymptotics} 
Applying the Berry-Esseen theorem to the finite blocklength bounds in Theorems \ref{ach:lossless} and \ref{lossless:converse}, one can obtain the second-order asymptotics, which provides a finer characterization of $M^*(n,\varepsilon)$ in \eqref{def:M*4lossless} beyond Shannon's classical first-order asymptotic result. To present the result, define the dispersion of the source $P_X$ as
\begin{align}
\rmV(P_X):=\mathrm{Var}[-\log P_X(X)]\label{def:sourcedispersion}.
\end{align}

\begin{theorem}
\label{second:lossless}
For any $\varepsilon\in(0,1)$,
\begin{align}
\log M^*(n,\varepsilon)
=nH(P_X)+\sqrt{n\rmV(P_X)}\rmQ^{-1}(\varepsilon)+O(\log n).
\end{align}
\end{theorem}
We remark that Theorem \ref{second:lossless} was first obtained by Yushkevich~\cite{yushkevich1953limit} for a Markov source and by Strassen~\cite{strassen1962asymptotische} for DMSes. Hayashi~\cite{hayashi2008source} rediscovered Theorem \ref{second:lossless}. The $O(\log n)$ term was found to be $-\frac{1}{2}\log n+O(1)$ by Kontoyiannis and Verd\'u~\cite{verdu14} and was recently further refined by Chen, Effros and Kostina~\cite[Theorem 5]{chen2020lossless} with explicit lower and upper bounds on the $O(1)$ term. In this monograph, we focus on the second-order asymptotics and further refined analyses for the remainder term as in~\cite{verdu14,chen2020lossless} are worthwhile future research directions but challenging for lossy source coding problems to be discussed in the result of this monograph.

Furthermore, the achievability part of Theorem \ref{second:lossless} can also be proved using the method of types~\cite[Chapter 11]{cover2012elements}, as demonstrated in~\cite[Chapter 3.3]{TanBook}. The achievability proof of second-order asymptotics based on the method of types finds applications in many other problems, including the point-to-point and multiterminal settings of lossy source coding problems to be discussed in this monograph.

To illustrate the tightness of the second-order asymptotic bound in Theorem \ref{second:lossless}, in Figure \ref{illus:tightsecond}, we plot the second-order asymptotic approximation in Theorem \ref{second:lossless} and compare the approximation with finite blocklength bounds in Theorems \ref{ach:lossless} and \ref{lossless:converse} for a Bernoulli source with parameter $0.2$ with the target error probability of $\varepsilon=0.01$. As observed from Figure \ref{illus:tightsecond}, for $n$ moderately large, the second-order asymptotic bound provides rather tight approximation to the finite blocklength performance. Furthermore, the gap between the second-order asymptotic result and the first-order asymptotic result of Shannon is significant unless $n\to\infty$.
\begin{figure}[t]
\centering
\includegraphics[width=.8\columnwidth]{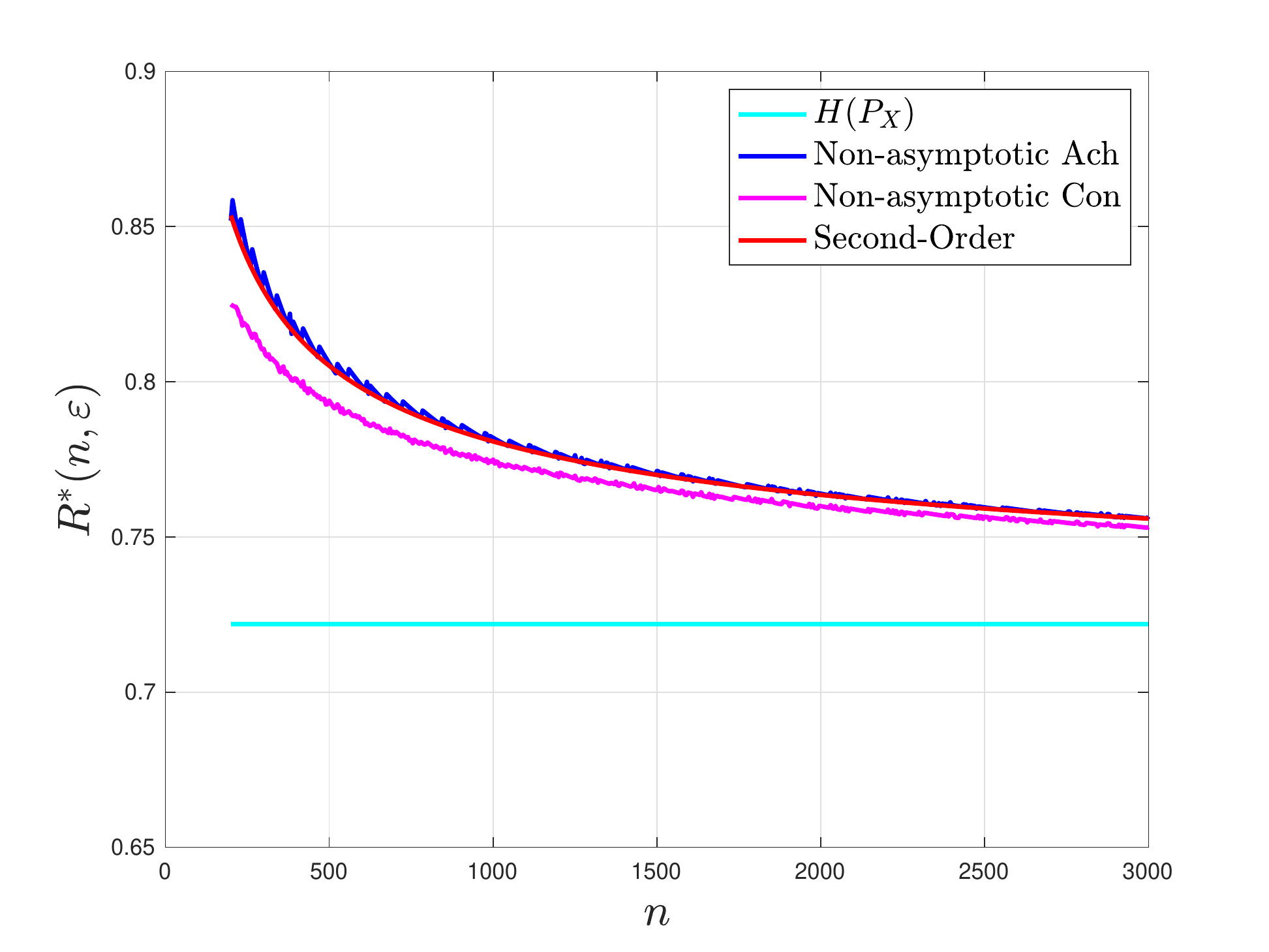}
\caption{Plots of the second-order asymptotic bound in Theorem \ref{second:lossless} and non-asymptotic bounds in Theorems \ref{ach:lossless} and \ref{lossless:converse} for a Bernoulli source with parameter $0.2$ and a target error probability of $\varepsilon=0.01$.}
\label{illus:tightsecond}
\end{figure}

Note that Theorem \ref{second:lossless} is known as the second-order asymptotic result because it characterizes the second dominant term in the expansion of $\log M^*(n,\varepsilon)$. An equivalent presentation of Theorem \ref{second:lossless} is to characterize the so called second-order coding rate coined by Hayashi~\cite{hayashi2008source,hayashi2009information}. For lossless source coding, the second-order coding rate is defined as follows.
\begin{definition}
Given any $\varepsilon\in[0,1)$, a real number $L\in\bbR$ is said to be a second-order achievable rate if there exists a sequence of $(n,M)$-codes such that
\begin{align}
\limsup_{n\to\infty}\frac{1}{\sqrt{n}}(\log M-nH(P_X))&\leq L,\\
\limsup_{n\to\infty}\rmP_{\rme,n}&\leq \varepsilon.
\end{align}
For any $\varepsilon\in[0,1)$, the infimum of all second-order achievable rates is called the optimal second-order coding rate and denoted by $L^*(\varepsilon)$.
\end{definition}
We remark that $L^*(\varepsilon)$ has the unit of nats per square root number of source symbols. With this definition, Theorem \ref{second:lossless} is equivalent to the following statement.
\begin{theorem}
\label{second:lossless2}
For any $\varepsilon\in[0,1)$, the optimal second-order rate coding is
\begin{align}
L^*(\varepsilon)=\sqrt{\rmV(P_X)}\rmQ^{-1}(\varepsilon).
\end{align}
\end{theorem}
In second-order asymptotics, by allowing a non-vanishing error probability $\varepsilon\in(0,1)$, we observe that the backoff of the non-asymptotic coding rate $R^*(n,\varepsilon):=\frac{1}{n}\log M^*(n,\varepsilon)$ from Shannon's asymptotic rate $H(P_X)$ is in the order of $\Theta\Big(\frac{1}{\sqrt{n}}\Big)$ with the coefficient determined by a function of the tolerable error probability and the source dispersion.

\section{Proof of Second-Order Asymptotics}
We next present the proof of Theorem \ref{second:lossless} by illustrating how one can apply the Berry-Esseen theorem (cf. Theorem \ref{berrytheorem}) to the non-asymptotic bounds in Theorems \ref{ach:lossless} and \ref{lossless:converse},

For the smooth presentation of the proof steps, let
\begin{align}
T(P_X):=\bbE[|\imath(X|P_X)-H(P_X)|^3].
\end{align}
\subsection{Achievability}
Given any $\varepsilon\in(0,1)$, let
\begin{align}
\varepsilon_n&=\varepsilon-\frac{T(P_X)}{(\rmV(P_X))^3\sqrt{n}},\\
\log M&=n H(P_X)+\sqrt{n\rmV}\rmQ^{-1}(\varepsilon_n).
\end{align}
It follows from Theorem \ref{ach:lossless} that the error probability of the code satisfies
\begin{align}
\rmP_{\rme,n}
&\leq \Pr\Big\{\sum_{i\in[n]}(\imath(X_i|P_X)-H(P_X))\geq \sqrt{n\rmV}\rmQ^{-1}(\varepsilon)\Big\}\\
&=\Pr\bigg\{\frac{1}{n}\sum_{i\in[n]}(\imath(X_i|P_X)-H(P_X))\geq \sqrt{\frac{\rmV}{n}}\rmQ^{-1}(\varepsilon)\bigg\}\\
&\leq \varepsilon_n+\frac{T(P_X)}{(\rmV(P_X))^3\sqrt{n}}\label{ach:useBE}\\
&=\varepsilon,
\end{align} 
where \eqref{ach:useBE} follows from the Berry-Esseen theorem for i.i.d. random variables in Theorem \ref{berrytheorem} since the random variables $\{\imath(X_i|P_X)-H(P_X)\}_{i\in[n]}$ are a sequence of i.i.d. random variables with mean $0$ and the identical variance $\rmV(P_X)$.

Thus, using the Taylor expansion of $\rmQ^{-1}(\varepsilon_n)$ around $\varepsilon$ that states $\rmQ^{-1}(\varepsilon_n)=
\rmQ^{-1}(\varepsilon)+O(\varepsilon-\varepsilon_n)$, we have
\begin{align}
\log M^*(n,\varepsilon)\leq n H(P_X)+\sqrt{n\rmV}\rmQ^{-1}(\varepsilon)+O(1).
\end{align}

\subsection{Converse}
For any $\varepsilon\in(0,1)$, let
\begin{align}
\varepsilon_n'&=\varepsilon+\frac{T(P_X)}{(\rmV(P_X))^3\sqrt{n}}+\frac{1}{n},\\
\log M&=n H(P_X)+\sqrt{n\rmV}\rmQ^{-1}(\varepsilon_n')-\log n.
\end{align}

Invoking Theorem \ref{lossless:converse} with $\gamma=\log n$ and using the Berry-Esseen theorem, the error probability of any $(n,M)$-code satisfies
\begin{align}
\rmP_{\rme,n}
&\geq \Pr\Big\{\sum_{i\in[n]}(\imath(X_i|P_X)-H(P_X))\geq \sqrt{n\rmV(P_X)}\rmQ^{-1}(\varepsilon_n')\Big\}-\frac{1}{n}\\
&\geq \varepsilon\label{useberry}.
\end{align}

Therefore,
\begin{align}
\log M^*(n,\varepsilon)
&\geq nH(P_X)+\sqrt{n\rmV}\rmQ^{-1}(\varepsilon_n')-\frac{1}{2}\log n+O(1)\\
&=nH(P_X)+\sqrt{n\rmV}\rmQ^{-1}(\varepsilon)+O(\log n).
\end{align}
The converse proof is now completed.

\section{Other Refined Asymptotic Analyses}
Besides second-order asymptotics, there are also other refined asymptotic analyses beyond Shannon's source coding theorem. Two examples are the large and moderate deviations analyses.

In large deviations, one characterizes the decay rate of the error probability $\rmP_{\rme,n}$ for any asymptotic rate greater than $H(P_X)$.
\begin{definition}
A non-negative number $E$ is said to be a rate-$R$ achievable error exponent if there exists a sequence of $(n,M)$-codes such that 
\begin{align}
\limsup_{n\to\infty}\frac{1}{n}\log M&\leq R,\\
\liminf_{n\to\infty}-\frac{1}{n}\log\rmP_{\rme,n}&\geq E.
\end{align}
The supremum of all rate-$R$ achievable error exponents is called the optimal error exponent and denoted by $E^*(R)$.
\end{definition}

The exact characterization of $E^*(R)$ was given by Gallager~\cite{gallager} and by Csisz\'ar and Longo~\cite{csiszar1971error}.
\begin{theorem}
\label{ee:lossless}
The optimal error exponent for the  lossless source coding problem is
\begin{align}
E^*(R)
&=\max_{\rho\geq 0} \bigg(\rho R-(1+\rho)\log\Big(\sum_{x\in\mathrm{supp}(P_X)} P_X^{\frac{1}{1+\rho}}(x)\Big)\bigg)\label{gallager}\\*
&=\min_{Q_X:H(Q_X)\geq R} D(Q_X\|P_X)\label{csislongo}.
\end{align}
\end{theorem}
As a result of Theorem \ref{ee:lossless}, we conclude that the error probability decays exponentially fast for any rate above the first-order coding rate, i.e., $R>H(P_X)$. The characterization in \eqref{gallager} was proved by Gallager using the maximum likelihood decoding with the $\rho$ trick and the characterization in \eqref{csislongo} was proved by Csisz\'ar and Longo~\cite{csiszar1971error} using the method of types. The equivalence of the two characterizations is hinted in~\cite[Problem 2.14]{csiszar2011information}.

The moderate deviations regime interpolates between the large deviations and second-order asymptotic regimes. In this regime, one is interested in a sequence of $(n,M)$-codes whose rates approach $H(P_X)$ and whose error probabilities decay to zero simultaneously.
\begin{definition}
Consider any sequence $\{\xi_n\}_{n\in\bbN}$ such that $\xi_n\to 0$ and $\sqrt{n\xi_n}\to\infty$ as $n\to\infty$. A non-negative number $\nu$ is said to be an achievable moderate deviations constant if there exists a sequence of $(n,M)$-codes such that
\begin{align}
\limsup_{n\to\infty}\frac{\log M-nH(P_X)}{n\xi_n}\leq 1,\\
\liminf_{n\to\infty}-\frac{1}{n\xi_n^2}\log \rmP_{\rme,n}\geq \nu.
\end{align}
The supremum of all moderate deviations constants is called the optimal moderate deviations constant and is denoted by $\nu^*$. 
\end{definition}
Note that in moderate deviations, the speed of the rate approaching $H(P_X)$ is in the order of $\xi_n$, which is slower than $O(\frac{1}{\sqrt{n}})$ in second-order asymptotics and the decay rate of the error probability is subexponential, which is slower than the exponential decay in large deviations. This is precisely the reason why moderate deviations is said to interpolate second-order and large deviations asymptotics.

The optimal moderate deviations constant for the  lossless source coding problem was obtained by Alt\u{u}g, Wagner and Kontoyiannis in \cite{altug2013lossless}. 
\begin{theorem}
\label{mdt:lossless}
The optimal moderate deviations constant is
\begin{align}
\nu^*=\frac{1}{2\rmV(P_X)}.
\end{align}
\end{theorem}
Theorem \ref{mdt:lossless} states that the sequence of optimal codes approaches $H(P_X)$ at the speed of $\xi_n$ with the error probability decaying subexponentially fast, which can be proved by applying the moderate deviations theorem~\cite[3.7.1]{dembo2009large} to the non-asymptotic bounds in Theorems \ref{ach:lossless} and \ref{lossless:converse}.

To illustrate the relationship between second-order, large and moderate deviations to the non-asymptotic bounds, we plot the relationship between the error probability and coding rate for different blocklengths for a binary memoryless source distributed according to a Bernoulli distribution with parameter $0.3$ in Figure \ref{illus:refined4lossless}, using the second-order asymptotic bound in Theorem \ref{second:lossless} as the approximation. Note that both large and moderate deviations theorems are tight for sufficiently large blocklength and thus violate the low-latency requirement of practical communication systems. In this monograph, for all lossy source coding problems to be covered, we focus on the second-order asymptotics that provide good approximations to the performance of optimal codes at finite blocklengths (cf.~\cite[Fig. 1]{chen2020lossless}), and we also present non-asymptotic bounds from which the second-order asymptotics are derived.
\begin{figure}[t]
\centering
\includegraphics[width=.8\columnwidth]{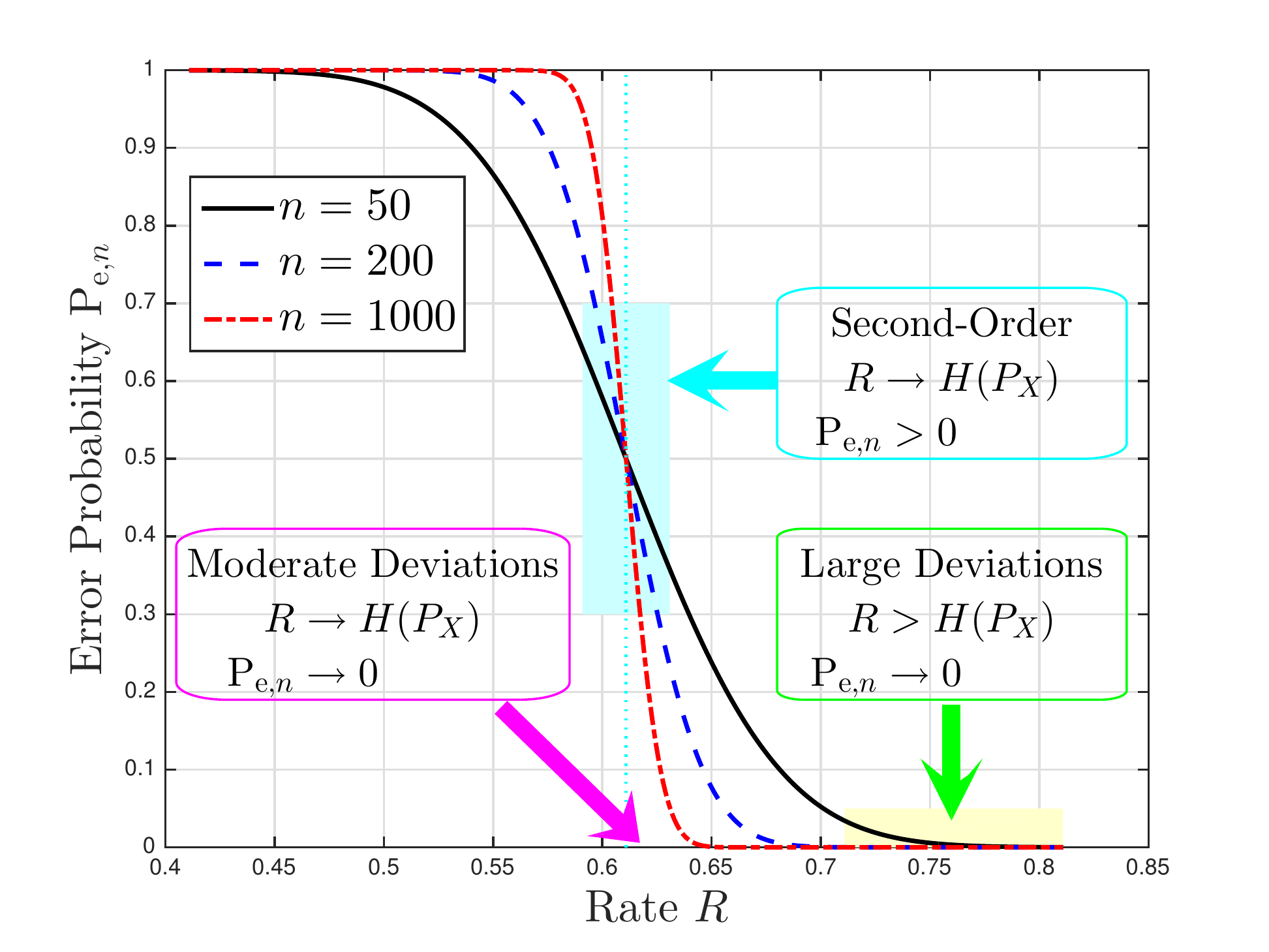}
\caption{Illustration of refined asymptotics for lossless source coding of a binary memoryless source. Note that both large and moderate deviations asymptotics provide tight characterization when $n$ is sufficiently large and violates the low latency requirement of practical communication systems. In contrast, second-order asymptotics provides approximations to the performance of optimal codes with finite blocklength.}
\label{illus:refined4lossless}
\end{figure}

\part{Point-to-Point Setting}

\chapter{Rate Distortion}
\label{chap:rd}

In this chapter, we study the rate-distortion problem of lossy source coding and present non-asymptotic and second-order asymptotic bounds for optimal codes~\cite{shannon1959coding,kostina2012fixed,ingber2011dispersion}. The rate-distortion problem has several motivations. Firstly, it is impossible to compress a continuous memoryless source and recover it losslessly with any finite rate. This is because an infinite number of bits is needed to perfectly represent a real number. Secondly, in image and video compression, imperfection is usually tolerable. For example, a $720$P video can convey the same episodes as a $1080$P or $2$K video and cannot be easily distinguished on a phone or tablet. Thirdly, in rate limited scenarios, a smaller compression rate is preferred and lossy data compression achieves rates smaller than the lossless counterpart. 

Shannon~\cite{shannon1959coding} proposed the system model of the rate-distortion problem and derived the first-order asymptotic optimal rate to ensure reliable compression in a lossy manner as the blocklength tends to infinity. A distortion measure is introduced to evaluate the difference of the source sequence and its reproduced version. Reliable lossy data compression is achieved if the distortion between the source $X^n$ and its reproduced version $\hatX^n$ is smaller than a tolerable distortion level $D$. For example, one can think of the source $X^n$ as a high quality $2$K video and set the distortion level $D$ so that the reproduced version $\hatX^n$ is acceptable as long as it is at least a $720$P video.

Shannon's asymptotic results were refined by Ingber and Kochman~\cite{ingber2011dispersion} and by Kostina and Verd\'u~\cite{kostina2012fixed} independently, where both papers defined the distortion-tilted information density that generalizes the entropy density and derived second-order asymptotics. Furthermore, Kostina and Verd\'u derived non-asymptotic achievability and converse bounds. This chapter is largely based on \cite{ingber2011dispersion,kostina2012fixed}.

\section{Problem Formulation and Shannon's Result}

\subsection{Problem Formulation}
Consider a memoryless source $X^n$ generated i.i.d. from a distribution $P_X$ defined on an alphabet $\calX$. Let $\hatcalX$ be the reproduced alphabet and let the distortion function be $d:\calX\times\hatcalX\to \bbR_+$. Given any two sequences $(x^n,\hatx^n)\in\calX^n\times\hatcalX^n$, the distortion function $d(x^n,\hatx^n)$ is assumed additive and defined as the average symbolwise distortion, i.e., $d(x^n,\hatx^n)=\frac{1}{n}\sum_{i\in[n]}d(x_i,\hatx_i)$. 

Some examples of the distortion functions are as follows.
\begin{definition}
A distortion function $d$ is said to be the Hamming distortion measure if $\calX=\hatcalX=\{0,1\}$ and for any $(x,\hatx)\in\calX\times\hatcalX$,
\begin{align}
d(x,\hatx)=\bbo\{x\neq \hatx\}\label{def:hamming}.
\end{align}
\end{definition}

\begin{definition}
A distortion function $d$ is said to be the quadratic distortion measure if $\calX=\hatcalX=\calR$ and for any $(x,\hatx)\in\calX\times\hatcalX$,
\begin{align}
d(x,\hatx)=(x-\hatx)^2\label{def:quadratic}.
\end{align}
\end{definition}

A code for the rate-distortion problem is defined as follows. 
\begin{definition}
\label{def:code4rd}
Given any $(n,M)\in\bbN^2$, an $(n,M)$-code for the rate-distortion problem consists of 
\begin{itemize}
\item an encoder $f:\calX^n\to\calM=[M]$,
\item a decoder $\phi:\calM\to\calX^n$.
\end{itemize}
\end{definition}
Let $\hatX^n$ denote the reproduced source sequence, i.e., $\hatX^n=\phi(f(X^n))$. Throughout the chapter, let $D\in\bbR_+$ be the target distortion level. The performance metric for the rate-distortion problem that we consider is the excess-distortion probability with respect to $D\in\bbR_+$, i.e.,
\begin{align}
\rmP_{\rme,n}(D)
&:=\Pr\{d(X^n,\hatX^n)>D\}.
\end{align}
Given any blocklength $n\in\bbN$, the distortion level $D$ and tolerable excess-distortion probability $\varepsilon$, let $M^*(n,D,\varepsilon)$ denote the minimum number $M$ such that one can construct an $(n,M)$-code with excess-distortion probability $\rmP_{\rme,n}(D)$ no greater than $\varepsilon$, i.e.,
\begin{align}
M^*(n,D,\varepsilon):=\inf\big\{M:~\exists\mathrm{~an~}(n,M)\mathrm{-code~s.t.~}\rmP_{\rme,n}(D)\leq \varepsilon\big\}\label{def:M*}.
\end{align}
In this chapter, we present non-asymptotic and asymptotic bounds on $M^*(n,D,\varepsilon)$. 

\subsection{Shannon's First-Order Asymptotic Result}
In this subsection, we recall Shannon's characterization of the first-order asymptotic coding rate, which is defined as follows.
\begin{definition}
A rate $R\in\bbR_+$ is said to be achievable for the rate-distortion problem with respect to distortion level $D$ if there exists a sequence of $(n,M)$-codes such that
\begin{align}
\liminf_{n\to\infty}\frac{1}{n}\log M&\geq R,\\
\limsup_{n\to\infty}\bbE[d(X^n,\hatX^n)]&\leq D\label{avgdistortion}.
\end{align}
The minimal achievable rate is denoted as $R^*(D)$.
\end{definition}

Shannon~\cite{shannon1959coding} proved the following theorem.
\begin{theorem}
\label{lossy:Shannon}
The minimal achievable rate $R^*(D)$ satisfies
\begin{align}
R^*(D)=\min_{P_{\hatX|X}:~\bbE[d(X,\hatX)]\leq D}I(X;\hatX)=:R(P_X,D)\label{def:rd}.
\end{align}
\end{theorem}
Note that $R(P_X,D)$ is known as the rate-distortion function. Although Shannon's coding theorem is derived for the average distortion criterion in \eqref{avgdistortion}, for bounded distortion measure where $\bard:=\max_{x,\hatx}d(x,\hatx)<\infty$, the same result holds also when \eqref{avgdistortion} is replaced by the vanishing excess-distortion probability criterion\footnote{The comment holds for any lossy source coding problem.}, i.e., \\$\lim_{n\to\infty}\rmP_{\rme,n}(D)=0$. Specifically, using the notation $M^*(n,D,\varepsilon)$, it follows that
\begin{align}
\lim_{\varepsilon\to 0}\lim_{n\to\infty}\frac{1}{n}\log M^*(n,D,\varepsilon)=R(P_X,D)\label{shanon:excessp}.
\end{align}
In other words, Shannon characterized the asymptotic minimal compression rate of any code for the rate-distortion problem as the blocklength tends to infinity to ensure that the joint excess-distortion probability with respect to $D$ vanishes or to ensure that that the average distortion between the source sequence $X^n$ and the reproduced version $\hatX^n$ no greater than $D$. In fact, \eqref{shanon:excessp} holds for any $\varepsilon\in(0,1)$~\cite[Theorem 7.3]{csiszar2011information}.

We then explain why \eqref{shanon:excessp} holds using Theorem \ref{lossy:Shannon} for bounded distortion measures with maximal distortion $\bard$.  Suppose that there exists a sequence of $(n,M)$-codes such that
\begin{align}
\liminf_{n\to\infty}\frac{1}{n}\log M&\geq R,\\
\lim_{n\to\infty}\rmP_{\rme,n}(D)&=0\label{vanishexcessp}.
\end{align}
Since
\begin{align}
\nn&\bbE[d(X^n,\hatX^n)]\\*
&=\Pr\{d(X^n,\hatX^n)\leq D\}D+\Pr\{d(X^n,\hatX^n)> D\}\bar d\\
&\leq D+\rmP_{\rme,n}(D)\bard,
\end{align}
it follows from \eqref{vanishexcessp} that
\begin{align}
\limsup_{n\to\infty}\bbE[d(X^n,\hatX^n)]&\leq D.
\end{align}
Thus, any rate $R$ that ensures vanishing excess-distortion probability $\rmP_{\rme,n}(D)\to 0$ also ensures that the average distortion criterion~\eqref{avgdistortion} is satisfied, which leads to $\lim_{\varepsilon\to 0}\lim_{n\to\infty}\frac{1}{n}\log M^*(n,D,\varepsilon)\leq R^*(D)$. On the other hand, if a rate $R$ is not achievable under the average distortion criterion, i.e.,
\begin{align}
\limsup_{n\to\infty}\bbE[d(X^n,\hatX^n)]&>D,
\end{align}
it follows from the weak law of large numbers (cf. Theorem \ref{wlln}) that
\begin{align}
\lim_{n\to\infty}\rmP_{\rme,n}(D)=1.
\end{align}
Thus, the rate $R$ is also not achievable under the excess-distortion probability criterion, which implies that  $\lim_{n\to\infty}\frac{1}{n}\log M^*(n,D,\varepsilon)\geq R^*(D)$. The justification is thus completed.

\section{Distortion-Tilted Information Density}
We next introduce the definition and present properties of the distortion-tilted information density that generalizes the entropy density in lossless source coding. Of particular interest is that the distortion-titled information density is closely related to the rate-distortion function $R(P_X,D)$ and it is essential in characterizing the second-order asymptotics for the rate-distortion problem. 

\subsection{Definition and An Example}
\label{sec:dtilted4rd}
Consider any source distribution $P_X$, distortion measure $d(\cdot)$ and distortion level $D$ such that i) $R(P_X,D)$ is finite and ii) $(Q_X,D')\rightarrow R(Q_X,D')$ is twice differentiable in the neighborhood of $(P_X,D)$ and the derivatives are bounded. Note that $R(P_X,D)$ in \eqref{def:rd} is the optimal value of a convex optimization problem. Assume that the conditional distribution $P_{\hatX|X}^*$ achieves $R(P_X,D)$. Let $P_{\hatX}^*$ be induced by the source distribution $P_X$ and $P_{\hatX|X}^*$. Furthermore, let $\lambda^*$ be the first derivative of $R(P_X,D')$ with respect to $D'$ at $D'=D$, i.e.,
\begin{align}
\lambda^*=-\frac{\partial R(P_X,D')}{\partial D'}|_{D'=D}\label{def:lambda^*}.
\end{align}
Note that $\lambda^*$ is well defined due to the above two assumptions and $\lambda^*\geq 0$ since $R(P_X,D)$ is non-increasing in $D$.

The distortion-tilted information density is then defined as follows.
\begin{definition}
\label{def:dtilted}
For any $x\in\calX$, the $D$-tilted information density is defined as
\begin{align}
\jmath(x|D,P_X):=-\log\bbE_{P_{\hatX}^*}[\exp(\lambda^*D-\lambda^*d(x,\hatX))]\label{def:dtilteddensity}.
\end{align}
\end{definition}
Definition \ref{def:dtilted} first appeared in \cite[Proposition 7]{ingber2011dispersion} for discrete memoryless sources and was generalized to arbitrary memoryless sources in~\cite[Definition 6]{kostina2012fixed}. One might find the definition of $\jmath(x,D)$ difficult to understand. To illustrate, an example is given for a binary memoryless source with distribution $P_X=\mathrm{Bern}(p)$ with $p<0.5$ under the Hamming distortion measure. Note that $P_X(1)=p$ and $P_X(0)=1-p$. It follows from \cite[Theorem 10.3.1]{cover2012elements} that the rate-distortion function for this case is
\begin{align}
R(P_X,D)
&=\left\{
\begin{array}{ll}
H_\rmb(p)-H_\rmb(D)&\mathrm{if}D< p,\\
0&\mathrm{otherwise},
\end{array}
\right.
\end{align}
where $H_\rmb(p)=H(\mathrm{Bern}(p))=-p\log p-(1-p)\log(1-p)$ denotes the binary entropy function. Furthermore, the induced marginal distribution $P_{\hatX^*}=\mathrm{Bern}(\frac{p-D}{1-2D})$ when $D<p$, i.e., $P_{\hatX^*}(1)=\frac{p-D}{1-2D}$ and $P_{\hatX^*}(0)=\frac{1-p-D}{1-2D}$. Thus, when $D\geq p$, $\lambda^*=0$ and $\jmath(x|D,P_X)=-\log 1=0$. We then consider the non-degenerate case of $D<p$. The derivative $\lambda^*$ satisfies
\begin{align}
\lambda^*=\frac{\partial H_\rmb(D)}{\partial D}=\log\frac{1-D}{D}.
\end{align}
It follows that
\begin{align}
\nn&\bbE_{P_{\hatX}^*}[\exp(\lambda^*D-\lambda^*d(0,\hatX))]\\*
&=\exp(\lambda^*D)\Big(P_{\hatX}^*(0)+P_{\hatX}^*(1)\exp(-\lambda^*)\Big)\\
&=\exp(\lambda^*D)\bigg(\frac{1-p-D}{1-2D}+\frac{(p-D)}{1-2D}\frac{D}{1-D}\bigg)\\
&=\exp(\lambda^*D)\frac{1-p}{1-D}.
\end{align}
Thus,
\begin{align}
\jmath(0|D,P_X)
&=-\log\bbE_{P_{\hatX}^*}[\exp(\lambda^*D-\lambda^*d(0,\hatX))]\\
&=-\lambda^*D-\log\frac{1-p}{1-D}\\
&=-D\log\frac{1-D}{D}-\log (1-p)+\log (1-D)\\
&=-\log (1-p)-H_\rmb(D).
\end{align}
Similarly, 
\begin{align}
\jmath(1|D,P_X)
=-\log p-H_\rmb(D).
\end{align}

Using the definition of the entropy density $\imath_X(\cdot)$ in \eqref{def:entropydensity}, for a Bernoulli memoryless source with parameter $p<D$, under the Hamming distortion measure, the $D$-tilted information density satisfies
\begin{align}
\jmath(x|D,P_X)=\imath(x|P_X)-H_\rmb(D)\label{j4bms}.
\end{align}

\subsection{Properties}
The distortion-tilted information density possess several interesting properties that connect it to the rate-distortion function and also pave the way for the proof of non-asymptotic converse bound. To present the properties of the distortion-tilted information density, we need the following definition of the mutual information density
\begin{align}
\imath_{X;\hatX^*}(x;\hatx):=\log\frac{P_{\hatX|X}^*(\hatx|x)}{P_{\hatX}^*(\hatx)}.
\end{align}

\begin{lemma}
\label{prop:dtilteddensity}
The following claims hold.
\begin{enumerate}
\item For $\hatx\in\mathrm{supp}(P_{\hatX^*}(x))$.
\begin{align}
\jmath(x|D,P_X)=\imath_{X;\hatX^*}(x;\hatx)+\lambda^* d(x,\hatx)-\lambda^*D.
\end{align}
\item The rate-distortion function is the expectation of the distortion-tilted information density, i.e.,
\begin{align}
R(P_X,D)=\bbE_{P_X}[\jmath(X|D,P_X)].
\end{align}
\item For any $\hatx\in\calX$,
\begin{align}
\bbE_{P_X}[\exp(\lambda^* D-\lambda^*d(X,\hatx)+\jmath(X|D,P_X))]&\leq 1,
\end{align}
where the inequality holds for $\hatx\in\mathrm{supp}(P_{\hatX^*}(x))$.
\item Suppose that for all $Q_X$ in some neighborhood of $P_X$, $\mathrm{supp}(Q_{\hatX}^*)\subset\mathrm{supp}(P_{\hatX}^*)$. For any $x\in\mathrm{supp}(P_X)$,
\begin{align}
\frac{\partial R(Q_X,D)}{\partial Q_X(x)}\bigg|_{Q_X=P_X}
&=\jmath(x|D,P_X)-1,
\end{align}
where $Q_{\hatX|X}^*$ is the optimal conditional distribution that achieves $R(Q_X,D)$, $\Gamma(Q_X)$ is the distribution on $\calX$ that orders elements of $Q_X$ in a decreasing order and $x_i\in\calX$ is the element that has $i$-th largest probability under the distribution $Q_X$.
\end{enumerate}
\end{lemma}
Lemma \ref{prop:dtilteddensity} was derived by Csisz\'ar~\cite{csiszar1974} and is also available in \cite[Chapter 2]{kostina2013lossy}. Claim (i) provides an alternative expression for the distortion-tilted information density~\cite{boyd2004convex}.  Claim (ii) shows that the distortion-tilted information density shares the property similar to the entropy density $\imath_X(\cdot)$ and the entropy function $H(P_X)$ and it is the reason why $\jmath(x|D,P_X)$ is named an information density. Claim (iii) is critical in deriving a non-asymptotic converse bound for the rate-distortion problem. Claim (iv) enables the proof of second-order asymptotics using the method of types, specifically Taylor expansions of the rate-distortion function of empirical distributions around the source distribution $P_X$ (cf. \eqref{useprops}).

\section{Non-Asymptotic Bounds}
In this section, we present non-asymptotic achievability and converse bounds for the rate-distortion problem~\cite{kostina2012fixed}.

For any $n\in\bbN$ and any $x^n\in\calX^n$, define the distortion ball $\calB_D(x^n)$ as
\begin{align}
\calB_D(x^n):=\{\hatx^n\in\hatcalX^n:~d(x^n,\hatx^n)\leq D\}\label{def:distortionball}.
\end{align}
The following achievability holds.
\begin{theorem}
\label{ach:fbl}
For any $P_{\hatX}$, there exists an $(n,M)$-code such that the excess-distortion probability satisfies
\begin{align}
\rmP_{\rme,n}(D)\leq \bbE_{P_X^n}[(1-P_{\hatX}^n(\calB_D(X^n)))^M].
\end{align}
\end{theorem}
\begin{proof}
The proof of Theorem \ref{ach:fbl} follows from the random coding idea and the minimum distance encoding. Specifically, let $\hatX^n(1),\ldots,\hatX^n(M)$ be a sequence of independent codewords, each generated i.i.d. from a distribution $P_{\hatX}$ define on the alphabet $\calX$. Consider the following $(n,M)$-code with encoder $f:\calX^n\to[M]$ such that
\begin{align}
f(X^n)=\argmin_{i\in[M]}d(X^n,\hatX^n(i)),
\end{align}
and decoder $\phi:[M]\to\hatcalX^n$ such that $\hatX^n=\phi(f(X^n))=\hatX^n(f(X^n))$.

The excess-distortion probability of the above code satisfies
\begin{align}
\rmP_{\rme,n}(D)
&=\Pr\{d(X^n,\hatX^n)>D\}\\
&=\Pr\{\forall~i\in[M]:~d(X^n,\hatX^n)>D\}\\
&=\bbE_{P_X^n}\bigg[\prod_{i\in[M]}\Pr_{P_{\hatX}^n}\{d(X^n,\hatX^n(i))>D|X^n\}\bigg]\label{useindependent}\\
&=\bbE_{P_X^n}\Big[(\Pr_{P_{\hatX}^n}\{d(X^n,\hatX^n(i))>D|X^n\})^M\Big]\label{identicaldist}\\
&=\bbE_{P_X^n}\Big[(1-P_{\hatX}^n(\calB_D(X^n)))^M\Big]\label{usedefball}
\end{align}
where \eqref{useindependent} since all codewords are generated independently of each other, \eqref{identicaldist} follows since each codeword is generated from the same distribution $P_{\hatX^n}$ and \eqref{usedefball} follows from the definition of the distortion ball that implies $\Pr_{P_{\hatX}^n}\{d(x^n,\hatX^n(i))>D\}=1-P_{\hatX}^n(\calB_D(x^n))$ for any $x^n\in\calX^n$. The existence of a deterministic code follows from the simple fact that $\bbE[X]<a$ implies that there exists an element $x\in\calX$ such that $x\leq a$ for any random variable $X$ with alphabet $\calX$ and for any real number $a$.
\end{proof}

Conversely, the excess-distortion probability of any $(n,M)$-code is lower bounded as follows.
\begin{theorem}
\label{converse:fbl}
Given any $\gamma\in\bbR_+$, any $(n,M)$-code satisfies that
\begin{align}
\rmP_{\rme,n}(D)
&\geq \Pr_{P_X^n}\Big\{\sum_{i\in[n]}\jmath(X_i|D,P_X)\geq \log M+n\gamma\Big\}-\exp(-n\gamma).
\end{align}
\end{theorem}
Note that Theorem \ref{converse:fbl} generalizes non-asymptotic converse bound in Theorem \ref{lossless:converse} to the lossy setting with the entropy density $\imath_X(\cdot)$ replaced by the distortion-tilted information density $\jmath(\cdot)$. The proof of Theorem \ref{converse:fbl} requires novel ideas beyond Theorem \ref{lossless:converse}, such as Claim (iii) of Lemma \ref{prop:dtilteddensity}.
\begin{proof}
Let $P_{W|X^n}$ and $P_{\hatX^n|W}$ be the conditional distributions induced by the encoder $f$ and the decoder $\phi$ respectively with $W$ being a random variable taking values in $[M]$. Specifically, for each $w\in[M]$, $x^n\in\calX^n$ and $\hatx^n\in\hatcalX^n$, $P_{W|X^n}(w|x^n)=\bbo(w=f(X^n))$ and $P_{\hatX^n|W}(\hatx^n|w)=\bbo(\hatx^n=\phi(w))$. Furthermore, let $Q_W$ be the uniform distribution over $[M]$ and let $Q_{\hatX^n}$ be induced by $Q_W$ and $P_{\hatX^n|W}$. In subsequent analyses, for simplicity, we drop the subscript of the probability terms when it is clear. For any $\gamma\in\bbR_+$, it follows that
\begin{align}
\nn&\Pr\Big\{\sum_{i\in[n]}\jmath(X_i|D,P_X)\geq \log M+n\gamma\Big\}\\*
&=\Pr\Big\{\sum_{i\in[n]}\jmath(X_i|D,P_X)\geq \log M+n\gamma,~d(X^n,\hatX^n)\leq D\Big\}\\*
&\qquad+\Pr\Big\{\sum_{i\in[n]}\jmath(X_i|D,P_X)\geq \log M+n\gamma,~d(X^n,\hatX^n)> D\Big\}\\
\nn&\leq \Pr\Big\{\sum_{i\in[n]}\jmath(X_i|D,P_X)\geq \log M+n\gamma,~d(X^n,\hatX^n)\leq D\Big\}\\*
&\qquad +\rmP_{\rme,n}(D)\label{converse:step1}.
\end{align}
Denote  the second term in \eqref{converse:step1} as $f(\gamma,D)$, which can be further upper bounded as follows:
\begin{align}
\nn&f(\gamma,D)\\*
\nn&=\sum_{x^n\in\calX^n}\sum_{w\in[M]}\sum_{\hatx^n\in\calB_D(x^n)}P_X^n(x^n)P_{\hatX^n|W}(\hatx^n|w)\\*
&\qquad\times \bbo\Big(\sum_{i\in[n]}\jmath(X_i|D,P_X)\geq \log M+n\gamma\Big)\\
\nn&\leq \sum_{x^n\in\calX^n}\sum_{w\in[M]}\sum_{\hatx^n\in\calB_D(x^n)}P_X^n(x^n)P_{\hatX^n|W}(\hatx^n|w)\\*
&\qquad\times \frac{\exp(-n\gamma)\exp(\sum_{i\in[n]}\jmath(X_i|D,P_X))}{M}\label{explain1}\\
&\leq \exp(-n\gamma)\sum_{x^n\in\calX^n}P_X^n(x^n)\exp\Big(\sum_{i\in[n]}\jmath(X_i|D,P_X)\Big)Q_{\hatX^n}(\calB_D(x^n))\label{explain2}\\
\nn&\leq \exp(-n\gamma)\sum_{x^n\in\calX^n}P_X^n(x^n)\exp\Big(\sum_{i\in[n]}\jmath(X_i|D,P_X)\Big)\\*
&\qquad\times\bbE_{Q_{\hatX^n}}[\exp(n\lambda^*D-n\lambda^*d(x^n,\hatX^n))]\label{explain3}\\
\nn&=\exp(-n\gamma)\sum_{\hatx^n\in\hatcalX^n}Q_{\hatX^n}(\hatx^n)\\*
&\qquad\times \bbE_{P_X^n}\Big[\exp\Big(\sum_{i\in[n]}\big(\lambda^*D-\lambda^*d(X_i,\hatx_i)+\jmath(X_i|D,P_X)\big)\Big)\Big]\\
&=\exp(-n\gamma)\sum_{\hatx^n\in\hatcalX^n}Q_{\hatX^n}(\hatx^n)\prod_{i\in[n]}\\*
&\qquad\times \bbE_{P_X}\Big[\exp\Big(\lambda^*D-\lambda^*d(X_i,\hatx_i)+\jmath(X_i|D,P_X)\Big)\Big]\\
&\leq \exp(-n\gamma)\sum_{\hatx^n\in\hatcalX^n}Q_{\hatX^n}(\hatx^n)\label{explain4}\\
&\leq \exp(-n\gamma),
\end{align}
where \eqref{explain1} follows since 
\begin{align}
&P_{W|X^n}(w|x^n)\leq 1,\\
\nn&\bbo\Big(\sum_{i\in[n]}\jmath(X_i|D,P_X)\geq \log M+n\gamma\Big)\\*
&\leq \exp\Big(\sum_{i\in[n]}\jmath(X_i|D,P_X)-\log M-n\gamma\Big)\\
&=\frac{\exp(-n\gamma)\exp(\sum_{i\in[n]}\jmath(X_i|D,P_X))}{M},
\end{align}
\eqref{explain2} follows from the definition of the distortion ball $\calB_D(x^n)$ in \eqref{def:distortionball} and the definition of $Q_{\hatX^n}$, \eqref{explain3} follows since it follows from the definition of $\calB_D(x^n)$ and the Markov inequality (cf. Theorem \ref{markovineq}) that
\begin{align}
Q_{\hatX^n}(\calB_D(x^n))
&=\Pr_{Q_{\hatX^n}}(d(x^n,\hatX^n)\leq D)\\
&\leq \bbE_{Q_{\hatX^n}}[\exp(n\lambda^*D-n\lambda^*d(x^n,\hatX^n))].
\end{align}
and \eqref{explain4} follows from Claim (iii) of Lemma \ref{prop:dtilteddensity}.

The proof of Theorem \ref{converse:fbl} is completed by combining \eqref{converse:step1} and \eqref{explain4}.
\end{proof}

\section{Second-Order Asymptotics}
\label{sec:second4rd}
\subsection{Result and Discussions}
Define the following the distortion-dispersion function 
\begin{align}
\rmV(P_X,D):=\Var_{P_X}[\jmath(X|D,P_X)]\label{def:disdispersion}.
\end{align}
In addition to the assumptions in Section \ref{sec:dtilted4rd}, we need the following further assumptions.
\begin{enumerate}
\item The distortion level $D\in(D_{\rm{min}},D_{\rm{max}})$
where
\begin{align}
D_{\rm{min}}&:=\inf\{D:~R(P_X,D)<\infty\},\label{def:dmin}\\
D_{\rm{max}}&:=\inf_{\hatx\in\hatcalX}\bbE_{P_X}[d(X,\hatx)]\label{def:dmax}.
\end{align}
\item The average $\bbE_{P_X\times P_{\hatX}^*}[d(X,\hatX)^9]<\infty$, where $P_{\hatX^*}$ is induced by the optimal test channel $P_{\hatX|X}^*$ of $R(P_X,D)$.
\item The dispersion $\rmV(P_X,D)$ is positive and finite.
\end{enumerate}
Note that any bounded distortion measure satisfies assumption (ii).

\begin{theorem}
\label{second4rd}
For any $\varepsilon\in(0,1)$, 
\begin{align}
\log M^*(n,D,\varepsilon)&=nR(P_X,D)+\sqrt{n\rmV(P_X,D)}\rmQ^{-1}(\varepsilon)+O(\log n).
\end{align}
\end{theorem}
Theorem \ref{second4rd} characterizes the second dominant term in the expansion of $\log M^*(n,D,\varepsilon)$. Upper and lower bounds on the remainder term $O(\log n)$ is available in \cite[Theorem 12]{kostina2012fixed}. However, the bounds on $O(\log n)$ term do not match even in the sign. The higher order terms than the second-order for the lossy source coding problems remain open. Novel ideas are required to derive matched bounds for the $O(\log n)$ term, generalizing the results for lossless source coding in~\cite{verdu14,chen2020lossless}.  

The result in Theorem \ref{second4rd} holds for any memoryless source and distortion function under mild conditions in~\cite[Theorem 12]{kostina2012fixed} beyond a discrete memoryless source under the bounded distortion measure that is considered in this monograph. For example, Theorem \ref{second4rd} holds for a Gaussian memoryless source under the quadratic distortion measure~\cite[Theorem 40]{kostina2012fixed}.

An equivalent presentation of Theorem \ref{second4rd} is to characterize the second-order codig rate. Similarly to the lossless source coding problem, the second-order coding rate for the rate-distortion problem is defined as follows.
\begin{definition}
\label{def:sr4lossy}
Given any $\varepsilon\in[0,1)$, a real number $L\in\bbR$ is said to be a second-order $(D,\varepsilon)$-achievable rate if there exists a sequence of $(n,M)$-codes such that
\begin{align}
\limsup_{n\to\infty}\frac{1}{\sqrt{n}}(\log M-nR(P_X,D))&\leq L,\\
\limsup_{n\to\infty}\rmP_{\rme,n}(D)&\leq \varepsilon.
\end{align}
For any $\varepsilon\in[0,1)$, the infimum of all second-order $(D,\varepsilon)$-achievable rates is called the optimal second-order coding rate and denoted by $L^*(D,\varepsilon)$.
\end{definition}
An alternative presentation of  Theorem \ref{second4rd} is 
\begin{align}
L^*(D,\varepsilon)=\sqrt{\rmV(P_X,D)}\rmQ^{-1}(\varepsilon)\label{second:rd},
\end{align}
which ignores the unmatched higher order terms. For subsequent chapters on multiterminal lossy source coding problems, we usually present the results in the form of \eqref{second:rd}.

In the next two subsections, we present a proof sketch of Theorem \ref{second4rd} of Kostina and Verd\'u~\cite{kostina2012fixed}, which follows by properly applying the Berry-Esseen theorem (cf. Theorem \ref{berrytheorem}) to the non-asymptotic bounds in Theorems \ref{ach:fbl} and \ref{converse:fbl}.

\subsection{Achievability Proof}
The following non-asymptotic refinement of asymptotic equipartition property (AEP) for lossy source coding~\cite[Lemma 2]{kostina2012fixed} is critical.
\begin{lemma}
\label{lossy:AEP}
There exists constants $(n_0,c,K)\in\bbR_+^3$ such that for all $n\geq n_0$,
\begin{align}
\nn&\Pr\Big\{-\log (P_{\hatX}^*)^n(\calB_D(X^n))\leq \sum_{j\in[n]}\jmath(X_i|D,P_X)+C\log n+c\Big\}\\*
&\geq 1-\frac{K}{\sqrt{n}},
\end{align}
where $C$ is a constant.
\end{lemma}
We remark that an early version of Lemma \ref{lossy:AEP} appeared in the analysis of the redundancy of fixed-to-variable length lossy source coding~\cite{yang1999redundancy}.

For ease of notation, given any $n\in\bbN$, let
\begin{align}
G_n:=\log M-\sum_{j\in[n]}\jmath(X_i|D,P_X)-C\log n-c.
\end{align}

Invoking Theorem \ref{ach:fbl} with $P_{\hatX}^*$ and using the inequality $(1-x)^M\leq \exp(-Mx)$, we conclude that there exists an $(n,M)$-code such that
\begin{align}
\rmP_{\rme,n}(D)
&\leq \bbE\Big[\exp\big(-M(P_{\hatX}^*)^n(\calB_D(X^n))\big)\Big]\\
&\leq \bbE\Big[\exp\big(-\exp(G_n)\big)\Big]+\frac{K}{\sqrt{n}}\label{uselossyaep},
\end{align}
where \eqref{uselossyaep} follows from Lemma \ref{lossy:AEP}. The first term in \eqref{uselossyaep} can be further upper bounded by
\begin{align}
\nn&\bbE\Big[\exp\big(-\exp(G_n)\big)\Big]\\*
\nn&=\bbE\Big[\exp\big(-\exp(G_n)\big)\bbo\Big(G_n<\log\frac{\log n}{2}\Big)\Big]\\*
&\qquad+\bbE\Big[\exp\big(-\exp(G_n)\big)\bbo\Big(G_n\geq \log\frac{\log n}{2}\Big)\Big]\\
&\leq \Pr\Big\{G_n<\log\frac{\log n}{2}\Big\}+\frac{1}{\sqrt{n}}\label{achproof:step1}.
\end{align}
It remains to bound the first term in \eqref{achproof:step1}. Note that for each $i\in[n]$, $\jmath(X_i|D,P_X)$ has the same mean $R(P_X,D)$ and variance $\rmV(P_X,D)$ since the source $X^n$ is memoryless and generated i.i.d. from $P_X$. Let $T(P_X,D)$ be the third absolute moment of $\jmath(X_1|D,P_X)$, i.e.,
\begin{align}
T(P_X,D)&:=\bbE_{P_X}\Big[\big|\jmath(X_1|D,P_X)-R(P_X,D)\big|^3\Big].
\end{align}
Furthermore, given any $\varepsilon\in(0,1)$, let
\begin{align}
B_n&:=\frac{6T(P_X,D)}{(\sqrt{\rmV(P_X,D)})^3},\\
\varepsilon_n&:=\varepsilon-\frac{B_n+K+1}{\sqrt{n}}.
\end{align}
Choose $M$ such that
\begin{align}
\log M
\nn&=nR(P_X,D)+\sqrt{n\rmV(P_X,D)}\rmQ^{-1}(\varepsilon_n)\\*
&\qquad+C\log n+\log\frac{\log n}{2}+c.
\end{align}
Applying the Berry-Esseen theorem to bound the first term in \eqref{achproof:step1} and combining \eqref{uselossyaep} and \eqref{achproof:step1}, it follows that
\begin{align}
\rmP_{\rme,n}(D)\leq \varepsilon.
\end{align}
Therefore, using the Taylor expansion of $\rmQ^{-1}(\cdot)$, we have
\begin{align}
\log M^*(n,D,\varepsilon)
\nn&\leq nR(P_X,D)+\sqrt{n\rmV(P_X,D)}\rmQ^{-1}(\varepsilon_n)\\*
&\qquad+C\log n+\log\frac{\log n}{2}+c\\
&=nR(P_X,D)+\sqrt{n\rmV(P_X,D)}\rmQ^{-1}(\varepsilon)+O(\log n).
\end{align}

\subsection{Converse Proof}
Invoking Theorem \ref{converse:fbl} with $\gamma=\log n$, we conclude that the excess-distortion probability of any $(n,M)$-code satisfies
\begin{align}
\rmP_{\rme,n}(D)
&\geq \Pr_{P_X^n}\Big\{\sum_{i\in[n]}\jmath(X_i|D,P_X)\geq \log M+\log n\Big\}-\frac{1}{n}\label{converse:step2}.
\end{align}
Given any $\varepsilon\in(0,1)$, let
\begin{align}
\varepsilon_n'&:=\varepsilon+\frac{B_n}{\sqrt{n}}+\frac{1}{n},\\
\log M&:=nR(P_X,D)+\sqrt{n\rmV(P_X,D)}\rmQ^{-1}(\varepsilon_n')-\log n.
\end{align}
Applying the Berry-Esseen theorem to the first term in \eqref{converse:step2} leads to
\begin{align}
\rmP_{\rme,n}(D)\geq \varepsilon.
\end{align}
Therefore, it follows from Taylor expansion of $\rmQ^{-1}(\cdot)$ that
\begin{align}
\log M^*(n,D,\varepsilon)
&\geq nR(P_X,D)+\sqrt{n\rmV(P_X,D)}\rmQ^{-1}(\varepsilon_n')-\log n\\
&=nR(P_X,D)+\sqrt{n\rmV(P_X,D)}\rmQ^{-1}(\varepsilon)+O(\log n).
\end{align}

\section{Alternative Proof Using the Method of Types}
We next present an alternative proof of Theorem \ref{second4rd} using the method of types. The achievability part follows from the derivation of Ingber and Kochman~\cite[Theorem 1]{ingber2011dispersion} that uses the type covering lemma for the rate-distortion problem~\cite[Lemma 3]{zhang1997tit}. In the converse part, instead of presenting the converse proof of \cite[Theorem 1]{ingber2011dispersion}, we present an alternative proof inspired by \cite{watanabe2015second} that uses the perturbation approach~\cite{wei2009strong} to prove a type-based strong converse and then lower bound the excess-distortion probability as desired. In our subsequent analyses of multiterminal lossy source coding problems, we mainly use the presented proof based on the method of types in this subsection to derive second-order asymptotics. Thus, the proof of second-order asymptotics for the rate-distortion problem in this subsection provides a solid foundation for further generalizations to more complicated multiterminal cases.

\subsection{Achievability}
Define the following constant
\begin{align}
c_1=4|\calX||\hatcalX|+9.
\end{align}
The following type covering lemma is crucial.
\begin{lemma}
\label{type:covering}
Given any type $Q_X\in\calP_n(\calX)$, for all $R_1\geq R(Q_X,D_1)$, there exists a codebook $\calC=\{\hatx^n(1),\ldots,\hatx^n(M)\}\in(\hatcalX^n)^M$ with $M$ codewords such that
\begin{align}
\log M\leq nR_1+c_1\log n,
\end{align}
and the type class $\calT_{Q_X}^n$ is $D$-covered by the codebook $\calC$, i.e.,
\begin{align}
\calT_{Q_X}^n\subseteq\bigcup_{\hatx^n\in\calC}\{x^n:d(x^n,\hatx^n)\leq D\}.
\end{align}
\end{lemma}

Using Lemma \ref{type:covering}, we can derive an upper bound of the excess-distortion probability of the $(n,M)$-code that uses the codebook $\calC$. Consider the following coding scheme. Given a source sequence $x^n$, the encoder $f$ first calculates the type $\calT_{x^n}$ and sends it to the decoder $\phi$ using at most $|\calX|\log (n+1)$ nats. This is because the number of $n$-length types is upper bounded by $(n+1)^{|\calX|}$ (cf. \eqref{number:types}). The encoder $f$ then calculates $R(\hatT_{x^n},D)$ and checks whether $nR(\hatT_{x^n},D)+(c_1+|\calX|)\log n>\log M$ or not. If the inequality holds, the system declares an error. Otherwise, encoder $f$ sends the index of the codeword from the codebook $\calC$ that minimizes the distortion measure, i.e., $f(x^n)=\argmin_{i\in[M]}d(x^n,\hatx^n(i))$ and the decoder could successfully recover the source sequence as $\hatx^n(i)$, which is then within distortion level $D$ with the source sequence $x^n$ as a result of Theorem \ref{type:covering}. Thus, we have constructed an $(n,M)$-code such that the excess-distortion probability $\rmP_{\rme,n}(D)$ satisfies
\begin{align}
\rmP_{\rme,n}(D)
&\leq \Pr\big\{nR(\hatT_{X^n},D)+(c_1+|\calX|)\log n>\log M\big\}\label{typeach:step1}.
\end{align}

For subsequent analysis, define the typical set
\begin{align}
\calA_n(P_X)
&:=\Bigg\{Q_X\in\calP_n(\calX):~\|Q_X-P_X\|_{\infty}\leq \sqrt{\frac{\log n}{n}}\Bigg\}\label{def:typicalset}.
\end{align}
It follows from \cite[Lemma 22]{tan2014state} that 
\begin{align}
\Pr\{\hatT_{X^n}\notin\calA_n(P_X)\}\leq \frac{2|\calX|}{n^2}\label{upp:atypical}.
\end{align}
Since $(Q_X,D')\rightarrow R(Q_X,D')$ is twice differentiable in the neighborhood of $(P_X,D)$ and the derivatives are bounded, using Claim (iv) of Lemma \ref{prop:dtilteddensity}, for any $x^n$ such that $\hatT_{x^n}\in\calA_n(P_X)$, it follows from Taylor's expansion $R(\hatT_{x^n},D)$ around $\hatT_{x^n}=P_X$ (cf. \cite[Eq. (30)]{ingber2011dispersion}) that
\begin{align}
R(\hatT_{x^n},D)
&=R(P_X,D)+\sum_{x\in\calX}(\hatT_{x^n(x)}-P_X(x))\jmath(x|D,P_X)+c_2\frac{\log n}{n}\\
&=\frac{1}{n}\sum_{i\in[n]}\jmath(x_i|D,P_X)+c_2\frac{\log n}{n}\label{useprops},
\end{align}
where $c_2$ is a bounded constant that accounts for the second derivative of $R(Q_X,D)$ with respect to $Q_X$ around $P_X$ and \eqref{useprops} follows from Claim (ii) of Lemma \ref{prop:dtilteddensity} and definition of the type $\calT_{x^n}$. Thus, using \eqref{typeach:step1}, we have
\begin{align}
\nn&\rmP_{\rme,n}(D)\\*
\nn&\leq \Pr\big\{nR(\hatT_{X^n},D)+(c_1+|\calX|)\log n>\log M,~\hatT_{X^n}\in\calA_n(P_X)\big\}\\*
&\qquad+\Pr\big\{\hatT_{X^n}\notin\calA_n(P_X)\big\}\\
&\leq \Pr\Big\{\sum_{i\in[n]}\jmath(X_i|D,P_X)>\log M-(c_1+c_2+|\calX|)\log n\Big\}+\frac{2|\calX|}{n^2}\label{typeach:step2}.
\end{align}
For any $\varepsilon\in(0,1)$, let
\begin{align}
\delta_n&:=-\frac{2|\calX|}{n^2}+\frac{B_n}{\sqrt{n}}\\
\log M\nn&=nR(P_X,D)+\sqrt{n\rmV(P_X,D)}\rmQ^{-1}(\varepsilon-\delta_n)\\*
&\qquad+(c_1+c_2+|\calX|)\log n.
\end{align}
Applying the Berry-Esseen theorem to the first term in \eqref{typeach:step2}, we conclude that $\rmP_{\rme,n}(D)\leq \varepsilon$ and the achievability proof is completed by using the Taylor expansion of $\rmQ^{-1}(\cdot)$ around $\varepsilon$. 

\subsection{Converse}
We first prove a type-based strong converse. For each $n\in\bbN$, let
\begin{align}
\vartheta_n&:=|\calX|\log (n+1)+2\log n,\\
D_n&:=D+\frac{\bard}{n}.
\end{align}

\begin{lemma}
\label{type:sc}
Given any $Q_X\in\calP_n(\calX)$, if the non-excess-distortion probability of an $(n,M)$-code satisfies
\begin{align}
\Pr\big\{d(X^n,\hatX^n)\leq D|X^n\in\calT_{Q_X}^n\big\}\geq \frac{1}{n}\label{sc:condition},
\end{align}
then
\begin{align}
\log M&\geq nR(Q_X,D_n)-\vartheta_n.
\end{align}
\end{lemma} 
The proof of Lemma \ref{type:sc} is inspired by \cite[Lemma 6]{watanabe2015second} and uses the perturbation approach of Gu and Effros~\cite{wei2009strong}. Lemma \ref{type:sc} implies that for any $(n,M)$-code such that $\log M<nR(Q_X,D)+\vartheta_n$, the conditional excess-resolution probability when the type of the source sequence is $Q_X$ is at least $\frac{n-1}{n}$. Such a result is known as the type-based strong converse theorem since it states that if the rate of any code is not large enough, the type-based excess-distortion probability diverges to one asymptotically, which is analogous to the usual strong converse theorem~\cite{wei2009strong}.

\begin{proof}
Given any type $Q_X\in\calP_n(\calX)$, define the set
\begin{align}
\calD_{Q_X}^n:=\{x^n\in\calT_{Q_X}^n:~d(x^n,\hatx^n)\leq D\},
\end{align}
where $\hatx^n=\phi(f(x^n))$ is the reproduced version of source sequence $x^n$ at the coder side. Let $U_{\calT_{Q_X}^n}$ denote the uniform distribution over the type class $\calT_{Q_X}^n$, let $\beta=\frac{\log n}{n}$ and define
\begin{align}
c(Q_X)
:=n^2U_{\calT_{Q_X}^n}(\calD_{Q_X}^n)+(1-U_{\calT_{Q_X}^n}(\calD_{Q_X}^n)).
\end{align}
Define the distribution $S_{T_{Q_X}^n}(x^n)$ on the type class $\calT_{Q_X}^n$ such that
\begin{align}
S_{T_{Q_X}^n}(x^n)=
\left\{
\begin{array}{ll}
\frac{n^2U_{\calT_{Q_X}^n}(x^n)}{c(Q_X)}&\mathrm{if~}x^n\in\calD_{Q_X}^n\\
\frac{U_{\calT_{Q_X}^n}(x^n)}{c(Q_X)}&\mathrm{otherwise}.
\end{array}
\right.
\end{align}
Note that
\begin{align}
\calU_{T_{Q_X}^n}\{\calD_{Q_X}^n\}
&=\Pr\big\{d(X^n,\hatX^n)\leq D|X^n\in\calT_{Q_X}^n\big\}.
\end{align}
It follows from \eqref{sc:condition} that
\begin{align}
S_{T_{Q_X}^n}(\calD_{Q_X}^n)
&=\frac{n^2}{c(Q_X)}\geq 1-\frac{1}{n}.
\end{align}
Thus, using the $(n,M)$-code satisfying \eqref{sc:condition}, the excess-distortion probability under the source distribution $S_{T_{Q_X}^n}$ satisfies
\begin{align}
S_{T_{Q_X}^n}((\calD_{Q_X}^n)^\rmc)\leq \frac{1}{n},
\end{align}
and the expected distortion satisfies
\begin{align}
\nn&\bbE[d(X^n,\hatX^n)]\\*
&=\sum_{x^n\in\calT_{Q_X}^n}S_{T_{Q_X}^n}(x^n)d(x^n,\hatx^n)\\
&=\sum_{x^n\in\calD_{Q_X}^n}S_{T_{Q_X}^n}(x^n)d(x^n,\hatx^n)+\sum_{x^n\not\in\calD_{Q_X}^n}S_{T_{Q_X}^n}(x^n)d(x^n,\hatx^n)\\
&\leq D+\bard S_{T_{Q_X}^n}((\calD_{Q_X}^n)^\rmc)\\
&=D+\frac{\bard}{n}=D_n.
\end{align}

Following the same steps as the weak converse~(cf. \cite{cover2012elements}) and similarly to \cite[Eq. (274)-(283)]{zhou2016second}, it follows that
\begin{align}
\log M
&\ge H(f(X^n))\\
&\ge I(X^n;f(X^n))\\
&\ge I(X^n;\hatX^n)\\ 
&\ge \sum_{i\in[n]}I(X_i;\hatX_i)-\sum_{i\in[n]}H(X_i)+H(X^n).
\end{align}
Let $J$ be the uniform random variable defined over the set $[n]$ independent of all other random variables. Similarly to \cite[Proof of Lemma 6]{watanabe2015second}, we conclude that the distribution of $X_J$ is the type $Q_X$ and there exists a conditional distribution $Q_{\hatX|X}$ such that
\begin{align}
\bbE_{Q_X\times Q_{\hatX|X}}[d(X;\hatX)]
&=\bbE[d(X_J,\hatX_J)]\\
&=\bbE[d(X^n,\hatX^n)]\\
&\leq D_n,
\end{align}
and
\begin{align}
I(Q_X,Q_{X|\hatX})
&=I(X_J;\hatX_J,J)\\
&=\sum_{i\in[n]}I(X_i;\hatX_i).
\end{align}
The proof of Lemma \ref{type:sc} is completed by recalling the definition of $R(Q_X,D_n)$ and noting that
\begin{align}
\sum_{i\in[n]}I(X_i;\hatX_i)&=I(X_J;\hatX_J,J),\\
\bbE[d(X^n,\hatX^n)]&=\bbE[d(X_J,\hatX_J)],\\
\Big|\sum_{i\in[n]}H(X_i)-H(X^n)\Big|&\leq \frac{|\calX|\log(n+1)}{n}+\alpha+\beta\label{closetoind},
\end{align}
where \eqref{closetoind} follows similarly to \cite[Eq. (34)-(36)]{watanabe2015second} using the method of types.
\end{proof}

Invoking Lemma \ref{type:sc}, we obtain the following lower bound on the excess-distortion probability of any $(n,M)$-code.
\begin{theorem}
\label{converse:types}
Given $(n,M)\in\bbN^2$, any $(n,M)$-code satisfies
\begin{align}
\rmP_{\rme,n}(D)&\geq \Pr\{\log M+\vartheta_n<nR(\hatX_{X^n},D_n)\}-\frac{1}{n}.
\end{align}
\end{theorem}
Note that Theorem \ref{converse:types} is dual to the achievability result in \eqref{typeach:step1}. Recall that $(Q_X,D')\rightarrow R(Q_X,D')$ is twice differentiable in the neighborhood of $(P_X,D)$ and the derivatives are bounded. Similarly to \eqref{useprops}, applying Taylor expansion of $R(\hatT_{x^n},D_n)$ around $(P_X,D)$ for $x^n\in\calA_n(P_X)$, we conclude that there exists some constant $c_3$ such that
\begin{align}
R(\hatT_{x^n},D_n)
&=\frac{1}{n}\sum_{i\in[n]}\jmath(x_i|D,P_X)+c_3\frac{\log n}{n}\label{taylorconverse}.
\end{align}
The rest of the converse proof is analogous to the achievability part from \eqref{typeach:step2} and is omitted for simplicity.


\chapter{Noisy Source}
\label{chap:noisy}
This chapter focuses on noisy lossy source coding, in which the source to be compressed is indirectly available over a noisy channel instead of a lossless channel as in the rate-distortion problem. Dobrushin and Tsybakov~\cite{Dobrushin1962} initialized the study of this problem by showing that the first-order asymptotic minimal achievable rate is similar to the rate-distortion function in the noiseless setting with the conditional average distortion measure~\cite[Chapter 3]{berger1971rate}. This problem finds applications in compression of data collecting from measurements, such as speech in noisy environments and is also known as quantization of noisy sources~\cite[Section V.G]{gray1998tit}.

The optimal encoder and decoder structure was proposed by Wolf and Ziv~\cite{wolf1970noisy}. The large deviations asymptotics for the problem was studied by Weissman~\cite{weissman2004noisy} who derived the universal achievable decay rate of the excess-distortion probability for lossy compression of a discrete memoryless source that is corrupted by a discrete noise. The results in~\cite{Dobrushin1962} were generalized to several other settings with names of indirect source coding or compression of remote sources~\cite{Witsenhausen1980tit,oohama2014indirect,oohama2012remote}.

However, all above results were established in the asymptotic limit of large blocklength, which violates the low-latency requirement of practical communication systems. To resolve this issue, Kostina and Verd\'u~\cite{kostina2016noisy} generalized the finite blocklength analysis of the rate-distortion problem to noisy lossy source coding by deriving non-asymptotic bounds and the second-order asymptotic approximation. In particular, it was shown that the second-order rate for the noisy lossy source coding problem is not equal to that of the rate-distortion problem under the conditional average distortion measure. This chapter is largely based on~\cite{kostina2016noisy}.

\section{Problem Formulation and Asymptotic Result}
The problem formulation of noisy lossy source coding is identical to the rate-distortion problem except that the input to the encoder is a noisy version of the source sequence.  Consider a memoryless source $X^n$ generated i.i.d. from a distribution $P_X$ define on the alphabet $\calX$. Let $P_{Y|X}\in\calP(\calY|\calX)$ be a noisy channel mapping from the set $\calX$ to another set $\calY$ and let $Y^n$ be the noisy output of passing $X^n$ through the memoryless channel $P_{Y|X}^n$. Furthermore, let $\hatcalX$ be the reproduction alphabet and let $d:\calX\times\hatcalX\to \bbR_+$ be the distortion measure. Given any two sequences $x^n$ and $\hatx^n$, the distortion $d(x^n,\hatx^n)$ is assumed additive such that $d(x^n,\hatx^n)=\frac{1}{n}\sum_{i\in[n]}d(x_i,\hatx_i)$.

A code for the noisy lossy source coding problem is defined as follows.
\begin{definition}
Given any $(n,M)\in\bbN^2$, an $(n,M)$-code consists of 
\begin{itemize}
\item an encoder $f:\calY^n\to\calM=[M]$,
\item a decoder $\phi:\calM\to\hatcalX^n$.
\end{itemize}
\end{definition}
The performance metric that we consider is the excess-distortion probability with respect to a distortion level $D\in\bbR_+$, i.e.,
\begin{align}
\rmP_{\rme,n}(D)
&:=\Pr\{d(X^n,\phi(f(Y^n)))>D\}\label{def:excessp4noisy}.
\end{align}
Note that the probability term in \eqref{def:excessp4noisy} is calculated with respect to the distributions of the source sequence and the noisy channel. Given any blocklength $n\in\bbN$, the distortion level $D$ and tolerable excess-distortion probability $\varepsilon$, let $M^*(n,D,\varepsilon)$ denote the minimum number $M$ such that one can construct an $(n,M)$-code with excess-distortion probability $\rmP_{\rme,n}(D)$ no greater than $\varepsilon$, i.e.,
\begin{align}
M^*(n,D,\varepsilon)
:=\inf\big\{M:~\exists\mathrm{~an~}(n,M)\mathrm{-code~s.t.~}\rmP_{\rme,n}(D)\leq \varepsilon\big\}\label{def:M^*4noisy}.
\end{align}

The studies of noisy lossy source coding concern characterization of $M^*(n,D,\varepsilon)$. The first-order asymptotics was derived by Dobrushin and Tsybakov~\cite{Dobrushin1962}. To present their result, define the following noisy rate-distortion function:
\begin{align}
R(P_{XY},D)
&:=\min_{\substack{P_{\hatX|Y}:\bbE[d(X,\hatX)]\leq D\\X-Y-\hatX}} I(Y;\hatX)\\
&=\min_{P_{\hatX|Y}:\bbE[\bard(Y,\hatX)]\leq D}  I(Y;\hatX),\label{rd4noisy}
\end{align}
where $\bard(y,\hatx):=\bbE_{P_{X|Y}}[d(X,\hatx)|Y=y]$ denotes the conditional average distortion measure and $(P_{XY},P_{X|Y})$ are induced by $P_X$ and $P_{Y|X}$. When $P_{Y|X}$ is the identity mapping, $Y=X$, $P_Y=P_X$ and the noisy rate-distortion function reduces to the rate-distortion function~\eqref{def:rd} for the noiseless setting.

Dobrushin and Tsybakov~\cite{Dobrushin1962} showed that $R(P_{XY},D)$ is the first-order asymptotic coding rate, i.e.,
\begin{align}
\lim_{n\to\infty}\frac{1}{n}\log M^*(n,D,\varepsilon)=R(P_{XY},D).
\end{align}
This result implies that asymptotically compressing a noisy source is equivalent to compressing the original source with a surrogate conditional average distortion measure. One might wonder whether the same conclusion holds in the finite blocklength regime. Kostina and Verd\'u~\cite{kostina2016noisy} answered this question negatively. 

\section{Noisy Distortion-Tilted Information Density}
Similar to the rate-distortion problem, the distortion-tilted information density plays a critical role in the presentation and proof of both non-asymptotic and second-order asymptotic bounds. We will present its definition and properties in this section.

Assume that the noisy rate-distortion function $R(P_{XY},D)$ is finite for some distortion level $D$ and let 
\begin{align}
D_{\mathrm{min}}:=\inf\{D\in\bbR_+:R(P_{XY},D)<\infty\}\label{def:dmin4noisy}.
\end{align}

Furthermore, assume that there exists a test channel $P_{\hatX|Y}^*$ that achieves $R(P_{XY},D)$ such that the constraint is satisfied with equality. With this definition, we define the following derivative of the noisy rate-distortion function:
\begin{align}
\lambda^*:=-\frac{\partial R(P_{XY},D)}{\partial D}.
\end{align}

\begin{definition}
For any $(x,y,\hatx)\in\calX\times\calY\times\hatcalX$, for any $D>D_{\mathrm{min}}$, the noisy distortion-tilted information density for noisy lossy source coding is defined as follows:
\begin{align}
\jmath(x,y,\hatx|D,P_{XY}):=\imath_{Y;\hatX}(y;\hatx)+\lambda^*(d(x,\hatx)-D)\label{dtilt4noisy},
\end{align}
where the mutual information density $\imath_{X;\hatX}(x;\hatx)$ is defined as follows 
\begin{align}
\imath_{Y;\hatX}(y;\hatx)=\log\frac{P_{\hatX|Y}^*(\hatx|y)}{P_{\hatX}^*(\hatx)},
\end{align}
and the marginal distribution $P_{\hatX}^*$ is induced by the optimal test channel $P_{\hatX|Y}^*$, the source distribution $P_X$ and the noisy observation channel $P_{Y|X}$.
\end{definition}

Taking expectation over $P_{X|Y}$ on the right hand side of \eqref{dtilt4noisy}, we obtain the surrogate distortion-tilted information density
\begin{align}
\bar{\jmath}(y,D|P_Y):=\imath_{Y;\hatX}(y;\hatx)+\lambda^*(\bard(y,\hatx)-D).
\end{align}
The noisy rate-distortion function is the expectation of the noisy distortion-tilted information density, i.e.,
\begin{align}
 R(P_{XY},D)=\bbE[\jmath(X,Y,\hatX|D,P_{XY})],
\end{align}
where the expectation is over $(X,Y,\hatX)\sim P_{XY\hatX^*}:=P_X\times P_{Y|X}\times P_{\hatX|Y}^*$. Other properties of the noisy distortion-tilted information density $\jmath(X,Y,\hatX|D,P_{XY})$ follow analogously to that of the distortion-tilted information density $\jmath(X|P_X,D)$ for the rate-distortion problem in \eqref{def:dtilteddensity} and are omitted for simplicity.

\section{Non-Asymptotic Bounds}

\subsection{Achievability}

We first present the non-asymptotic achievability bound in~\cite[Theorem 3]{kostina2016noisy}, which generalizes Theorem \ref{ach:fbl} for the rate-distortion problem.

\begin{theorem}
\label{fblach4noisy}
For any $P_{\hatX}$ defined on $\calX$, there exits an $(M,D)$-code such that the excess-distortion probability satisfies
\begin{align}
\rmP_{\rme,n}(D)\leq \int_0^1 \bbE[(\bbP_{P_{\hatX}^n}\{\pi(Y^n,\hatX^n)>t\})^M]\rmd t,
\end{align}
where the expectation is calculated according to $(Y^n,\hatX^n)\sim P_Y^n\times P_{\hatX}^n$, $P_Y$ is induced by $P_X$ and $P_{Y|X}$ and the function $\pi:\calY^n\times\hatcalX^n\to\bbR_+$ is defined as follows:
\begin{align}
\pi(y^n,\hatx^n):=\Pr_{P_{X|Y}^n}\{d(X^n,\hatx^n)>D|Y^n=y^n\},
\end{align}
and  $P_{X|Y}$ is induced by $P_X$ and $P_{Y|X}$.
\end{theorem}
Note that Theorem \ref{fblach4noisy} reduces to Theorem \ref{ach:fbl} for the noiseless case since $\pi(y^n,\hatx^n)=1\{d(x^n,\hatx^n)>D\}$ almost surely when $x^n=y^n$ and $P_{Y|X}$ is the identity matrix.

\begin{proof}
The proof of Theorem \ref{fblach4noisy} parallels Theorem \ref{ach:fbl} for the noiseless case and uses the random coding argument. Let $(\hatX^n(1),\ldots,\hatX^n(M))\in(\hatcalX^n)^M$ be a sequence of $M$ codewords, each of which is generated i.i.d. from $P_{\hatX}$. Upon observing the noisy sequence $y^n$, the encoder $f$ chooses index $i^*$ if
\begin{align}
i^*=\argmin_{i\in[M]}\pi(y^n,\hatX^n(i)).
\end{align}
If there are multiple such minimizers, $i^*$ is chosen arbitrarily among them. The decoder $\phi$ simply outputs $\hatX^n(i^*)$ as the estimation of the source sequence $X^n$.

We next derive an upper bound on the excess-distortion probability of the coding scheme using the encoder $f$ and the decoder $\phi$ described above. Note that
\begin{align}
\rmP_{\rme,n}(D)
&:=\Pr\{d(X^n,f(\phi(Y^n)))>D\}\\
&=\Pr\{d(X^n,\hatX^n(i^*))>D\}\\
&=\bbE[\pi(Y^n,\hatX^n(i^*))]\label{lawofexpectation}\\
&=\int_0^1\Pr\{\pi(Y^n,\hatX^n(i^*))>t\}\rmd t\label{usedefExpectation}\\
&=\int_0^1\bbE\big[\Pr\{\pi(Y^n,\hatX^n(i^*))>t|Y^n\}\big]\rmd t\\
&=\int_0^1 \bbE\Big[\prod_{i\in[M]}\Pr\{\pi(Y^n,\hatX^n(i))>t|Y^n\}\Big]\rmd t\label{indecodebook4noisy}\\
&=\int_0^1 \bbE\Big[(\Pr\{\pi(Y^n,\hatX^n)>t|Y^n\})^M\Big]\rmd t\label{samedist4noisy}
\end{align}
where \eqref{lawofexpectation} follows since $\bbE[f(X,Y)]=\bbE[\bbE[f(X,Y)|X]]$ for any two random variables $(X,Y)$, \eqref{usedefExpectation} follows from the definition of the expectation and the fact that $\phi(\cdot)\in[0,1]$, \eqref{indecodebook4noisy} follows since each codeword $\hatX^n(i)$ is generated independently and $\phi(y^n,\hatx^n(i^*))>t$ implies that $\phi(y^n,\hatx^n(i))>t$ for all $i\in[M]$ and \eqref{samedist4noisy} follows since each codeword $\hatX^n(i)$ is generated from the same distribution.

The proof of Theorem \ref{fblach4noisy} is completed by noting that $\bbE[X]\leq a$ implies that there exists $x\leq a$ for any random variable $X$ and real number $a$.
\end{proof}

To derive a tight second-order approximation to the finite blocklength performance, we need the following corollary of Theorem \ref{fblach4noisy}~\cite[Theorem 4]{kostina2016noisy}.
\begin{corollary}
\label{weakfblach4noisy}
For any $P_{\hatX}$, there exists an $(n,M)$-code such that 
\begin{align}
\rmP_{\rme,n}(D)\leq \exp\left(-\frac{M}{\gamma}\right)+\int_0^1 \Pr\{g(Y^n,t|P_{\hatX})\geq \log \gamma\}\rmd t, 
\end{align}
where the function $g:\calY^n\times\bbR_+\to\bbR_+$ is defined as
\begin{align}
g(y^n,t|P_{\hatX}):=\inf_{\substack{\barP_{\hatX^n}:\pi(y^n,\hatx^n)\leq t\\\forall~\hatx^n\in\supp(\barP_{\hatX^n})}} D(\barP_{\hatX^n}\|P_{\hatX}^n)\label{def:gynt4noisy}.
\end{align}
\end{corollary}
\begin{proof}
It follows from \cite[Eq. (26.18) on page 278]{polyanskiy2014lecture} that for any $(p,\gamma,M)\in\bbR_+\times\bbN$,
\begin{align}
(1-p)^M\leq \exp(-Mp)\leq \exp\left(-\frac{M}{\gamma}\right)+|1-\gamma p|^+\label{ineq4noisy}.
\end{align} 
Combining \eqref{ineq4noisy} and Theorem \ref{fblach4noisy}, we conclude that there exists an $(n,M)$-code such that
\begin{align}
\rmP_{\rme,n}(D)
&\leq \int_0^1 \bbE\Big[(1-\Pr\{\pi(Y^n,\hatX^n)\leq t|Y^n\})^M\Big]\rmd t\\
&\leq \exp\left(-\frac{M}{\gamma}\right)+\int_0^1 \bbE\Big[\big|1-\gamma\Pr\{\pi(Y^n,\hatX^n)\leq t|Y^n\}\big|^+\Big]\rmd t\label{useineq4noisy}.
\end{align}
We next bound the second term in \eqref{useineq4noisy}. For any $y^n$ and $t$, given any $\barP_{\hatX}$ such that $\pi(y^n,\hatX^n)\leq t$ a.s, let $\imath(\hatx^n|P_{\hatX},\barP_{\hatX}):=\sum_{i\in[n]}\log\frac{\barP_{\hatX}(\hatx_i)}{P_{\hatX}(\hatx_i)}$. It follows that
\begin{align}
\nn&\big|1-\gamma\Pr_{P_{\hatX}^n}\{\pi(y^n,\hatX^n)\leq t|Y^n\}\big|^+\\
&\leq \Big|1-\gamma\bbE_{\barP_{\hatX}^n}\Big[\exp(-\imath(\hatX^n|P_{\hatX},\barP_{\hatX}))\bbo(\pi(y^n,\hatX^n)\leq t)\Big]\Big|^+\label{cof4noisy}\\
&\leq \Big|1-\gamma\bbE_{\barP_{\hatX}^n}\Big[\exp(-\imath(\hatX^n|P_{\hatX},\barP_{\hatX}))\Big|^+\\
&\leq \Big|1-\gamma\exp(-D(\barP_{\hatX}^n\|P_{\hatX}^n))\Big|^+\label{usejesen4nooisy}\\
&\leq \bbo(D(\barP_{\hatX}^n\|P_{\hatX}^n)\geq \log \gamma)\label{algebra4noisy},
\end{align}
where \eqref{cof4noisy} follows from the change-of-measure technique~\cite{csiszar2011information}, \eqref{usejesen4nooisy} follows from Jensen's inequality and the fact that $\exp(-x)$ is convex in $x$, and \eqref{algebra4noisy} follows from simple algebra.

Note that \eqref{algebra4noisy} holds for any $\barP_{\hatX}$ such that $\pi(y^n,\hatX^n)\leq t$. The tightness upper bound would be for any $y^n$ and $t$ 
\begin{align}
\big|1-\gamma\Pr_{P_{\hatX}^n}\{\pi(y^n,\hatX^n)\leq t|Y^n\}\big|^+\leq \bbo(g(y^n,t|P_{\hatX})\geq \log \gamma)\label{weakachstep2}.
\end{align}
The proof of Corollary \ref{weakfblach4noisy} is completed by combining \eqref{useineq4noisy} and \eqref{weakachstep2}.
\end{proof}

\subsection{Converse}
We next present a non-asymptotic converse bound~
\cite[Theorem 2]{kostina2016noisy}, using which the optimality of the coding scheme in Theorem \ref{fblach4noisy} is proved in the second-order asymptotic sense. To do so, we need the following definition. For any distribution $P_Y$ and any conditional distribution $\barP_{Y|\hatX}$, given any $(x^n,y^n,\hatx^n)$, let 
\begin{align}
\imath(y^n,\hatx^n|P_Y,\barP_{Y|\hatX}):=\log\frac{\barP_{Y|\hatX}^n(y^n|\hatx^n)}{P_Y^n(y^n)},
\end{align}
and define the following function
\begin{align}
\nn&\imath(x^n,y^n,\hatx^n|P_Y,\barP_{Y|\hatX})\\*
&:=\imath(y^n,\hatx^n|P_Y,\barP_{Y|\hatX})+\sup_{\lambda\in\bbR_+}\lambda (d(x^n,\hatx^n)-D)-\log M.
\end{align}

The next theorem generalizes Theorem \ref{converse:fbl} for the noiseless setting.
\begin{theorem}
\label{fblcon4noisy}
Any $(n,M)$-code satisfies that
\begin{align}
\rmP_{\rme,n}(D)
\nn&\geq \inf_{P_{\hatX^n|Y^n}}\sup_{\barP_{Y|\hatX}}\sup_{\gamma\in\bbR_+} \Big\{\Pr\{\imath(X^n,Y^n,\hatX^n|P_Y,\barP_{Y|\hatX})\geq \gamma\}\\*
&\qquad\qquad\qquad\qquad\qquad-\exp(-\gamma)\Big\},
\end{align}
where $(X^n,Y^n,\hatX^n)\sim P_X^n\times P_{Y|X}^n\times P_{\hatX^n|Y^n}$.
\end{theorem}
\begin{proof}
Let $S$ be a random variable on $\calM$ that denotes the output of the encoder. Let $P_{S|Y^n}$ and $P_{\hatX^n|S}$ be the stochastic mapping of the encoder $f$ and decoder $\phi$ respectively. Furthermore, let $P_{\hatX^n|Y^n}$ be the conditional distribution induced by $P_X$, $P_{Y|X}$, $P_{S|Y^n}$ and $P_{\hatX^n|S}$, i.e.,
\begin{align}
P_{\hatX^n|Y^n}(\hatx^n|y^n)
&=\frac{\sum_{x^n,s}P_X^n(x^n)P_{Y|X}^n(y|x^n)P_{S|Y^n}(s|y^n)P_{\hatX^n|S}(\hatx^n|s)}{P_Y^n(y^n)}\\
&=\frac{\sum_{s}P_Y^n(y^n)P_{S|Y^n}(s|y^n)P_{\hatX^n|S}(\hatx^n|s)}{P_Y^n(y^n)}\label{jointd4noisy}.
\end{align}
For any $\gamma\in\bbR_+$, 
\begin{align}
\nn&\Pr\{\imath(X^n,Y^n,\hatX^n|P_Y,\barP_{Y|\hatX})\geq \gamma\}\\*
\nn&=\Pr\{\imath(X^n,Y^n,\hatX^n|P_Y,\barP_{Y|\hatX})\geq \gamma,~d(X^n,\hatX^n)\leq D\}\\*
&\qquad+\Pr\{\imath(X^n,Y^n,\hatX^n|P_Y,\barP_{Y|\hatX})\geq \gamma,~d(X^n,\hatX^n)> D\}\\
&\leq \Pr\{\imath(X^n,Y^n,\hatX^n|P_Y,\barP_{Y|\hatX})\geq \gamma,~d(X^n,\hatX^n)\leq D\}+\rmP_{\rme,n}(D)\label{converse4noisy:step1}.
\end{align}
We next further upper bound the second term in \eqref{converse4noisy:step1} as follows:
\begin{align}
\nn&\mathrm{second~term~of~\eqref{converse4noisy:step1}}\\*
&=\Pr\{\imath(Y^n,\hatX^n|P_Y,\barP_{Y|\hatX})\geq \log M+\gamma,~d(X^n,\hatX^n)\leq D\}\label{explain14noisy}\\
&\leq \Pr\{\imath(Y^n,\hatX^n|P_Y,\barP_{Y|\hatX})\geq \log M+\gamma\}\\
&\leq \frac{\exp(-\gamma)}{M}\bbE[\exp(\imath(Y^n,\hatX^n|P_Y,\barP_{Y|\hatX}))]\label{useMarkov4noisy}\\
&=\frac{\exp(-\gamma)}{M}\sum_{(y^n,\hatx^n)}\sum_{s\in[M]}P_Y^n(y^n)P_{S|Y^n}(s|y^n)P_{\hatX^n|S}(\hatx^n|s)\frac{\barP_{Y|\hatX}^n(y^n|\hatx^n)}{P_Y^n(y^n)}\label{usejointd4noisy}\\
&\leq \frac{\exp(-\gamma)}{M}\sum_{(y^n,\hatx^n)}\sum_{s\in[M]}P_{\hatX^n|S}(\hatx^n|s)\barP_{Y|\hatX}^n(y^n|\hatx^n)\label{prob<=14noisy}\\
&\leq \exp(-\gamma)\label{converse4noisy:step2},
\end{align}
where \eqref{explain14noisy} follows since when $d(x^n,\hatx^n)\leq D$, 
\begin{align}
\imath(x^n,y^n,\hatx^n|P_Y,\barP_{Y|\hatX})=\imath(y^n,\hatx^n|P_Y,\barP_{Y|\hatX})-\log M,
\end{align}
\eqref{useMarkov4noisy} follows from the Markov inequality (cf. Theorem \ref{markovineq}), \eqref{usejointd4noisy} follows from the distribution in \eqref{jointd4noisy},  and \eqref{prob<=14noisy} follow since $P_{S|Y^n}(s|y^n)\leq 1$.

Combining \eqref{converse4noisy:step1} and \eqref{converse4noisy:step2}, we conclude that 
\begin{align}
\rmP_{\rme,n}(D)
&\geq \Pr\{\imath(X^n,Y^n,\hatX^n|P_Y,\barP_{Y|\hatX})\geq \gamma\}-\exp(\gamma).
\end{align}
The proof is completed by optimizing over parameters $(\barP_{Y|\hatX},\gamma)$ for the tightest bound and by optimizing over $P_{\hatX^n|Y^n}$ to yield a code-independent bound.
\end{proof}

Note that the optimization over $P_{\hatX^n|Y^n}$ is intractable for large $n$. To derive tight second-order asymptotics, the following relaxation was proposed~\cite[Corollary 1]{kostina2016noisy}.
\begin{corollary}
\label{fblcon:coro}
Any $(n,M)$-code satisfies
\begin{align}
\rmP_{\rme,n}(D)
\nn&\geq \sup_{\barP_{Y|\hatX}}\sup_{\gamma\in\bbR_+} \Big\{\inf_{\hatx^n}\bbE\big[\Pr\{\imath(X^n,Y^n,\hatx^n|P_Y,\barP_{Y|\hatX})\geq \gamma|Y^n\}\big]\\*
&\qquad\qquad\qquad\qquad-\exp(-\gamma)\Big\}.
\end{align}
\end{corollary}
Corollary \ref{fblcon:coro} follows by lower bounding the minimax bound in Theorem \ref{fblcon4noisy} with a maximin lower bound and applying the law of iterative expectation with simple algebra.

The infimum over $\hatx^n$ could still be challenging. However, as shown in \cite[Remark 2]{kostina2016noisy}, the inner probability term is a constant function of $\hatx^n$ under certain choice of $\barP_{Y|X}$ for a GMS corrupted by an AWGN channel under the quadratic distortion measure and for a uniform discrete source corrupted by a symmetric channel under the Hamming distortion measure. An example for the latter case will be presented to illustrate the result.

\section{Second-Order Asymptotics}

\subsection{Result and Discussions}
In this section, we present a second-order approximation to the finite blocklength performance~\cite[Theorem 5]{kostina2016noisy}. To present the result, several assumptions are needed.
\begin{enumerate}
\item Let the distortion level $D\in(D_{\rm{min}},D_{\max})$, where $D_{\rm{min}}$ was defined in \eqref{def:dmin4noisy} and $D_{\rm{max}}:=\inf_{\hatx\in\hatcalX}\bbE_{P_X}[d(X,\hatx)]$.
\item Given any $Q_Y$, let $Q_{XY}=P_{X|Y}Q_Y$. Suppose that for all $Q_Y$ in the neighborhood of $P_Y$, $R(Q_{XY},D)$ is twice continuously differentiable with respect to $Q_Y$ and $\supp(Q_{\hatX}^*)=\supp(P_{\hatX}^*)$, where $Q_{\hatX^*}$ is induced by $Q_Y$ and the optimal test channel $Q_{\hatX|Y}^*$ for $R(Q_{XY},D)$.
\end{enumerate}
Define the following dispersion function for noisy lossy source coding.
\begin{align}
\tilrmV(P_{XY},D):=\mathrm{Var}[\jmath(X,Y,\hatX|D,P_{XY})].
\end{align}
The second-order asymptotics states as follows.
\begin{theorem}
\label{second4noisy}
For any $\varepsilon\in(0,1)$, 
\begin{align}
\log M^*(n,D,\varepsilon)=nR(P_{XY},D)+\sqrt{n\tilrmV(P_{XY},D)}\rmQ^{-1}(\varepsilon)+O(\log n).
\end{align}
\end{theorem}
The proof of Theorem \ref{second4noisy} is omitted since it follows similarly to the rate-distortion problem by applying the Berry-Esseen theorem to the non-asymptotic bounds. Readers could refer to \cite[Appendies C-D]{kostina2016noisy} or \cite[Appendix D]{kostina2013lossy} for details. Note that Theorem \ref{second4noisy} was only proved for discrete memoryless sources using non-asymptotic bounds in Corollaries \ref{weakfblach4noisy} and \ref{fblcon:coro} involving the distortion-tilted information density in both directions. It would be interesting to provide an alternative proof using method of types for a DMS and to generalize the results to a GMS.

A critical remark is that the dispersion function has the following equivalent form 
\begin{align}
\tilrmV(P_{XY},D)
=\mathrm{Var}[\bar{\jmath}(Y,D|P_Y)]+(\lambda^*)^2\mathrm{Var}_{P_{XY\hatX}^*}[d(X,\hatX)| Y,\hatX]\label{altdispersion4noisy},
\end{align}
where the conditional variance for any two variables $(U,V)$ with joint distribution $P_{UV}$ is $\mathrm{Var}_{P_{UV}}[U|V]=\bbE_{P_U}[(U-\bbE_{P_{U|V}}[U|V])^2]$ and the joint distribution $P_{XY\hatX^*}$ is induced by $P_X$, $P_{Y|X}$ and the optimal test channel $P_{\hatX|Y}$ for $R(P_{XY},D)$.
 This implies that unlike the first-order asymptotics, the second-order coding rate for noisy lossy source coding is not equivalent to the noiseless case with the surrogate conditional average distortion measure since the additional second term in \eqref{altdispersion4noisy} is non-zero unless $P_{Y|X}$ is the identity matrix.

\subsection{A Numerical Example}
We next present a numerical example to illustrate Theorem \ref{second4noisy}~\cite[Section VI]{kostina2016noisy}. Let $\calX=\{0,1\}$, $\calY=\{0,1,\rme\}$, $P_X(0)=P_X(1)=0.5$ and let $P_{Y|X}$ be a binary erasure channel with erasure probability $\delta\in\bbR_+$, i.e., for any $(x,y)\in\calX\times\calY$,
\begin{align}
P_{Y|X}(y=x)=(1-\delta)\bbo(y=x)+\delta \bbo(y-\rme).
\end{align}
Let $P_{XY}$ be the induced joint distribution. For any $D\in(0.5\delta,0.5)$, under the Hamming distortion measure, the noisy rate-distortion function is
\begin{align}
R(P_{XY},D)=(1-\delta)\Big(\log 2-H_\rmb\Big(\frac{D-0.5\delta}{1-\delta}\Big)\Big),
\end{align}
where $H_\rmb(\cdot)$ is the binary entropy function and the optimal test channel $P_{\hatX|Y}^*$ satisfies that the marginal distribution $P_{\hatX^*}(0)=P_{\hatX}^*(1)=0.5$ and
\begin{align}
P_{Y|\hatX}^*(y|\hatx)
&=\left\{
\begin{array}{ll}
1-D-0.5\delta&\mathrm{if~}\hatx=y\\
D-0.5\delta&\mathrm{if~} \hatx\neq y,~y\neq \rme\\
\delta&\mathrm{otherwise}.
\end{array}
\right.
\end{align}

The derivative $\lambda^*$ satisfies
\begin{align}
\lambda^*=\log\frac{1-D-0.5\delta}{D-0.5\delta}
\end{align}
and the noisy distortion-tilted information density satisfies
\begin{align}
\jmath(x,y,\hatx|P_{XY},D)
&=-\lambda^*D+\left\{
    \begin{array}{ll}
    \log\frac{2}{1+\exp(-\lambda^*)}&\mathrm{if~}x=y,\\
    \lambda^*&\mathrm{if~} \hatx\neq y,~y\neq \rme,\\
    0&\mathrm{otherwise}.
    \end{array}
\right.
\end{align}
Thus, the noisy dispersion is 
\begin{align}
\tilrmV(P_{XY},D)=\delta(1-\delta)\bigg(\log\cosh\Big(\frac{\lambda^*}{2\log e}\Big)\bigg)^2+\frac{\delta\lambda^*}{4}.
\end{align}

We next calculate the rate-distortion function and the dispersion function under the surrogate conditional average distortion measure. From the definition of $\bard(\cdot)$, we find
\begin{align}
\bard(y,\hatx)=
\left\{
\begin{array}{ll}
0&\mathrm{if~}\hatx=y,\\
1&\mathrm{if~} \hatx\neq y,~y\neq \rme,\\
0.5&\mathrm{otherwise}.
\end{array}
\right.
\end{align}
By taking expectation of $\jmath(x,y,\hatx|P_{XY},D)$ over $P_{X|Y}$, we obtain the surrogate distortion-tilted information density as follows:
\begin{align}
\bar{\jmath}(y|P_Y,D)
&=\left\{
\begin{array}{ll}
0.5\lambda^*D&\mathrm{if~}y=\rme,\\
-\lambda^*D+\log\frac{2}{1+\exp(-\lambda^*)}&\mathrm{otherwise},
\end{array}
\right.
\end{align}
and its dispersion satisfies
\begin{align}
\mathrm{Var}[\bar{\jmath}(Y|P_Y,D)]
&=\delta(1-\delta)\bigg(\log\cosh\Big(\frac{\lambda^*}{2\log e}\Big)\bigg)^2\\
&=\tilrmV(P_{XY},D)-\frac{\delta\lambda^*}{4}\\
&<\tilrmV(P_{XY},D).
\end{align}


\chapter{Noisy Channel}
\label{chap:jscc}

This chapter concerns lossy joint source channel coding, where one aims to transmit an information source over a noisy channel and recover it in a lossy manner. Such a problem generalizes the rate-distortion problem by having a noisy channel to convey the encoded source information instead of a lossless channel. It is also known as vector quantization for noisy channels~\cite[Section V.G]{gray1998tit}. 

Asymptotically, Shannon~\cite{shannon1959coding} proved that it is optimal to use separate source and channel coding (SSCC) to achieve the optimal rate if one targets for vanishing error probability. That is, at the transmitter side, one first compresses the source using a source encoder and subsequently encodes the output of the source encoder using a channel encoder. Analogously, at the receiver side, one first decodes the output of the source encoder using a channel decoder and subsequently produces a source estimate using a source decoder. Optimal performance could be achieved when optimal codes are used for both source and channel coding.

However, one might wonder whether this claim holds in the refined asymptotic analysis or in the finite blocklength regime. Csisz\'ar~\cite{csiszar1980joint} answered this question negatively in the large deviations regime by showing that joint source channel coding (JSCC) achieves a larger error exponent than SSCC. Wang, Ingber and Kochman~\cite{wang2011dispersion} provided further evidence in second-order asymptotics, which provides an approximation to the finite blocklength performance. Kostina and Verd\'u derived non-asymptotic bounds valid for any blocklength and recovered the result in \cite{wang2011dispersion} by applying the Berry-Esseen theorem to their non-asymptotic bounds for moderately large blocklengths. To reveal the impact of a noisy channel on the rate-distortion problem, we present the non-asymptotic bounds in~\cite{kostinajscc} and the second-order asymptotics in~\cite{wang2011dispersion,kostinajscc}.

\section{Problem Formulation and Asymptotic Result}

\subsection{Problem Formulation}
Consider a discrete memoryless source $P_X$ defined on the finite alphabet $\calX$ and a discrete memoryless channel with transition probability matrix $P_{Z|Y}\in\calP(\calZ|\calY)$ where $\calY$ is the input alphabet and $\calZ$ is the output alphabet of the channel. Furthermore, similarly to the rate-distortion problem, let $\hatcalX$ be the reproduction alphabet and let $d:\calX\times\hat\calX\to[0,\infty)$ be a bounded distortion measure. Given any $n\in\bbN$, let the distortion between $x^n$ and $\hatx^n$ be defined as $d(x^n,\hatx^n)=\frac{1}{n}\sum_{i\in[n]}d(x_i,\hatx_i)$. 

In the joint source channel coding problem, one wishes to transmit a source sequence $X^n$ reliably over the DMC $P_{Z|Y}$ reliably over $k$ channel uses in a lossy manner. Formally, a code is defined as follows.
\begin{definition}
\label{code4jscc}
An $(n,k)$-code for the lossy joint source channel coding problem consists of
\begin{itemize}
\item an encoder $f:\calX^n\to\calY^k$
\item a decoder $\phi:\calZ^k\to \hatcalX^n$.
\end{itemize}
\end{definition}
For simplicity, we let $Y^k:=f(X^n)$ be the output of the encoder and $\hatX^n$ be the output of the decoder. From problem formulation, the Markov chain $X^n-Y^k-Z^k-\hatX^n$ holds. Similarly to the rate-distortion problem, the performance criterion is the excess-distortion probability with respect to a target distortion level $D$:
\begin{align}
\rmP_{\rme,n,k}(D):=\Pr\{d(X^n,\hatX^n)>D\},
\end{align}
where the probability is calculated with respect to the joint distribution of $(X^n,Y^k,Z^k,\hatX^n)$ that are induced by the source distribution $P_X$, the noisy channel $P_{Z|Y}$ and the $(n,k)$-code. 

The fundamental limit of lossy JSCC is the maximum number of symbols that can be reliably transmitted over $k$ channel uses with excess-distortion probability no greater than $\varepsilon\in(0,1)$, i.e.,
\begin{align}
n^*(k,D,\varepsilon):=\sup\{n\in\bbN:~\exists~\mathrm{an~}(n,k)\mathrm{-code~s.t.~}\rmP_{\rme,n,k}(D)\leq \varepsilon\}.
\end{align} 
In case there is a cost constraint on the channel input $Y^k$, let $c:\calY^k\to\bbR_+$ be the cost function. We can define an $(n,k,\alpha)$-code similarly to Definition \ref{code4jscc} except that a cost constraint $c(Y^k)\leq \alpha$ should be added. Analogously, we can define the corresponding fundamental limit $n^*(k,D,\alpha,\varepsilon)$. In this chapter, for simplicity, we focus on the case without a cost constraint on the channel input.

When the noisy channel $P_{Z|Y}$ is the identity matrix and when $\calZ=\calY=\calM:=[M]$ for some $M\in\bbN$, the lossy JSCC problem reduces to the rate-distortion problem. When the distortion measure is the Hamming distortion measure and when $D=0$, the lossy JSCC problem reduces to the lossless case where one wishes to transmit $X^n$ in an almost perfect manner.

Note that the problem formulation here is slightly different from \cite{wang2011dispersion,kostinajscc} where they use $S^k$ to denote the source sequence of length $k$ and use $P_{Y|X}^n$ to denote the noisy channel with $n$ channel uses. However, to be consistent with other chapters, we keep $X^n$ as the source sequence of length $n$ and use $P_{Z|Y}^k$ to denote the noisy channel with $k$ channel uses.

\subsection{Asymptotic Result}
We next present Shannon's first-order asymptotic characterization of $n^*(k,D,\varepsilon)$. Recall the definition of the rate-distortion function $R(P_X,D)$ in \eqref{def:rd}. Furthermore, define the following capacity function:
\begin{align}
C(P_{Z|Y}):=\max_{P_Y\in\calP(\calX)}I(P_Y,P_{Y|Z}),
\end{align}
where $P_{Y|Z}$ is induced by $P_Y$ and $P_{Z|Y}$.

With these definitions, Shannon~\cite{shannon1959coding} proved the following result.
\begin{theorem}
\label{shanon4jscc}
For any target distortion $D$ such that $R(P_X,D)<\infty$,
\begin{align}
\lim_{\varepsilon\to 0}\lim_{k\to\infty}\frac{n^*(k,D,\varepsilon)}{k}=\frac{C(P_{Z|Y})}{R(P_X,D)}=:\rho(P_X,P_{Z|Y},D)\label{def:rho4jscc}.
\end{align}
\end{theorem}
Theorem \ref{shanon4jscc} implies that asymptotically, the maximum number of source symbols $n^*(k,D,\varepsilon)$ that can be transmitted reliably in a lossy manner, scales in the same order as the number of channel uses $k$ and the optimal ratio is $\rho(P_X,P_{Z|Y},D)$. Intuitively, one can concatenate an optimal lossy source coding with an optimal channel code to achieve this goal. That is, one can first compress $X^n$ into $nR(P_X,D)$ nats reliably in a lossy manner and then transmit the $nR(P_X,D)$ nats reliably over the noisy memoryless channel $P_{Z|Y}$ supposed that $nR(P_X,D)\leq kC(P_{Y|X})$. However, as we shall show, such a separation based coding scheme is suboptimal in the second-order asymptotics.

\section{Non-Asymptotic Bounds}
In this section, we present non-asymptotic achievability and converse bound that hold for any $(n,k)\in\bbN^2$.

\subsection{Achievability}
We first present an achievability bond. Recall the definition of the distortion ball $\calB_D(x^n)$ in \eqref{def:distortionball}. Furthermore, given any input distribution $P_Y\in\calP(\calY)$ and any channel $P_{Z|Y}$, for each $(y,z)\in\calY\times\calZ$, define the following information density
\begin{align}
\imath(y^k,z^k|P_Y,P_{Z|Y}):=\sum_{i\in[k]}\log\frac{P_{Z|Y}(z_i|y_i)}{P_Z(z_i)},
\end{align}
where $P_Z$ is induced by $P_Y$ and $P_{Z|Y}$.

Kostina and Verd\'u prove the following theorem~\cite[Theorem 7]{kostinajscc}.
\begin{theorem}
\label{fblach4jscc}
There exists an $(n,k)$-code such that
\begin{align}
\rmP_{\rme,n,k}(D)
\nn&\leq \inf_{P_{Y^k},P_{\hatX^n},P_{W|X^n}}\bigg\{\bbE\Big[\exp\Big(-\big|\imath(Y^k;Z^k|P_Y,P_{Z|Y})-\log W\big|^+\Big)\Big]\\*
&\qquad+\bbE\big[(1-P_{\hatX^n}(\calB_D(X^n))^W\big]\bigg\},
\end{align}
where the random variable $W$ takes values in $\bbN$ and random variables distribute as $(X^n,Y^k,Z^k,\hatX^n,W)\sim P_X^nP_{Y^k}P_{Z|Y}^kP_{\hatX^n}P_{W|X^n}$.
\end{theorem}
The proof of Theorem \ref{fblach4jscc} follows by analyzing a careful concatenation of a source and channel code. In particular, the second expectation term corresponds to the non-asymptotic achievability bound in Theorem \ref{ach:fbl} for the rate-distortion problem and the first expectation term corresponds to the error probability of a channel code. Note that the two codes are connected via the critical random variable $W$ which depends only on the source sequence $X^n$. Such a design is analogous to the JSCC coding scheme based on the unequal error protection idea by Wang, Ingber and Kochman~\cite{wang2011dispersion}.
\begin{proof}
For simplicity, we only present the code used to prove Theorem \ref{achcoro4jscc} since the analysis of the code is similar to that of the rate-distortion problem in Theorem \ref{ach:fbl} and its noisy version in Theorem \ref{fblach4noisy}. Readers could refer to \cite[Eq. (89)-106]{kostinajscc} for details.

To present the code, let $M\in\bbN$ be arbitrary and let $W\in[M]$ be a random variable that depends only on the source sequence $X^n$. Let $\hatbx^n=(\hatx^n(1),\ldots,\hatx^n(M))$ be a sequence of source codewords and let $y^k(1),\ldots,y^k(M)$ be a sequence of channel codewords. The encoder $f$ is a concatenation of a source encoder $f_\rms:\calX^n\to[M]$ and a channel encoder $f_\rmc:[M]\to\calY^k$. Specifically, given the source sequence $x^n$, the source encoder $f_\rms$ generates a random variable $W$ using a stochastic mapping $P_{W|X^n=x^n}$, outputs the index $m$ if it is the smallest index from $[W]$ such that $(x^n,\hatx^n(i))\leq D$ and outputs $W$ if there is no such index. The channel encoder then outputs $y^k(m)$. Thus, given $x^n$, the encoder $f$ outputs $y^k(f_\rms(x^n))$. 

The decoder is also a concatenation of a source decoder $\phi_\rms:[M]\to\hatcalX^n$ and a channel decoder $\phi_\rmc:\calZ^k\to[M]$. Let $U\in[M+1]$ be a random variable that depends on $X^n,W$ and the source codebook such that $U=f_\rms(X^n)$ if $d(X^n,f_\rms(X^n))\leq D$ and $U=M+1$ otherwise. Given the channel output $z^k$ that is the output of passing $y^k(m)$ over the memoryless channel $P_{Z|Y}^k$, the channel decoder outputs $\hatm$ if
\begin{align}
\hatm:=\argmax_{j\in[M]}P_{U|\hatbX^n}(j|\hatbx^n)P_{Z^k|Y^k}(z^k|y^k(j)).
\end{align}
Subsequently, the source decoder outputs $\hatx^n(\hatm)$ as the source estimate.
\end{proof}

To derive second-order asymptotics, we need the following weakening version of Theorem \ref{fblach4jscc}.
\begin{corollary}
\label{achcoro4jscc}
There exists an $(n,k)$-code such that
\begin{align}
\rmP_{\rme,n}
\nn&\leq \inf_{P_{Y^k},P_{\hatX^n}}\inf_{\gamma>0}\bigg\{
\bbE\Big[\exp\Big(-\Big|\imath(Y^k;Z^k|P_Y,P_{Z|Y})-\log \frac{\gamma}{P_{\hatX^n}(\calB_D(X^n)}\Big|^+\Big]\\*
&\qquad+\exp(1-\gamma)\bigg\}.
\end{align}
\end{corollary}
\begin{proof}
Let $\gamma>0$ be arbitrary, $\tau(x^n):=P_{\hatX^n}(\calB_D(x^n))$ and choose $W=\left\lfloor \frac{\gamma}{\tau(X^n)} \right\rfloor$. It follows that for any $x^n$.
\begin{align}
(1-\tau(x^n))^{\lfloor \frac{\gamma}{\tau(x^n)}\rfloor}
&\leq (1-\tau(x^n))^{\frac{\gamma}{\tau(x^n)}-1}\\
&\leq \exp(1-\gamma)\label{weaken4jscc}.
\end{align}
The proof of Corollary \ref{achcoro4jscc} follows by invoking Theorem \ref{fblach4jscc} and the result \eqref{weaken4jscc}.
\end{proof}

\subsection{Converse}
We next present the non-asymptotic converse bound~\cite[Theorem 3]{kostinajscc}. Recall the definition of the distortion-tilted information density $\jmath(x|D,P_X)$ in \eqref{def:dtilteddensity}. 
\begin{theorem}
\label{fblcon4jscc}
Given any $(\gamma,T)\in\bbR_+\times\bbN$, any $(n,k)$-code satisfies that
\begin{align}
\nn&\rmP_{\rme,n,k}(D)\\*
\nn&\geq \inf_{P_{Y^k|X^n}}\bigg\{-T\exp(-\gamma)\\*
&\qquad+\sup_{\barP_{Z^k|W}}\!\!\!\sup_{\substack{W\in[T]:\\X^n-(Y^k,W)-Z^k}}\!\!\!\!\!\!\Pr\Big\{\sum_{i\in[n]}\jmath(X_i|D,P_X)-\imath(Y^k;Z^k|W)\geq\gamma\Big\}\bigg\}\label{conform14jscc}\\
\nn&\geq -T\exp(-\gamma)\\*
&\qquad+\sup_{\barP_{Z^k|W}}\sup_{\substack{W\in[T]:\\X^n-(Y^k,W)-Z^k}}\bbE\Big[\inf_{y^k}\Pr\Big\{\sum_{i\in[n]}\jmath(X_i|D,P_X)-\imath(y^k;Z^k|W)\geq\gamma|X^n\Big\}\Big]\label{conform24jscc},
\end{align}
where the information density function $\imath(y^k;z^k|w)$ is defined as
\begin{align}
\imath(y^k;z^k|w):=\log\frac{P_{Z^k|Y^kW}(z^k|y^k,w)}{\barP_{Z^k|W}(z^k|t)}\label{infdensity4jscc}
\end{align}
and the distribution $P_{Z^k|Y^kW}$ is induced by the joint distribution of $(X^n,Y^k,W,Z^k)\sim P_X^nP_{W|X^n}P_{Y^k|X^nW}P_{Z^k|Y^kW}$ for some distributions $(P_{W|X^n},P_{Y^k|X^nW},P_{Z^k|Y^kW})$ such that the marginal conditional distribution satisfies $P_{Z^k|Y^k}=P_{Z|Y}^k$.
\end{theorem}
We remark that Kostina and Verd\'u derived several other non-asymptotic converse bounds in \cite[Section III.B]{kostinajscc} using list decoding and hypothesis testing. However, Theorem \ref{fblcon4jscc} and its weakened versions suffice to derive a second-order asymptotic converse bound. Furthermore, Theorem \ref{fblcon4jscc} reduces to the non-asymptotic converse bound in Theorem \ref{converse:fbl} for the rate-distortion problem when $T=1$, $\calZ=\calY=\calM:=[M]$, $P_Y$ is the uniform distribution over $[M]$ and $P_{Z|Y}$ is the identity matrix.

\begin{proof}
Let $\gamma\in\bbR_+$ and $T\in\bbN$ be arbitrary. Consider any potentially stochastic encoder $P_{Y^k|X^n}$ and decoder $P_{\hatX^n|Z^k}$. Let $W\in[T]$ be an auxiliary random variable such that $X^n-(Y^k,W)-Z^k$ forms a Markov chain so that the joint distribution of $(X^n,Y^k,W,Z^k)$ satisfies $P_{X^nY^kWZ^k}=P_{X}^nP_{W|X^n}P_{Y^k|X^nW}P_{Z^k|Y^kW}$ for some distributions $(P_{W|X^n},P_{Y^k|X^nW},P_{Z^k|Y^kW})$ induced by the encoder $P_{Y^k|S}$ and the noisy channel $P_{Z|Y}^k$. Furthermore, let $\barP_{Z^k|W}$ be a conditional distribution and let $\barP_{\hatX^n|W}$ be induced by $\barP_{Y^k|W}$ and $P_{\hatX^n|Z^k}$, i.e., for each $w\in[T]$ and $\hatx^n\in\hatcalX^n$,
\begin{align}
\barP_{\hatX^n|W}(\hatx^n|w)=\sum_{z^k}P_{\hatX^n|Z^k}(\hatx^n|z^k)\barP_{Z^k|W}(z^k|w)\label{defbarPhatXn}.
\end{align}
It follows that 
\begin{align}
\nn&\Pr\Big\{\sum_{i\in[n]}\jmath(X_i|D,P_X)-\imath(Y^k;Z^k|W)\geq\gamma\Big\}\\*
\nn&=\Pr\Big\{\sum_{i\in[n]}\jmath(X_i|D,P_X)-\imath(Y^k;Z^k|W)\geq\gamma,d(X^n,\hatX^n)>D \Big\}\\*
&\qquad+\Pr\Big\{\sum_{i\in[n]}\jmath(X_i|D,P_X)-\imath(Y^k;Z^k|W)\geq\gamma,~d(X^n,\hatX^n)\leq D \Big\}\\
&\leq \rmP_{\rme,n,k}(D)+\Pr\Big\{\sum_{i\in[n]}\jmath(X_i|D,P_X)-\imath(Y^k;Z^k|W)\geq\gamma,~d(X^n,\hatX^n)\leq D \Big\}\label{constep14jscc}.
\end{align}
The second term in \eqref{constep14jscc} can be further upper bounded as follows:
\begin{align}
\nn&\mathrm{second~term~in}\eqref{constep14jscc}\\*
\nn&\leq \sum_{x^n,y^k,w,z^k}\sum_{\hatx^n\in\calB_D(x^n)}P_{X^nY^kWZ^k}(x^n,y^k,w,z^k)P_{\hatX^n|Z^k}(\hatx^n|z^k)\\*
&\qquad\times\bbo\Big(\sum_{i\in[n]}\jmath(x_i|D,P_X)-\gamma\geq\imath(y^k;z^k|w)\Big)\\
\nn&\leq \exp(-\gamma)\sum_{x^n,y^k,z^k}\sum_{w\in[T]}\sum_{\hatx^n\in\calB_D(x^n)}P_X^n(x^n)\exp\Big(\sum_{i\in[n]}\jmath(x_i|D,P_X)\Big)\\
&\qquad\times P_{W|X^n}(w|x^n)P_{Y^k|WX^n}(y^k|w,x^n)\barP_{Z^k|W}(z^k|w)P_{\hatX^n|Z^k}(\hatx^n|z^k)\label{useidcon4jscc}\\
\nn&\leq\exp(-\gamma)\sum_{x^n,z^k}\sum_{w\in[T]}\sum_{\hatx^n\in\calB_D(x^n)}P_X^n(x^n)\exp\Big(\sum_{i\in[n]}\jmath(x_i|D,P_X)\Big)\\*
&\qquad\qquad\times\barP_{Z^k|W}(z^k|w)P_{\hatX^n|Z^k}(\hatx^n|z^k)\label{ineq4jscc}\\
&=\exp(-\gamma)\sum_{x^n}\sum_{w\in[T]}\sum_{\hatx^n\in\calB_D(x^n)}P_X^n(x^n)\barP_{\hatX^n|W}(\hatx^n|w)\exp\Big(\sum_{i\in[n]}\jmath(x_i|D,P_X)\Big)\label{usebarPhatXn}\\
\nn&\leq \exp(-\gamma)\sum_{x^n}\sum_{w\in[T]}\sum_{\hatx^n\in\calB_D(x^n)}P_X^n(x^n)\barP_{\hatX^n|W}(\hatx^n|w)\\*
&\qquad\times\exp\Big(\sum_{i\in[n]}\jmath(x_i|D,P_X)+n\lambda^*(D-d(x^n,\hatx^n))\Big)\label{uselambda*}\\
\nn&\leq \exp(-\gamma)\sum_{x^n}\sum_{w\in[T]}\sum_{\hatx^n}P_X^n(x^n)\barP_{\hatX^n|W}(\hatx^n|w)\\*
&\qquad\times\exp\Big(\sum_{i\in[n]}\jmath(x_i|D,P_X)+n\lambda^*(D-d(x^n,\hatx^n))\Big)\\
&=\exp(-\gamma)\sum_{\hatx^n,w}\barP_{\hatX^n|W}(\hatx^n|w)\bbE_{P_X^n}\Big[\exp\Big(\sum_{i\in[n]}\big(\jmath(X_i|D,P_X)+\lambda^*(D-d(X_i,\hatx_i))\big)\Big)\Big]\\
&\leq \exp(-\gamma)\sum_{w\in[T]}\sum_{\hatx^n}\barP_{\hatX^n|W}(\hatx^n|w)\label{usejxdprop4jscc}\\
&\leq T\exp(-\gamma)\label{constep24jscc},
\end{align}
where \eqref{useidcon4jscc} follows from the definition of $\imath(y^k;z^k|w)$ in \eqref{infdensity4jscc} which implies an upper bound on $P_{Z^k|Y^kW}(z^k|y^k,w)$, \eqref{ineq4jscc} follows since $P_{W|X^n}(w|x^n)\leq 1$ and $\sum_{y^k}P_{Y^k|WX^n}(y^k|w,x^n)=1$, \eqref{usebarPhatXn} follows from the definition of $\barP_{\hatX^n|W}$ in \eqref{defbarPhatXn}, \eqref{uselambda*} follows since $d(x^n,\hatx^n)\leq D$ for $\hatx^n\in\calB_D(\hatx^n)$ and $\lambda^*$ defined in \eqref{def:lambda^*} is non-negative and \eqref{usejxdprop4jscc} follows from Claim (iii) of Lemma \ref{prop:dtilteddensity} similarly to \eqref{explain4} for the rate-distortion problem.

The proof of \eqref{conform14jscc} is completed by combining \eqref{constep14jscc} and \eqref{constep24jscc} and optimize over different parameters to obtain the tightest bound that does not depend on the code design. The proof of \eqref{conform24jscc} follows from \eqref{conform14jscc} by changing the order of infimum and supremum and applying the algebra in~\cite[Eq. (48)-(50)]{kostinajscc}.
\end{proof}

\section{Second-Order Asymptotics}
Applying the Berry-Esseen theorem to the non-asymptotic bounds in Corollary \ref{achcoro4jscc} and Theorem \ref{fblcon4jscc}, we obtain second-order asymptotic approximation to the finite blocklength performance of optimal codes.

Recall the definition of the dispersion function $\rmV(P_X,D)$ for the rate-distortion problem. Furthermore, for any $P_{Z|Y}$, let $\calP_Y^*$ be the input distribution that achieves $C(P_{Y|X})$ and define the channel dispersion function
\begin{align}
\rmV_\rmc(P_{Z|Y})&:=\mathrm{Var}\bigg[\log\frac{P_{Z|Y}(Z|Y)}{P_{\hatY}^*(Y)}\bigg].
\end{align}
Under the same condition above Theorem \ref{second4rd} for the rate-distortion problem, the following theorem holds~\cite[Theorem 10]{kostinajscc}.
\begin{theorem}
\label{second4jscc}
For any $\varepsilon\in(0,1)$, the optimal $(n,k)$-code satisfies
\begin{align}
kC(P_{Z|Y})-nR(P_X,D)
&=\sqrt{k\rmV_\rmc(P_{Z|Y})+n\rmV(P_X,D)}\rmQ^{-1}(\varepsilon)+O(\log k).
\end{align}
\end{theorem}
As discussed in \cite[Remark 9]{kostinajscc}, an equivalent form of Theorem \ref{second4jscc} is
\begin{align}
\frac{n^*(k,D,\varepsilon)}{k}
&=\rho(P_X,P_{Z|Y},D)+\frac{L(P_X,P_{Z|Y},D|\varepsilon)}{\sqrt{k}}+O(\log k),
\end{align}
where the second-order coding rate $L(P_X,P_{Z|Y},D|\varepsilon)$ for the lossy JSCC is
\begin{align}
L(P_X,P_{Z|Y},D|\varepsilon)
&:=-\frac{\sqrt{\rmV_\rmc(P_{Z|Y})+\rho(P_X,P_{Z|Y},D)\rmV(P_X,D)}\rmQ^{-1}(\varepsilon)}{R(P_X,D)}.
\end{align}

In this monograph, we only consider the discrete memoryless source and the discrete memoryless channel with no cost constraint. The proof of Theorem \ref{second4jscc} is omitted due to its similarity to the rate-distortion problem, readers could refer to \cite[Appendix D-B]{kostinajscc} for the achievability proof using Corollary \ref{achcoro4jscc} and \cite[Appendix C]{kostinajscc} for the converse proof using the non-asymptotic bound in Theorem \ref{fblcon4jscc}. An independent proof of Theorem \ref{second4jscc} was provided by Wang, Ingber and Kochman using the method of types and the idea of unequal error protection~\cite{wang2011dispersion}.
For other types of source and channels, e.g., a Gaussian source and an AWGN channel, one could refer to \cite[Section V]{kostinajscc} for details. Generally speaking, the same conclusion holds but the dispersion function and the remainder $O(\log n)$ term differ.

A critical question is the cost of separation in the second-order asymptotics. As discussed in \cite[Remark 8]{kostinajscc}, combining the second-order asymptotics for the rate-distortion problem in~\cite{kostina2012fixed} (see also Theorem \ref{second4rd}) and the channel coding problem in~\cite[Theorem 49]{polyanskiy2010finite}, the achievable performance of separate source-channel coding satisfies
\begin{align}{}
\nn&kC(P_{Z|Y})-nR(P_X,D)\\*
\nn&\leq \min_{\substack{(\varepsilon_1,\varepsilon_2)\in(0,1)^2:\\\varepsilon_1+\varepsilon_2\leq \varepsilon}}
\Big\{\sqrt{k\rmV_\rmc(P_{Z|Y})}\rmQ^{-1}(\varepsilon_1)+\sqrt{n\rmV(P_X,D)}\rmQ^{-1}(\varepsilon_2)\Big\}\\*
&\qquad+O(\log k).
\end{align}
In other words,
\begin{align}
\frac{n_{\rm{sscc}}^*(k,D,\varepsilon)}{k}
&\leq\frac{C(P_{Z|Y})}{R(P_X,D)}+\frac{L_{\rm{sscc}}(P_X,P_{Z|Y},D|\varepsilon)}{\sqrt{k}}+O(\log k),
\end{align}
where the achievable second-order coding rate $\rmV_{\rm{sscc}}(P_X,P_{Z|Y},D,\varepsilon)$ for SSCC is
\begin{align}
L_{\rm{sscc}}(P_X,P_{Z|Y},D|\varepsilon)
\nn&:=\max_{\substack{(\varepsilon_1,\varepsilon_2)\in(0,1)^2:\\\varepsilon_1+\varepsilon_2\leq \varepsilon}}\Bigg(\frac{-\sqrt{\rho(P_X,P_{Z|Y},D)\rmV(P_X,D)}\rmQ^{-1}(\varepsilon_1)}{R(P_X,D)}\\*
&\qquad\qquad\qquad\qquad-\frac{\sqrt{\rmV_\rmc(P_{Z|Y})}\rmQ^{-1}(\varepsilon_2)}{R(P_X,D)}\Bigg).
\end{align}
Note that $L_{\rm{sscc}}(P_X,P_{Z|Y},D|\varepsilon)\leq L(P_X,P_{Z|Y},D|\varepsilon)$ unless the source dispersion function $\rmV(P_X,D)=0$ or the channel dispersion function $\rmV(P_{Z|Y})=0$. To illustrate this point, in Fig. \ref{illus:costsscc}, we plot the second-order coding rates for a Bernoulli source with parameter $p$, the Hamming distortion measure and a distortion level $D=0.05$ and a binary symmetric channel $P_{Z|Y}$ with parameter $q=0.1$, i.e., $\calX=\calY=\calZ=\{0,1\}$, $P_{X}(1)=p$, $P_X(0)=1-p$, and
\begin{align}
P_{Z|Y}(z|y)=(1-q)\bbo(z=y)+q\bbo(z\neq y).
\end{align}
\begin{figure}[bt]
\centering
\includegraphics[width=.8\columnwidth]{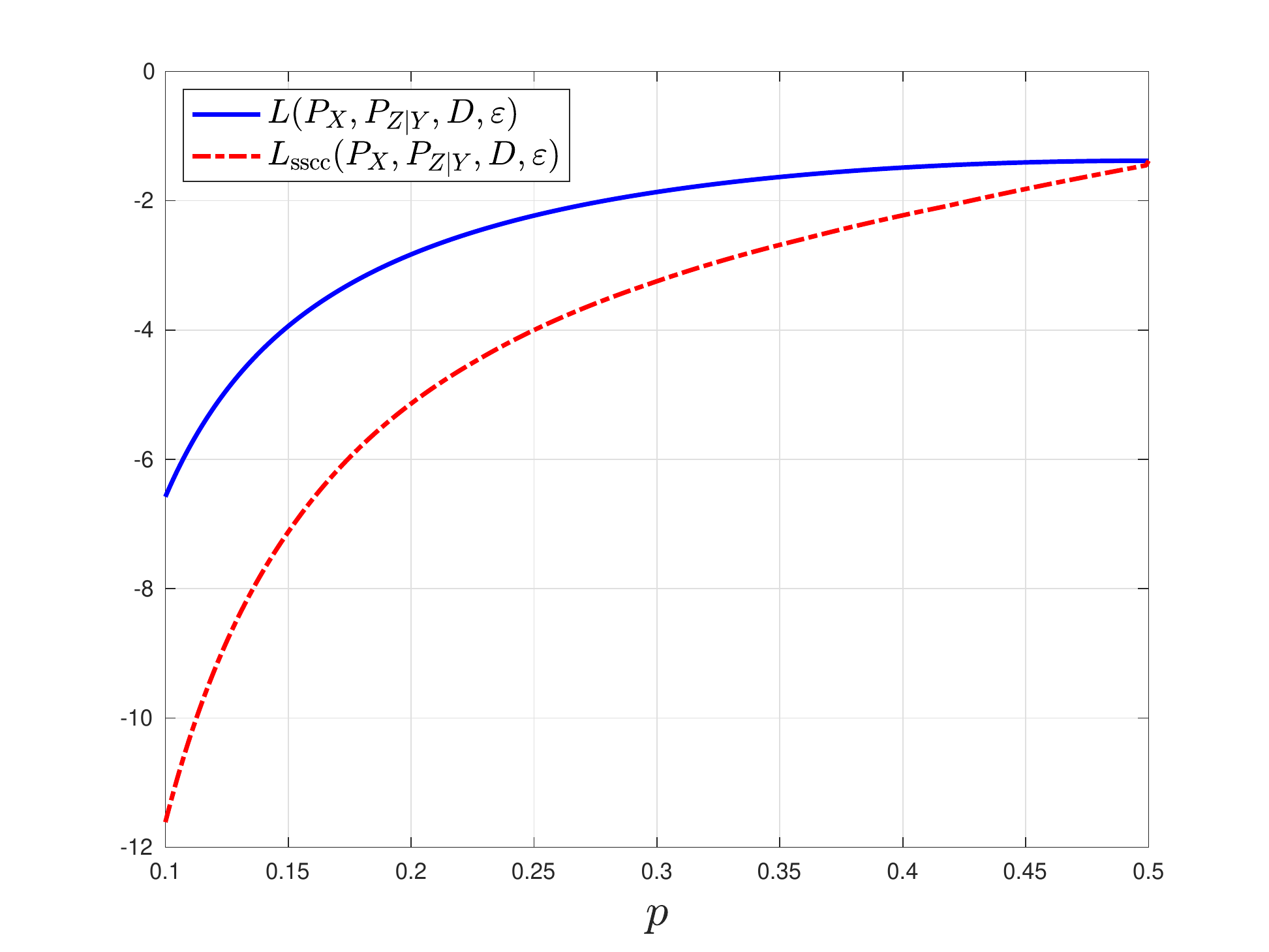}
\caption{Plots of second-order coding rates for JSCC $L(P_X,P_{Z|Y},D|\varepsilon)$ and for SSCC $L_{\rm{sscc}}(P_X,P_{Z|Y},D|\varepsilon)$ when $D=0.05$ for a Bernoulli source with parameter $p$, a binary symmetric channel $P_{Z|Y}$ with parameter $q=0.1$ under the Hamming distortion measure.}
\label{illus:costsscc}
\end{figure}


\chapter{Mismatch}
\label{chap:mismatch}

This chapter concerns the mismatched rate-distortion aspect that tackles a practical problem of lossy data compression: the generating distribution of the information source is \emph{unknown}. Although the results in previous chapters are very insightful, the assumption of the exact knowledge of the source distribution $P_X$ is highly impractical. To tackle this problem, Lapidoth~\cite[Theorem 3]{lapidoth1997} proposed to use the spherical codebook and minimum Euclidean distance encoding to compress an arbitrary memoryless source under the quadratic distortion measure. Note that such a codebook is optimal for the Gaussian memoryless source under the quadratic distortion measure, asymptotically~\cite{shannon1959coding} and second-order asymptotically~\cite[Theorem 40]{kostina2012fixed}.

Lapidoth showed that for any ergodic source with known and finite second moment $\sigma^2$, the rate-distortion function for the GMS $\calN(0,\sigma^2)$ is achievable and ensemble tight as the blocklength tends to infinity. Here ensemble tight means that the analysis of the code is optimal. Lapidoth's codebook design only requires the knowledge of the second-moment of the information source, which is much more accessible than the exact source distribution and can be estimated from observed source sequence. 

The results of Lapidoth were refined by Zhou, Tan and Motani~\cite{zhou2017refined} who derived ensemble tight second-order asymptotics and also considered the i.i.d. Gaussian codebook. Specifically, the authors of \cite[Theorem 1]{zhou2017refined} showed that both i.i.d. Gaussian and spherical codebooks achieve the same second-order coding rate. This chapter is largely based on \cite{zhou2017refined}.

\section{Problem Formulation and Asymptotic Result}

Consider a memoryless source $X^n$ with distribution (either probability mass function or probability density function) $P_X$  satisfying
\begin{align}
\bbE [X^2]=\sigma^2,~\zeta:=\bbE [X^4]<\infty,~\bbE[X^6]<\infty\label{sourceconstraint},
\end{align} 
Consider any distortion level $D\in(0,\sigma^2)$. Let $\hatcalX$ be the reproduction alphabet and let $d:\calX\times\hatcalX\to\bbR_+$ be the quadratic distortion defined in \eqref{def:quadratic}. Given any source sequence $x^n$ and the reproduced source sequence $\hatx^n$, let $d(x^n,\hatx^n)=\frac{1}{n}\sum_{i\in[n]}d(x_i,\hatx_i)$, i.e.,
\begin{align}
d(x^n,\hatx^n)
&=\frac{1}{n}\sum_{i\in[n]}(\hatx_i-x_i)^2\\
&=\frac{1}{n}\|\hatx^n-x^n\|^2.
\end{align}
Lapidoth's coding scheme is as follows~\cite[Theorem 3]{lapidoth1997}.
\begin{definition}
An $(n,M)$-code for the mismatched rate-distortion problem consists of 
\begin{itemize}
\item A set of $M$ codewords $\{\hatX^n(i)\}_{i=1}^M$ known by both the encoder and decoder;
\item An encoder $f$ which maps the source sequence $X^n$ into the index of the codeword that minimizes the quadratic distortion with respect to the source sequence $X^n$, i.e.,
\begin{align}
f(X^n)
&:=\argmin_{i\in[1:M]} d\big(X^n,\hatX^n(i) \big).
\end{align}
\item A decoder $\phi$ which declares the reproduced sequence as the codeword with index $f(X^n)$, i.e.,
\begin{align}
\phi(f(X^n))=\hatX^n(f(X^n)).
\end{align}
\end{itemize}
\end{definition}

We consider two following types of codebooks $\{\hatX^n(i)\}_{i=1}^M$. 
\begin{itemize}
\item  First, we consider the {\em spherical codebook} where each codeword $\hatX^n$ is generated independently and uniformly over a sphere with radius $\sqrt{n(\sigma^2-D)}$, i.e.,
\begin{align}
\hatX^n\sim f_{\hatX^n}^{\rm{sp}}(\hatx^n)=\frac{1\{\|\hatx^n\|^2-n(\sigma^2-D)\}}{S_n(\sqrt{n(\sigma^2-D) })}\label{spherecodebook},
\end{align}
where $1\{\cdot\}$ is the indicator function, $S_n(r)= {n\pi^{n/2}} r^{n-1} / \Gamma(\frac{n+2}{2})$ is the surface area of an $n$-dimensional sphere with radius $r$, and $\Gamma(\cdot)$ is the Gamma function. 
\item  Second, we consider the {\em  i.i.d.\  Gaussian codebook} where each codeword $\hatX^n$ is generated independently from the product Gaussian distribution with mean $0$ and variance $\sigma^2-D$, i.e.,
\begin{align}
\hatX^n\sim f_{\hatX^n}^{\rm{iid}}(\hatx^n)
=\frac{\exp\Big(-\frac{\|\hatx^n\|^2}{2(\sigma^2-D)}\Big)}{\big(\sqrt{2\pi(\sigma^2-D)}\big)^n}
\label{iidcodebook}.
\end{align}
\end{itemize}

To evaluate the performance of the above code, we consider the following {\em ensemble excess-distortion probability} with $M$ codewords:
\begin{align}
\rmP_{\rme,n}(M)
&:=\Pr\{d(X^n,\phi(f(X^n)))>D\}\label{def:excessdp}\\
&=\bbE_{P_X^n} \left[\big(1-\Pr\{d(X^n,\hatX^n)\leq D|\,X^n\}\big)^M\right]\label{peforboth},
\end{align}
where \eqref{peforboth} follows from \cite[Theorem 9]{kostina2012fixed} and the inner probability is calculated either with respect to the right hand side of~\eqref{spherecodebook} if we use a spherical codebook or the right hand side of \eqref{iidcodebook} if we use an i.i.d.\  Gaussian codebook. Note that the probability in \eqref{def:excessdp} is averaged over the source distribution $P_X^n$ {\em as well as the distribution random codebook $\{\hatX^n(i)\}_{i=1}^M$}. This is in contrast to the rate-distortion problem and its generalizations in previous chapters, where the excess-distortion probability is averaged over the source distribution only. The additional average over the codebook enables one to pose questions concerning ensemble tightness in the spirit of~\cite{gallager_ensemble,scarlett2017mismatch}, which guarantees the optimality of the analysis of the code performance.

Analogous to \eqref{def:M*} for the rate-distortion problem, we next define the non-asymptotic fundamental limit. Let $M_{\rm{sp}}^*(n,\varepsilon,\sigma^2,D)$ be the minimum number of codewords required to compress a length-$n$ source sequence so that the ensemble excess-distortion probability with respect to the distortion level $D$ is no larger than $\varepsilon\in(0,1)$ when a spherical codebook is used, i.e.,
\begin{align}
\nn&M_{\rm{sp}}^*(n,\varepsilon,\sigma^2,D)\\*
&:=\inf\{M:~\exists~(n,M)~\mathrm{code~using~a~SC~s.t.}~\rmP_{\rme,n}\leq \varepsilon\}\label{def:Msp*},
\end{align}
where $\mathrm{SC}$ is short for spherical codebook. Similarly, let $M_{\rm{iid}}^*(n,\varepsilon,\sigma^2,D)$ be the corresponding quantity when an i.i.d.\  Gaussian codebook is used. 

Lapidoth~\cite[Theorem 3]{lapidoth1997} showed that for any ergodic source with finite second moment $\sigma^2$ and any $\varepsilon\in(0,1)$, 
\begin{align}
\lim_{n\to\infty}\frac{1}{n}\log M_{\rm{sp}}^*(n,\varepsilon,\sigma^2,D)=\frac{1}{2}\log\frac{\sigma^2}{D}\label{first:sp}.
\end{align}
As a by-product of the second-order asymptotics presented in Theorem~\ref{secondorder4mismatch} below, for any source satisfying \eqref{sourceconstraint} and any $\varepsilon\in(0,1)$, 
\begin{align}
\lim_{n\to\infty}\frac{1}{n}\log M_{\rm{iid}}^*(n,\varepsilon,\sigma^2,D)=\frac{1}{2}\log\frac{\sigma^2}{D}\label{iidfotrue}.
\end{align}
Thus, asymptotically, both spherical and i.i.d. Gaussian codebooks achieve the same first-order asymptotic performance. When specialized to a GMS with distribution $P_X=\calN(0,\sigma^2)$, the asymptotic result was first established by Shannon~\cite[Page 346]{shannon1959coding}.

\section{Second-Order Asymptotics}
Recall the definitions of $\sigma^2$ and $\zeta$ in \eqref{sourceconstraint}. Let the {\em mismatched dispersion} be defined as
\begin{align}
\rmV(\sigma^2,\zeta)&:=\frac{\zeta-\sigma^4}{4\,\sigma^4} =  \frac{{\rm{Var}}[X^2]}{ 4 \, (  \bbE[X^2]^2) }.\label{def:rmvsphere}
\end{align}
The second-order asymptotic result states as follows.
\begin{theorem}
\label{secondorder4mismatch}
Consider an arbitrary memoryless source $X$ satisfying \eqref{sourceconstraint}. For any $\varepsilon\in [0,1)$ and any $\dagger\in\rm\{sp,iid\}$,
\begin{align}
\log M_{\dagger}^*(n,\varepsilon,\sigma^2,D)
\nn&=\frac{n}{2}\log\frac{\sigma^2}{D}+\sqrt{n\rmV(\sigma^2,\zeta)}\rmQ^{-1}(\varepsilon)\\*
&\qquad+O(\log n)\label{spheresecond}.
\end{align}
\end{theorem} 
Theorem \ref{secondorder4mismatch} generalizes Lapidoth's result~\cite[Theorem 3]{lapidoth1997} in two aspects. Firstly, Theorem \ref{secondorder4mismatch} establishes a second-order asymptotic approximation to the finite blocklength performance and thus refines the first-order asymptotic result valid only for infinite blocklength. Secondly, Theorem \ref{secondorder4mismatch} holds for both spherical and i.i.d. Gaussian codebooks and shows that both codebooks achieve the same first and second-order coding rates. When specialized to the GMS with $P_X=\calN(0,\sigma^2)$, $\rmV(\sigma^2,\zeta)=\frac{1}{2}$ and equals the dispersion established in~\cite{kostina2012fixed,ingber2011dispersion}. Thus, the proof of Theorem \ref{secondorder4mismatch} provides an alternative second-order achievability proof for a GMS.

The proof of Theorem \ref{secondorder4mismatch} differs significantly from Lapidoth's analyses. Lapidoth~\cite[Theorem 3]{lapidoth1997} used a theorem of Wyner~\cite{Wyner67} concerning packings and coverings of $n$-spheres to derive the first-order asymptotic result. The proof of Theorem \ref{secondorder4mismatch} requires finer analyses as demonstrated in Section \ref{sec:proof4mismatch}. Specifically, the dominant error event is the atypicality of the source sequence regardless which codebook ensemble is used. Intuitively, it is sufficient to use roughly $ \exp(\frac{n}{2}\log\frac{\sigma^2}{D})$ codewords to cover the set of typical source sequences decaying super-exponentially.  Theorem \ref{secondorder4mismatch} then follows by judiciously analyzing the probability of the set of atypical source sequences with appropriate choices of the minimum number of codewords.

Dual to the problem of the mismatched rate-distortion problem considered in this chapter, Lapidoth also considered a mismatched channel coding problem~\cite{lapidoth1996}. Specifically, Lapidoth proposed to use i.i.d. Gaussian or the spherical codebook and nearest neighbor decoding to transmit a message over an additive noise channel where the distribution of the noise is unknown. Lapidoth derived the ensemble tight first-order asymptotic result and showed that both codebooks achieve the asymptotic rate of the capacity of an AWGN channel while Scarlett, Tan and Durisi~\cite{scarlett2017mismatch} derived the second-order approximation to the finite blocklength performance. In particular, the authors of \cite{scarlett2017mismatch} showed that the spherical codebook achieves larger second-order coding rate, consistent with the analysis for AWGN channels~\cite{tantomamichel2015,polyanskiy2010finite}.

Analogous to the noisy channel problem considered in Chapter \ref{chap:jscc}, a mismatched version of the noisy channel problem was studied in~\cite{zhou2019jscc} where the authors proposed a coding scheme based on the unequal error protection idea~\cite{wang2011dispersion} to transmit a memoryless source with unknown distribution over an additive noise channel with unknown noise distribution and derived ensemble tight second-order asymptotics. We provide some intuition why two codebooks lead to different behavior in the mismatched rate-distortion and channel coding problem. In channel coding, every codeword in the codebook is used to transmit a uniformly distributed message. If some codewords have powers that deviate from a fixed power $P$, the overall performance will be adversely affected. However, in the rate-distortion problem, to compress each source sequence, we select only the codeword which minimizes the distortion with respect to the source sequence. As a result, even if there are many codewords with power bounded away from $\sigma^2-D$, the performance is unaffected.

In \cite[Section II]{lapidoth1997}, Lapidoth also considered another type of mismatch for the rate-distortion problem where the encoder and decoder use different distortion measures. However, no tight characterization of the rate-distortion function was available. Readers could refer to \cite[Chapter 4]{scarlett2020book} for detailed discussions.

\section{Proof of Second-Order Asymptotics}
\label{sec:proof4mismatch}

\subsection{Preliminaries for the Spherical Codebook}
\label{sec:presphere}
In this subsection, we present some definitions and preliminary results for the spherical codebook. For simplicity, let $P_\rmc:=\sigma^2-D$. Furthermore, for any $\varepsilon\in(0,1)$, let 
\begin{align}
\rmV&:=\mathrm{Var}[X^2]=\zeta-\sigma^4\label{def:rmV},\\
a_n&:=\sqrt{\rmV\frac{\log n}{n}},\label{def:an}\\
b_n&:=\sqrt{\frac{\rmV}{n}}\rmQ^{-1}(\varepsilon)\label{def:bn},
\end{align}
where the second equality in \eqref{def:rmV} follows from the definition in \eqref{sourceconstraint}. Note that for any $x^n$, $\Pr\{d(x^n,\hatX^n)\leq D\}$ depends on $x^n$ only through its norm $\|x^n\|$. For any $x^n$ such that $\frac{1}{n}  {\|x^n\|^2}=z>0$, let 
\begin{align}
\Psi(n,z)&:=\Pr\{d(x^n,\hatX^n)\leq D\}\\
&=\Pr\{\|x^n-\hatX^n\|^2\leq nD\}\\
&=\Pr\{\|x^n\|^2+\|\hatX^n\|^2-2\langle x^n,\hatX^n \rangle\leq nD\}\label{ensemble}\\
&=\Pr\{nz+nP_\rmc-2\langle x^n,\hatX^n\rangle \leq nD\}\\
&=\Pr\{2\langle x^n,\hatX^n\rangle \geq n(z+P_\rmc-D)\}\\
&=\Pr\bigg\{\hatX_1\geq \frac{\sqrt{n}(z+P_\rmc-D)}{2\sqrt{z}}\bigg\}\label{y1symm},
\end{align}
where $\hatX_1$ is the first element of sequence $\hatX^n = (\hatX_1,\ldots, Y_n)$ and \eqref{y1symm} follows because $\hatX^n$ is spherically symmetric so we may take $x^n =( \sqrt{nz},0,\ldots, 0)$ (cf.~\cite{scarlett2017mismatch}).

Let $Z:=\frac{1}{n}{\|X^n\|^2}$ be the random variable representing the average power of the source  $X^n$. Furthermore, let $f_Z$ be the corresponding probability distribution function (pdf) of $Z$. Let
\begin{align}
r_1&:=\sqrt{P_\rmc}-\sqrt{D}\label{def:r1},\\
r_2&:=\sqrt{P_\rmc}+\sqrt{D}\label{def:r2}.
\end{align}

Kostina and Verd\'u~\cite[Theorem 37]{kostina2012fixed} showed that for any $z$ such that $\sqrt{z}<r_1$ or $\sqrt{z}>r_2$,
\begin{align}
\Psi(n,z)=0,\label{psi=1foratp}
\end{align}
and otherwise
\begin{align}
\Psi(n,z)
&\geq \frac{\Gamma(\frac{n+2}{2})}{\sqrt{\pi} n\Gamma(\frac{n+1}{2})}\bigg(1-\frac{(z+P_\rmc-D)^2}{4zP_\rmc}\bigg)^{\frac{n-1}{2}}=:\underline{g}(n,z),\label{def:undergnz}
\end{align}
where $\Gamma(\cdot)$ is the Gamma function. Hence, from \eqref{peforboth} and \eqref{psi=1foratp}, we conclude that the excess-distortion probability for the spherical codebook is
\begin{align}
\rmP_{\rme,n}(M)
\nn&=\Pr\{Z<(\max\{0,r_1\})^2\}+\Pr\{Z>r_2^2\}\\*
&\qquad+\int_{(\max\{0,r_1\})^2}^{r_2^2}\!\!\!(1-\Psi(n,z))^M f_Z(z)\, \rmd z\label{pe4sphere}.
\end{align}

\subsection{Achievability Proof for the Spherical Codebook}
\label{sec:sphereach}

Using the definition of $\underline{g}(\cdot)$ in \eqref{def:undergnz}, we conclude that $\underline{g}(n,z)$ is a decreasing function of $z$ if $z\geq |P_\rmc-D|$. Invoking the definitions of $b_n$ in \eqref{def:bn}, $r_1$ in \eqref{def:r1} and $r_2$ in \eqref{def:r2}, we conclude that $r_1^2\leq |P_\rmc-D|$ and $r_2^2\geq \sigma^2+b_n$ for $n$ large enough. Thus, combining \eqref{def:undergnz}, \eqref{pe4sphere} and noting that $\Psi(n,z)\geq 0$, for sufficiently large $n$, we can upper bound the excess-distortion probability as follows:
\begin{align}
\rmP_{\rme,n}(M)
\nn&\leq \Pr\{Z<|P_\rmc-D|\}+\Pr\{Z>\sigma^2+b_n\}\\*
&\qquad+\int_{|P_\rmc-D|}^{\sigma^2+b_n}(1-\underline{g}(n,z))^M f_Z(z)\, \rmd z\label{r2c4}\\
\nn&\leq \Pr\{Z<|P_\rmc-D|\}+\Pr\{Z>\sigma^2+b_n\}\\*
&\qquad+\int_{|P_\rmc-D|}^{\sigma^2+b_n}\exp\{-M\underline{g}(n,z)\}f_Z(z)\, \rmd z\label{upplog1-x}\\
&\leq \Pr\{Z<|P_\rmc-D|\}+\Pr\{Z>\sigma^2+b_n\}+\exp\{-M\underline{g}(n,\sigma^2+b_n)\}\label{usenondecreaseug},
\end{align}
where \eqref{upplog1-x} follows since $(1-a)^M\leq \exp\{-Ma\}$ for any $a\in[0,1)$, and \eqref{usenondecreaseug} follows since $\underline{g}(n,z)$ is decreasing in $z$ for $z\geq |P_\rmc-D|$. Let the third central moment of $X^2$ be defined as
\begin{align}
T:=\bbE\big[|X^2-\sigma^2|^3\big]\label{def:T}.
\end{align}
Using the definitions of $\rmV$ in \eqref{def:rmV}, $T$ in \eqref{def:T} and the Berry-Esseen theorem (cf. Theorem \ref{berrytheorem}), we conclude that
\begin{align}
\Pr\{Z<|P_\rmc-D|\}
&=\Pr\bigg\{\frac{1}{n}\sum_{i=1}^n X_i^2<|\sigma^2-2D|\bigg\}\\
&\leq \frac{6T}{\sqrt{n}\, \rmV^{3/2}}+\rmQ\bigg( \big(\sigma^2-|\sigma^2-2D|\big)\sqrt{\frac{n}{\rmV}}\bigg)\\
&\leq \frac{6T}{\sqrt{n}\, \rmV^{3/2}}+\exp\bigg\{-\frac{2n\big(\sigma^2-|\sigma^2-2D|\big)^2}{\rmV}\bigg\}\label{upprmq2}\\
&=O\bigg( \frac{1}{\sqrt{n}}\bigg),\label{achstep1}
\end{align}
where \eqref{upprmq2} follows since $\rmQ(a)\leq \exp\{-\frac{a^2}{2}\}$ while \eqref{achstep1} follows since $T$ (cf. \eqref{def:T}) is finite for sources satisfying \eqref{sourceconstraint} and $\sigma^2-|\sigma^2-2D|>0$ due to the fact that $\sigma^2>D$. Similarly, using the definition of $b_n$ in \eqref{def:bn} and the Berry-Esseen theorem, we have
\begin{align}
\Pr\{Z>\sigma^2+b_n\}
&=\Pr\bigg\{\frac{1}{n}\sum_{i=1}^nX_i^2>\sigma^2+b_n\bigg\}\\
&\leq \varepsilon+\frac{6T}{\sqrt{n}\, \rmV^{3/2}}\\
&=\varepsilon+O\bigg( \frac{1}{\sqrt{n}}\bigg)\label{achstep2}.
\end{align}

Choose $M$ such that 
\begin{align}
\log M
&=-\log \underline{g}(n,\sigma^2+b_n)+\log\bigg( \frac{1}{2}\log n\bigg)\label{hi:chosem}\\
&=n\bigg( \frac{1}{2}\log \frac{\sigma^2}{D}+\frac{b_n}{2\sigma^2}+O\bigg( \frac{\log n}{n}\bigg)\bigg)\label{tayloragain}\\
&=\frac{n}{2}\log \frac{\sigma^2}{D}+\sqrt{n\rmV(\sigma^2,\zeta)}\rmQ^{-1}(\varepsilon)+O(\log n)\label{usebnha},
\end{align}
where \eqref{tayloragain} follows from the Taylor expansion of $\underline{g}(n,\sigma^2+b_n)$ (cf. \eqref{def:undergnz}) and noting that $ {\Gamma(\frac{n+2}{2})}/{\Gamma(\frac{n+1}{2})}=\Theta(\sqrt{n})$, and \eqref{usebnha} follows from the definition of $b_n$ (cf. \eqref{def:bn}) and $\rmV(\sigma^2,D)$ (cf. \eqref{def:rmvsphere}). Thus, with the choice of $M$ in \eqref{hi:chosem}, we conclude that
\begin{align}
\exp\{-M\underline{g}(n,\sigma^2+b_n)\}&=\frac{1}{\sqrt{n}}\label{achstep3}.
\end{align}
Hence, combining \eqref{usenondecreaseug}, \eqref{achstep1}, \eqref{achstep2}, \eqref{usebnha} and \eqref{achstep3}, we have shown  that
\begin{align}
\log M_{\rm{sp}}^*(n,\varepsilon,\sigma^2,D)
&\geq \frac{n}{2}\log \frac{\sigma^2}{D}+\sqrt{n\rmV(\sigma^2,\zeta)}\rmQ^{-1}(\varepsilon)+O(\log n).
\end{align}

\subsection{Ensemble Converse for the Spherical Codebook}
\label{sec:spheresecconverse}
We now show that the result in \eqref{spheresecond} is ensemble tight. From Stam's paper \cite[Eq. (4)]{stam1982limit}, the distribution of $\hatX_1$ is 
\begin{align}
f_{\hatX_1}(\hatx)=\frac{1}{\sqrt{\pi n P_\rmc}}\frac{\Gamma(\frac{n}{2})}{\Gamma(\frac{n-1}{2})}\bigg( 1-\frac{\hatx^2}{nP_\rmc}\bigg)^{\frac{n-3}{2}}1\{\hatx^2\leq nP_\rmc\}\label{pdfy1}.
\end{align}
Recall the definitions of $a_n$ in \eqref{def:an} and $b_n$ in \eqref{def:bn}. Define the sets 
\begin{align}
\calP&:=\{r\in\bbR:b_n<r-\sigma^2\leq a_n\}\label{def:calp},\\
\calQ&:=\{r\in\bbR:r+P_\rmc-D\geq 0\}\label{def:calq}.
\end{align}
Then, for any $z\in\calP\cap\calQ$ satisfying $\frac{\sqrt{n}(z+P_\rmc-D)}{2\sqrt{z}}\leq \sqrt{nP_\rmc}$, using the definition of $\Psi(\cdot)$ in \eqref{y1symm}, we obtain that
\begin{align}
\Psi(n,z)
&=\Pr\bigg\{\hatX_1\geq \frac{\sqrt{n}(z+P_\rmc-D)}{2\sqrt{z}}\bigg\}\\
&= \int_{\frac{\sqrt{n}(z+P_\rmc-D)}{2\sqrt{z}}}^{\sqrt{nP_\rmc}}\frac{1}{\sqrt{\pi n P_\rmc}}\frac{\Gamma(\frac{n}{2})}{\Gamma(\frac{n-1}{2})}\bigg(1-\frac{\hatx^2}{nP_\rmc}\bigg)^{\frac{n-3}{2}}\rmd\hatx\label{usepdfy1}\\
&\leq \int_{\frac{\sqrt{n}(z+P_\rmc-D)}{2\sqrt{z}}}^{\sqrt{nP_\rmc}}\frac{1}{\sqrt{\pi n P_\rmc}}\frac{\Gamma(\frac{n}{2})}{\Gamma(\frac{n-1}{2})}\bigg(1-\frac{(z+P_\rmc-D)^2}{4zP_\rmc}\bigg)^{\frac{n-3}{2}}\rmd\hatx\label{nonincreasing}\\
&\leq \frac{1}{\sqrt{\pi}}\frac{\Gamma(\frac{n}{2})}{\Gamma(\frac{n-1}{2})}\bigg(1-\frac{(z+P_\rmc-D)^2}{4zP_\rmc}\bigg)^{\frac{n-3}{2}}\label{enlargeinterval}\\
&=\frac{1}{\sqrt{\pi}}\frac{\Gamma(\frac{n}{2})}{\Gamma(\frac{n-1}{2})}\exp\!\bigg\{\!\frac{n-3}{2}\log \!\bigg(1\!-\!\frac{(z\!+P_\rmc\!-D)^2}{4zP_\rmc}\bigg)\!\bigg\}\\
&=:\overline{g}(n,z)\label{def:overlineg},
\end{align}
where \eqref{usepdfy1} follows from the definition in \eqref{pdfy1} and the condition that $z\in\calQ$ (cf. \eqref{def:calq}) which implies $\frac{\sqrt{n}(z+P_\rmc-D)}{2\sqrt{z}}\geq 0>-\sqrt{nP_\rmc}$, \eqref{nonincreasing} follows since $(1-\frac{\hatx^2}{nP_\rmc})$ is decreasing in $y$ for positive $y$, and \eqref{enlargeinterval} follows by enlarging the integration region (recall that $\frac{\sqrt{n}(z+P_\rmc-D)}{2\sqrt{z}}\geq 0$). Note that $\overline{g}(n,z)$ is decreasing in $z$ for $z\geq |P_\rmc-D|$ and $\overline{g}(n,z)\geq 0$ for all $z\in\calP$. Hence, for any $z\in\calP\cap\calQ$ such that $\frac{\sqrt{n}(z+P_\rmc-D)}{2\sqrt{z}}>\sqrt{nP_\rmc}$, we still have $\overline{g}(n,z)\geq \Psi(n,z)$. 

Recall that $Z=\frac{1}{n} {\|X^n\|^2}$ and $f_Z$ is the corresponding pdf of $Z$. Thus, according to \eqref{peforboth}, for $n$ sufficiently large, we have 
\begin{align}
\rmP_{\rme,n}(M)
&=\bbE_{X^n}[(1-\Pr\{d(X^n,\hatX^n)\leq D|X^n\})^M]\label{fromhere}\\
&=\int_0^\infty (1-\Psi(n,z))^M f_Z(z)\, \rmd z\label{usepsidef}\\
&\geq \int_0^\infty (1-\overline{g}(n,z))^M 1\{z\in\calP\cap\calQ\}f_Z(z)\, \rmd z\label{usedefg}\\
&\geq \int_{z\in\calP\cap\calQ} (1-\overline{g}(n,\sigma^2+b_n))^M f_Z(z)\, \rmd z\label{usenonincreaseoverg}\\
&\geq \int_{z\in\calP\cap\calQ} \exp\bigg\{-M\frac{\overline{g}(n,\sigma^2+b_n)}{1-\overline{g}(n,\sigma^2+b_n)}\bigg\}f_Z(z)\, \rmd z\label{lblog}\\
\nn&\geq \int_{z\in\calP\cap\calQ}\exp\bigg\{-M\frac{\overline{g}(n,\sigma^2+b_n)}{1-\overline{g}(n,\sigma^2+b_n)}\bigg\}\\*
&\qquad\quad\times 1\bigg\{M\frac{\overline{g}(n,\sigma^2+b_n)}{1-\overline{g}(n,\sigma^2+b_n)}\leq \frac{1}{\sqrt{n}}\bigg\} \, f_Z(z)\rmd z\\
\nn&\geq \bigg( 1-\frac{1}{\sqrt{n}}\bigg)\int_{z\in\calP\cap\calQ}\!\!\!
1\bigg\{M\frac{\overline{g}(n,\sigma^2+b_n)}{1-\overline{g}(n,\sigma^2+b_n)}\leq \frac{1}{\sqrt{n}}\bigg\}\\*
&\qquad\qquad\qquad\qquad\qquad\times f_Z(z)\rmd z\label{addconstraint}\\
&=\bigg( 1-\frac{1}{\sqrt{n}}\bigg)\times\Pr\bigg\{Z\in\calP\cap\calQ,M\leq \frac{1-\overline{g}(n,\sigma^2+b_n)}{\overline{g}(n,\sigma^2+b_n)}\frac{1}{\sqrt{n}}\bigg\}\\
\nn&\geq \bigg( 1-\frac{1}{\sqrt{n}}\bigg)\Pr\bigg\{Z\in\calP\cap\calQ,\log M\leq -\log 2
\\*
&\qquad\qquad\qquad\qquad\qquad-\log \overline{g}(n,\sigma^2+b_n)-\frac{1}{2}\log n\bigg\}\label{nlarge},
\end{align}
where \eqref{usepsidef} follows from the definition of $\Psi(n,z)$ in \eqref{y1symm}, \eqref{usedefg} follows by restricting $z\in\calP\cap\calQ$ and using the definition of $\overline{g}(\cdot)$ in \eqref{def:overlineg}, \eqref{usenonincreaseoverg} follows since $\overline{g}(n,z)$ is decreasing in $z$ for $z\in\calP\cap\calQ$, \eqref{lblog} follows since $(1-a)^M\geq \exp\{-M\frac{a}{1-a}\}$ for any $a\in[0,1)$, \eqref{addconstraint} follows since $M\frac{\overline{g}(n,z)}{1-\overline{g}(n,z)}\leq \frac{1}{\sqrt{n}}$, $\exp\{-a\}$ is decreasing in $a$, and $\exp\{-a\}\geq 1-a$ for $a\geq 0$, and \eqref{nlarge} follows since $\overline{g}(n,z)\leq \frac{1}{2}$ for $n$ large enough if $z>\sigma^2$.

Combining \eqref{def:overlineg}, \eqref{nlarge} and applying a Taylor expansion of $\overline{g}(n,\sigma^2+b_n)$ similarly to \eqref{tayloragain}, we conclude that for any $(n,M)$-code such that
\begin{align}
\log M
&\leq -\log 2-\frac{1}{2}\log n-\log \overline{g}(n,\sigma^2+b_n)\\*
&=n\bigg( \frac{1}{2}\log\frac{\sigma^2}{D}+\frac{b_n}{2\sigma^2}+O\bigg( \frac{\log n}{n}\bigg)\bigg)\label{usezinp},
\end{align}
we have
\begin{align}
\rmP_{\rme,n}(M)
&\geq \bigg(1-\frac{1}{\sqrt{n}}\bigg)\Pr\{Z\in\calP\cap\calQ\}\label{conversestep3}.
\end{align}
The following lemma is essential to complete  the converse proof.
\begin{lemma}
\label{concentrate}
Consider any source distribution $P_X$ such that \eqref{sourceconstraint} is satisfied and $\sigma^2<\infty$. Then, we have 
\begin{align}
\Pr\{Z\in\calP\cap\calQ\}
\geq \varepsilon+O\bigg( \frac{1}{\sqrt{n}}\bigg).
\end{align}
\end{lemma}
The proof of Lemma \ref{concentrate} is available in~\cite[Appendix A]{zhou2017refined} using the Berry-Esseen theorem.

Combining the definition of $\rmV(\sigma^2,\zeta)$ in \eqref{def:rmvsphere}, the definition of $b_n$ in \eqref{def:bn}, the bounds in \eqref{usezinp}, \eqref{conversestep3}, and Lemma \ref{concentrate}, we conclude that 
\begin{align}
\log M_{\rm{sp}}^*(n,\varepsilon,\sigma^2,D)
\leq \frac{n}{2}\log \frac{\sigma^2}{D}+\sqrt{n\rmV(\sigma^2,\zeta)}\rmQ^{-1}(\varepsilon)+O(\log n).
\end{align}


\subsection{Preliminaries for the  i.i.d. Gaussian Codebook}
\label{sec:preiid}
Now we consider the i.i.d.\ Gaussian codebook (cf.~\eqref{iidcodebook}). Note that $\Pr\{d(x^n,\hatX^n)\leq D\}$ depends on $x^n$ only through its norm $\|x^n\|$ (cf. \cite{ihara2000error}). Given any   $x^n$ such that $\frac{1}{n} {\|x^n\|^2} =z$, define
\begin{align}
\Upsilon(n,z)
&:=\Pr\{d(x^n,\hatX^n)\leq D\}\label{def:upsilonnz}.
\end{align}
From \eqref{iidcodebook}, we obtain that
\begin{align}
f_{\hatX^n}^{\rm{iid}}(\hatX^n)&=\frac{1}{(2\pi(\sigma^2-D))^{n/2}}\exp\bigg\{-\frac{\|\hatX^n\|^2}{2(\sigma^2-D)}\bigg\}.
\end{align}
Since $f_{\hatX^n}^{\rm{iid}}(\hatX^n)$ is decreasing in $\|\hatX^n\|$, we conclude that $\Upsilon(n,z)$ is a decreasing function of $z$ (cf. \cite{ihara2000error}). Using the definition of $\Upsilon(\cdot)$ in \eqref{def:upsilonnz}, we have
\begin{align}
\Upsilon(n,z)
&=\Pr\{\|x^n-\hatX^n\|^2\leq nD\}\\
&=\Pr\bigg\{\sum_{i=1}^n(Y_i-\sqrt{z})^2\leq nD\bigg\}\label{powerinvariant}\\*
&=\Pr\bigg\{-\frac{1}{nP_\rmc}\sum_{i=1}^n(Y_i-\sqrt{z})^2\geq -\frac{D}{P_\rmc}\bigg\}.
\end{align}
where \eqref{powerinvariant} follows since the probability depends on $x^n$ only through its power and thus we can choose $x^n$ such that  $x_i=\sqrt{z}$ for all $i\in[1:n]$ (cf.~\cite[Eq. (94)]{scarlett2017mismatch}). For the i.i.d.\ Gaussian codebook, each $Y_i\sim\calN(0,P_\rmc)$ and hence $\frac{1}{P_\rmc} {(Y_i-\sqrt{z})^2}$ is distributed according to a non-central $\chi^2$ distribution with one degree of freedom.

Given $s\in\bbR$ and any non-negative number $z$, define
\begin{align}
\kappa(s,z)&:=\frac{(P_\rmc(1+2s)+2z)^2}{P_\rmc(1+2s)^3}\label{def:lambda2s},\\
R_{\rm{iid}}(s,z)
&:=\frac{1}{2}\log(1+2s)+\frac{sz}{(1+2s)(\sigma^2-D)}-\frac{sD}{\sigma^2-D}\label{def:riidsapha},\\
s^*(z)&:=\max\bigg\{0,\frac{\sigma^2-3D+\sqrt{(\sigma^2-D)^2+4zD}}{4D}\bigg\}\label{def:sstar}.
\end{align}

Using the result of \cite[Section 2.2.12]{tanizaki2004computational} concerning the cumulant generating function of a non-central $\chi^2$ distribution, the definition of $R_{\rm{iid}}(\cdot)$ in \eqref{def:riidsapha}, the definition of $s^*(\cdot)$ in \eqref{def:sstar}, and the Bahadur-Ranga Rao (strong large deviations) theorem for non-lattice random variables~\cite[Theorem~3.7.4]{dembo2009large},  we obtain  
\begin{align}
\Upsilon(n,z)
&\sim\frac{\exp\{-nR_{\rm{iid}}(s^*(z),z)\}}{s^*(z)\sqrt{\kappa(s^*(z),z)2\pi n}} ,\quad n\to\infty\label{iidach1}.
\end{align}

\subsection{Achievability Proof for the I.I.D. Gaussian Codebook}
\label{sec:iidach}
According to \eqref{peforboth}, the excess-distortion probability under the i.i.d.\ Gaussian codebook can be upper bounded as follows:
\begin{align}
\rmP_{\rme,n}(M)
&=\bbE\Big[(1-\Pr\{d(X^n,\hatX^n)\leq D\,\big|\,X^n\})^M \Big]\label{ce1}\\
&=\int_0^\infty (1-\Upsilon(n,z))^Mf_Z(z)\, \rmd z\\
\nn&\leq \int_0^{\sigma^2-a_n}f_Z(z)\, \rmd z+\int_{\sigma^2+b_n}^\infty f_Z(z)\, \rmd z\\*
&\qquad+\int_{\sigma^2-a_n}^{\sigma^2+b_n} (1-\Upsilon(n,z))^Mf_Z(z)\, \rmd z\label{uppby1}\\
\nn&\leq \Pr\{Z<\sigma^2-a_n\}+\Pr\{Z>\sigma^2+b_n\}\\*
&\qquad+\int_{\sigma^2-a_n}^{\sigma^2+b_n} \exp\{-M\Upsilon(n,z)\}f_Z(z)\, \rmd z\label{useineq2}\\
&\leq \Pr\{Z<\sigma^2-a_n\}+\Pr\{Z>\sigma^2+b_n\}+\exp\{-M\Upsilon(n,\sigma^2+b_n)\}\label{iidach2},
\end{align}
where \eqref{uppby1} follows since $\Upsilon(n,z)\geq 0$, \eqref{useineq2} follows since $(1-a)^M\leq \exp\{-Ma\}$, and \eqref{iidach2} follows since $\Upsilon(n,z)$ is decreasing in $z$ and $\Pr\{\sigma^2-a_n\leq Z\leq \sigma^2+b_n\}\leq 1$.

Using the definitions of $R_{\rm{iid}}(\cdot)$ in \eqref{def:riidsapha} and $s^*(\cdot)$ in \eqref{def:sstar}, we have
\begin{align}
\nn R_{\rm{iid}}(s^*(\sigma^2+b_n),\sigma^2+b_n)
&=\frac{1}{2}\log \frac{P_\rmc+\sqrt{P_\rmc^2+4(\sigma^2+b_n)D}}{2D}\\*
\nn&\qquad+\frac{z(P_\rmc-2D+\sqrt{P_\rmc^2+4(\sigma^2+b_n)D})}{2P_\rmc(P_\rmc+\sqrt{P_\rmc^2+4(\sigma^2+b_n)D})}\\*
&\qquad-\frac{P_\rmc-2D+\sqrt{P_\rmc^2+4(\sigma^2+b_n)D}}{4P_\rmc}\label{concavebys*}\\
&=\frac{1}{2}\log\frac{\sigma^2}{D}+\frac{b_n}{2\sigma^2}+O(b_n^2)\label{taylor3}\\
&=\frac{1}{2}\log\frac{\sigma^2}{D}+\sqrt{\frac{\rmV(\sigma^2,\zeta)}{n}}\rmQ^{-1}(\varepsilon)+O\bigg( \frac{1}{n}\bigg),\label{aftertaylot3}
\end{align}
where \eqref{taylor3} follows from a Taylor expansion at $z=\sigma^2$ and recalling that $P_\rmc=\sigma^2-D$, and \eqref{aftertaylot3} follows from the definitions of $\rmV(\sigma^2,\zeta)$ in \eqref{def:rmvsphere} and $b_n$ in \eqref{def:bn}.

Choose $M$ such that
\begin{align}
\log M\geq -\log \Upsilon(n,\sigma^2+b_n)+\log\bigg( \frac{1}{2}\log n\bigg).
\end{align}
Then, we have
\begin{align}
\exp\{-M\Upsilon(n,\sigma^2+b_n)\}
\leq \frac{1}{\sqrt{n}}\label{iidach3}.
\end{align}
Furthermore, using the result in \eqref{iidach1} and \eqref{aftertaylot3}, we obtain
\begin{align}
\!\!\!\!\log M
&\geq \frac{n}{2}\log\frac{\sigma^2}{d}+\sqrt{n\rmV(\sigma^2,\zeta)}\rmQ^{-1}(\varepsilon)+O(\log n)\label{iidachlogm4second}.
\end{align}
Similarly to the proof of Lemma \ref{concentrate}, using the Berry-Esseen theorem and the definition of $a_n$ in \eqref{def:an}, we obtain
\begin{align}
\Pr\{Z<\sigma^2-a_n\}
&=\Pr\bigg\{\frac{1}{n}\sum_{i=1}^n(X_i^2-\sigma^2)<\sqrt{\rmV\frac{\log n}{n}}\bigg\}\\
&\leq \rmQ(\sqrt{\log n})+\frac{6T}{\sqrt{n}\, \rmV^{3/2}}=O\bigg( \frac{1}{\sqrt{n}}\bigg)\label{iidach4}.
\end{align}

Hence, combining \eqref{achstep2}, \eqref{iidach2}, \eqref{iidach3}, \eqref{iidachlogm4second} and \eqref{iidach4}, we conclude that 
\begin{align}
\log M_{\rm{iid}}^*(n,\varepsilon,\sigma^2,D)
&\geq \frac{n}{2}\log \frac{\sigma^2}{D}+\sqrt{n\rmV(\sigma^2,\zeta)}\rmQ^{-1}(\varepsilon)+O(\log n).
\end{align}

\subsection{Ensemble Converse for the i.i.d. Gaussian Codebook}
\label{sec:iidconverse}

The ensemble converse proof for the i.i.d.\ Gaussian codebook is omitted since it is similar to the ensemble converse proof for the spherical codebook in Section \ref{sec:spheresecconverse} starting from \eqref{fromhere} except for the following two points: i) replace $\overline{g}(n,z)$ with $\Upsilon(n,z)$, and ii) replace $\calP\cap\calQ$ with $\calP$.


\chapter{The Guass-Markov Source}
\label{chap:gm}

This chapter concerns lossy data compression of the Gauss-Markov source, which is a Gaussian source with first-order Markovian memory~\cite{davisson1072proc}. This chapter generalizes the rate-distortion study for memoryless sources pioneered by~Shannon~\cite{shannon1959coding} to the sources with memory. Such analyses find practical applications in image and video applications since the pixels and frames are usually correlated. The Gauss-Markov source is a special case of the Gaussian autoregressive source~\cite{kolmogorov1956IRE,gray1970tit}.

In contrast to memoryless sources, the Shannon theoretical study for sources with memory are very limited. Kolmogorov~\cite{kolmogorov1956IRE} initiated the study by deriving the rate-distortion function for a stationary Gaussian autoregressive source under the quadratic distortion measure by using an orthogonal coordinate transformation~(see also \cite{davisson1072proc}) that decomposes the autoregressive source into memoryless sources. Berger generalized the result in~\cite{kolmogorov1956IRE} to the Wiener process, which is a non-stationary case of the Gaussian autoregressive source. Gray~\cite{gray1970tit} generalized the result in~\cite{kolmogorov1956IRE} to general non-stationary Gaussian autoregressive processes and first-order binary symmetric Markov sources (BSMS). Subsequently, Gray~\cite{gray1971tit} generalized his result for BSMS in~\cite{gray1970tit} to finite-state finite-alphabet Markov sources. A critical result of Gray states that for distortions less than a certain value, the achievable rate for Markov sources is identical to the rate-distortion function for memoryless sources that generate the autoregressive sources. For further discussions on rate-distortion theory of sources with memory, readers could refer to \cite[Section IX]{berger1998lossy} or \cite[Section II.D]{gray1998tit} for more details.

All the above results are insightful. However, all tight results are asymptotic and only provide exact guidance when one compresses an infinitely long source sequence. It is natural to wonder what is the penalty in the practical finite blocklength regime. To date, the only known such result is the second-order asymptotics for a Gauss-Markov source by Tian and Kostina, who considered both stationary~\cite{tian2019tit} and nonstationary cases~\cite{tian2021nonstationary}. In this chapter, we present the results in~\cite{tian2019tit,tian2021nonstationary} with proof sketches.

\section{Problem Formulation and Asymptotic Result}

The problem formulation is exactly the same to the rate-distortion problem in Chapter \ref{chap:rd} except that we consider a Gauss-Markov source to be specified. Let $a\in\bbR_+$ be a non-negative real number and let $\calN(0,\sigma^2)$ be the Gaussian distribution with mean $0$ and variance $\sigma^2$. The Gauss-Markov source $\{X_i\}_{i\in\bbN}$ satisfies the following equation
\begin{align}
X_i=aX_{i-1}+Z_i,~\forall~i\geq 1\label{def:gaussmarkov},
\end{align}
where $X_0=0$ and $\{Z_i\}_{\in\bbN}$ is a GMS generated i.i.d. from $\sim\calN(0,\sigma^2)$. Note that when $a=0$, the random process $\{X_i\}_{i\in\bbN}$ reduces to a Gaussian memoryless source. The Gauss-Markov source is $\{X_i\}_{i\in\bbN}$ stationary when $a\in(0,1)$ and becomes nonstationary when $a\geq 1$. The special case of $a=1$ is also known as the Wiener process~\cite{berger1970tit}. Let $P_{X^n}$ denote the distribution of the Gauss-Markov source with length-$n$.

Recall from Definition \ref{def:code4rd} that an $(n,M)$-code consists of an encoder $f:\bbR^n\to[M]$ and a decoder $\phi:[M]\to\bbR^n$ that compresses a length-$n$ source sequence $X^n$ into a index over $[M]$ and reproduces it as $\hatX^n$ from the compressed index, respectively. Furthermore, let $\rmP_{\rme,n}(D)$ be the excess-distortion probability when compressing a length-$n$ source sequence under the quadratic distortion measure, i.e.,
\begin{align}
\rmP_{\rme,n}(D)=\Pr\{\|\hatX^n-X^n\|_2^2>D\}.
\end{align}
For any $\varepsilon\in(0,1)$ and target distortion level $D\in\bbR_+$, recall that $M^*(n,D,\varepsilon|a,\sigma^2)$ \eqref{def:M*} is the non-asymptotic fundamental limit of the rate-distortion problem and corresponds to the minimal number of codewords $M$ such that one can construct an $(n,M)$-code with excess-distortion probability satisfying $\rmP_{\rme,n}(D)\leq \varepsilon$.

Gray~\cite[Section II]{gray1970tit} characterized the first-order asymptotic coding rate $R^*(D|a,\sigma^2):=\lim_{\varepsilon\to 0}\lim_{n\to\infty}\frac{1}{n}\log M^*(n,D,\varepsilon|a,\sigma^2)$ for the Gaussian autoregressive source that includes the Gauss-Markov source as a special case. To present Gray's result when specialized to the Gauss-Markov source, we need the following definitions.  Define the function $h:[-\pi,\pi]\to\bbR+$ such that
\begin{align}
h(w|a,\sigma^2):=\frac{\sigma^2}{1+a^2-2a\cos w}\label{def:hw}.
\end{align}
Let $\theta_D$ be the solution of $\theta$ to the following equality
\begin{align}
\frac{1}{2\pi}\int_{-\pi}^{\pi}\min\{\theta,h(w|a,\sigma^2)\}\rmd w=D\label{reversewaterfilling}.
\end{align}
We can then define the following rate-distortion function $R_{\rm{GM}}(a,\sigma^2,D)$ for the Gauss-Markov source:
\begin{align}
R_{\rm{GM}}(a,\sigma^2,D)
&=\frac{1}{2\pi}\int_{-\pi}^{\pi}\bigg|\frac{1}{2}\log \frac{h(w|a,\sigma^2)}{\theta_D}\bigg|^+\rmd w\label{fromGray}.
\end{align}
Eq. \eqref{reversewaterfilling} is named reverse water filling since one needs to find a water level $\theta_D$ to satisfy the distortion constraint $D$.

With these definitions, the first-order asymptotic rate $R^*(D|a,\sigma^2)$ is characterized in the following theorem.
\begin{theorem}
For the Gauss-Markov source, it follows that
\begin{align}
R^*(D|a,\sigma^2)=R_{\rm{GM}}(a,\sigma^2,D).
\end{align}
\end{theorem}
When $a=0$, we have $h(w|a,\sigma^2)=\sigma^2$ and $R_{\rm{GM}}(a,\sigma^2,D)=R_{\rmG}(\sigma^2,D)=|\frac{1}{2}\log\frac{\sigma^2}{D}|^+$ is the rate-distortion function for the GMS with distribution $\calN(0,\sigma^2)$ under the quadratic distortion measure~\cite{shannon1959coding}. Furthermore, Gray~\cite[Eq. (24)]{gray1970tit} showed that for any $a\geq 0$,
\begin{align}
R_{\rm{GM}}(a,\sigma^2,D)
\geq R_{\rmG}(\sigma^2,D)
\end{align}
with equality if and only if $D\in[0,D_\rmc]$ where the critical distortion level, 
\begin{align}
D_\rmc
&=\min_{w\in[-\pi,\pi]}h(w|a,\sigma^2)=\frac{\sigma^2}{1+a^2}\label{def:Dc}.
\end{align}
Thus, at low distortion levels, the rate-distortion function of the Gauss-Markov source equals that of the GMS.

\section{Second-Order Asymptotics}
Recall the definition of $h(w|a,\sigma^2)$ in \eqref{def:hw} and the definition of $\theta_D$ as the reverse waterfilling level in \eqref{reversewaterfilling}. Define the following dispersion function 
\begin{align}
\rmV_{\rm{GM}}(a,\sigma^2,D)
&:=\frac{1}{\pi}\int_{-\pi}^{\pi}\min\bigg\{1,\bigg(\frac{h(w|a,\sigma^2)}{\theta_D}\bigg)^2\bigg\}\rmd w.
\end{align}
Furthermore, let 
\begin{align}
D_{\rm{max}}:=
\left\{
\begin{array}{ll}
\frac{\sigma^2}{1-a^2}&\mathrm{if~}a<1,\\
\frac{\sigma^2}{a^2-1}&\mathrm{if~}a>1.
\end{array}
\right.
\end{align}

Tian and Kostina~\cite{tian2019tit,tian2021nonstationary} proved the following result.
\begin{theorem}
\label{second4gaussmarkov}
For any $\varepsilon\in(0,1)$ and any $D\in(0,D_{\rm{max}})$,
\begin{align}
\log M^*(n,D,\varepsilon|a,\sigma^2)
&=nR_{\rm{GM}}(a,\sigma^2,D)+\sqrt{n\rmV_{\rm{GM}}(a,\sigma^2,D)}\rmQ^{-1}(\varepsilon)+o(\sqrt{n}).
\end{align}
\end{theorem}
Note that the non-stationary case of $a=1$ for the Wiener process is not addressed by Tian and Kostina due to a technical challenge pointed out in~\cite[Footnote 1]{tian2021nonstationary}. The proofs of the stationary case when $a\in(0,1)$ and the non-stationary case are available in \cite{tian2019tit} and \cite{tian2021nonstationary}, respectively. Generally speaking, in both cases, the proof follows from generalizations of the non-asymptotic bounds for the rate-distortion problem in Chapter \ref{chap:rd} with proper modifications. In the achievability part, a generalization of the lossy AEP in Lemma \ref{lossy:AEP} to the Gauss-Markov source is critical and in both directions, a decomposition of the Gauss-Markov source into independent source is vital. In the next section, we provide a proof sketch for the stationary case.

Theorem \ref{second4gaussmarkov} refined the classical first-order asymptotic result of Gray~\cite{gray1970tit} for the Gauss-Markov source by deriving the exact order and coefficient of the second-order coding rate. The dispersion function $\rmV_{\rm{GM}}(a,\sigma^2,D)$ follows from the same reverse waterfilling solution as the rate-distortion function $R_{\rm{GM}}(a,\sigma^2,D)$. Tian and Kostina showed that the dispersion term $\rmV_{\rm{GM}}(a,\sigma^2,D)$ relates to the dispersion function $\rmV_{\rmG}(\sigma^2,D)=\frac{1}{2}$ for the GMS $\calN(0,\sigma^2)$~\cite{kostina2012fixed,ingber2011dispersion}, analogously to how the rate-distortion function of both cases are related~\cite[Eq. (24)]{gray1970tit}. Specifically, it holds that
\begin{align}
\rmV_{\rm{GM}}(a,\sigma^2,D)\leq \rmV_{\rmG}(\sigma^2,D),
\end{align}
with equality if and only if $D\in(0,D_\rmc)$, where $D_\rmc$ is the critical distortion level defined in \eqref{def:Dc}. Therefore, Theorem \ref{second4gaussmarkov}, combined with \cite[Eq. (24)]{gray1970tit}, shows that at low distortion levels $D\in(0,D_\rmc)$, the second-order coding rate of the Gauss-Markov source in \eqref{def:gaussmarkov} equals that of the GMS $Z^n\sim\calN(0,\sigma^2)$. On the other hand, if $D\in(D_\rmc,D_{\rm{max}})$, the first order coding rate $R_{\rm{GM}}(a,\sigma^2,D)$ is greater than the rate-distortion function $R_\rmG(P_X,D)$ of the GMS but the the second-order coding rate $\rmV_{\rm{GM}}(a,\sigma^2,D)$ is smaller than $\rmV_{\rmG}(P_X,D)$. This implies that the memory of the source makes it first-order asymptotically more difficult to compress the Gauss-Markov source but ensures a smaller gap between the non-asymptotic rate and the first-order asymptotic coding rate when $\varepsilon\in(0,0.5)$.

It would be of interest to consider a mismatched version of the Gauss-Markov source where the innovation process $\{Z_i\}_{i\in\bbN}$ might not be Gaussian, as in Chapter \ref{chap:mismatch}. A critical question is then whether one could propose a coding scheme ignorant of the distribution of $\{Z_i\}_{i\in\bbN}$ and the parameter $a$ to achieve universally good performance. The second-order asymptotic analysis for such a case would be interesting. A even more practical setting would be to incorporate the above mismatched scenario with the noisy source setting in Chapter \ref{chap:noisy} to consider the case where the source sequence is also corrupted by some additional noise with unknown distribution.

\section{Proof Sketch}
We only present the proof sketch for the stationary case~\cite{tian2019tit} when $a\in(0,1)$. For the nonstationary case of $a>1$, readers can refer to \cite{tian2021nonstationary} for details.

\subsection{Decorrelation of the Gauss-Markov source}
As pointed at the beginning of \cite[Section IX]{berger1998lossy}, given the knowledge of how to compress a memoryless source,  a direct approach to compress a source with memory is to transform the source with memory into several independent memoryless sources. An explicit approach of this kind was given by Davisson~\cite[Eq. (15)]{davisson1072proc} for correlated stationary Gaussian sources. In this section, we present the application of the transformation to the Gauss-Markov source~\cite[Section III.A]{tian2019tit} and decompose it into independent Gaussian sources.

For any $n\in\bbN$, let $\bA$ be the $n\times n$ lower triangular matrix such that for each $(i,j)\in[n]^2$,
\begin{align}
A_{i,j}
&=\left
\{
\begin{array}{ll}
1&\mathrm{if~}i=j,\\
-a&\mathrm{if~}j=i-1,~i\geq 1,\\
0&\mathrm{otherwise}.
\end{array}
\right.
\end{align}
A pictorial illustration of the matrix $\bA$ is 
\begin{align}
\bA:=
\left[
\begin{array}{lllll}
1&0&0&\ldots&0\\
-a&1&0&\ldots&0\\
0&-a&1&\ldots&0\\
\vdots&\ddots&\ddots&\ddots&\vdots\\
0&\ldots&0&-a&1
\end{array}
\right].
\end{align}
Let $\bX$ denote the column vector of $X^n=(X_1,\ldots,X_n)$ of the Gaussian-Markov source $X^n$ in \eqref{def:gaussmarkov} and let $\bZ$ denote the column vector of $Z^n=(Z_1,\ldots,Z_n)$, where $Z^n$ is the GMS with distribution $\calN(0,\sigma^2)$. It follows that $\bZ=\bA\bX$ and the covariance matrix $\Sigma_{\bX}$ satisfies
\begin{align}
\Sigma_{\bX}=\bbE[\bA^{-1}\bZ\bZ^\rmT\bA]=\sigma^2(\bA^\rmT\bA)^{-1}.
\end{align}
Let $\bV$ be the unitary matrix corresponding to the eigendecomposition of $(\bA^\rmT\bA)^{-1}$, i.e.,
\begin{align}
(\bA^\rmT\bA)^{-1}=\bV\Lambda\bV^\rmT,
\end{align}
where $\Lambda=\rm{diag}(\frac{1}{\mu_1},\ldots,\frac{1}{\mu_n})$ is the diagonal matrix where $\mu_1,\ldots,\mu_n$ are the eigen values of $(\bA^\rmT\bA)^{-1}$. 

Define the vector $\bS$ such that
\begin{align}
\bS=\bV^\rmT\bX\label{def:bS}.
\end{align}
It follows that $\bS~\sim\calN(\bzero_n,\sigma^2\Lambda)$, i.e., $(S_1,\ldots,S_n)$ are independent Gaussian random variables with zero mean and different variances $(\sigma_1^2,\ldots,\sigma_n^2)$ where $\sigma_i^2:=\frac{\sigma^2}{\mu_i}$ for each $i\in[n]$. This way, the Gauss-Markov source $\bX$ is decomposed into independent random variables $\bS$ with the product Gaussian distribution $P_{\bS}=\prod_{i\in[n]}\calN(0,\frac{\sigma^2}{\mu_i})$, which eases the analysis of second-order asymptotics.

\subsection{Preliminaries}
In this section, we present necessary definitions and preliminary results used to prove Theorem \ref{second4gaussmarkov}. For each $n\in\bbN$, define the following $n$-th order multi-letter rate-distortion function 
\begin{align}
R_{\rm{GM}}(P_{\bX},D,n)
&:=\min_{P_{\hat{\bX}|\bX}:\bbE[d(\bX,\hat{\bX})]\leq D} I(P_{\bX},P_{\bX|\hat{\bX}}),
\end{align}
where $P_{\bX}$ is the distribution of the Gauss-Markov source $X^n$ up to time $n$ and $P_{\bX|\hat{\bX}}$ is induced by $P_{\bX}$ and the test channel $P_{\hat{\bX}|\bX}$. It was shown by Gray~\cite{gray1970tit} that the rate-distortion function $R_{\rm{GM}}(a,\sigma^2,D)$ in \eqref{fromGray} is the limit value of $R_{\rm{GM}}(P_X,D,n)$ as $n\to\infty$. Analogous to \eqref{def:dtilteddensity}, for any $\bx\in\bbR^n$, define the distortion-tilted information density for the Gauss-Markov source as follows:
\begin{align}
\jmath_{\rm{GM}}(\bx|P_{\bX},D)
&:=-\log\bbE_{P_{\hat{\bX}}^*}[\exp(n\lambda^*D-\lambda_n^*d(\bx,\hat{\bX}))],
\end{align}
where $\lambda_n^*:=-\frac{\partial R_{\rm{GM}}(P_{\bX},D,n) }{\partial D}$ is the negative first derivative of \\$R_{\rm{GM}}(P_{\bX},D,n)$ with respect to the distortion level $D$ and $P_{\hat{\bX}}^*$ is induced by the source distribution $P_{\bX}$ and the optimal test channel $P_{\hat{\bX}|\bX}^*$ that achieves $R_{\rm{GM}}(P_{\bX},D,n)$. Similar to Lemma \ref{prop:dtilteddensity}, one can show that
\begin{align}
R_{\rm{GM}}(P_{\bX},D,n)=\bbE_{P_{\bX}}[\jmath_{\rm{GM}}(\bX|P_{\bX},D)].
\end{align}

In the proof of Theorem \ref{second4gaussmarkov}, instead of considering the Gauss-Markov source $\bX$ with memory, we use the decomposed independent source $\bS$. It follows from \eqref{def:bS} that 
\begin{align}
R_{\rm{GM}}(P_{\bS},D,n)&=R_{\rm{GM}}(P_{\bX},D,n),
\end{align}
and for any $\bx\in\bb^n$ and $\bs=\bV^\rmT\bx$,
\begin{align}
\jmath_{\rm{GM}}(\bs|P_{\bS},D)&=\jmath_{\rm{GM}}(\bx|P_{\bX},D).
\end{align}

For each $n\in\bbN$, let $\theta_D^n$ be the solution of $\theta_n$ to
\begin{align}
\frac{1}{n}\sum_{i\in[n]}\min\{\theta_n,\sigma_i^2\}=D.
\end{align}
Since $\bS$ is memoryless, for any $\bs\in\bbR^n$, 
\begin{align}
\jmath_{\rm{GM}}(\bs|P_{\bS},D)
&=\sum_{i\in[n]}\jmath(s_i|P_{S_i},\min\{\theta_D^n,\sigma_i^2\}),
\end{align}
where $\jmath(\cdot)$ is the distortion-tilted information density defined in \eqref{def:dtilteddensity} for the rate-distortion problem and $P_{S_i}=\calN(0,\sigma_i^2)$ is the induced marginal distribution of the random variable $S_i$. We need the following alternative distortion-tilted information density that approximates $\jmath_{\rm{GM}}(\bs|P_{\bS},D)$:
\begin{align}
\jmath_{\rm{alt}}(\bs|P_{\bS},D_n)
&:=\sum_{i\in[n]}\jmath(s_i|P_{S_i},\min\{\theta_D,\sigma_i^2\}),
\end{align}
where $\theta_D$ is defined as the solution of $\theta$ to \eqref{reversewaterfilling}, which is independent of $n$ and $D_n$ is defined as
\begin{align}
D_n&:=\frac{1}{n}\sum_{i\in[n]}\min\{\theta_D,\sigma_i^2\}.
\end{align}
Let $E_i$, $V_i$ and $T_i$ be the expectation, the variance and the third absolute moment of the random variable $\jmath(S_i|P_{S_i},\min\{\theta_D,\sigma_i^2\})$ with respect to $P_{S_i}$, respectively. 
It follows from \cite[Theorem 4]{tian2019tit} that there exists constants $c_r$ and $c_v$ such that
\begin{align}
\Big|nR_{\rm{GM}}(a,\sigma^2,D)-\sum_{i\in[n]}E_i\Big|&\leq c_r,\label{approxmean}\\
\bigg|\sqrt{nV_{\rm{GM}}(a,\sigma^2,D)}-\sqrt{\sum_{i\in[n]}V_i}\bigg|&\leq c_v\label{approxvariance}.
\end{align}
Define $C_{\rm{BE}}$ as
\begin{align}
C_{\rm{BE}}
&:=\frac{\frac{6}{n}\sum_{i\in[n]}T_i}{(\frac{1}{n}\sum_{i\in[n]}V_i^2)^{3/2}}.
\end{align}
It follows from \cite[Appendix A]{tian2019tit} that for the Gauss-Markov source, $C_{\rm{BE}}$ is a finite positive constant in the tail probability of the Berry-Esseen theorem.

Finally, let $\tilc\in(0,1)$ be the constant in \cite[Theorem 10]{tian2019tit} and define the event
\begin{align}
\calE:=\bigg\{\big|\jmath_{\rm{GM}}(\bS|P_{\bS},D)-\jmath_{\rm{alt}}(\bS|P_{\bS},D_n)\big|\leq \frac{4C_D}{\tilc \theta_D}\bigg\}\label{def:calE}.
\end{align}
It follows from \cite[Theorem 10]{tian2019tit} that
\begin{align}
\Pr\{\calE\}\geq 1-\frac{1}{n}\label{pofE}.
\end{align}

\subsection{Achievability}
The achievability proof parallels that of Theorem \ref{second4rd} for the rate-distortion problem. Recall from \eqref{def:distortionball} the definition of the distortion ball $\calB_D(\bs)=\{\hat{\bs}\in\bbR^n:~d(\bs,\hat{\bs})\leq D\}$ and recall that $\bS$ is the decomposed independent source sequence with distribution $P_{\bS}$. Similar to Theorem \ref{ach:fbl}, we have the following result.
\begin{lemma}
\label{achfbl4gm}
For any $P_{\hat{\bS}}\in\calP(\bbR^n)$, there exists an $(n,M)$-code such that the excess-distortion probability satisfies
\begin{align}
\rmP_{\rme,n}(D)
&\leq \bbE_{P_{\bS}}[(1-P_{\hat{\bS}}(\calB_D(\bS)))^M].
\end{align}
\end{lemma}

Define the following constants
\begin{align}
\eta_n&:=\sqrt{\frac{a\log\log n}{n}}.
\end{align}
The following lemma generalizes the lossy AEP in Lemma \ref{lossy:AEP} for memoryless sources to the Gauss-Markov source.
\begin{lemma}
\label{aep4gm}
Let $\alpha>0$ and let $q>1$, $\beta_1>0$, $\beta_2$ and $\kappa$ be constants defined in \cite[Lemma 3]{tian2019tit}. There exists a constant $K>0$ such that
\begin{align}
\nn&\Pr\Big\{-\log (P_{\hat{\bS}}^*)^n(\calB_D(S^n))\leq \jmath_{\rm{GM}}(\bS|P_{\bS},D)+\beta_1(\log n)^q+\beta_2\Big\}\\*
&\geq 1-\frac{K}{(\log n)^{\kappa\alpha}},
\end{align}
where $P_{\hat{\bS}}^*$ is induced by $P_{\bS}$ and the optimal test channel $P_{\hat{\bS}|\bS}$ that achieves $R_{\rm{GM}}(P_{\bS},D,n)$.
\end{lemma}
Lemma \ref{aep4gm} relates the probability of the distortion ball under $P_{\hat{\bS}}^*$ with the distortion-tilted information density $\jmath_{\rm{GM}}(\bS|P_{\bS},D)$.

Define the event 
\begin{align}
\calL:=\bigg\{-\log (P_{\hat{\bS}}^*)^n(\calB_D(S^n))&\leq \jmath_{\rm{alt}}(\bS|P_{\bS},D_n)+\beta_1(\log n)^q+\beta_2+\frac{4C_D}{\tilc \theta_D}\bigg\}\label{def:calL}.
\end{align}
Using Lemma \ref{achfbl4gm}, we conclude that there exists an $(n,M)$-code for the Gauss-Markov source such that
\begin{align}
\rmP_{\rme,n}(D)
&\leq \bbE_{P_{\bS}}[(1-P_{\hat{\bS}^*}(\calB_D(\bS)))^M]\\
&\leq \bbE_{P_{\bS}}\Big[\exp\big(-MP_{\hat{\bS}^*}(\calB_D(\bS))\big)\Big]\label{useineq4gm}\\
\nn&=\bbE_{P_{\bS}}\Big[\exp\big(-MP_{\hat{\bS}^*}(\calB_D(\bS))\big)\bbo(\calL)\Big]\\*
&\qquad+\bbE_{P_{\bS}}\Big[\exp\big(-MP_{\hat{\bS}^*}(\calB_D(\bS))\big)\bbo(\calL^\rmc)\Big]\\
&\leq \bbE_{P_{\bS}}\Big[\exp\big(-MP_{\hat{\bS}^*}(\calB_D(\bS))\big)\bbo(\calL)\Big]+\frac{1}{n}+\frac{K}{(\log n)^{\kappa\alpha}}\label{step1:ach4gm},
\end{align}
where \eqref{useineq4gm} follows from the inequality $(1-x)^M\leq \exp(-Mx)$ and \eqref{step1:ach4gm} follows from \eqref{pofE} and Lemma \ref{aep4gm}, which implies that
\begin{align}
\Pr\{\calL^\rmc\}\leq \frac{1}{n}+\frac{K}{(\log n)^{\kappa\alpha}}.
\end{align}
We next upper bound the first term in \eqref{step1:ach4gm}. Let $\varepsilon_n$ be defined as
\begin{align}
\varepsilon_n&:=\varepsilon-\frac{C_{\rm{BE}+1}}{\sqrt{n}}-\frac{1}{n}-\frac{K}{(\log n)^{\kappa\alpha}}\label{def:epsilon4gm}.
\end{align}
Choose $M$ such that
\begin{align}
\log M
\nn&=nR_{\rm{GM}}(a,\sigma^2,D)+\sqrt{n\rmV_{\rm{GM}}(a,\sigma^2,D)}\rmQ^{-1}(\varepsilon_n)\\*
\nn&\qquad+\log\Big(\frac{\log n}{2}\Big)+\beta_1(\log n)^q+\beta_2+c_r\\*
&\qquad+c_v|Q^{-1}(\varepsilon_n)|+\frac{4C_D}{\tilc \theta_D}\label{choosem4gm}.
\end{align}
Define the random variable $\rmG_n$ such that
\begin{align}
G_n:=\log M-\jmath_{\rm{alt}}(\bS|P_{\bS},D_n)-\beta_1(\log n)^q-\beta_2-\frac{4C_D}{\tilc \theta_D}\label{def:Gn}.
\end{align}
Combining \eqref{approxmean} and \eqref{approxvariance}, we conclude that
\begin{align}
G_n\geq \sum_{i\in[n]}E_i+\sqrt{\sum_{i\in[n]}V_i}\rmQ^{-1}(\varepsilon_n)-\jmath_{\rm{alt}}(\bS|P_{\bS},D_n)+\log\Big(\frac{\log n}{2}\Big).
\end{align}
Define the event $\calG$ such that
\begin{align}
\calG&:=\bigg\{G_n<\log\Big(\frac{\log n}{2}\Big)\bigg\}\label{def:calG}.
\end{align}
Using the Berry-Esseen theorem for independent but not identically distributed random variables in Theorem \ref{berrytheorem4general}, we have
\begin{align}
\Pr\{\calG\}
&\leq \Pr\bigg\{\jmath_{\rm{alt}}(\bS|P_{\bS},D_n)>\sum_{i\in[n]}E_i+\sqrt{\sum_{i\in[n]}V_i}\rmQ^{-1}(\varepsilon_n)\bigg\}\\
&\leq \varepsilon_n+\frac{C_{\rm{BE}}}{\sqrt{n}}\label{useBE4gm}.
\end{align}
Therefore, 
\begin{align}
\nn&\bbE_{P_{\bS}}\Big[\exp\big(-MP_{\hat{\bS}^*}(\calB_D(\bS))\big)\bbo(\calL)\Big]\\*
&\leq \bbE_{P_{\bS}}[\exp(-\exp(G_n))]\label{useGn}\\
&=\bbE_{P_{\bS}}[\exp(-\exp(G_n))\bbo(\calG)]+\bbE_{P_{\bS}}[\exp(-\exp(G_n))\bbo(\calG^\rmc)]\\
&\leq \Pr\{\calG\}+\frac{1}{n}\Pr\{\calG^\rmc\}\label{usecalG}\\
&\leq \varepsilon_n+\frac{C_{\rm{BE}+1}}{\sqrt{n}}\label{stepn:ach4gm},
\end{align}
where \eqref{useGn} follows from  the definitions of the event $\calL$ in \eqref{def:calL} and $G_n$ in \eqref{def:Gn}, \eqref{usecalG} follows from the definition of the event $\calG$ in \eqref{def:calG} and \eqref{stepn:ach4gm} follows from \eqref{useBE4gm}.

The achievability proof is completed by combining \eqref{step1:ach4gm}, \eqref{def:epsilon4gm} and \eqref{stepn:ach4gm} and applying the Taylor expansion of $\rmQ^{-1}(\varepsilon_n)$ around $\varepsilon$.

\subsection{Converse}
Similar to Theorem \ref{converse:fbl} for memoryless sources, for the Gauss-Markov source, we conclude that any $(n,M)$-code satisfies that
\begin{align}
\rmP_{\rme,n}(D)
&\geq \Pr\Big\{\jmath_{\rm{GM}}(\bS|P_{\bS},D)\geq \log M+\log n\Big\}-\frac{1}{n}\label{converse:step14GM},
\end{align}
where the probability is calculated with respect to $P_{\bS}$ of the decomposed independent source sequence $\bS\in\bbR^n$. The first term in \eqref{converse:step14GM} can be further lower bounded as follows:
\begin{align}
\nn&\Pr\Big\{\jmath_{\rm{GM}}(\bS|P_{\bS},D)\geq \log M+\log n\Big\}\\*
&\geq \Pr\Big\{\jmath_{\rm{GM}}(\bS|P_{\bS},D)\geq \log M+\log n\mathrm{~and~}\calE\Big\}\\
&\geq \Pr\bigg\{\jmath_{\rm{alt}}(\bS|P_{\bS},D_n)\geq \log M+\log n+\frac{4C_D}{\tilc \theta_D}\bigg\}-\Pr\{\calE^\rmc\}\label{usecalEagain4gm}\\
&\geq \Pr\bigg\{\jmath_{\rm{alt}}(\bS|P_{\bS},D_n)\geq \log M+\log n+\frac{4C_D}{\tilc \theta_D}\bigg\}-\frac{1}{n}\label{usepoeagain4gm},
\end{align}
where \eqref{usecalEagain4gm} follows from the definition of $\calE$ in \eqref{def:calE} and \eqref{usepoeagain4gm} follows from the result in \eqref{pofE}.

The rest of the converse proof is omitted since it is similar to the achievability proof from Eq. \eqref{choosem4gm} to \eqref{stepn:ach4gm} by applying the Berry-Esseen theorem to the first term in \eqref{usepoeagain4gm}.


\chapter{Variable Length Compression}
\label{chap:variable}

This chapter concerns variable length lossy compression, where given each source sequence, a binary string with a variable length is output as the compressed codewords. This problem generalizes the rate-distortion problem in Chapter \ref{chap:rd} by allowing a flexible codeword length. The excess-length probability criterion was used in variable length compression~\cite{he2009redundancy,Kontoyiannis2014tit,kosut2017tit}. However, as pointed out by Verd\'u~\cite{verdu2011}, the fundamental limit in such a setting is exactly the same as the fixed-length compression allowing errors. Therefore, the more meaningful fundamental limit for variable length compression is usually the average codeword length of an optimal code subject to a certain excess-distortion (error) probability constraint. 

The study of variable length compression focuses on the lossless case and dates back to Shannon. By relating the codeword length reversely proportional to the probability of the source sequence, Shannon~\cite[Section 10]{shannon1948mathematical} showed that the average codeword length is bounded by the entropy of the source sequence with a deviation of at most one bit, implying that the asymptotic coding rate per source symbol equals the source entropy with zero error probability. Han~\cite{han2000tit} initiated the study of variable length compression allowing errors by considering a vanishing error probability. The result of Han was later generalized by Koga and Yamamoto~\cite{Koga2005tit}, who showed that the asymptotic average codeword length per source symbol of an optimal code is less than the source entropy if a non-vanishing error probability is tolerated, demonstrating the asymptotic advantage of variable length compression allowing errors. The result in~\cite{Koga2005tit} was refined by Kostina, Polyanskiy and Verd\'u~\cite{kostina2015tit} who derived the second-order asymptotic approximation and further refined by Sakai, Yavas, Tan~\cite{sakai2021tit} who further derived the third order asymptotic approximation. The above studies were also generalized to the case with side information~\cite{watanabe2015variable,sakai2020tit}. 

For lossy compression, Zhang, Yang and Wei~\cite{zhang1997tit} studied the deviation of the expected codeword length per source symbol to the rate-distortion function for discrete memoryless sources, which was later generalized by Yang and Zhang~\cite{yang1999redundancy} to abstract sources. The results in~\cite{zhang1997tit,yang1999redundancy} assumed zero excess-distortion probability and were lossy counterparts to \cite{shannon1948mathematical}. By tolerating a non-vanishing excess-distortion probability, Kostina, Polyanskiy and Verd\'u~\cite[Section III]{kostina2015tit} derived a second-order asymptotic approximation to the average codeword length per source symbol. In particular, it follows from \cite[Theorem 9]{kostina2015tit} that the asymptotic coding rate is smaller than the the rate-distortion function if the excess-distortion probability is not zero and the deviation of the non-asymptotic coding rate from the asymptotic one is always negative, which implies the great advantage of variable length compression allowing errors in the finite blocklength regime. In this chapter, we present the non-asymptotic and second-order asymptotic bounds in~\cite[Section III]{kostina2015tit} with proof sketches.

\section{Problem Formulation and Existing Results}
Similar to the rate-distortion problem in Chapter \ref{chap:rd}, let $X^n$ be a memoryless source generated i.i.d. from the distribution $P_X$ defined on the alphabet $\calX$ and let $\hatcalX$ be the reproduced alphabet. Recall that $d:\calX\times\hatcalX\to\bbR_+$ is the distortion measure and $d(x^n,\hatx^n)$ denotes the symbolwise average distortion between a source sequence $x^n\in\calX^n$ and its reproduction $\hatx^n\in\hatcalX^n$. Let $D\in\bbR_+$ be the target distortion level. In variable length compression, we need to use the set of all binary strings, denoted by $\calB:=\{0,1\}^*=\{\emptyset,0,1,00,01,10,11,\ldots\}$. For any $b\in\calB$, let $l(b)$ be the length of the binary string, e.g., $l(\emptyset)=0$, $l(00)=1$, $l(1001)=4$. For simplicity, we let $\calB_{\emptyset}$ to denote $(\calB\setminus\{\emptyset\})$, i.e., the elements of all binary strings except the empty one. Different from other chapters, the logarithm in this chapter is base $2$ instead of $e$ to account for the fact that the length of a binary string should be in bits.

With above definitions, a code for variable length lossy compression is defined as follows.
\begin{definition}
\label{def:code4variable}
An $(n,L)$-code consists of a potentially stochastic pair of encoder $P_{B|X^n}\in\calP(\calB|\calX^n)$ and decoder $P_{\hatX^n|B}\in\calP(\hatcalX^n|\calB)$ such that the average codeword length is upper bounded by $L$, i.e.,
\begin{align}
\bbE[l(f(X^n))]\leq L,
\end{align}
where the expectation are calculated with respect to the distribution $P_{X^n}$ and the potentially stochastic encoder and decoders.
\end{definition}
Note that when $\calB$ is replaced by the set $\calM:=[M]$, Def. \ref{def:code4variable} reduces to the $(n,M)$-code in Def. \ref{def:code4rd} for the rate-distortion problem. To evaluate the performance of an $(n,L)$-code, we consider the excess-distortion probability with respect to the target distortion level $D$, i.e.,
\begin{align}
\rmP_{\rme,n}(D)=\Pr\{d(X^n,\hatX^n)>D\}.
\end{align}

The fundamental limit of variable length compression is the achievable minimal average codeword length such that the excess-distortion probability is bounded by a constant $\varepsilon\in[0,1]$, i.e.,
\begin{align}
L^*(n,D,\varepsilon)
&:=\min\{L:~\exists~\mathrm{an~}(n,L)\mathrm{-code~s.t.~}\rmP_{\rme,n}(D)\leq \varepsilon\}.
\end{align}
In practice, encoders and decoders are usually deterministic. To account for this case, let $L_{\rm{det}}^*(n,D,\varepsilon)$ denote the fundamental limit when both encoder and decoders are deterministic, i.e., for any $(x^n,b,\hatx^n)\in\calX^n\times\calB\times\hatcalX^n$, $P_{B|X^n}(b|x^n)=\bbo(b=f(x^n))$ and $P_{\hatX^n|B}(\hatx^n|b)=\bbo(\hatx^n=\phi(b))$ for some deterministic functions $f:\calX^n\to\calB$ and $\phi:\calB\to\hatcalX^n$.

Recall the definitions of  the rate-distortion function $R(P_X,D)$, $D_{\rm{min}}$ and $D_{\rm{max}}$ in \eqref{def:dmin}, \eqref{def:dmax} and \eqref{def:rd}), respectively, i.e.,
\begin{align}
R(P_X,D)&=\inf_{P_{\hatX|X}:\bbE[d(X,\hatX)]\leq D}I(P_X,P_{X|\hatX}),\label{rdf4variable}\\
D_{\rm{min}}&=\inf\{D\in\bbR_+:~R(P_X,D)<\infty\},\\
D_{\rm{max}}&=\inf_{\hatx}\bbE[d(X,\hatx)].
\end{align}
Zhang, Yang and Wei~\cite[Theorems 4 and 5]{zhang1997tit} derived the following result.
\begin{theorem}
\label{asymp4variable}
For any $D\in(D_{\rm{min}},D_{\rm{max}})$,
\begin{align}
L_{\rm{det}}^*(n,D,0)=nR(P_X,D)+\frac{\log n}{2}+o(\log n).
\end{align}
\end{theorem}
Theorem \ref{asymp4variable} derives an approximation to the non-asymptotic performance of an optimal deterministic code when \emph{zero} excess-distortion probability is tolerated. In the rest of the chapter, we present generalizations of Theorem \ref{asymp4variable} to stochastic codes and demonstrate the great advantage of tolerating a non-zero excess-distortion probability.

\section{Properties of Optimal Codes}
Different from simple coding scheme in the lossless case where one can order source sequences with decreasing probabilities and assign binary strings in $\calB$ with increased length, an optimal code for the lossy case does not have explicit simple descriptions. Instead, we recall the properties of the optimal stochastic codes and discuss the relationship between the fundamental limits of optimal deterministic and stochastic codes in~\cite[Section III.B]{kostina2015tit}.

\subsection{Zero Excess-Distortion Probability}
\label{prop:zero}
Let $\calB_D(\hatx^n)=\{x^n\in\calX^n:~d(x^n,\hatx^n)\leq D\}$ be the distortion ball for a reproduced source sequence $\hatx^n$. For any two distinct binary strings $(b_1,b_2)\in\calB^2$, we say $b_1<b_2$ if $l(b_1)<l(b_2)$ or if $l(b_1)=l(b_2)$ but $b_1$ has more number of zeros till the first one appears, i,e., $\emptyset<0$, $0<1$, $0100<1101$. For each $i\in\bbN$, let $b_i\in\calB$ be the $i$-th largest element of $\calB$, i.e., $b_1=\emptyset$, $b_4=00$.

An optimal $(n,L)$ code with $\rmP_{\rme,n}(D)=0$ satisfies
\begin{enumerate}
\item the optimal code has deterministic encoder and decoder, i.e., for each $(x^n,b,\hatx^n)\in\calX^n\times\calB\times\hatcalX^n$, $P_{B|X^n}^*(b|x^n)=\bbo(b=f^*(x^n))$ and $P_{\hatX^n|B}^*(\hatx^n|b)=\bbo(\hatx^n=\phi^*(b))$ for some deterministic functions $f^*:\calX^n\to\calB$ and $\phi^*:\calB\to\hatcalX^n$;
\item the output $B=f^*(X^n)$ of the optimal deterministic encoder $f^*$ orders binary strings in $\calB$ with probability reversely proportional to lengths, i.e., $P_B^*(b_i)\ge P_B^*(b_j)$ if and only if $i\le j$, where $P_B^*$ is induced by the source distribution $P_{X^n}$ and the optimal encoder $f^*$;
\item Given each $b\in\calB$, for all $x^n\in(\calB_{\phi^*(b)}\setminus\bigcup_{\barb\in\calB:\barb<b}\calB_{\phi^*(\barb)})$, $f^*(x^n)=b$.
\end{enumerate}
All above three claims can be proved via contradiction since violation of any claim would increase the average codeword length $\bbE[L(X^n)]$ and an optimal code has the smallest average codeword length. 

The explicit code construction is challenging to describe. However, the relationship between the codeword length and the probability of the source sequence can be made explicit. Property (iii) implies that for each $b\in\calB$, given any $x^n\in\calB_{\phi^*(b)}$ is mapped into the binary string $b$ and thus
\begin{align}
P_B^*(b)=\sum_{x^n\in\calB_{\phi^*(b)}}P_X^n(x^n).
\end{align}
Let $x_1^n,x_2^n,\ldots$ be the ordering of all possible source sequences with decreasing probabilities. Properties (ii) and (iii) imply that for each $i\in[n]$,
\begin{align}
l(f^*(x_i^n))\leq \log i\leq -\log{P_X^n(x_i^n)}.
\end{align}
In other words, for each $x^n$,
\begin{align}
l(f^*(x^n))\leq -\log P_X^n(x^n)\label{uppl4variable}.
\end{align}

\subsection{Non-Zero Excess-Distortion Probability}
\label{prop:nonzero}
Similarly, if one tolerates a non-zero excess-distortion probability of $\varepsilon\in(0,1]$, one can show that an optimal $(n,L)$-code with positive excess-distortion probability satisfies property (ii) and the following two properties:
\begin{enumerate}
\item the optimal decoder $P_{\hatX^n|B}^*(\hatx^n|b)=\bbo(\hatx^n=\phi^*(b))$ is deterministic via the function $\phi^*$ and the optimal encoder $P_{B|X^n}^*$ is stochastic such that $P_{B|X^n}^*(b|x^n)=1-P_{B|X^n}^*(\emptyset|x^n)$ for all $x^n\in\calX^n$ and $b\in\calB_\emptyset$;
\item there exists $\eta\in\bbR_+$ and $\alpha\in[0,1)$ such that for each $b\in\calB_\emptyset$
\begin{align}
P_{B|X^n}^*(b|x^n)
&=\left\{
\begin{array}{ll}
1&\mathrm{if~}x^n\in(\calB_{\phi^*(b)}\setminus\bigcup_{\barb\in\calB:\barb<b}\calB_{\phi^*(\barb)}\\*
&\mathrm{~and~}l(b)<\eta,\\
1-\alpha&\mathrm{if~}x^n\in(\calB_{\phi^*(b)}\setminus\bigcup_{\barb\in\calB:\barb<b}\calB_{\phi^*(\barb)})\\
&\mathrm{~and~}l(b)=\eta,
\end{array}
\right.
\label{opcode4variable1}\\
P_{B|X^n}^*(\emptyset|x^n)
&=\left\{
\begin{array}{ll}
1&\mathrm{if~}x^n\notin\bigcup_{b\in\calB}\calB_{\phi^*(b)},\\
\alpha&\mathrm{if~}x^n\in\bigcup_{b\in\calB}\calB_{\phi^*(b)}\mathrm{~and~}l(b)=\eta.
\end{array}
\right.
\label{opcode4variable2}
\end{align}
and
\begin{align}
\Pr\bigg\{X^n\notin\bigcup_{b\in\calB:b<\eta}\calB_{\phi^*(b)}\bigg\}+\alpha\Pr\{l(f^*(X^n))=\eta\}=\varepsilon\label{exp:guarantee4variable}.
\end{align}
\end{enumerate}
Note that \eqref{opcode4variable2} implies that for any $x^n$ that incurs an excess-distortion event with respect to the distortion level $D$, the encoder maps $x^n$ into $\emptyset$ and incurs no penalty to the average codeword length. Furthermore, \eqref{exp:guarantee4variable} states that the excess-distortion probability of the above code is exactly $\varepsilon$ as desired. From the above code construction, we find that when a non-zero excess-distortion probability is allowed, the optimal code is no longer deterministic since the optimal encoder is stochastic. 

\subsection{Deterministic and Stochastic Codes}
Since both optimal encoders and decoders are deterministic when $\varepsilon=0$ as shown in Section \ref{prop:zero}, it follows that
\begin{align}
L^*(n,D,0)=L_{\rm{det}}^*(n,D,0).
\end{align}
Thus, Theorem \ref{asymp4variable}  also holds for optimal stochastic codes. In other words, under the zero excess-distortion probability criterion, the randomization of encoders or decoders does not improve the performance.

When $\varepsilon\in(0,1]$ is strictly non-zero, it follows from Section \ref{prop:nonzero} (cf. \cite[Eq. (98)-(99)]{kostina2015tit}) that
\begin{align}
L^*(n,D,\varepsilon)\leq L_{\rm{det}}^*(n,D,\varepsilon)\leq L^*(n,D,\varepsilon)+1\label{linkcodes}.
\end{align}
Note that in the upper bound in \cite[Eq. (99)]{kostina2015tit} has the constant of $\phi(\min\{\varepsilon,1/e\})$ instead of $1$, where $\phi(x):=-x\log x$. We believe it is easier to present the bound in the simpler form since $\max_{x\in[0,1]}\phi(x)\leq \frac{\log e}{e}\approx 0.531<1$, 

The lower bound in \eqref{linkcodes} follows since any deterministic code is a special case of a stochastic code and thus the minimal average codeword length of an optimal stochastic code is no larger than the average codeword length of an optimal deterministic code. The upper bound in \eqref{linkcodes} is justified by analyzing the stochastic nature of the optimal encoder in \eqref{opcode4variable1} and \eqref{opcode4variable2}. Note that the randomization of the encoder $P_{B|X^n}^*$ occurs if and only if $x^n\in\bigcup_{b\in\calB}\calB_{\phi^*(b)}\mathrm{~and~}l(b)=\eta$. Specifically, the randomization is applied only to one source sequence. To clarify, let $\calD:=\{x^n:~l(f^*(x^n))\}=\eta$, let $m:=|\calD|$ and let $x_1^n,\ldots,x_m^n$ be the order of elements in $\calD$ with with decreasing order of probabilities, i.e., $P_X^n(x_i^n)\geq P_X^n(x_j^n)$ if $i\le j$. Furthermore, let $i^*\in[m]$ be the smallest value such that the sum probabilities of elements $\{x_i^n\}_{i>i^*}$ is no greater than $\alpha$, i.e.,
\begin{align}
i^*:=\argmin_{i\in[m]}\sum_{i\in[m]:i>i^*}P_X^n(x^n)\leq \alpha.
\end{align}
Using the optimal deterministic encoder $f^*$ for the case of $\varepsilon=0$, the optimal encoder $P_{B|X^n}^*$ can be described as follows: for any $x^n\in\calX^n$,
\begin{enumerate}
\item if $l(f^*(x^n))<\eta$ or $l(f^*(x^n))=\eta$, $P_X^n(x^n)>P_X^n(x_{i^*}^n)$, and $P_{B|X^n}^*(b|x^n)=\bbo(b=f^*(x^n))$;
\item if $l(f^*(x^n))>\eta$ or $l(f^*(x^n))=\eta$, $P_X^n(x^n)<P_X^n(x_{i^*}^n)$, and $P_{B|X^n}^*(b|x^n)=\bbo(b=\emptyset)$;
\item if $x^n=x_{i^*}^n$, 
\begin{align}
P_{B|X^n}^*(b|x^n)=\beta \bbo(b=f^*(x^n))+(1-\beta)\bbo(b=\emptyset),
\end{align}
where $\beta\in[0,1]$ satisfies
\begin{align}
\beta P_X^n(x_{i^*}^n)+\sum_{i\in[m]:i>i^*}P_X^n(x^n)=\alpha.
\end{align}
\end{enumerate}
Thus, we can construct a deterministic code by mapping $x_{i^*}^n$ to $f^*(x_{i^*}^n)$. This way, the excess-distortion probability is of the deterministic code is upper bounded by $\varepsilon$ and the average codeword length is upper bounded by
\begin{align}
\nn&L^*(n,D,\varepsilon)+(1-\beta)P_X^n(x_{i^*}^n)l(f^*(x_{i^*}^n))\\*
&\leq L^*(n,D,\varepsilon)-P_X^n(x_{i^*}^n)\log P_X^n(x_{i^*}^n)\label{useuppl4variable}\\
&\leq L^*(n,D,\varepsilon)+1\label{usephi4var},
\end{align}
where \eqref{useuppl4variable} follows from the result in \eqref{uppl4variable} and \eqref{usephi4var} follows since $\max_{x\in[0,1]}\phi(x)\le 1$.

Therefore, with the relationship in \eqref{linkcodes}, it suffices to derive bounds for $L^*(n,D,\varepsilon)$ to fully understand the fundamental limit of variable length lossy compression. 

\section{Non-Asymptotic Bounds}
We need the following definitions. Given any $\varepsilon\in[0,1]$, let
\begin{align}
R_n(P_X,D,\varepsilon)
&:=\min_{P_{\hatX^n|X^n}:\Pr\{d(X^n,\hatX^n)\}>D} I(P_X^n,P_{X^n|\hatX^n})\label{def:rn4var}.
\end{align}
When $n=1$, $R_1(P_X,D,\varepsilon)$ is analogous to the rate-distortion function $R(P_X,D)$ with the only exception that the constraint on the average distortion is replaced by a constraint on the excess-distortion probability.

For any $\varepsilon\in(0,1)$, given any real valued random variable $Y$, define the $\varepsilon$-cutoff random variable $\langle Y\rangle_{\varepsilon}$ such that $\Pr\{\langle Y\rangle_{\varepsilon}>\eta\}=0$,
\begin{align}
\Pr\{\langle Y\rangle_{\varepsilon}=Y\}
&=\left\{
\begin{array}{ll}
1&\mathrm{if~}Y<\eta,\\
1-\alpha&\mathrm{if~}Y=\eta,
\end{array}
\right.\\
\Pr\{\langle Y\rangle_{\varepsilon}=0\}
&=\left\{
\begin{array}{ll}
1&\mathrm{if~}Y>\eta,\\
\alpha&\mathrm{if~}Y=\eta,
\end{array}
\right.
\end{align}
where $\eta\in\bbR$ and $\alpha\in[0,1]$ are chosen such that 
\begin{align}
\Pr\{Y>\eta\}+\alpha\Pr\{Y=\eta\}=\varepsilon.
\end{align}
Recall that $\calB_D(x^n)=\{\hatx^n\in\hatcalX^n:~d(x^n,\hatx^n)\leq D\}$ is the distortion ball around $x^n$. The following cutoff random variable for the probabilities of the distortion ball around the source sequence $X^n$ is critical:
\begin{align}
\barR_n(P_X,D,\varepsilon)
&:=\min_{P_{\hatX^n}} \bbE_{P_X^n}\Big[\big\langle-\log P_{\hatX^n}(\calB_D(X^n))\big\rangle_{\varepsilon}\Big].
\end{align}

With above definitions, Kostina \emph{et al.}~\cite[Theorem 7]{kostina2015tit} proved the following result.
\begin{theorem}
\label{fbl4variable}
For any $\varepsilon\in[0,1]$,
\begin{align}
R_n(P_X,D,\varepsilon)-\log(R_n(P_X,D,\varepsilon)+1)-\log e
&\leq L^*(n,D,\varepsilon)\label{confbl4variable}\\
&\leq \barR_n(P_X,D,\varepsilon)\label{achbl4variable}.
\end{align}
\end{theorem}
\begin{proof}
The converse bound in \eqref{confbl4variable} follows from the same argument as in the lossless case and uses a critical results that lower bound the expected codeword length of variable lossless compression with zero error in~\cite{alon1994lower,wyner1972upper}. Specifically, consider any code with output string $B\in\calB$ such that $\rmP_{\rme,n}(D)\leq \varepsilon$. Note that $X^n-B-\hatX^n$ forms a Markov chain and $\Pr\{d(X^n,\hatX^n)>D\}\leq \varepsilon$. Thus,
\begin{align}
H(B)
&\geq I(X^n;B)\\
&\geq I(X^n;\hatX^n)\label{usemar4var}\\
&\geq R_n(P_X,D,\varepsilon)\label{usern4var},
\end{align}
where \eqref{usephi4var} follows since $X^n-B-\hatX^n$ is a Markov chain and \eqref{usern4var} follows from the definition of $R_n(P_X,D,\varepsilon)$. Furthermore, it follows from ~\cite{alon1994lower,wyner1972upper} that
\begin{align}
\bbE[l(B)]
&\ge H(B)-\log(H(B)+1)-\log e\label{lossless4var}.
\end{align}
The proof of \eqref{confbl4variable} is completed by combining \eqref{usern4var} and \eqref{lossless4var}.

The achievability bound in \eqref{achbl4variable} is derived as follows. Let $P_{\hatX^n}$ be arbitrary and let $\hatbx=(\hatx_1^n,\hatx_2^n,\ldots)$ be a sequence of reproduced codewords generated independently from $P_{\hatX^n}$. Consider a code with encoder $f$ and decoder $\phi$ that operates as follows. Given any source sequence $X^n$, the encoder $f$ maps it into the binary string $b_W$ where
\begin{align}
W
&:=\left\{
\begin{array}{ll}
\min\{i\in[\bbN]:~d(X^n,\hatx_i^n)\leq D\}&\mathrm{if~}\langle-\log P_{\hatX^n}(\calB_D(x^n))\big\rangle_{\varepsilon}>0,\\
1&\mathrm{otherwise}.
\end{array}
\right.
\end{align}
Upon receiving $b_W$, the decoder outputs $\hatx_W^n$ as the reproduced source sequence.

Using the random coding idea, by averaging over distributions of the source sequence $X^n$ and random codewords $\hatbX^n=(\hatX_1^n,\hatX_2^n,\ldots)$, the average codeword length of the code satisfies
\begin{align}
\bbE[l(f(X^n))]
&=\bbE[\lfloor \log W\rfloor]\\
&\leq \bbE\Big[\log W\bbo\big(\langle-\log P_{\hatX^n}(\calB_D(x^n))\big\rangle_{\varepsilon}>0\big)\Big]\\
&=\bbE\Big[\bbE[\log W|X^n]\bbo\big(\langle-\log P_{\hatX^n}(\calB_D(x^n))\big\rangle_{\varepsilon}>0\big)\Big]\\
&\leq \bbE\Big[\log \bbE[W|X^n]\bbo\big(\langle-\log P_{\hatX^n}(\calB_D(x^n))\big\rangle_{\varepsilon}>0\big)\Big]\label{usejensen4var}\\
&\leq \bbE\Big[\big\langle-\log P_{\hatX^n}(\calB_D(x^n))\big\rangle_{\varepsilon}\Big]\label{usedefw4var},
\end{align}
where \eqref{usejensen4var} follows by applying the Jensen's inequality to the concave function $\log x$ and \eqref{usedefw4var} follows from the definition of the $\varepsilon$-cutoff random variable and the fact that if $\langle-\log P_{\hatX^n}(\calB_D(x^n))\big\rangle_{\varepsilon}>0$ holds, the random variable $W$ is a geometric random variable with success probability $P_{\hatX^n}(\calB_D(x^n))$, which implies that $\log\bbE[W|X^n]=-\log P_{\hatX^n}(\calB_D(x^n))$. By optimizing over all $P_{\hatX^n}$, we obtain the bound in \eqref{achbl4variable}.

\end{proof}

\section{Second-Order Asymptotics}

\subsection{Result and Discussions}
Recall the definition of the distortion-dispersion function $\rmV(P_X,D)$ in \eqref{def:disdispersion}.
Let $P_{\hatX^*}$ be the induced marginal distribution that achieves the rate-distortion function $R(P_X,D)$ in \eqref{rdf4variable}. 

Assume that $\bbE_{P_X\times P_{\hatX^*}}[(d(X,\hatX))^{12}]<\infty$.  Kostina, Polyanskiy and Verd\'u~\cite[Theorem 9]{kostina2015tit} proved the following result.
\begin{theorem}
\label{second4variable}
For any $\varepsilon\in[0,1]$ and $D\in(D_{\rm{min}},D_{\rm{max}})$,
\begin{align}
L^*(n,D,\varepsilon)
&=(1-\varepsilon)nR(P_X,D)-\sqrt{\frac{n\rmV(P_X,D)}{2\pi}}\exp\bigg(-\frac{\rmQ^{-1}(\varepsilon)^2}{2}\bigg)+O(\log n).
\end{align}
\end{theorem}
The proof of Theorem \ref{second4variable} follows by applying Berry-Esseen theorem to the non-asymptotic bounds in Theorem \ref{fbl4variable} with proper choice of parameters and is available in Section \ref{proof4variable}.

Theorem \ref{second4variable} establishes a second-order asymptotic approximation to the average codeword length of an optimal variable-length lossy compression code that tolerates an excess-distortion probability of $\varepsilon\in[0,1]$. In light of \eqref{linkcodes}, the same second-order asymptotic bound  also holds for optimal deterministic codes. 

Compared with Theorem \ref{asymp4variable} that tackles zero excess-distortion probability, tolerating a non-zero excess-distortion probability significantly reduces the average codeword length. Specifically, the asymptotic average codeword rate per source is reduced by a multiplicative factor of $\varepsilon$, i.e., $\lim_{n\to\infty}\frac{L^*(n,D,\varepsilon)}{n}=(1-\varepsilon)R(P_X,D)$ and the negative second-order coding rate implies that the non-asymptotic rate approaches the first-order asymptotic rate from below regardless of $\varepsilon\in(0,1)$. This is in stark contrast to the fixed-length case in Theorem \ref{second4rd} where the non-asymptotic rate approaches the asymptotic rate from above if $\varepsilon\in(0,1)$, implying a finite blocklength penalty. In summary, the average bit required per source symbol for variable-length lossy compression allowing errors is significantly reduced compared with the fixed-length lossy compression and variable-length compression with zero error, both asymptotically and non-asymptotically. In Fig. \ref{illus:gainofvar}, we plot the second-order asymptotic approximation to $R^*(n,D,\varepsilon):=\frac{L^*(n,D,\varepsilon)}{n}$ for different values of $\varepsilon\in[0,1]$ and compare with the second-order asymptotics for the fixed-length compression in Theorem \ref{second4rd} for a Bernoulli source with distribution $P_X=\mathrm{Bern}(0.2)$ under the Hamming distortion measure with target excess-distortion probability $\varepsilon=0.05$ with respect to the distortion level $D=0.02$.
\begin{figure}[bt]
\centering
\includegraphics[width=.8\columnwidth]{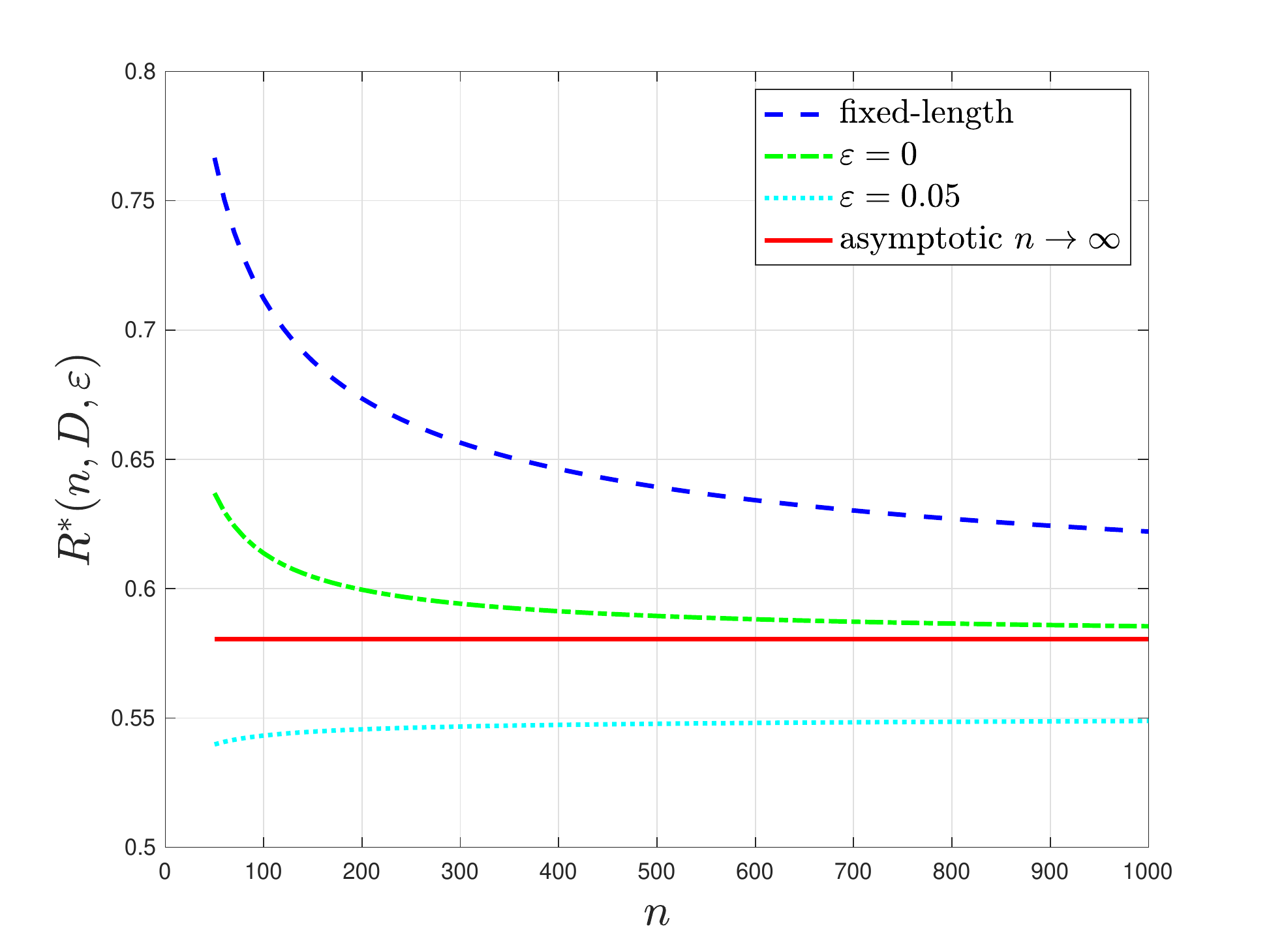}
\caption{Plots of the second-order asymptotic approximation in Theorem \ref{second4variable} to the average codeword length per source symbol $R^*(n,D,\varepsilon)$ for different values of the blocklength $n$ and comparison with the fixed-length compression counterpart in Theorem \ref{second4rd}.}
\label{illus:gainofvar}
\end{figure}

Similar to Theorem \ref{second4rd}, Theorem \ref{second4variable} holds for both discrete and continuous sources that satisfy the assumptions. It would be interesting to generalize Theorem \ref{second4variable} to account for the noisy source, noisy channel, the mismatched setting or the source with memory.

\subsection{Proof Sketch}
\label{proof4variable}
The following Lemma~\cite[Lemmas 1]{kostina2015tit} bounds the expectation of the cutoff random variable and is critical in the proof of second-order asymptotics.
\begin{lemma}
\label{clt4cutoff}
Let $Y^n$ be an i.i.d. sequence generated from a distribution $P_Y$ with finite third absolute moment, i.e., $\bbE[(Y-\bbE[Y])^3]<\infty$. For any $\varepsilon\in[0,1]$,
\begin{align}
\bbE\Big[\Big\langle \sum_{i\in[n]}Y_i \Big\rangle_{\varepsilon}\Big]&=(1-\varepsilon)n\bbE[Y]-\sqrt{\frac{n\rm{Var}[Y]}{2\pi}}\exp\bigg(-\frac{\rmQ^{-1}(\varepsilon)^2}{2}\bigg)+O(1).
\end{align}
\end{lemma}
The proof of Lemma \ref{clt4cutoff} follows from the Berry-Esseen theorem for independent random variables and algebra. Readers could refer to \cite[Appendix A]{kostina2015tit} for details.

The achievability part follows by weakening \eqref{achbl4variable} with $\bbE_{P_X^n}\Big[\big\langle-\log P_{\hatX^n}^*(\calB_D(X^n))\big\rangle_{\varepsilon}\Big]$ and using the following lemma that bounds the expectation term explicitly, where $P_{\hatX^n}^*$ is the product distribution of $P_{\hatX}^*$, i.e., for any $\hatx^n\in\hatcalX^n$,
\begin{align}
P_{\hatx^n}^*(\hatx^n)=\prod_{i\in[n]}P_{\hatX}^*(x_i).
\end{align}

\begin{lemma}
\label{achlemma4var}
For any $\varepsilon\in[0,1]$,
\begin{align}
\nn&\bbE_{P_X^n}\Big[\big\langle-\log P_{\hatX^n}^*(\calB_D(X^n))\big\rangle_{\varepsilon}\Big]\\*
&=(1-\varepsilon)nR(P_X,D)-\sqrt{\frac{\rmV(P_X,D)}{2\pi}}\exp\bigg(-\frac{\rmQ^{-1}(\varepsilon)^2}{2}\bigg)+O(\log n).
\end{align}
\end{lemma}
The proof of Lemma \ref{achlemma4var} follows from Lemma \ref{clt4cutoff} and the refined version~\cite[Lemma 4]{kostina2015tit} of the lossy AEP in Lemma \ref{lossy:AEP}, which was implicitly presented in \cite{yang1999redundancy}.
  
We next present the proof sketch for the converse part. Recall the definition of the distortion-tilted information density $\jmath(x|D,P_X)$ in \eqref{def:dtilteddensity} and recall that $\lambda^*$ defined in \eqref{def:lambda^*} is defined as the first negative derivative of $R(P_X,D)$ with respect to $D$. For any $x^n\in\calX^n$, let $\jmath(x^n|D,P_X)=\sum_{i\in[n]}\jmath(x_i|D,P_X)$.

The following lemma was derived in~\cite[Theorem 8]{kostina2015tit}.
\begin{lemma}
\label{conlemma4var}
For any $\varepsilon\in[0,1]$ and $D\in(D_{\rm{min}},D_{\rm{max}})$,
\begin{align}
R_n(P_X,D,\varepsilon)
\nn&=\bbE\Big[\big\langle\jmath(X^n|D,P_X)\big\rangle_{\varepsilon}\Big]-\log(nR(P_X,D)+n\lambda^*D+1)\\*
&\qquad-\log e-H_\rmb(\varepsilon).
\end{align}
\end{lemma}
The converse proof of Theorem \ref{second4variable} follows from the non-asymptotic bound in \eqref{confbl4variable}, Lemma \ref{clt4cutoff} and Lemma \ref{conlemma4var}.

\part{Multiterminal Setting}
\chapter{Kaspi Problem}
\label{chap:kaspi}
In this chapter, we study the lossy source coding problem with one encoder and two decoders, where side information is available at the encoder and one of the two decoders. We term the problem as the Kaspi problem since this problem was first introduced by Kaspi, who derived the asymptotically optimal achievable rate to ensure reliable lossy reconstruction at both decoders~\cite[Theorem 1]{kaspi1994}. Analogous to the rate-distortion problem, we term the asymptotic optimal achievable rate as the Kaspi rate-distortion function. The Kaspi problem generalizes the rate-distortion problem by adding one additional decoder and allowing the encoder and the additional decoder to access to some correlated side information. 

Kaspi's asymptotic results were recently refined by Zhou and Motani in \cite{zhou2017kaspishort,zhou2017non}, in which the authors derived non-asymptotic and second-order asymptotics bounds for the Kaspi problem. In this chapter, we present the results in \cite{zhou2017kaspishort,zhou2017non} and illustrate the role of side information on lossy data compression in the finite blocklength regime. Specifically, we first present a parametric representation for the Kaspi rate-distortion function. Subsequently, we generalize the notion of the distortion-tilted information density for the rate-distortion problem in Chapter \ref{chap:rd} to the Kaspi problem and present a non-asymptotic converse bound. Finally, for a DMS under bounded distortion measures, we present second-order asymptotics and illustrate the results via two numerical examples. 

Since the Kaspi problem generalizes the rate-distortion problem, the results for the Kaspi problem generalizes those in Chapter \ref{chap:rd}. Furthermore, another special case of the Kaspi problem is the conditional rate-distortion problem where side information is available to both the encoder and decoder in the rate-distortion problem. Thus, the results for the Kaspi problem generalize those for the conditional rate-distortion problem~\cite{le2014second} as well.

\section{Problem Formulation and Asymptotic Result}
\label{sec:pform}

The setting of the Kaspi problem is shown in Figure \ref{systemmodel}. There are one encoder $f$ and two decoders $\phi_1,\phi_2$. The side information $Y^n$ is available to the encoder $f$ and the decoder $\phi_2$ but \emph{not} to the decoder $\phi_1$. The encoder $f$ compresses the source $X^n$ into a message $S$ given the side information $Y^n$. Decoder $\phi_1$ aims to recover source sequence $X^n$ within distortion level $D_1$ under distortion measure $d_1$ using the message $S$. Decoder $\phi_2$ aims to recover $X^n$ within distortion level $D_2$ under distortion measure $d_2$ using the message $S$ and the side information $Y^n$. Consider a correlated memoryless source with distribution $P_{XY}$ defined on the alphabet $\calX\times\calY$. Assume that the source sequence and side information $(X^n,Y^n)$ is generated i.i.d. from $P_{XY}$. Furthermore, assume that the reproduction alphabets for decoders $\phi_1$ and $\phi_2$ are $\hat{\calX}_1$ and $\hat{\calX}_2$ respectively.

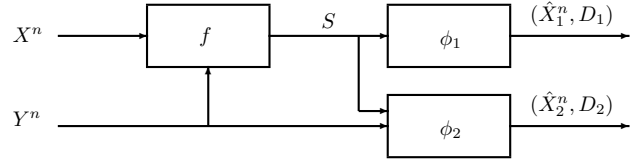
\begin{figure}[htbp]
\centering
\setlength{\unitlength}{0.5cm}
\scalebox{0.8}{
\begin{picture}(20,6)
\linethickness{1pt}
\put(0.5,4.8){\makebox{$X^n$}}
\put(0.5,2){\makebox{$Y^n$}}
\put(5,4){\framebox(4,2)}
\put(6.7,4.8){\makebox{$f$}}
\put(2,5){\vector(1,0){3}}
\put(2,2){\vector(1,0){11}}
\put(7,2){\vector(0,1){2}}
\put(13,1){\framebox(4,2)}
\put(13,4){\framebox(4,2)}
\put(14.7,4.7){\makebox{$\phi_1$}}
\put(14.7,1.7){\makebox{$\phi_2$}}
\put(9,5){\vector(1,0){4}}
\put(11,5.5){\makebox(0,0){$S$}}
\put(12,5){\line(0,-1){2.5}}
\put(12,2.5){\vector(1,0){1}}
\put(17,2){\vector(1,0){4}}
\put(17.7,2.5){\makebox{$(\hatX_2^n,D_2$)}}
\put(17,5){\vector(1,0){4}}
\put(17.7,5.5){\makebox{$(\hatX_1^n,D_1)$}}
\end{picture}}
\caption{System model for the Kaspi problem of lossy source with side information at the encoder and one of the two decoders~\cite[Theorem 1]{kaspi1994}.}
\label{systemmodel}
\end{figure}

\begin{definition}
\label{def:code4kaspi}
An $(n,M)$-code for the Kaspi problem consists of one encoder 
\begin{align}
f:\calX^n\times\calY^n\to \calM=[M],
\end{align}
and two decoders
\begin{align}
\phi_1&:\calM\to\hat{\calX}_1^n\\
\phi_2&:\calM\times\calY^n\to \hat{\calX}_2^n.
\end{align}
\end{definition}
For simplicity, let $\hatX_1^n=\phi_1\big(f(X^n,Y^n)\big)$ and $\hatX_2^n=\phi_2\big(f(X^n,Y^n),Y^n\big)$. For $i\in[2]$, let $d_i:\calX\times\hat{\calX}_i\to [0,\infty]$ be two distortion measures. For any $x^n\in\calX^n$ and $\hatx_i^n\in\hatcalX_i^n$, let the distortion between $x^n$ and $\hatx_i^n$ be additive and defined as $d_i(x^n,\hatx_i^n):=\frac{1}{n}\sum_{j\in[n]} d_i(x_j,\hatx_{i,j})$. 

Following \cite{kaspi1994}, the rate-distortion function of the Kaspi problem is defined as follows, which characterizes the asymptotically minimal rate to ensure reliable lossy compression at both decoders as the blocklength tends to infinity.
\begin{definition}
\label{firstorder}
A rate $R$ is said to be $(D_1,D_2)$-achievable for the Kaspi problem if there exists a sequence of $(n,M)$-codes such that
\begin{align}
\limsup_{n\to\infty} \frac{\log M}{n}\leq R,
\end{align}
and
\begin{align}
\limsup_{n\to\infty} \mathbb{E}\big[d_i(X^n,\hatX_i^n)\big]\leq D_i,~i\in[2]\label{kaspi:avgdistortion}.
\end{align}
The minimum $(D_1,D_2)$-achievable rate is called the Kaspi rate-distortion function and denoted as $R^*(D_1,D_2)$.
\end{definition}

Define
\begin{align}
\nn&R(P_{XY},D_1,D_2)\\*
&:=\min_{\substack{P_{\hatX_1|XY},~P_{\hatX_2|XY\hatX_1}:\\\mathbb{E}[d_1(X,\hatX_1)]\leq D_1,\\\mathbb{E}[d_2(X^n,\hatX_2^n)]\leq D_2}} I(XY;\hatX_1)+I(X;\hatX_2|Y\hatX_1)\label{kaspiratefunc}.
\end{align}
Kaspi~\cite[Theorem 1]{kaspi1994} derived the following result.

\begin{theorem}
The minimum $(D_1,D_2)$-achievable rate for the Kaspi problem satisfies
\begin{align}
R^*(D_1,D_2)=R(P_{XY},D_1,D_2).
\end{align}
\end{theorem}
We refer to $R(P_{XY},D_1,D_2)$ as the Kaspi rate-distortion function. Note that $R(P_{XY},D_1,D_2)$ is convex and non-increasing in both $D_1$ and $D_2$. We remark that the explicit formulas of the Kaspi rate-distortion function was derived by Perron, Diggavi and Telatar for a GMS under quadratic distortion measures~\cite{perron2005kaspi} and a binary memoryless erasure source under Hamming distortion measures~\cite{perron2006kaspi}.

To derive non-asymptotic and second-order asymptotic bounds, instead of using the average distortion criterion, we adopt the following joint excess-distortion probability as the performance criterion:
\begin{align}
\rmP_{\rme,n}(D_1,D_2)
&:=\Pr\Big\{d_1(X^n,\hatX_1^n)>D_1~\mathrm{or}~d_2(X^n,\hatX_2^n)>D_2\Big\}\label{defexcessp}.
\end{align}
Note that the probability in \eqref{defexcessp} is calculated with respect to the distribution of the source sequences for a fixed $(n,M)$-code. For bounded distortion measures, the asymptotically minimal rate to ensure vanishing joint excess-distortion probability $\rmP_{\rme,n}(D_1,D_2)$ is also $R(P_{XY},D_1,D_2)$. The justification is similar to the case of the rate-distortion problem below Theorem \ref{lossy:Shannon}.

The Kaspi rate-distortion function $R(P_{XY},D_1,D_2)$ equals the rate-distortion function $R(P_X,D_2)$ if $D_1$ is large enough, and equals the conditional rate-distortion function $R(P_{X|Y},D_1|P_Y)$ if $D_2$ is large enough where
\begin{align}
R(P_{X|Y},D_1|P_Y)
&:=\min_{P_{\hatX_1|XY}:\bbE[d_1(X,\hatX_1)]\leq D_1}I(X;\hatX_1|Y).
\end{align}
Note that the conditional rate-distortion function $R(P_{X|Y},D_1|P_Y)$ is the minimal achievable rate of lossy compression when side information is available at both the encoder and the decoder, which is also known as the conditional rate-distortion problem. Similarly, for second-order asymptotics, the results for the Kaspi problem specialize to either the rate-distortion problem or the conditional rate-distortion problem.

\section{Properties of the Rate-Distortion Function}
\label{sec:mainresults4kaspi}
We first present the properties of the Kaspi rate-distortion function, which allows us to define the distortions-tilted information density for the Kaspi problem and derive a non-asymptotic converse bound that generalizes the non-asymptotic converse bound for the rate-distortion problem in Theorem \ref{converse:fbl} of the rate-distortion problem.

Given any (conditional) distributions $(P_{\hatX_1|XY},P_{\hatX_2|XY\hatX_1})$, let $P_{\hatX_1}$, $P_{X\hatX_1}$, $P_{X\hatX_2}$ and $P_{Y\hatX_1}$, $P_{\hatX_2|Y\hatX_1}$, $P_{Y\hatX_1\hatX_2}$ be induced by $P_{XY}$, $P_{\hatX_1|XY}$ and $P_{\hatX_2|XY\hatX_1}$. Consider the distortion levels $(D_1,D_2)$ such that \\$R(P_{XY},D_1,D_2)$ is finite and there exists test channels $(P_{\hatX_1|XY}^*,P_{\hatX_2|XY\hatX_1}^*)$ that achieve $R(P_{XY},D_1,D_2)$. Note that $R(P_{XY},D_1,D_2)$ (see \eqref{kaspiratefunc}) corresponds to a convex optimization problem and the dual problem is given by 
\begin{align}
\nn&\sup_{(\lambda_1,\lambda_2)\in\bbR_+^2} \min_{P_{\hatX_1|XY},P_{\hatX_2|XY\hatX_1}}\Big(
I(XY;\hatX_1)+I(X;\hatX_2|Y\hatX_1)\\*
&\qquad+\lambda_1(\mathbb{E}[d_1(X,\hatX_1)-D_1])+\lambda_2 (\mathbb{E}[d_2(X^n,\hatX_2^n)-D_2])\Big)\label{dualproblem}.
\end{align}
For any given distortion levels $(D_1,D_2)$, the optimal solutions to the dual problem of $R(P_{XY},D_1,D_2)$ are
\begin{align}
\lambda_1^*&:=\frac{\partial R(P_{XY},D,D_2)}{\partial D}\Big|_{D=D_1}\label{dualoptimal1},\\
\lambda_2^*&:=\frac{\partial R(P_{XY},D_1,D)}{\partial D}\Big|_{D=D_2}\label{dualoptimal2}.
\end{align}
Given any $(x,y,\hatx_1)\in\calX\times\calY\times\hatcalX_1$ and distributions $(Q_{\hatX_1},Q_{\hatX_2|Y\hatX_1})\in\calP(\hatcalX_1)\times\calP(\hatcalX_2|\calY,\hatcalX_1)$, let
\begin{align}
\nn&\alpha_2(x,y,\hatx_1|Q_{\hatX_2|Y\hatX_1})\\*
&:=\Big\{\mathbb{E}_{Q_{\hatX_2|Y\hatX_1}}\Big[\exp(-\lambda_2^*d_2(X^n,\hatX_2^n))\Big|Y=y,\hatX_1=\hatx_1\Big]\Big\}^{-1},\label{def:a2q4kaspi}
\end{align}
and let
\begin{align}
\alpha(x,y|Q_{\hatX_1},Q_{\hatX_2|Y\hatX_1})
&:=\Bigg\{\mathbb{E}_{Q_{\hatX_1}}\Bigg[\frac{\exp\big(-\lambda_1^*d_1(x,\hatX_1)\big)}{\alpha_2(x,y,\hatX_1|Q_{\hatX_2|Y\hatX_1})}\Bigg]\Bigg\}^{-1}\label{def:a1q4kaspi}.
\end{align}

\begin{lemma}
\label{opttest4kaspi}
A pair of conditional distributions $(P_{\hatX_1|XY}^*,P_{\hatX_2|XY\hatX_1}^*)$ achieves $R(P_{XY},D_1,D_2)$ if and only if 
\begin{itemize}
\item  For all $(x,y,\hatx_1)$,
\begin{align}
P_{\hatX_1|XY}^*(\hatx_1|x,y)
&=\frac{\alpha(x,y|P_{\hatX_1}^*,P_{\hatX_2|Y\hatX_1}^*)P_{\hatX_1}^*(\hatx_1)\exp(-\lambda_1^*d_1(x,\hatx_1))}{\alpha_2(x,y,\hatx_1|P_{\hatX_2|Y\hatX_1}^*)}\label{optcond14kaspi},
\end{align}
\item For all $(x,y,\hatx_1,\hatx_2)$ such that $P_{\hatX_1}^*(\hatx_1)>0$,
\begin{align}
P_{\hatX_2|XY\hatX_1}^*(\hatx_2|x,y,\hatx_1)
\nn&=\alpha_2(x,y,\hatx_1|P_{\hatX_2|Y\hatX_1}^*)P_{\hatX_2|Y\hatX_1}^*(\hatx_2|y,\hatx_1)\\*
&\qquad\times\exp(-\lambda_2^*d_2(X^n,\hatX_2^n))\label{optcond24kaspi}.
\end{align}
\end{itemize}
Furthermore, if the pair of distributions $(P_{\hatX_1|XY}^*,P_{\hatX_2|XY\hatX_1}^*)$ achieves $R(P_{XY},D_1,D_2)$,
\begin{align}
R(P_{XY},D_1,D_2)&=\mathbb{E}[\log \alpha(x,y|P_{\hatX_1}^*,P_{\hatX_2|Y\hatX_1}^*)]-\lambda_1^*D_1-\lambda_2^*D_2\label{kaspipara},
\end{align}
\end{lemma}
The proof of Lemma \ref{opttest4kaspi} is similar to \cite[Properties 1-3]{kostina2012converse} for the rate-distortion problem that mainly uses the KKT conditions for convex optimization problems. Lemma \ref{opttest4kaspi} paves the way for the definition of the distortions-tilted information density for the Kaspi problem and also implies critical properties for the Kaspi distortions-tilted information density that parallel Lemma \ref{prop:dtilteddensity} for the rate-distortion problem.

We remark that for any pair of optimal test channels $(P_{\hatX_1|XY}^*,P_{\hatX_2|XY\hatX_1}^*)$, similarly to \cite[Lemma 2]{watanabe2015second}, one can verify that the values of \\$\alpha_2(x,y,\hatx_1|P_{\hatX_2|Y\hatX_1}^*)$ and $\alpha(x,y|P_{\hatX_1}^*,P_{\hatX_2|Y\hatX_1}^*)$ remain the same. Hence, for simplicity, we define
\begin{align}
\alpha_2(x,y,\hatx_1)&:=\alpha_2(x,y,\hatx_1|P_{\hatX_2|Y\hatX_1}^*)\label{def:a24kaspi},\\
\alpha(x,y)&:=\alpha(x,y|P_{\hatX_1}^*,P_{\hatX_2|Y\hatX_1}^*)\label{def:a4kaspi}.
\end{align}
Furthermore, for any $\hatx_1\in\hatcalX_1$ and distribution $Q_{\hatX_2|Y\hatX_1}\in\calP(\hatcalX_2|\calY,\hatcalX_1)$, define the following function
\begin{align}
\nn&\nu(\hatx_1,Q_{\hatX_2|Y\hatX_1})\\*
&:=\mathbb{E}_{P_{XY}\times Q_{\hatX_2|Y\hatX_1}}\Big[\alpha(X,Y)\exp(-\lambda_1^*d_1(X,\hatX_1)
-\lambda_2^*d_2(X,\hatX_2))\big|\hatX_1=\hatx_1\Big]\label{def:nu2}.
\end{align}
The following lemma holds.
\begin{lemma}
\label{nule1}
For any $\hatx_1$ and arbitrary distribution $Q_{\hatX_2|Y\hatX_1}$, we have
\begin{align}
\nu(\hatx_1,Q_{\hatX_2|Y\hatX_1})&\leq 1.\label{nu2le1}
\end{align}
\end{lemma}
We remark that Lemma \ref{nule1} holds for both discrete and continuous memoryless sources. The proof of Lemma \ref{nule1} is inspired by \cite[Lemma 1.4]{csiszar1974}, \cite[Lemma 5]{tuncel2003comp} and \cite{kostina2019sr} and available in \cite[Appendix A]{zhou2017non}. Invoking Lemma \ref{nule1}, we prove a non-asymptotic converse bound for the Kaspi problem in Theorem \ref{fblconverse4kaspi}.

\section{Distortions-Tilted Information Density} 

Now we introduce the distortions-tilted information density for the Kaspi problem that generalizes distortion-tilted information density for the lossy source coding problem \cite{kostina2012fixed,ingber2011dispersion}. Recall the definition of $\alpha(\cdot)$ in \eqref{def:a4kaspi}. 

\begin{definition}
For any $(x,y)\in\calX\times\calY$, the $(D_1,D_2)$-tilted information density for the Kaspi problem is defined as
\begin{align}
\jmath_{\mathrm{K}}(x,y|D_1,D_2,P_{XY})
&:=\log \alpha(x,y)-\lambda_1^*D_1-\lambda_2^*D_2\label{def:kaspitilt}.
\end{align}
\end{definition}

The properties of the $(D_1,D_2)$-tilted information density follows from Lemma \ref{opttest4kaspi}. For example, invoking \eqref{optcond14kaspi} and \eqref{optcond24kaspi}, we conclude that for all $(x,y,\hatx_1,\hatx_2)$ such that $P_{\hatX_1}^*(\hatx_1)P_{\hatX_2|Y\hatX_1}^*(\hatx_2|y,\hatx_1)>0$,
\begin{align}
\jmath_{\mathrm{K}}(x,y|D_1,D_2,P_{XY})
\nn&=\log\frac{P_{\hatX_1|XY}^*(\hatx_1|x,y)}{P_{\hatX_1}^*(\hatx_1)}+\log\frac{P_{\hatX_2|XY\hatX_1}^*(\hatx_2|x,y,\hatx_1)}{P_{\hatX_2|Y\hatX_1}^*(\hatx_2|y,\hatx_1)}\\
&\qquad+\lambda_1^*(d_1(x,\hatx_1)-D_1)+\lambda_2^*(d_2(X^n,\hatX_2^n)-D_2)\label{expandj4kaspi}.
\end{align}
Furthermore, it follows from \eqref{kaspipara} that
\begin{align}
R(P_{XY},D_1,D_2)=\bbE[\jmath_{\mathrm{K}}(X,Y|D_1,D_2,P_{XY})]\label{kaspi:avgprop}.
\end{align}

Finally, we have the following lemma that further relates the distortions-tilted information density with the derivative of the Kaspi rate-distortion function with respect to the distribution $P_{XY}$. Given a joint probability mass function $P_{XY}$, recall that $m=|\supp(P_{XY})|$ and $\Gamma(P_{XY})$ is the sorted distribution such that for each $i\in[m]$, $\Gamma_i(P_{XY})=P_{XY}(x_i,y_i)$ is the $i$-th largest value of $\{P_{XY}(x,y):~(x,y)\in\calX\times\calY\}$.
\begin{lemma}
\label{kaspifirstderive}
Suppose that for all $Q_{XY}$ in the neighborhood of $P_{XY}$, $\mathrm{supp}(Q_{\hat{X}_1\hat{X}_2}^*)=\mathrm{supp}(P_{\hat{X}_1\hat{X}_2}^*)$. Then, for each $i\in[m-1]$,
\begin{align}
\frac{\partial R(Q_{XY},D_1,D_2)}{\partial \Gamma_i(Q_{XY})}\Big|_{Q_{XY}=P_{XY}}
&=\jmath(x_i,y_i|D_1,D_2,P_{XY})-\jmath(x_m,y_m|D_1,D_2,P_{XY}).
\end{align}
\end{lemma}
The proof of Lemma \ref{kaspifirstderive} is available in \cite[Appendix I]{zhou2017non}.
Lemma \ref{kaspifirstderive} parallels Claim (iv) in Lemma \ref{prop:dtilteddensity} for the rate-distortion problem and is critical in the achievability proof of second-order asymptotics.

\section{A Non-Asymptotic Converse Bound}
Invoking Lemma \ref{nule1}, we obtain the following non-asymptotic converse bound for the Kaspi problem that generalizes Theorem \ref{converse:fbl} for the rate-distortion problem.
\begin{theorem}
\label{fblconverse4kaspi}
Given any $\gamma>0$, the joint excess-distortion probability of any $(n,M)$-code for the Kaspi problem satisfies
\begin{align}
\nn\rmP_{\rme,n}(D_1,D_2)&\geq \Pr\Big\{\sum_{i\in[n]}\jmath(X_i,Y_i|D_1,D_2,P_{XY})\geq \log M+n\gamma\Big\}\\*
&\qquad-\exp(-n\gamma)\label{nshot}.
\end{align}
\end{theorem}
We remark that Theorem \ref{fblconverse4kaspi} plays a central role in the converse proof the second-order asymptotics and holds for any memoryless sources.
\begin{proof}
The proof of Theorem \ref{fblconverse4kaspi} is similar to that of Theorem \ref{converse:fbl}.
Given any $(n,M)$-code with encoder $f$ and decoders $(\phi_1,\phi_2)$, let $S=f(X^n)$ be the compressed index that takes values in $\calM$, let $P_{S|X^n}$ be the conditional distribution induced by the encoder $f$ and let $P_{\hatX_1^n|S}$ and let the conditional distributions $P_{\hatX_1^n|S}$ and $P_{\hatX_2^n|S,Y^n}$ be induced by the decoders $\phi_1$ and $\phi_2$, respectively. Furthermore, let $Q_S$ be the uniform distribution over $\calM$ and let
\begin{align}
Q_{\hatX_1^n}(\hatx_1^n)&:=\sum_{s\in\calM}Q_S(s)P_{\hatX_1^n|S}(\hatx_1^n|s),\\
Q_{\hatX_2^n|Y^n}(\hatx_2^n|y^n)&:=\frac{\sum_s Q_S(s)P_{\hatX_1^n|S}(\hatx_1^n|s)P_{\hatX_2^n|SY^n}(\hatx_2^n|s,y^n)}{Q_{\hatX_1^n}(\hatx_1^n)}.
\end{align}
For ease of notation, we use $\calC(D_1,D_2)$ to denote the non-excess-distortion event, i.e., the event that $\{d_1(X^n,\hatX_1^n)\leq D_1,~d_2(X^n,\hatX_2^n)\leq D_2\}$ and use $\calE(D_1,D_2)$ to denote the excess-distortion event $\{d_1(X^n,\hatX_1^n)> D_1\mathrm{~or~}d_2(X^n,\hatX_2^n)> D_2\}$.  For any $\gamma>0$, it follows that
\begin{align}
\nn&\Pr\Big\{\sum_{i\in[n]}\jmath(X_i,Y_i|D_1,D_2,P_{XY})\geq \log M+n\gamma\Big\}\\
\nn&\leq \Pr\Big\{\sum_{i\in[n]}\jmath(X_i,Y_i|D_1,D_2,P_{XY})\geq \log M+n\gamma\mathrm{~and~}\calC(D_1,D_2)\Big\}\\
&\quad+\Pr\left\{\calE(D_1,D_2)\right\}\label{tobeupperbd},
\end{align}
where the second term in \eqref{tobeupperbd} is exactly the joint excess-distortion probability $\rmP_{\rme,n}(D_1,D_2)$.

The first term in \eqref{tobeupperbd} can be upper bounded as follows:
\begin{align}
\nn&\Pr\Big\{\sum_{i\in[n]}\jmath(X_i,Y_i|D_1,D_2,P_{XY})\geq \log M+n\gamma,~\calC(D_1,D_2)\Big\}\\
&=\Pr\Big\{M\leq \exp\big(\sum_{i\in[n]}\jmath(X_i,Y_i|D_1,D_2,P_{XY})-n\gamma\big)\bbo(\calC(D_1,D_2))\Big\}\\
&\leq \frac{\exp(-\gamma)}{M}\mathbb{E}\Big[\exp\big(\sum_{i\in[n]}\jmath(X_i,Y_i|D_1,D_2,P_{XY})\big)\bbo(\calC(D_1,D_2))\Big]\label{usemarkovineq}\\
&\leq \frac{\exp(-n\gamma)}{M}\mathbb{E}\Big[\exp\Big(\sum_{i\in[n]}\jmath(X_i,Y_i|D_1,D_2,P_{XY})+\sum_{i\in[2]}\lambda_i^*(D_i-d_i(X^n,\hatX_i^n))\Big)\Big]\label{nonnegativelambda},\\
\nn&=\exp(-n\gamma)\sum_{s}\sum_{(x^n,y^n)}\sum_{\hatx_1^n,\hatx_2^n}Q_S(s)P_{XY}^n(x^n,y^n)P_{S|X^n}(s|x^n)P_{\hatX_1^n|S}(\hatx_1^n|s)\\*
&\qquad\times P_{\hatX_2^n|Y^n,S}(\hatx_2^n|y^n,s)\prod_{i\in[n]}\alpha(x_i,y_i)\times\exp(-\lambda_1^*d(x_i,\hatx_{1,i})-\lambda_2^*(d(x_i,\hatx_{2,i})))\label{usealpha}\\
\nn&\leq \exp(-n\gamma)\sum_{(x^n,y^n)}\sum_{\hatx_1^n,\hatx_2^n}Q_{\hatX_1^n}(\hatx_1^n)Q_{\hatX_2^n|\hatX_1^n}(\hatx_1^n|
\hatx_1^n)P_{XY}^n(x^n,y^n)\label{usedefQandP<1}\\*
&\qquad\times \prod_{i\in[n]}\alpha(x_i,y_i)\exp(-\lambda_1^*d(x_i,\hatx_{1,i})-\lambda_2^*(d(x_i,\hatx_{2,i})))\\
&=\exp(-n\gamma)\sum_{\hatx_1^n}Q_{\hatX_1^n}(\hatx_1^n)\Big[\prod_{i\in[n]}\nu(\hatx_{1,i},Q_{\hatX_{2,i}|Y_i\hatX_{1,i}})\Big]\label{defanotherQ}\\
&\leq \exp(-n\gamma)\label{usesumine},
\end{align}
where \eqref{usemarkovineq} follows from Markov's inequality and \eqref{nonnegativelambda} follows since $\lambda_i^*\geq 0$ for $i\in[2]$, \eqref{usealpha} follows from the definitions of $\alpha(\cdot)$ in \eqref{def:a4kaspi} and $\jmath(\cdot)$ in \eqref{def:kaspitilt}, \eqref{usedefQandP<1} follows from the fact $P_{S|X^n}(s|x^n)\leq 1$ and the definitions of distributions $(Q_{\hatX_1^n},Q_{\hatX_2^n|\hatX_1^n},Q_{\hatX_2^n|Y^n})$, \eqref{defanotherQ} since we define $Q_{\hatX_{2,i}|Y_i\hatX_{1,i}}$ as the marginal distribution $Q_{\hatX_{2,i}|Y_i}$ of $Q_{\hatX_2^n|Y^n}$ and use the definition of $\nu(\cdot)$ in \eqref{def:nu2} and \eqref{usesumine} follows from the result in \eqref{nu2le1}.

The proof of Theorem  \ref{fblconverse4kaspi} is completed by combining \eqref{tobeupperbd} and \eqref{usesumine}.
\end{proof}

\section{Second-Order Asymptotics}
In this section, we define and present second-order asymptotics of the Kaspi problem for a DMS under bounded distortion measures. In other words, we assume that $\calX$, $\calY$, $\hat{\calX_1}$, $\hat{\calX_2}$ are all finite sets and $\max_{x,\hatx_i}d_i(x,\hatx_i)~,i\in[2]$ is finite.

\subsection{Definition, Main Result and Discussions}
Let $\varepsilon\in(0,1)$ be fixed.
\begin{definition}
\label{secondorder}
A rate $L$ is said to be second-order $(D_1,D_2,\varepsilon)$-achievable for the Kaspi problem if there exists a sequence of $(n,M)$-codes such that
\begin{align}
\limsup_{n\to\infty} \frac{\log M-nR(P_{XY},D_1,D_2)}{\sqrt{n}}\leq L,
\end{align}
and
\begin{align}
\limsup_{n\to\infty} \rmP_{\rme,n}(D_1,D_2)\leq \varepsilon.
\end{align}
The infimum second-order $(D_1,D_2,\varepsilon)$-achievable rate is called the optimal second-order coding rate and denoted as $L^*(D_1,D_2,\varepsilon)$.
\end{definition}
Note that in Definition \ref{firstorder} of the rate-distortion region, the average distortion criterion is used, while in Definition \ref{secondorder}, the excess-distortion probability is considered. The reason is that for second-order asymptotics, second-order asymptotics always companies with the probability of a certain event. To be specific, the excess-distortion probability plays a similar role as error probability for the lossless source coding problem~\cite{hayashi2008} or channel coding problems~\cite{hayashi2009information,polyanskiy2010thesis}. Let $\rmV(D_1,D_2,P_{XY})$ be the distortions-dispersion function for the Kaspi problem, i.e., 
\begin{align}
\rmV(D_1,D_2,P_{XY})&:=\mathrm{Var}\big[\jmath(X,Y|D_1,D_2,P_{XY})\big]\label{dispersion4kaspi}.
\end{align}

We impose following conditions:
\begin{enumerate}
\item\label{kaspi:cond1} The distortion levels are chosen such that $R(P_{XY},D_1,D_2)>0$ is finite;
\item\label{kaspi:condend} $Q_{XY}\to R(Q_{XY},D_1,D_2)$ is twice differentiable in the neighborhood of $P_{XY}$ and the derivatives are bounded.
\end{enumerate}
\begin{theorem}
\label{kaspisecond}
Under conditions (\ref{kaspi:cond1}) and (\ref{kaspi:condend}), the optimal second-order coding rate for the Kaspi problem is
\begin{align}
L^*(D_1,D_2,\varepsilon)=\sqrt{\rmV(D_1,D_2,P_{XY})}\rmQ^{-1}(\varepsilon).
\end{align}
\end{theorem}
The converse proof of Theorem \ref{kaspisecond} follows by applying the Berry-Esseen Theorem to the non-asymptotic bound in Theorem \ref{fblconverse4kaspi}. In the achievability proof, we first prove a type-covering lemma tailored for the Kaspi problem. Subsequently, we make use of the properties of $\jmath(x,y|D_1,D_2,P_{XY})$ in Lemma \ref{opttest4kaspi} and appropriate Taylor expansions.

We remark that the distortions-tilted information density for the Kaspi problem $\jmath(x,y|D_1,D_2,P_{XY})$ reduces to the distortion-tilted information density for the lossy source coding problem~\cite{kostina2012fixed}, or the distortion-tilted information density for the lossy source coding problem with encoder and decoder side information~\cite{le2014second} for particular choices of distortion levels $(D_1,D_2)$. Hence, our result in Theorem \ref{kaspisecond} is a strict generalization of the second-order coding rate for the lossy source coding problem~\cite{kostina2012fixed} and the conditional lossy source coding problem~\cite{le2014second} for a DMS under bounded distortion measures. We also illustrate this point in Section \ref{kaspi:dsbs} via a numerical example for the doubly symmetric binary source.

In the next two subsections, we illustrate Theorem \ref{kaspisecond} via two numerical examples by calculating the second-order coding rate $L^*(D_1,D_2,\varepsilon)$ in close form.

\subsection{Numerical Examples}
\subsubsection{Asymmetric Correlated Source}
In order to illustrate our results in Lemma \ref{opttest4kaspi} and Theorem \ref{kaspisecond}, we consider the following source. Let $\calX=\{0,1\}$, $\calY=\{0,1,\rme\}$ and $P_X(0)=P_X(1)=\frac{1}{2}$. Let $Y$ be the output of passing $X$ through a Binary Erasure Channel (BEC) with erasure probability $p$, i.e., $P_{Y|X}(y|x)=1-p$ if $x=y$ and $P_{Y|X}(\rme|x)=p$. The explicit formula of the Kaspi rate-distortion function for the above correlated source under Hamming distortion measures was derived by Perron, Diggavi and Telatar in \cite{perron2006kaspi}. Here we only recall the \emph{non-degenerate} result, i.e., the case where the distortion levels $(D_1,D_2)$ are chosen such that $\lambda_1^*>0$ and $\lambda_2^*>0$.

Define the set
\begin{align}
\calD_{\mathrm{bec}}
&:=\Big\{(D_1,D_2)\in\bbR_+^2: D_1\leq \frac{1}{2},~D_1-\frac{1-p}{2}\leq D_2\leq p D_1\Big\}.
\end{align}
\begin{lemma}
If $(D_1,D_2)\in\calD_{\mathrm{bec}}$, then the Kaspi rate-distortion function for the above asymmetric correlated source under Hamming distortion measures is
\begin{align}
R(P_{XY},D_1,D_2)
&=\log 2-(1-p)H_b\Bigg(\frac{D_1-D_2}{1-p}\Bigg)-pH_b\Bigg(\frac{D_2}{p}\Bigg).
\end{align}
\end{lemma}
Hence, for $(D_1,D_2)\in\calD_{\mathrm{bec}}$, using the definitions of $\lambda_1^*$ in \eqref{dualoptimal1} and $\lambda_2^*$ in \eqref{dualoptimal2}, we obtain
\begin{align}
\lambda_1^*&=\log\frac{(1-p)-(D_1-D_2)}{1-p}-\log\frac{D_1-D_2}{1-p}\\
&=\log\frac{(1-p)-(D_1-D_2)}{D_1-D_2}\label{call1*},\\
\lambda_2^*
&=\log\frac{p-D_2}{p}+\log\frac{D_1-D_2}{1-p}-\log\frac{(1-p)-(D_1-D_2)}{1-p}-\log\frac{D_2}{p}\\
&=-\lambda_1^*+\log\frac{p-D_2}{D_2}\label{call2*}.
\end{align}
Then, using the definitions of $\alpha_2(\cdot)$ in \eqref{def:a24kaspi} and $\alpha(\cdot)$ in \eqref{def:a4kaspi}, we have
\begin{align}
\alpha_2(0,0,0)
\nn&=\alpha_2(0,0,1)=\alpha_2(1,1,0)=\alpha_2(1,1,1)\\
&=\alpha_2(0,\rme,0)=\alpha_2(1,\rme,1)=1,\\
\alpha_2(1,0,0)\nn&=\alpha_2(1,0,1)=\alpha_2(0,1,0)=\alpha_2(0,1,1)\\
&=\alpha_2(1,\rme,0)=\alpha_2(0,\rme,1)=\exp(\lambda_2^*),
\end{align}
and
\begin{align}
\alpha(0,0)&=\alpha(1,1)=\frac{2}{1+\exp(-\lambda_1^*)},\\
\alpha(0,\rme)&=\alpha(1,\rme)=\frac{2}{1+\exp(-\lambda_1^*-\lambda_2^*)}.
\end{align}

It can be verified easily that \eqref{optcond14kaspi}, \eqref{optcond24kaspi}, \eqref{kaspipara} hold. In the following, we will verify that \eqref{nu2le1} holds for arbitrary $Q_{\hatX_2|Y\hatX_1}$ and $\hatx_1$. As a first step, we can verify that for any $(y,\hatx_1,\hatx_2)$, we have
\begin{align}
\nn&\sum_x P_{XY}(x,y)\alpha(x,y)\exp(-\lambda_1^*d_1(x,\hatx_1)-\lambda_2^*d_2(X^n,\hatX_2^n))\\*
&\leq \sum_x P_{XY}(x,y)\frac{\alpha(x,y)}{\alpha_2(x,y,\hatx_1)}\exp(-\lambda_1^*d_1(x,\hatx_1))\label{verify1}.
\end{align}
Then, for any distribution $Q_{\hatX_2|Y\hatX_1}$, using the definition of $\nu(\cdot)$ in \eqref{def:nu2}, multiplying $Q_{\hatX_2|Y\hatX_1}(\hatx_2|y,\hatx_1)$ over both sides of \eqref{verify1}, and summing over $(y,\hatx_2)$, we obtain that
\begin{align}
\nu(\hatx_1,Q_{\hatX_2|Y\hatX_1})\leq 1\label{verify2}.
\end{align}

Using the definition of $\jmath(\cdot)$ in \eqref{def:kaspitilt}, we have
\begin{align}
\jmath(0,0|D_1,D_2,P_{XY})
&=\jmath(1,1|D_1,D_2,P_{XY})\\
&=\log \alpha(0,0)-\lambda_1^*D_1-\lambda_2^*D_2\label{j1cal4kaspi},
\end{align}
and
\begin{align}
\jmath(0,\rme|D_1,D_2,P_{XY})
&=\jmath(1,\rme|D_1,D_2,P_{XY})\\
&=\log \alpha(0,\rme)-\lambda_1^*D_1-\lambda_2^*D_2\label{j2cal4kaspi}.
\end{align}
Furthermore, using the definition of the distortion-dispersion function $\rmV(D_1,D_2,P_{XY})$ in \eqref{dispersion4kaspi}, we have
\begin{align}
\rmV(D_1,D_2,P_{XY})
&=\mathrm{Var}[\jmath(X,Y|D_1,D_2,P_{XY})]\\
&=p(1-p)\Bigg(\log\frac{p-D_2}{p}-\log\frac{(1-p)-(D_1-D_2)}{1-p}\Bigg)^2.
\end{align}
Thus,
\begin{align}
L^*(D_1,D_2,\varepsilon)=\sqrt{\rmV(D_1,D_2,P_{XY})}\rmQ^{-1}(\varepsilon).
\end{align}

\subsubsection{Doubly Symmetric Binary Source (DSBS)}
\label{kaspi:dsbs}
In this example, we show that under certain distortion levels, the Kaspi rate-distortion function reduces to the rate-distortion function~\cite{shannon1959coding} (see also \cite[Theorem 3.5]{el2011network}) and the conditional rate-distortion function~\cite[Eq. (11.2)]{el2011network}. We consider the DSBS where $\calX=\calY=\{0,1\}$, $P_{XY}(0,0)=P_{XY}(1,1)=\frac{1-p}{2}$ and $P_{XY}(0,1)=P_{XY}(1,0)=\frac{p}{2}$ for some $p\in[0,\frac{1}{2}]$. 
\begin{lemma}
\label{kaspidsbs}
Depending on the distortion levels $(D_1,D_2)$, the Kaspi rate-distortion function for the DSBS with Hamming distortion measures satisfies
\begin{itemize}
\item $D_1\geq \frac{1}{2}$ and $D_2\geq p$
\begin{align}
R(P_{XY},D_1,D_2)=0.
\end{align}
\item $D_1<\frac{1}{2}$ and $D_2\geq \min\{p,D_1\}$
\begin{align}
R(P_{XY},D_1,D_2)=\log 2-H_b(D_1),
\end{align}
where $H_b(x)=-x\log x-(1-x)\log (1-x)$ is the binary entropy function.
\item $D_1\geq D_2+\frac{1-2p}{2}$ and $D_2<p$ 
\begin{align}
R(P_{XY},D_1,D_2)&=H_b(p)-H_b(D_2).
\end{align}
\end{itemize}
\end{lemma}

When $D_1<\frac{1}{2}$ and $D_2<pD_1$, the Kaspi rate-distortion function reduces to the rate-distortion function for the lossy source coding problem. Thus, the distortion-tilted information density for the Kaspi problem reduces to the $D_1$-tilted information density in \eqref{def:d1tilt}, i.e.,
\begin{align}
\jmath(x,y|D_1,D_2,P_{XY})&=\log 2-H_b(D_1)\label{def:d1tilt}.
\end{align} 
Hence, $L^*(D_1,D_1|P_{XY})=0$. When $D_1\geq D_2+\frac{1-2p}{2}$ and $D_2<p$, the Kaspi rate-distortion function reduces to the conditional rate-distortion function. Under the optimal test channel, we have $\hatX_1=0/1$ and $X\to\hatX_2\to Y$ forms a Markov chain. In this case, the distortion-tilted information density for the Kaspi problem reduces to the conditional distortion-tilted information density~\cite[Definition 5]{kostina2012converse} (see also \cite{le2014second}), i.e., 
\begin{align}
\jmath(x,y|D_1,D_2,P_{XY})
&=-\log P_{X|Y}(x|y)-H_b(D_2).
\end{align}
Hence,
\begin{align}
\rmV(D_1,D_2,P_{XY})
&=\mathrm{Var}[-\log P_{X|Y}(X|Y)]\\
&=(1-p)(-\log(1-p)-H_b(p))^2+p(-\log p-H_b(p))^2\\
&:=\rmV(p),
\end{align}
and
\begin{align}
L^*(D_1,D_2,\varepsilon)=\sqrt{\rmV(p)}\rmQ^{-1}(\varepsilon).
\end{align}

\section{Proof of Second-Order Asymptotics}
\subsection{Achievability}
We first prove a type covering lemma for the Kaspi problem, based on which we derive an upper bound on the excess-distortion probability. Subsequently, using the Berry-Esseen theorem together with proper Taylor expansions, we manage to prove the desired achievable second-order coding rate.

To present our type covering lemma, define the following constant
\begin{align}
c=\Big(8|\calX|\cdot|\calY|\cdot|\hat{\calX}_1|\cdot|\hat{\calX}_2|+6\Big).
\end{align}

\begin{lemma}
\label{coveringkaspi}
There exists a set $\calB\subset\hatcalX_1^n$ such that for each $(x^n,y^n)\in\calT_{Q_{XY}}$, if 
\begin{align}
(z^n)^*=\argmin_{\hat{x}_1^n\in\calB}d_1(x^n,\hat{x}_1^n),
\end{align}
then the following conclusion hold.
\begin{enumerate}
\item the distortion between $x^n$ and $(z^n)^*$ is upper bounded by $D_1$, i.e., 
\begin{align}
d_1(x^n,(z^n)^*)\leq D_1,
\end{align}
\item there exists a set $\calB((z^n)^*,y^n)\subset\hat\calX_2^n$ such that
\begin{align}
\min_{\hat{x}_2^n\in\calB((z^n)^*,y^n)}d_2(x^n,\hat{x}_2^n)\leq D_2.
\end{align}
\item and the size of the set $\calB\cup\calB((z^n)^*,y^n)$ satisfies
\begin{align}
\nn&\log \big|\calB\cup\calB((z^n)^*,y^n)\big|\\*
&\leq nR(Q_{XY},D_1,D_2)+c\log (n+1)\label{newinapp}.
\end{align}
\end{enumerate}
\end{lemma}
The proof of Lemma \ref{coveringkaspi} is similar to the proof of type covering lemmas for rate-distortion problem.

Invoking Lemma \ref{coveringkaspi}, we can upper bound the excess-distortion probability of an $(n,M)$-code. To do so, for any $(n,M)\in\bbN^2$, define
\begin{align}
R_n:=\frac{1}{n}\log M-(c+|\calX|\cdot|\calY|)\frac{\log (n+1)}{n}\label{defrn}.
\end{align}
\begin{lemma}
\label{uppexcess:kaspi}
There exists an $(n,M)$-code whose excess-distortion probability satisfies
\begin{align}
\rmP_{\rme,n}(D_1,D_2)&\leq \Pr\Big\{R_n<R(\hat{T}_{X^nY^n},D_1,D_2)\Big\}.
\end{align}
\end{lemma}
\begin{proof}
Consider the following coding scheme. Given source sequence pair $(x^n,y^n)$, the encoder first calculates the joint type $\hat{T}_{x^ny^n}$, which can be transmitted reliably using at most $|\calX|\cdot|\calY|\log (n+1)$ nats. Then the encoder calculates $R(\hat{T}_{x^ny^n},D_1,D_2)$ and declares an error if $nR(\hat{T}_{x^ny^n},D_1,D_2)+c\log (n+1)+|\calX|\cdot|\calY|\log (n+1)>\log M$. Otherwise, the encoder chooses a set $\calB$ satisfying the properties specified in Lemma \ref{coveringkaspi} and sends the index of $(z^n)^*=\argmin_{\hat{x}_1^n\in\calB} d_1(x^n,\hat{x}_1^n)$. Subsequently, the decoder chooses a set $\calB((z^n)^*,y^n)$ satisfying the properties specified in Lemma \ref{coveringkaspi} and sends the index of \\$\argmin_{\hat{x}_2^n\in\calB((z^n)^*,y^n)}d_2(x^n,\hat{x}_2^n)$. Lemma \ref{coveringkaspi} implies that the decoding is error free if $nR(\hat{T}_{x^ny^n},D_1,D_2)+c\log (n+1)+|\calX|\cdot|\calY|\log (n+1)\leq \log M$. The proof of Lemma \ref{uppexcess:kaspi} is now completed.
\end{proof}

Given any distribution $P_{XY}$ on the finite set $\calX\times\calY$, define the typical set
\begin{align}
\calA_n(P_{XY})
&:=\Bigg\{Q_{XY}\in\calP_n(\calX\times\calY):\|Q_{XY}-P_{XY}\|_{\infty}\leq \sqrt{\frac{\log n}{n}}\Bigg\}\label{def:calan}.
\end{align}
It follows from \cite[Lemma 22]{tan2014state} that
\begin{align}
\Pr\Big\{\hat{T}_{X^nY^n}\notin \calA_n(P_{XY})\Big\}&\leq \frac{2|\calX||\calY|}{n^2}\label{atypicalprob}.
\end{align}
If we choose
\begin{align}
\frac{1}{n}\log M
&=R(P_{XY},D_1,D_2)+\frac{L}{\sqrt{n}}+\Big(c+|\calX|\cdot|\calY|\Big)\frac{\log (n+1)}{n},
\end{align}
then
\begin{align}
R_n&=R(P_{XY},D_1,D_2)+\frac{L}{\sqrt{n}}.
\end{align}

For any $(x^n,y^n)$ such that $\hat{T}_{x^ny^n}\in\calA_n(P_{XY})$, since the mapping \\$Q_{XY}\to R(Q_{XY},D_1,D_2)$ is twice differentiable in the neighborhood of $P_{XY}$ and the derivative is bounded, applying Taylor expansion of $R(\hat{T}_{x^ny^n},D_1,D_2)$ around $\hat{T}_{x^ny^n}=P_{XY}$ and using Lemma \ref{kaspifirstderive}, we have
\begin{align}
R(\hat{T}_{x^ny^n},D_1,D_2)
&=\frac{1}{n}\sum_{i\in[n]}\jmath(x_i,y_i|D_1,D_2,P_{XY})+O\Big(\frac{\log n}{n}\Big)\label{taylorexpand}.
\end{align}

Define $\xi_n=\frac{\log n}{n}$. It follows from Lemma \ref{uppexcess:kaspi} that
\begin{align}
\nn&\rmP_{\rme,n}(D_1,D_2)\\*
\nn&\leq \Pr\Big\{R_n<R(\hat{T}_{X^nY^n},D_1,D_2),\hat{T}_{X^nY^n}\in\calA_n(P_{XY})\Big\}\\*
&\qquad+\Pr\Big\{\hat{T}_{X^nY^n}\notin\calA_n(P_{XY})\Big\}\\
\nn&\leq \Pr\Big\{R(P_{XY},D_1,D_2)+\frac{L}{\sqrt{n}}\\*
&\qquad\quad<\frac{1}{n}\sum_{i\in[n]} \jmath(X_i,Y_i|D_1,D_2,P_{XY})+O(\xi_n)\Big\}+\frac{2|\calX||\calY|}{n^2}\label{uselemma&atypical}\\
\nn&\leq \Pr\bigg\{\frac{1}{\sqrt{n}}\sum_{i\in[n]} \big(\jmath(X_i,Y_i|D_1,D_2,P_{XY})\\*
&\qquad\qquad-R(P_{XY},D_1,D_2)\big)>L+O(\xi_n\sqrt{n})\bigg\}+\frac{2|\calX||\calY|}{n^2}\\
&\leq \rmQ\Bigg(\frac{L+O(\xi_n\sqrt{n})}{\sqrt{\rmV(D_1,D_2|P_{XY})}}\Bigg)+\frac{6\rmT(D_1,D_2|P_{XY})}{\sqrt{n}\rmV^{3/2}(D_1,D_2|P_{XY})}+\frac{2|\calX||\calY|}{n^2}\label{achfinal},
\end{align}
where \eqref{uselemma&atypical} follows from the results in \eqref{atypicalprob} and Lemma \ref{taylorexpand} and \eqref{achfinal} follows from Berry-Esseen theorem, where $\rmT(D_1,D_2|P_{XY})$ is the third absolute moment of $\jmath(X,Y|D_1,D_2,P_{XY})$, which is finite for a DMS.

Therefore, if $L$ satisfies 
\begin{align}
L\geq \sqrt{\rmV(D_1,D_2|P_{XY})}\rmQ^{-1}(\varepsilon),
\end{align}
by noting that $O(\xi_n\sqrt{n})=O(\log n/\sqrt{n})$, it follows that
\begin{align}
\limsup_{n\to\infty} \rmP_{\rme,n}(D_1,D_2)\leq \varepsilon.
\end{align}
Thus, the optimal second-order coding rate satisfies 
\begin{align}
L^*(\varepsilon,D_1,D_2)\leq\sqrt{\rmV(D_1,D_2|P_{XY})}Q^{-1}(\varepsilon).
\end{align}

\subsection{Converse}
The converse part follows by applying the Berry-Esseen theorem to the non-asymptotic converse bound in Theorem \ref{fblconverse4kaspi}. Let 
\begin{align}
\log M:=nR(P_{XY},D_1,D_2)+L\sqrt{n}-\frac{1}{2}\log n.
\end{align}

Invoking \eqref{nshot} with $\varepsilon=\frac{\log n}{2n}$, we obtain 
\begin{align}
\nn&\rmP_{\rme,n}(D_1,D_2)+\frac{1}{\sqrt{n}}\\*
&\geq \Pr\Big\{\sum_{i\in[n]}\jmath(x,y|D_1,D_2,P_{XY})\geq nR(P_{XY},D_1,D_2)+L\sqrt{n}\Big)\\
&\geq \rmQ\bigg(\frac{L}{\sqrt{\rmV(D_1,D_2|P_{XY})}}\bigg)-\frac{6\rmT(D_1,D_2|P_{XY})}{\sqrt{n}\rmV^{3/2}(D_1,D_2|P_{XY})}\label{conversefinal},
\end{align}
where \eqref{conversefinal} follows from the Berry-Esseen theorem. If 
\begin{align}
L<\sqrt{\rmV(D_1,D_2|P_{XY})}Q^{-1}(\varepsilon),
\end{align}
then
\begin{align}
\limsup_{n\to\infty} \rmP_{\rme,n}(D_1,D_2)>\varepsilon.
\end{align}
The converse proof is thus completed.

\chapter{Successive Refinement}
\label{chap:sr}

In this chapter, we study the successive refinement problem with two encoders and two decoders, which generalizes the rate-distortion problem by introducing an additional pair of encoders and decoders. Based on the encoding process of the original encoder, the additional encoder further compresses the source sequence and the additional decoder uses compressed information from both encoders to produce a finer estimate of the source sequence than the first decoder that only accesses the original encoder. The optimal rate-distortion region for a DMS under bounded distortion measures was derived by Rimoldi in~\cite{rimoldi1994}, which collects rate pairs of encoders with vanishing joint excess-distortion probabilities.

Successive refinement is the first lossy source coding problem with multiple encoders studied in this monograph. The successive refinement problem is an information-theoretic formulation of whether it is possible to interrupt a transmission to provide a finer reconstruction of the source sequence without any loss of optimality for lossy compression. For such a problem, in order to derive the second-order asymptotics, we need to study the backoff of the encoders' rates from a boundary rate-point on the rate-distortion region, analogously to the study of the backoff of the encoder's rate from the rate-distortion function in second-order asymptotics for the rate-distortion problem. For a DMS under bounded distortion measures, we derive the optimal second-order coding region under a joint excess-distortion criterion (JEP)~\cite{zhou2016second}. We also recall the second-order asymptotics under the separate excess-distortion probabilities (SEP) criteria by No, Ingber and Weissman~\cite{no2016}. For successively refinable discrete memoryless source-distortion measure triplets~\cite{koshelev1981estimation,equitz1991successive}, under SEP, the second-order region is significantly simplified and the notion of successive refinability~\cite{equitz1991successive,koshelev1981estimation} is generalized to the second-order asymptotic regime under the SEP criterion. This chapter is largely based on \cite{zhou2016second,no2016}.

There are several new insights on the second-order coding region that we can glean when we consider the joint excess-distortion probability (cf.\ Section \ref{sec:examples}). For example, under the joint excess-distortion probability criterion, the second-order region is curved for successively refinable source-distortion triplets, which implies that if one second-order coding rate is small, the other is necessarily large. This reveals a fundamental tradeoff that cannot be observed if one adopts the separate excess-distortion probability criterion. Therefore, in subsequent chapters that involve more complicated multiterminal lossy source coding problems, we only consider the joint excess-distortion probability criterion that better captures the rate tradeoff of multiple encoders.

\section{Problem Formulation and Asymptotic Result}
\label{sec:pf}

\subsection{Problem Formulation}
The successive refinement source coding problem~\cite{rimoldi1994,equitz1991successive} is shown in Figure~\ref{systemmodel_sr}. There are two encoders and two decoders. Encoder $f_i,~i=1,2$ has access to a source sequence $X^n$ and compresses it into a message $S_i,~i=1,2$. Decoder $\phi_1$ aims to recover source sequence $X^n$ under distortion measure $d_1$ and distortion level $D_1$ with the encoded message $S_1$ from encoder $f_1$. The decoder $\phi_2$ aims to recover $X^n$ under distortion measure $d_2$ and distortion level $D_2$ with messages $S_1$ and $S_2$. 
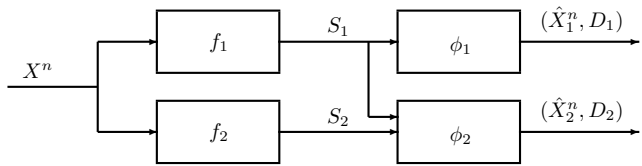
\begin{figure}[htbp]
\centering
\setlength{\unitlength}{0.5cm}
\scalebox{0.8}{
\begin{picture}(26,6)
\linethickness{1pt}
\put(1.5,3.8){\makebox{$X^n$}}
\put(6,1){\framebox(4,2)}
\put(6,4){\framebox(4,2)}
\put(7.7,1.8){\makebox{$f_2$}}
\put(7.7,4.8){\makebox{$f_1$}}
\put(1,3.5){\line(1,0){3}}
\put(4,5){\vector(1,0){2}}
\put(4,2){\line(0,1){3}}
\put(4,2){\vector(1,0){2}}
\put(14,1){\framebox(4,2)}
\put(14,4){\framebox(4,2)}

\put(15.7,4.7){\makebox{$\phi_1$}}
\put(15.7,1.7){\makebox{$\phi_2$}}
\put(10,2){\vector(1,0){4}}
\put(12,2.5){\makebox(0,0){$S_2$}}
\put(10,5){\vector(1,0){4}}
\put(12,5.5){\makebox(0,0){$S_1$}}
\put(13,5){\line(0,-1){2.5}}
\put(13,2.5){\vector(1,0){1}}
\put(18,2){\vector(1,0){4}}
\put(18.7,2.5){\makebox{$(\hatX_2^n,D_2$)}}
\put(18,5){\vector(1,0){4}}
\put(18.7,5.5){\makebox{$(\hatX_1^n,D_1)$}}
\end{picture}}
\caption{System model for the successive refinement problem~\cite{rimoldi1994}.}
\label{systemmodel_sr}
\end{figure}

We consider a memoryless source with distribution $P_X$ supported on a finite alphabet $\calX$. Thus, $X^n$ is an i.i.d. sequence where each $X_i$ is generated according to $P_X$. We assume the reproduction alphabets for decoder $\phi_1,\phi_2$ are respectively alphabets $\hatcalX_1$ and $\hatcalX_2$. We follow the definitions in~\cite{rimoldi1994} for codes and the achievable rate region.
\begin{definition}
\label{code:sr}
An $(n,M_1,M_2)$-code for successive refinement source coding consists of two encoders:
\begin{align}
f_1:\calX^n\to\calM_1=[M_1],\\
f_2:\calX^n\to\calM_2=[M_2],
\end{align}
and two decoders:
\begin{align}
\phi_1&:\calM_1\to \hatcalX_1^n,\\*
\phi_2&:\calM_1\times\calM_2\to \hatcalX_2^n.
\end{align}
\end{definition}

For each $i\in[2]$, define a distortion measure $d_i:\calX\times\hatcalX_i\to[0,\infty)$ and let the distortion between $x^n$ and $\hatx_i^n$ be defined as $d_i(x^n,\hatx_i^n):=\frac{1}{n}\sum_{i\in[n]}d_i(x_i,\hatx_i)$. Define the joint excess-distortion probability as
\begin{align}
\label{defexcessprob_sr}
\rmP_{\rme,n}(D_1,D_2)&:=\Pr\left\{d_1(X^n,\hatX_1^n)>D_1~\mathrm{or}~d_2(X^n,\hatX_2^n)> D_2\right\},
\end{align}
where $\hatX_1^n=\phi_1(f_1(X^n))$ and $\hatX_2^n = \phi_2(f_1(X^n), f_2(X^n))$ are the reconstructed sequences.
\begin{definition}
\label{deffirst_sr}
A rate pair $(R_1,R_2)$ is said to be $(D_1,D_2)$-achievable for the successive refinement source coding if there exists a sequence of $(n,M_1,M_2)$-codes such that
\begin{align}
\limsup_{n\to\infty}\frac{1}{n}\log M_1&\leq R_1,\\*
\limsup_{n\to\infty}\frac{1}{n}\log(M_1M_2)&\leq R_1+R_2\label{eqn:R2},
\end{align}
and
\begin{align}
\lim_{n\to\infty} \rmP_{\rme,n}(D_1,D_2)=0.
\end{align}
The closure of the set of all $(D_1,D_2)$-achievable rate pairs is called optimal  $(D_1,D_2)$-achievable rate region and denoted as $\calR(D_1,D_2|P_X)$.
\end{definition}
Note that in the original work by Rimoldi~\cite{rimoldi1994}, the rate $R_2$ corresponds to the {\em sum} rate $R_1+R_2$. In this monograph, to be consistent with other chapters, we use $R_2$ to denote the rate of message $S_2$ in Figure~\ref{systemmodel_sr}.

\subsection{Rimoldi's Rate-Distortion Region}
The optimal rate region for a DMS with arbitrary distortion measures was characterized in~\cite{rimoldi1994}. Let $\calP(P_X,D_1,D_2)$ be the set of joint distributions $P_{X\hatX_1\hatX_2}$ such that the $\calX$-marginal is $P_X$, $\mathbb{E}[d_1(X,\hatX_1)]\leq D_1$ and $\mathbb{E}[d_{2}(X,\hatX_2)]\leq D_2$. Given $P_{X\hatX_1\hatX_2}\in\calP(P_X,D_1,D_2)$, let
\begin{align}
\nn&\calR(P_{X\hatX_1\hatX_2})\\*
&:=\Big\{(R_1,R_2):R_1\geq I(X;\hatX_1),~R_1+R_2\geq I(X;\hatX_1,\hatX_2)\Big\}.
\end{align}
\begin{theorem}
\label{rimoldi}
The optimal $(D_1,D_2)$-achievable rate region for a DMS with arbitrary distortion measures under successive refinement source coding is
\begin{align}
\calR(D_1,D_2|P_X)=\bigcup_{P_{X\hatX_1\hatX_2}\in\calP(P_X,D_1,D_2)}\calR(P_{X\hatX_1\hatX_2})\label{eqn:opt_rr}.
\end{align}
\end{theorem}

Now we introduce an important quantity for subsequent analyses for a DMS. Given a rate $R_1$ and distortion pair $(D_1,D_2)$, let the minimal sum rate $R_1+R_2$ such that $(R_1,R_2)\in\calR(D_1,D_2|P_X)$ be $\rvR(R_1,D_1,D_2|P_X)$, i.e.,
\begin{align}
\label{minimumr2}
\rvR(R_1,D_1,D_2|P_X)
&:=\min\left\{R_1+R_2:(R_1,R_2)\in\calR(D_1,D_2|P_X)\right\}\\*
&=\inf_{\substack{P_{\hatX_1\hatX_2|X}:\mathbb{E}[d_1(X,\hatX_1)]\leq D_1\\ \mathbb{E}[d_2(X,\hatX_2)]\leq D_2, I(X;\hatX_1)\leq R_1}} I(X;\hatX_1,\hatX_2)\label{minr2},
\end{align}
where~\eqref{minr2} follows from~\cite[Corollary 1]{kanlis1996error}.

Let $R(P_X,D_1)$ and $R(P_X,D_2)$ be the rate-distortion functions~\cite[Chapter~3]{el2011network} (see also \eqref{def:rd}) when the reproduction alphabets are $\hatcalX_1$ and $\hatcalX_2$ respectively, i.e., for each $i\in[2]$,
\begin{align}
R(P_X,D_i) &:=\inf_{P_{\hatX_i|X}:\mathbb{E}[d_i(X,\hatX_i)]\leq D_i} I(X;\hatX_i).
\end{align}
Note that if $R_1<R(P_X,D_1)$, then the convex optimization in~\eqref{minr2} is infeasible. Otherwise, since $\rvR(R_1,D_1,D_2|P_X)$ is a convex optimization problem, the minimization in~\eqref{minr2} is attained for some test channel $P_{\hatX_1\hatX_2|X}$ satisfying  
\begin{align}
\sum_{x,y,z}P_X(x)P_{\hatX_1\hatX_2|X}(\hatx_1,\hatx_2|x)d_1(x,\hatx_1)&=D_1,\\
\sum_{x,y,z}P_X(x)P_{\hatX_1\hatX_2|X}(\hatx_1,\hatx_2|x)d_2(x,\hatx_2)&=D_2,\\
I(P_X,P_{\hatX_1|X})&=R_1.
\end{align}
Therefore, a rate pair $(R_1^*,R_2^*)$ lies on the boundary of the rate-distortion region $\calR(D_1,D_2|P_X)$ if and only if $R_1^*=R(P_X,D_1)$ or $R_1^*+R_2^*=\rvR(R_1^*,D_1,D_2|P_X)$.

\subsection{Successive Refinability}
Next we introduce the notion of a {\em successively refinable source-distortion measure triplet}~\cite{koshelev1981estimation,equitz1991successive}. We recall the definitions with a slight generalization in accordance to~\cite[Definition~2]{no2016}.

\begin{definition}
\label{prop:sr}
Given distortion measures $d_1,d_2$ and a source $X$ with distribution $P_X$, the source-distortion measure triplet $(X,d_1,d_2)$ is said to be $(D_1,D_2)$-successively refinable if the rate pair $(R(P_X,D_1),R(P_X,D_2))$ is $(D_1,D_2)$-achievable. If the source-distortion measure triplet is $(D_1,D_2)$-successively refinable for all $(D_1,D_2)$ such that $R(P_X,D_1)<R(P_X,D_2)$, then it  is said to be successively refinable.
\end{definition}
For a successively refinable source-distortion measure triplet, the minimal sum rate $R_1+R_2$ given $R_1$ in a certain interval is exactly the rate-distortion function (see~\eqref{srminr2} to follow). This reduces the computation of the optimal rate region in~\eqref{eqn:opt_rr}.

Koshelev~\cite{koshelev1981estimation} presented a sufficient condition for a source-distortion measure triplet to be successively refinable while Equitz and Cover~\cite[Theorem 2]{equitz1991successive} presented a  
necessary and sufficient condition which we reproduce below. 

\begin{theorem}
A memoryless source-distortion measure triplet is successively refinable if and only if there exists a conditional distribution $P_{\hatX_1\hatX_2|X}^*$ such that
\begin{align}
R(P_X,D_1)&=I(P_X,P_{\hatX_1|X}^*),~
\mathbb{E}_{P_X\times P_{\hatX_1|X}^*}[d_1(X,\hatX_1)]\leq D_1,\\
R(P_X,D_2)&=I(P_X,P_{\hatX_2|X}^*),~
\mathbb{E}_{P_X\times P_{\hatX_2|X}^*}[d_2(X,\hatX_2)]\leq D_2,
\end{align}
and
\begin{align}
P_{\hatX_1\hatX_2|X}^*=P_{\hatX_1|X}^*P_{\hatX_2|X}^*.
\end{align}
\end{theorem}
In~\cite{equitz1991successive}, it was shown that a DMS with Hamming distortion measures, a GMS with quadratic distortion measures, and a Laplacian source with absolute distortion measures are successively refinable. Note that in the original paper of Equitz and Cover~\cite{equitz1991successive}, the authors only considered the case where both decoders use the same distortion measure, i.e., $d_1=d_2=d$. Interestingly, as pointed out in~\cite[Theorem 4]{no2016}, the result still holds even when $d_1\neq d_2$. This can be verified easily for a DMS by invoking~\cite[Theorem 1]{rimoldi1994}.

\section{Rate-Distortions-Tilted Information Density}

Throughout the section, we assume that $R(P_X,D_1)\leq R_1^*<R(P_X,D_2)$ and $\calR(D_1,D_2|P_X)$ is smooth on a boundary rate pair $(R_1^*,R_2^*)$ of our interest, i.e.,
\begin{align}
\label{deflambda}
\xi^*&:=-\frac{\rvR(R,D_1,D_2|P_X)}{\partial R}\bigg|_{R=R_1^*}, 
\end{align}
is well-defined. Note that $\xi^*\geq 0$ since $\rvR(R_1,D_2,D_2)$ is convex and non-increasing in $R_1$. Further, for  a positive distortion pair $(D_1,D_2)$, define
\begin{align}
\nu_1^*&:=-\frac{\rvR(P_X,R_1,D,D_2)}{\partial D}\bigg|_{D=D_1}\label{defnu1},\\*
\nu_2^*&:=-\frac{\rvR(P_X,R_1,D_1,D)}{\partial D}\bigg|_{D=D_2}\label{defnu2}.
\end{align}
Note that for a successively refinable discrete memoryless source-distortion measure triplet, from~\eqref{srminr2}, we obtain $\xi^*=0$ and $\nu_1^*=0$. Let $P_{\hatX_1\hatX_2|X}^*$ be the optimal test channel achieving $\rvR(R_1,D_1,D_2|P_X)$ in~\eqref{minimumr2} (assuming it is unique)\footnote{If optimal test channels are not unique, then following the proof of \cite[Lemma 2]{watanabe2015second}, we can argue that the tilted information density is still well defined.}. Let $P_{X\hatX_1}^*,P_{X\hatX_2}^*,P^*_{\hatX_1\hatX_2}$, $P_{\hatX_1}^*$, and $P_{\hatX_1|X}^*$ be the induced (conditional) marginal distributions. We are now ready to define the tilted information density for successive refinement source coding problem.

Let $(R_1^*,R_2^*)$ be any boundary rate pair of the rate-distortion region $\calR(D_1,D_2|P_X)$. 
\begin{definition}
For any $x\in\calX$, the rate-distortions tilted information density for the successive refinement problem is defined as
\begin{align}
\label{srtilted}
\jmath(x,R_1^*,D_1,D_2|P_X)
\nn&:=-\log\mathbb{E}_{P_{\hatX_1\hatX_2}^*}\Bigg(\exp\bigg\{-\xi^*\bigg(\log\frac{P_{\hatX_1|X}^*(\hatX_1|x)}{P_{\hatX_1}^*(\hatX_1)}-R_1^*\bigg)\\*
&\qquad\qquad\qquad-\sum_{i\in[2]}\nu_i^*(d_i(x,\hatX_i)-D_i)\bigg\}\Bigg).
\end{align}
\end{definition}

The properties of $\jmath(x,R_1^*,D_1,D_2|P_X)$ are summarized in the following lemma.
\begin{lemma}
\label{propertytilted_sr}
The following claims hold.
\begin{enumerate}
\item For any $(\hatx_1,\hatx_2)$ such that $P_{\hatX_1\hatX_2}^*(\hatx_1,\hatx_2)>0$,
\begin{align}
 \jmath(x,R_1^*,D_1,D_2|P_X)
 &=\log\frac{P_{\hatX_1\hatX_2|X}^*(\hatx_1,\hatx_2|x)}{P_{\hatX_1\hatX_2}^*(\hatx_1,\hatx_2)}+\xi^*\left(\log\frac{P_{\hatX_1|X}^*(\hatx_1|x)}{P_{\hatX_1}^*(\hatx_1)} -R_1^*\right)\nn\\* 
 &\qquad-\nu_1^*(d_1(x,\hatx_1)-D_1)-\nu_2^*(d_2(x,\hatx_2)-D_2)\label{expand}. 
\end{align}
\item The minimal sum rate $\rvR(R_1^*,D_1,D_2|P_X)$ equals the expectation of the rate-distortions-tilted information density, i.e.,
\begin{align}
\rvR(R_1^*,D_1,D_2|P_X)=\mathbb{E}_{P_X}\left[\jmath(X,R_1^*,D_1,D_2|P_X)\right]\label{expectlemma}.
\end{align}
\item Suppose that for all $Q_X$ in the neighborhood of $P_X$, $\mathrm{supp}(Q^*_{\hatX_1\hatX_2})=\mathrm{supp}(P^*_{\hatX_1\hatX_2})$. Then for all $a\in\calX$,
\begin{align}
\frac{\partial \rvR(R_1^*,D_1,D_2|Q_X)}{\partial Q_X(a)}\bigg|_{Q_X=P_X}
&=\jmath(a,R_1^*,D_1,D_2|P_X)-(1+\xi^*).
\end{align}
\end{enumerate}
\end{lemma}
Lemma \ref{propertytilted_sr} generalizes the properties of the distortion-tilted information density for the rate-distortion problem in Lemma \ref{prop:dtilteddensity}, which are also available in~\cite[Properties 1-3]{kostina2012converse} and~\cite[Theorems 2.1-2.2]{kostina2013lossy}.

For a successively refinable discrete memoryless source-distortion measure triplet, it follows from Definition \ref{prop:sr} that if $R(P_X,D_1)\leq R_1< R(P_X,D_2)$,
\begin{align}
\rvR(R_1,D_1,D_2|P_X)=R(P_X,D_2)\label{srminr2}.
\end{align}
In this case, $\xi^*=0$, $\nu_1^*=0$. The rate-distortions-tilted information density $\jmath(x,R_1^*,D_1,D_2|P_X)$ reduces to the distortion-tilted information density $\jmath(x,D_2|P_X)$ in \eqref{def:dtilteddensity} for the rate-distortion problem, where
\begin{align}
\label{kostinatilt}
\jmath(x,D|P_X)&=-\log \mathbb{E}_{P_{\hatX}^*}[\exp(-\lambda_1^*(d(x,\hatX)-D))],\\
\lambda^*&=-\frac{\partial R(P_X,D')}{\partial D'}\bigg|_{D'=D}.
\end{align}

\section{Second-Order Asymptotics}
\subsection{Definitions and Discussions}
Recall that for the rate-distortion problem with only one encoder, the second-order coding rate is defined as the backoff from the minimal achievable rate, i.e.,~the rate-distortion function $R(P_X,D)$ (cf. Definition \ref{def:sr4lossy}). Analogously, for a multiterminal lossy source coding problem such as successive refinement, in order to derive the second-order asymptotics, we need to study the backoff of the rates of encoders from a boundary point on the rate-distortion region, which is a minimal achievable rate pair and takes role of the rate-distortion function for the rate-distortion problem. 

Formally, let $(R_1^*, R_2^*)$ be a rate pair on the boundary of the rate-distortion region $\calR(D_1,D_2|P_X)$. The second-order coding region for the successive refinement problem is defined as follows.
\begin{definition}
\label{defsecond_sr}
Given any $\varepsilon\in(0,1)$, a pair $(L_1,L_2)$ is said to be second-order $(R_1^*,R_2^*,D_1,D_2,\varepsilon)$-achievable if there exists a sequence of $(n,M_1,M_2)$-codes such that
\begin{align}
\limsup_{n\to\infty}\frac{1}{\sqrt{n}}\left(\log M_1-nR_1^*\right)\leq L_1,\\
\limsup_{n\to\infty}\frac{1}{\sqrt{n}}\left(\log(M_1M_2)-n(R_1^*+R_2^*)\right)\leq L_2,
\end{align}
and
\begin{align}
\limsup_{n\to\infty}\rmP_{\rme,n}(D_1,D_2)\leq \varepsilon\label{sr:jep}.
\end{align}
The closure of the set of all second-order $(R_1^*,R_2^*,D_1,D_2,\varepsilon)$-achievable pairs is called the second-order $(R_1^*,R_2^*,D_1,D_2,\varepsilon)$ coding region and denoted as \\$\calL(R_1^*,R_2^*,D_1,D_2,\varepsilon)$.
\end{definition}

We emphasize that the JEP criterion~\eqref{sr:jep} is consistent with original setting of successive refinement in Rimoldi's work~\cite{rimoldi1994} and the error exponent analysis of Kanlis and Narayan~\cite{kanlis1996error}. In contrast, Tuncel and Rose~\cite{tuncel2003} considered the separate excess-distortion events and probabilities and derived the tradeoff between exponents of two excess-distortion probabilities. Note that the rate-distortion region remains the same~\cite{rimoldi1994,equitz1991successive} regardless whether we consider vanishing joint or the separate excess-distortion probabilities. In the study of second-order asymptotics, the second-order coding region can also be defined under the SEP criterion~\cite{no2016}. Specifically, the second-order coding region $\calL_{\mathrm{sep}}(R_1^*,R_2^*,D_1,D_2,\varepsilon_1,\varepsilon_2)$ is defined similarly to Definition \ref{defsecond_sr}, except that \eqref{sr:jep} is replaced by
\begin{align}
\limsup_{n\to\infty}\Pr\left\{d_1(X^n,\hatX_1^n)>D_1\right\}&\le\varepsilon_1,\label{eqn:eta1}\\*
\limsup_{n\to\infty}\Pr\left\{d_2(X^n,\hatX_2^n)>D_2\right\}&\le\varepsilon_2,\label{eqn:eta2}
\end{align}
for some fixed $(\varepsilon_1,\varepsilon_2)\in (0,1)^2$ and the boundary rate-pair $(R_1^*,R_2^*)$ is fixed as $R_1^*=R(P_X,D_1)$ and 
$R_1^*+R_2^*=\rvR(R(P_X,D),D_1,D_2|P_X)$, which corresponds to the case where both encoders respectively use their own optimal (i.e., minimum possible) asymptotic rates. 

The main content of this chapter is the characterization of \\$\calL(R_1^*,R_2^*,D_1,D_2,\varepsilon)$ and $\calL_{\mathrm{sep}}(R_1^*,R_2^*,D_1,D_2,\varepsilon_1,\varepsilon_2)$ for a DMS under bounded distortion measures, e.g., a binary source with Hamming distortion measures. We note that $\calL(R_1,R_2 ,D_1,D_2,\varepsilon)$ can, in principle, be evaluated for rate pairs that are not on the boundary of the first-order region  $\calR(D_1,D_2|P_X)$. However, this would lead to degenerate solutions.

We next explain some advantages of using the JEP criterion over the SEP criterion in second-order asymptotics.
\begin{enumerate}
\item The JEP criterion is consistent with recent works in the second-order literature~\cite{tan2014dispersions,le2015inter,watanabe2015second}. For example, in \cite{le2015inter}, Le, Tan and Motani established the second-order asymptotics for the Gaussian interference channel in the strictly very strong interference regime under the joint error probability criterion. If in \cite{le2015inter}, one adopts the separate error probabilities criterion, one would {\em not} be able to observe the performance tradeoff between the two decoders.
\item In Section \ref{sec:examples},  we show, via different proof techniques compared to existing works, that the second-order region is curved for successively refinable source-distortion triplets. This shows that if one second-order coding rate is small, the other is necessarily large. This reveals a fundamental tradeoff that cannot be observed if one adopts the separate excess-distortion probability criterion.
\end{enumerate}

\subsection{A General DMS}

Recall that $\Psi(x_1,x_2;\bmu,\mathbf{\Sigma})$ is the bivariate generalization of the Gaussian cdf. Given each $i\in[2]$, let $\mathrm{V}(D_i|P_X):=\mathrm{Var}[\jmath(X,D_i|P_X)]$ be the {\em rate-dispersion function}~(cf. \eqref{def:disdispersion}). Given a rate pair $(R_1^*,R_2^*)$ on the boundary of $\calR(D_1,D_2|P_X)$, also define another  rate-dispersion function $\mathrm{V}(R_1^*,D_1,D_2|P_X):=\mathrm{Var}\left[\jmath(X,R_1^*,D_1,D_2|P_X)\right]$. Let $\mathbf{V}(R_1^*,D_1,D_2|P_X) \succeq 0$ be the covariance matrix of the  two-dimensional random vector $[\jmath(X,D_1|P_X),\jmath(X,R_1^*,D_1,D_2|P_X)]^\top$, i.e., the {\em rate-dispersion matrix}.

We impose the following conditions on the rate pair $(R_1^*,R_2^*)$, the distortion measures $(d_1,d_2)$, the distortion levels $(D_1,D_2)$ and the source distribution $P_X$:
\begin{enumerate}
\item \label{cond1_sr} $\rvR(R_1^*,D_1,D_2|P_X)$ is finite;
\item $\xi^*\geq 0$ in~\eqref{deflambda} and $\nu_i^*,~i=1,2$ in \eqref{defnu1}, \eqref{defnu2} are well-defined;
\item $(Q_X,D_1')\mapsto R(Q_X,D_1')$ is twice differentiable in the neighborhood of $(P_X,D_1)$ and the derivatives are bounded (i.e., the spectral norm of the Hessian matrix is bounded);
\item \label{cond3} $(R_1,D_1',D_2',Q_{X})\mapsto \rvR(R_1,D_1',D_2'|Q_{X})$ is twice differentiable in the neighborhood of $(R_1^*,D_1,D_2,P_X)$ and the derivatives are bounded;
\end{enumerate}
Note that similar regularity assumptions were made on second-order asymptotics for the rate-distortion and Kaspi problems.

We first present the second-order asymptotics under the JEP criterion.
\begin{theorem}
\label{mainresult_sr}
Under conditions (\ref{cond1_sr}) to (\ref{cond3}), depending on the values of $(R_1^*,R_2^*)$, for any $\varepsilon\in(0,1)$, the second-order coding region satisfies:
\begin{itemize}
\item Case (i): $R(P_X,D_1)<R_1^*<\rvR(R_1^*,D_1,D_2|P_X)$ and $R_1^*+R_2^*=\rvR(R_1^*,D_1,D_2|P_X)$
\begin{align}
\nn&\calL(R_1^*,R_2^*,D_1,D_2,\varepsilon)\\*
&=\Big\{(L_1,L_2): \xi^*L_1+L_2\geq \sqrt{\mathrm{V}(R_1^*,D_1,D_2|P_X)}\mathrm{Q}^{-1}(\varepsilon)\Big\}.
\end{align}
\item Case (ii): $R_1^*=R(P_X,D_1)$ and $R_1^*+R_2^*>\rvR(R_1^*,D_1,D_2|P_X)$
\begin{align}
\calL(R_1^*,R_2^*,D_1,D_2,\varepsilon)&=\left\{(L_1,L_2):L_1 \geq \sqrt{\mathrm{V}(D_1|P_X)} \rm\rmQ^{-1}(\varepsilon)\right\}.
\end{align}
\item Case (iii): $R_1^*=R(P_X,D_1)$, $R_1^*+R_2^*=\rvR(R_1^*,D_1,D_2|P_X)$ \\and $\mathrm{rank}(\mathbf{V}(R_1^*,D_1,D_2|P_X))\ge 1$,
\begin{align}
\nn&\calL(R_1^*,R_2^*,D_1,D_2,\varepsilon)\\*
&=\Big\{(L_1,L_2):\Psi(L_1,\xi^*L_1+L_2;\mathbf{0},\mathbf{V}(R_1^*,D_1,D_2|P_X))\geq 1-\varepsilon\Big\}. \label{eqn:calLiii}
\end{align}
\end{itemize} 
\end{theorem}
The proof of Theorem~\ref{mainresult_sr} is provided  in Section~\ref{secondproof}. In the achievability part, we leverage the type covering lemma~\cite[Lemma 8]{no2016}. In the converse part, we follow the perturbation approach proposed by Gu and Effros in their proof for the strong converse of Gray-Wyner problem~\cite{wei2009strong}, leading to a type-based strong converse. In the proofs of both directions, we leverage the properties of appropriately defined rate-distortions-tilted information densities and use the (multi-variate) Berry-Esseen theorem. An alternative converse proof of Theorem \ref{mainresult_sr} is possible by applying the Berry-Esseen theorem to the non-asymptotic converse bound in \cite[Corollary 2]{kostina2019sr} (see also Lemma \ref{lemma:oneshot4sr} from our analysis of the Fu-Yeung problem), analogously to the converse proof of second-order asymptotics for the rate-distortion and Kaspi problems. We omit the alternative converse proof of Theorem \ref{mainresult_sr}.

In both Cases (i) and (ii), the code is operating at a rate bounded away from one of the first-order fundamental limits. Hence, a univariate Gaussian suffices to characterize the second-order behavior. In contrast, for Case (iii), the code is operating at precisely the two first-order fundamental limits. Hence,  in general, we need a bivariate Gaussian to characterize the second-order behavior. Using an argument by Tan and Kosut~\cite[Theorem 6]{tan2014dispersions}, we note that  this result holds for both positive definite and rank deficient  rate-dispersion matrices $\mathbf{V}(R_1^*,D_1,D_2|P_X)$. However, we exclude the  degenerate case in which $\mathrm{rank}(\mathbf{V}(R_1^*,D_1,D_2|P_X))=0$. Note that if the rank of $\bV(R_1^*,D_1,D_2|P_X)$ is $0$, it means that the dispersion matrix is all zeros matrix, i.e., $\mathrm{Cov}[\jmath(X,D_1|P_X),\jmath(X,R_1^*,D_1,D_2|P_X)]=0$, $\rmV(D_1|P_X)=0$, and \\$\rmV(R_1^*,D_1,D_2|P_X)=0$. This implies that $\jmath(x,D_1|P_X)$ and $\jmath(x,R_1^*,D_1,D_2|P_X)$ are both deterministic. In this case, the second-order term (dispersion) vanishes, and if one seeks refined asymptotic estimates for the optimal finite blocklength  coding rates, one  would then be interested to analyze the {\em third-order} or $\Theta(\log n)$ asymptotics (cf.\ \cite[Theorem 18]{kostina2012fixed}).

We next present inner (achievability) and outer (converse) bounds on the second-order coding region under the SEP criterion. 
\begin{theorem}
\label{sep:sr}
Under conditions (\ref{cond1_sr}) to (\ref{cond3}), for any $(\varepsilon_1,\varepsilon_2)\in(0,1)^2$,  the second-order coding region $\calL_{\mathrm{sep}}(R_1^*,R_2^*,D_1,D_2,\varepsilon_1,\varepsilon_2)$ satisfies that when $R_1^*=R(P_X,D_1)$ and $R_1^*+R_2^*=\rmR(R(P_X,D_1),D_1,D_2|P_X)$,
\begin{align}
\nn&\Big\{(L_1,L_2):~L_1\geq \sqrt{\rmV(D_1|P_X)}\rmQ^{-1}(\min\{\varepsilon_1,\varepsilon_2\}),\\*
\nn&\qquad\qquad L_2\geq \sqrt{\rmV(R_1^*,D_1,D_2|P_X)}\mathrm{Q}^{-1}(\min\{\varepsilon_1,\varepsilon_2\})\Big\}\\*
\nn&\subseteq\calL_{\mathrm{sep}}(R_1^*,R_2^*,D_1,D_2,\varepsilon_1,\varepsilon_2)\\*
\nn&\subseteq\Big\{(L_1,L_2):~L_1\geq \sqrt{\rmV(D_1|P_X)}\rmQ^{-1}(\varepsilon_1),\\*
&\qquad\qquad\qquad L_2\geq \sqrt{\rmV(R_1^*,D_1,D_2|P_X)}\mathrm{Q}^{-1}(\varepsilon_2)\Big\}.
\end{align}
\end{theorem}
The achievability proof of Theorem \ref{sep:sr} was proved by No, Ingber and Weissman using the type covering lemma for the successive refinement problem~\cite[Section V]{no2016} and the converse part follows by applying the Berry-Esseen theorem to the non-asymptotic converse bound by Kostina and Tuncel~\cite[Theorem 3]{kostina2019sr}. The inner bound could also be obtained similarly to the proof Case (iii) of Theorem \ref{mainresult_sr} with $\varepsilon$ replaced by $\min\{\varepsilon_1,\varepsilon_2\}$. 

The inner and outer bounds match when $\varepsilon_1=\varepsilon_2$. It was claimed by No, Ingber and Weissman~\cite{no2016} that the outer bound was achievable for any $(\varepsilon_1,\varepsilon_2)\in(0,1)^2$. However, a careful check suggests that it is impossible. This is because, in order not to incur an excess-distortion event at decoder $\phi_2$ for a sequence $x^n$, decoder $\phi_1$ should not incur an excess-distortion constraint since otherwise, the ``correct'' decoding of decoder $\phi_2$ is not guaranteed.

\subsection{A Successively Refinable DMS}
\label{srsource}
In this subsection, we specialize the results in Theorem~\ref{mainresult_sr} to successively refinable discrete memoryless source-distortion measure triplets. Note that for such source-distortion measure triplets, $\rvR(R_1^*,D_1,D_2|P_X)=R(P_X,D_2)$ if $R(P_X,D_1)\leq R_1^*<R(P_X,D_2)$. Hence, $\xi^*=0$ and $\nu^*_1=0$ and $\jmath(X,R_1^*,D_1,D_2|P_X)=\jmath(X,D_2|P_X)$. The covariance matrix $\mathbf{V}(R_1^*,D_1,D_2|P_X)$ is also simplified to $\mathbf{V}(D_1,D_2|P_X)$ with diagonal elements being $\mathrm{V}(D_1|P_X)$ and $\mathrm{V}(D_2|P_X)$ and off-diagonal element being the covariance $\mathrm{Cov}[\jmath(X,D_1|P_X),\jmath(X,D_2|P_X)]$. The conditions in Theorem~\ref{mainresult_sr} are also now simplified to: $(Q_X,D_1')\mapsto R(Q_X,D_1')$ and $(Q_X,D_2')\mapsto R(Q_X,D_2')$  are twice differentiable in the neighborhood of $(P_X,D_1,D_2)$ and the derivatives are bounded.
\begin{corollary}
\label{srmainresult_sr}
Under the conditions stated above, depending on $(R_1^*,R_2^*)$, the optimal second-order $(R_1^*,R_2^*,D_1,D_2,\varepsilon)$ coding region for a successively refinable discrete memoryless source-distortion measure triplet is as follows:
\begin{itemize}
\item Case (i): $R(P_X,D_1)<R_1^*<R(P_X,D_2)$ and $R_1^*+R_2^*=R(P_X,D_2)$
\begin{align}
\calL(R_1^*,R_2^*,D_1,D_2,\varepsilon)
&=\left\{(L_1,L_2):L_2 \geq \sqrt{\mathrm{V}(D_2|P_X)} \rm\rmQ^{-1}(\varepsilon)\right\}.
\end{align}
\item Case (ii): $R_1^*=R(P_X,D_2)$ and $R_1^*+R_2^*>R(P_X,D_2)$
\begin{align}
\calL(R_1^*,R_2^*,D_1,D_2,\varepsilon)&=\left\{(L_1,L_2):L_1 \geq \sqrt{\mathrm{V}(D_1|P_X)} \rm\rmQ^{-1}(\varepsilon)\right\}.
\end{align}
\item Case (iii): $R_1^*=R(P_X,D_2)$ and $R_1^*+R_2^*=R(P_X,D_2)$ \\and $\mathrm{rank}(\mathbf{V}(D_1,D_2|P_X))\geq 1$,
\begin{align}
\calL(R_1^*,R_2^*,D_1,D_2,\varepsilon)&=\Big\{(L_1,L_2):\Psi(L_1,L_2;\mathbf{0},\mathbf{V}(D_1,D_2|P_X))\geq 1-\varepsilon\Big\}.\label{eqn:pd}
\end{align}
Specifically, if~$\mathbf{V}(D_1,D_2|P_X)=\mathrm{V}(D_1|P_X)\cdot \mathrm{ones}(2,2)$, or equivalently $\jmath(X,D_1|P_X)-R_1^*=\jmath(X,D_2|P_X)-R_2^*$ almost surely,
\begin{align}
\calL(R_1^*,R_2^*,D_1,D_2,\varepsilon)&=\Big\{(L_1,L_2):\min\{L_1,L_2\} \geq \sqrt{\mathrm{V}(D_1|P_X)} \rm\rmQ^{-1}(\varepsilon)\Big\}\label{eqn:ones}.
\end{align}
\end{itemize} 
\end{corollary}
Corollary~\ref{srmainresult_sr} results from specializations of Theorem~\ref{mainresult_sr}. The special case in \eqref{eqn:ones} is proved in Section~\ref{proofsrmain}. We notice that the expressions in the second-order regions are simplified for successively refinable discrete memoryless source-distortion measure triplets. In particular, the optimization to compute the optimal test channel $P_{\hatX_1\hatX_2|X}^*$ in $\rvR(R_1, D_1, D_2|P_X)$, defined in \eqref{minimumr2}--\eqref{minr2}, is no longer necessary since the Markov chain $X-Z-Y$ holds for $P_{\hatX_1\hatX_2|X}^*$~\cite{equitz1991successive}.

Furthermore, in Section \ref{altproof}, we provide an alternative converse proof of Corollary \ref{srmainresult_sr} by generalizing the one-shot converse bound of Kostina and Verd\'u in~\cite[Theorem 1]{kostina2012converse}. We remark that the alternative converse proof is also applicable to successively refinable continuous memoryless source-distortion measure triplets such as the a GMS with quadratic distortion measures. 

The case in~\eqref{eqn:ones} pertains, for example, to a binary source with Hamming distortion measures. For such a source-distortion measure triplet, $\mathbf{V}(D_1,D_2|P_X)$ is rank $1$ and proportional to the all ones matrix. See Section \ref{sec:bin}. The result in~\eqref{eqn:ones} implies that both excess-distortion events in~\eqref{defexcessprob_sr} are perfectly correlated so that the one consisting of the {\em smaller} second-order rate $L_i,~i=1,2$ dominates, since the first-order rates are fixed at the first-order fundamental limits $(R(P_X, D_1),R(P_X, D_2))$. In fact, our result in \eqref{eqn:ones} specializes to the scenario where one considers the {\em separate} excess-distortion criterion~\cite{no2016} in~\eqref{eqn:eta1}--\eqref{eqn:eta2} with $\varepsilon_1=\varepsilon_2=\varepsilon$ and $\rmV(D_1|P_X)=\rmV(D_2|P_X)$.   More importantly, the case in \eqref{eqn:pd} when $\mathbf{V}(D_1,D_2|P_X)$ is full rank pertains to a source-distortion measure triplets with more ``degrees-of-freedom''. See Section \ref{sec:quat} for a concrete example.  Thus our work is a strict generalization of that in \cite{no2016}.

The result under the SEP criterion follows from Theorem \ref{sep:sr}.
\begin{corollary}
\label{coro:sr}
Under conditions (\ref{cond1_sr}) to (\ref{cond3}), for any $(\varepsilon_1,\varepsilon_2)\in(0,1)^2$,  the second-order coding region $\calL_{\mathrm{sep}}(R_1^*,R_2^*,D_1,D_2,\varepsilon_1,\varepsilon_2)$ satisfies that when $R_1^*=R(P_X,D_1)$ and $R_1^*+R_2^*=R(P_X,D_2)$,
\begin{align}
\nn&\Big\{(L_1,L_2):~L_i\geq \sqrt{\rmV(D_i|P_X)}\rmQ^{-1}(\min\{\varepsilon_1,\varepsilon_2\})\Big\}\\*
\nn&\subseteq\calL_{\mathrm{sep}}(R_1^*,R_2^*,D_1,D_2,\varepsilon_1,\varepsilon_2)\\*
&\subseteq\Big\{(L_1,L_2):~L_i\geq \sqrt{\rmV(D_i|P_X)}\rmQ^{-1}(\varepsilon_i)\Big\}.
\end{align}
\end{corollary}
The converse part also follows from the converse proof of second-order asymptotics for the rate-distortion problem in Theorem \ref{second4rd}. Corollary \ref{coro:sr} implies that when $\varepsilon_1=\varepsilon_2$, for a successively refinable DMS, under the SEP criterion, the second-order coding rates are also successively refinable since the pair $L_1=\sqrt{\rmV(D_1|P_X)}\rmQ^{-1}(\varepsilon_1)$ and $L_2=\sqrt{\rmV(D_2|P_X)}\rmQ^{-1}(\varepsilon_2)$ is second-order achievable for the boundary rate pair $(R_1^*,R_2^*)=(R(P_X,D_1),R(P_X,D_2))$. Such a result implies that it is optimal to interrupt a transmission to provide a finer reconstruction of the source sequence without any loss in terms of second-order asymptotics, which is stronger than the original definition of successively refinability in terms of first-order asymptotics and coined ``strong successive refinability'' in~\cite{no2016}.

\subsection{Numerical Examples}
\label{sec:examples}
Recall that any discrete memoryless source with Hamming distortion measures is successively refinable~\cite{equitz1991successive}. In this subsection, we consider two such numerical examples originated in~\cite{kostina2012fixed} to illustrate Corollary~\ref{srmainresult_sr}. We use the logarithm with base $2$ in this subsection.
\subsubsection{A Binary Memoryless Source with Hamming Distortion Measures} \label{sec:bin}
Fix $p\in[0,1]$. We consider a binary source with $P_X(0)=p$. For any distortion levels $D_2<D_1<p$, it follows from \eqref{j4bms} that for each $i\in[2]$, 
\begin{align}
 \jmath(x,D_i|P_X)=\imath(x|P_X)- H_\rmb(D_i).
\end{align}
Hence,
\begin{align}
\rmV(D_1|P_X)&=\mathrm{V}(D_2|P_X)=p(1-p)\log^2\left(\frac{1-p}{p}\right),
\end{align}
and the rate-dispersion matrix is
\begin{align}
\mathbf{V}(D_1,D_2|P_X)
&=\mathrm{V}(D_1|P_X)\cdot \mathrm{ones}(2,2)\\*
&=p(1-p)\log^2\left(\frac{1-p}{p}\right) \cdot \mathrm{ones}(2,2) , \label{eqn:binary}
\end{align} which does not depend on $(D_1, D_2)$. 
From the above considerations, we see that a binary source with Hamming distortion measures is an example that falls under \eqref{eqn:ones} in Corollary~\ref{srmainresult_sr}.

\subsubsection{A Quaternary Memoryless Source with Hamming Distortion Measures} \label{sec:quat}
\begin{figure}[tb]
\centering 
\includegraphics[width=.6\columnwidth]{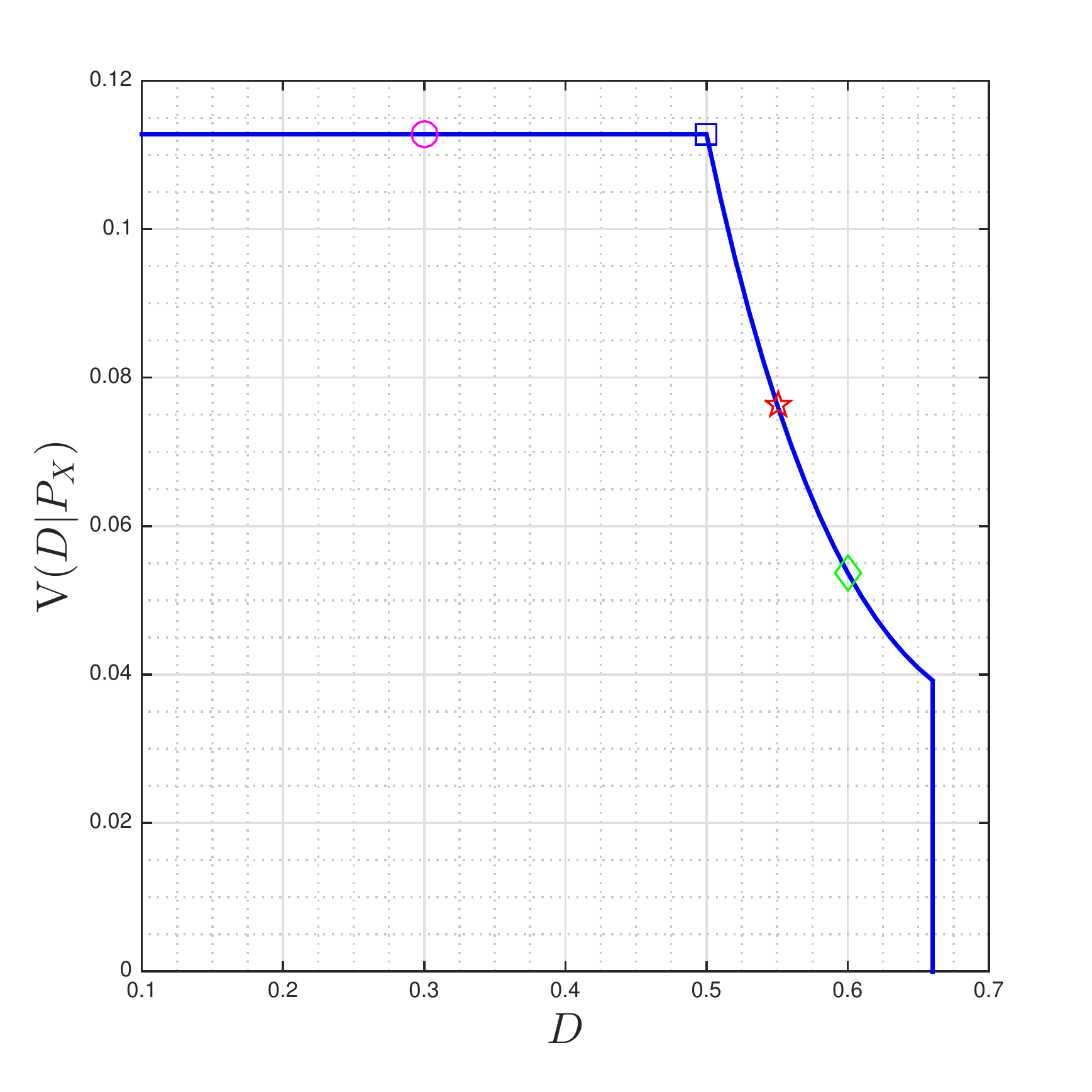}
\caption{Rate-dispersion  function $\rmV(D|P_X)$  for the source $P_X=[1/3,1/4,1/4,1/6]$~\cite[Section~VII.B]{kostina2012fixed} as a function of the distortion $D$.}
\label{rate_dispersion}
\end{figure}

\begin{figure}[tp]
\centering
\includegraphics[width=.6\columnwidth]{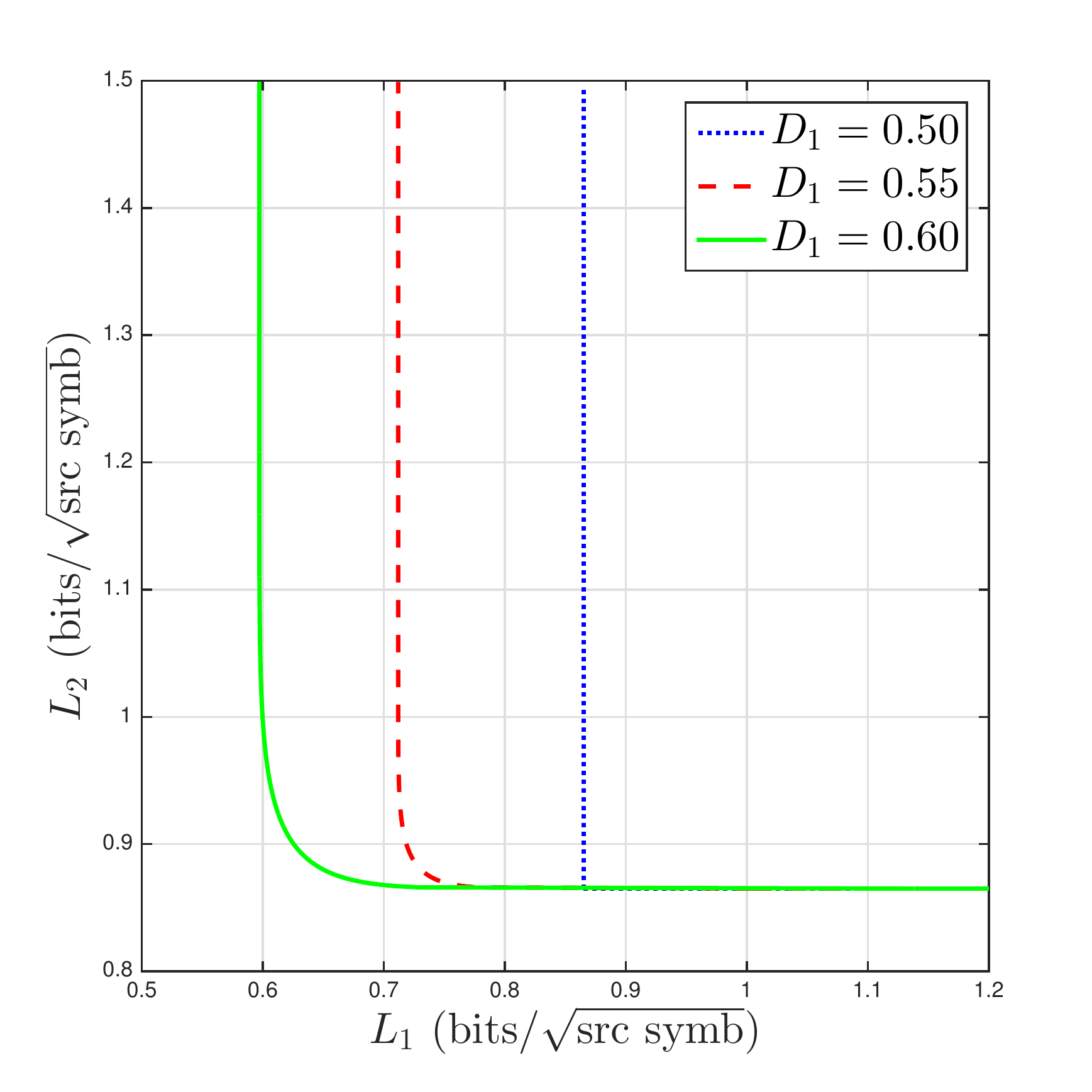}
\caption{Boundaries of the second-order coding region $\calL(R_1^*,R_2^*,D_1,D_2,\varepsilon)$ for Case (iii) in Corollary~\ref{srmainresult_sr}. The regions are to the top right of the boundaries.}
\label{plot_srregion}
\end{figure}

We next consider a more interesting source with the joint excess-distortion probability upper bounded by $\varepsilon=0.005$. In particular, we consider a quaternary memoryless source with distribution $P_X=[1/3,1/4,1/4,1/6]$. This example illustrates Case (iii)  of Corollary \ref{srmainresult_sr} and is adopted from~\cite[Section~VII.B]{kostina2012fixed}. The expressions for  the rate-distortion function  and the distortion-tilted information density are given in~\cite[Section~VII.B]{kostina2012fixed} (and will not be reproduced here as they are not important for our discussion). Since $\jmath(x,D_1|P_X)=\jmath(x,D_2|P_X)$ when $D_1=D_2=D$, we use $\jmath(x,D|P_X)$ to denote the common value of the distortion-tilted information density. Similarly, let $\mathrm{V}(D|P_X)$ be the common value of $\mathrm{V}(D_1|P_X)$ and $\mathrm{V}(D_2|P_X)$ when $D_1=D_2=D$. As shown in Figure~\ref{rate_dispersion} (reproduced from~\cite[Section~VII.B,~Figure~4]{kostina2012fixed}), the rate-dispersion function $\mathrm{V}(D|P_X)$ is dependent on  the distortion level $D$, unlike the binary example in Section \ref{sec:bin}.

In this numerical example, we fix $D_2=0.3$, which is denoted by the circle in Figure~\ref{rate_dispersion}. Then we decrease $D_1$ from $0.6$ to $0.55$ and finally to $0.5$. These points are denoted respectively by the diamond, the pentagram and the square in Figure~\ref{rate_dispersion}. Given these values of $(D_1,D_2)$, we plot the second-order coding rate for Case (iii) of Corollary~\ref{srmainresult_sr} in Figure~\ref{plot_srregion}. 

From Figure~\ref{plot_srregion}, we make the following observations and conclusions.
\begin{itemize}
\item The minimum $L_1$ converges to $\sqrt{V(D_1|P_X)}\rmQ^{-1}(\varepsilon)$ as $L_2\uparrow \infty$. This is because as $L_2$ increases, the bivariate Gaussian cdf asymptotically degenerates to the univariate Gaussian cdf with mean $0$ and variance $\mathrm{V}(D_1|P_X)$. A similar observation was made for the Slepian-Wolf problem in \cite{tan2014dispersions}.

\item As we decrease the value of $D_1$, the second-order coding region shrinks. We remark that there is  a  transition from \eqref{eqn:pd} with \\$\mathrm{rank}(\bV(D_1,D_2|P_X))=2$ to \eqref{eqn:ones} (where $\mathrm{rank}(\bV(D_1,D_2|P_X))=1$) as we decrease $D_1$ with the critical value of $D_1$ being $0.5$.

\item When $D_2<D_1\le 0.5$, the rate-dispersion matrix $\mathbf{V}(D_1,D_2|P_X)$ is rank $1$ (and proportional to the all ones matrix). Correspondingly, the result in \eqref{eqn:ones} applies. Here,  the second-order region is a (unbounded) rectangle with a sharp corner at the left bottom since the smaller $L_i,~i=1,2$ dominates. The second-order region remains unchanged as we decrease $D_1$ towards $D_2$ for fixed $D_2=0.3$. 

\item When $0.5<D_1<2/3$, the result in \eqref{eqn:pd} applies. In this case, neither $L_1$ nor $L_2$ dominates. The second-order coding rates $(L_1,L_2)$  are coupled  together by the  full rank   rate-dispersion matrix $\mathbf{V}(D_1,D_2|P_X)$, resulting the smooth boundary at the left bottom. 
 
\end{itemize}

We conclude that depending on the value of the distortion levels, the rate-dispersion matrix is either rank $1$ or rank $2$, illustrating Case (iii) of Corollary~\ref{srmainresult_sr}. These interesting observations cannot be gleaned from the work of No, Ingber and Weissman~\cite{no2016} in which the separate excess-distortion criteria are employed for the successive refinement problem. When $\bV(D_1,D_2|P_X)$ is rank $1$, exactly one excess-distortion event dominates the probability in~\eqref{defexcessprob_sr} entirely; when $\bV(D_1,D_2|P_X)$ is rank $2$, both excess-distortion events contribute non-trivially to the probability and a bivariate Gaussian is required to characterize the second-order fundamental limit.

\section{Proof of Second-Order Asymptotics}
\label{secondproof}
\subsection{Achievability}
\label{secondach}
We make use of the type covering lemma~\cite[Lemma 8]{no2016}, which is modified from~\cite[Lemma 1]{kanlis1996error}. Leveraging  the type covering lemma, we can then upper bound the excess-distortion probability. Finally, we Taylor expand appropriate terms and invoke the Berry-Essen theorem to obtain an achievable second-order coding region.

Define two constants:
\begin{align}
c_1&=4|\calX||\hatcalX_1|+9,\\*
c_2&=6|\calX||\hatcalX_1||\hatcalX_2|+2|\calX||\hatcalX_1|+17.
\end{align}
We are now ready to recall the \emph{discrete} type covering lemma for successive refinement.
\begin{lemma}
\label{typecovering}
Given type $Q_X\in\calP_n(\calX)$, for all $R_1\geq R(Q_X,D_1)$, the following holds:
\begin{itemize}
\item There exists a set $\calB_1\subset\hatcalX_1^n$ such that 
\begin{align}
\frac{1}{n}\log|\calB_1|\leq R_1+c_1\frac{\log n}{n}
\end{align}
and the type class is $D_1$-covered by the set $\calB_1$, i.e.,
\begin{align}
\calT_{Q_{X}}\subset\bigcup_{\hatx_1^n\in\calB_1}\left\{x^n:d_{1}(x^n,\hatx_1^n)\leq D_1\right\}.
\end{align}
\item For each $x^n\in\calT_{Q}$ and each $\hatx_1^n\in\calB_1$, there exists a set $\calB_2(\hatx_1^n)\subset\hatcalX_2^n$ such that
\begin{align}
\frac{1}{n}\log \left(\sum_{\hatx_1^n\in\calB}|\calB_2(\hatx_1^n)|\right)\leq \rvR(R_1,D_1,D_2|Q_X)+c_2\frac{\log n}{n}
\end{align}
and the $D_1$-distortion ball $\calN_1(\hatx_1^n,D_1):=\left\{x^n:d_{1}(x^n,\hatx_1^n)\leq D_1\right\}$ is $D_2$-covered by the set $\calB_2(\hatx_1^n)$ i.e.,
\begin{align}
\calN_1(\hatx_1^n,D_1)\subset\bigcup_{\hatx_2^n\in\calB_2(\hatx_1^n)}\left\{x^n:d_2(x^n,\hatx_2^n)\leq D_2\right\}.
\end{align}
\end{itemize}
\end{lemma}

Invoking Lemma~\ref{typecovering}, we can then upper bound the excess-distortion probability for some $(n,M_1,M_2)$-code. Given any $(n,M_1,M_2)$-code, define
\begin{align}
R_{1,n}&:=\frac{1}{n}\Bigg(\log M_1-c_1\log n-|\calX|\log(n+1)\Bigg),\\
R_{2,n}&:=\frac{1}{n}\Bigg(\log(M_1M_2)-c_2\log n\Bigg)-R_{1,n}.
\end{align}

\begin{lemma}
\label{uppexcess:sr}
There exists an $(n,M_1,M_2)$-code such that
\begin{align}
\rmP_{\rme,n}(D_1,D_2)
\nn&\leq \Pr\Big\{R_{1,n}<R(\hat{T}_{X^n},D_1)~\mathrm{or}~\\*
&\qquad\qquad R_{1,n}+R_{2,n}<\rvR(R_{1,n},D_1,D_2|\hat{T}_{X^n})\Big\}.
\end{align}
\end{lemma}
The proof of Lemma~\ref{uppexcess:sr} is similar to~\cite[Lemma 5]{watanabe2015second} and available in \cite[Appendix D]{zhou2016second}.

Recall the definition of the typical set in \eqref{def:typicalset} and the result in \eqref{upp:atypical} that
\begin{align}
\Pr\left\{\hat{T}_{X^n}\notin\calA_n(P_X)\right\}\leq \frac{2|\calX|}{n^2}.
\end{align}
For a rate pair $(R_1^*,R_2^*)$ satisfying the conditions in Theorem~\ref{mainresult_sr}, we choose 
\begin{align}
\frac{1}{n}\log M_1&=R_1^*+\frac{L_1}{\sqrt{n}}+\frac{c_1\log n+|\calX|\log (n+1)}{n}\label{achm1},\\
\frac{1}{n}\log(M_1M_2)&=R_1^*+R_2^*+\frac{L_2}{\sqrt{n}}+c_2\frac{\log n}{n}\label{achm2}.
\end{align}

Hence,
\begin{align}
R_{1,n}&=R_1^*+\frac{L_1}{\sqrt{n}},\\
R_{1,n}+R_{2,n}&=R_1^*+R_2^*+\frac{L_2}{\sqrt{n}}.
\end{align}
From the conditions in Theorem~\ref{mainresult_sr}, we know that the second derivative  of $R(Q_X,D_1)$ is bounded in the neighborhood of $P_X$, and that the second derivative  of $\rvR(R_1,D_1,D_2|Q_X)$ with respect to $(R_1,R_2,Q_X)$ is bounded around a neighborhood of $(R_1^*,P_X)$. Hence, for any $x^n$ such that $\hat{T}_{x^n}\in\calA_n(P_X)$, applying Taylor's expansion and invoking Lemmas \ref{prop:dtilteddensity} and \ref{propertytilted_sr}, we obtain
\begin{align}
\nn&R(\hat{T}_{x^n},D_1)\\*
&=R(P_X,D_1)+\sum_{x}\left(\hat{T}_{x^n}(x)-P_X(x)\right)\jmath_X(x,D_1|P_X)+O\left(\frac{\log n}{n}\right)\\
&=\frac{1}{n}\sum_{i\in[n]} \jmath(x_i,D_1|P_X)+O\left(\frac{\log n}{n}\right)\label{taylor_1},
\end{align}
and
\begin{align}
\nn&\rvR(R_{1,n},D_1,D_2|\hat{T}_{x^n})\\*
\nn&=\rvR(R_1^*,D_1,D_2|P_{XY})-\xi^*\frac{L_1}{\sqrt{n}}+O\left(\frac{\log n}{n}\right)\\
&\quad+\sum_{x}\left(\hat{T}_{x^n}(x)-P_X(x)\right)\jmath(x,R_1^*,D_1,D_2|P_X)\\*
&=\frac{1}{n}\sum_{i\in[n]}\jmath(x_i,R_1^*,D_1,D_2|P_X)-\xi^*\frac{L_1}{\sqrt{n}}+O\left(\frac{\log n}{n}\right)\label{taylor_2}.
\end{align}
Define $\eta_n=\frac{\log n}{n}$. 

In subsequent analyses, for ease of notation, we use $\jmath(x,R_1^*)$ and $\jmath(x,R_1^*,D_1,D_2|P_X)$ interchangeably. It follows from Lemma~\ref{uppexcess:sr} that
\begin{align}
\nn&\rmP_{\rme,n}(D_1,D_2)\\*
&\leq \Pr\left\{R_{1,n}<R(\hat{T}_{X^n},D_1)~\mathrm{or}~R_{1,n}+R_{2,n}<\rvR(R_{1,n},D_1,D_2|\hat{T}_{X^n})\right\}\label{taylorexpandstep000}\\
\nn&\leq \Pr\Big\{R_{1,n}<R(\hat{T}_{X^n},D_1)~\mathrm{or}~R_{1,n}+R_{2,n}<\rvR(R_{1,n},D_1,D_2|\hat{T}_{X^n}),\\*
&\qquad\qquad\mathrm{and~}\hat{T}_{X^n}\in\calA_{n}(P_X)\Big\}+\Pr\left\{\hat{T}_{X^n}\notin\calA_{n}(P_X)\right\}\\
&\nn\leq \Pr\bigg\{R_1^*+\frac{L_1}{\sqrt{n}}<\frac{1}{n}\sum_{i\in[n]}\jmath(X_i,D_1|P_X)+O\left(\eta_n\right)~\mathrm{or}\\*
&\qquad\quad R_1^*+R_2^*+\frac{L_2}{\sqrt{n}}<\frac{1}{n}\sum_{i\in[n]}\jmath(X_i,R_1^*)-\xi^*\frac{L_1}{\sqrt{n}}+O(\eta_n)\bigg\}+\frac{2|\calX|}{n^2}\\
&\nn=\Pr\bigg\{R_1^*+\frac{L_1}{\sqrt{n}}<\frac{1}{n}\sum_{i\in[n]}\jmath(X_i,D_1|P_X)+O\left(\eta_n\right)~\mathrm{or}~\\*
&\qquad\quad R_1^*+R_2^*+\xi^*\frac{L_1}{\sqrt{n}}+\frac{L_2}{\sqrt{n}}<\frac{1}{n}\sum_{i\in[n]}\jmath(X_i,R_1^*)+O(\eta_n)\bigg\}+\frac{2|\calX|}{n^2}\label{taylorexpand:sr}.
\end{align}
Thus,
\begin{align}
1-\rmP_{\rme,n}(D_1,D_2)\nn&\geq 
\Pr\bigg\{\frac{1}{n}\sum_{i\in[n]}\jmath(X_i,D_1|P_X)\leq R_1^*+\frac{L_1}{\sqrt{n}}+O\left(\eta_n\right),\\*
\nn&\qquad\quad \frac{1}{n}\sum_{i\in[n]}\jmath(X_i,R_1^*)\leq R_1^*+R_2^*+\xi^*\frac{L_1}{\sqrt{n}}+\frac{L_2}{\sqrt{n}}+O(\eta_n)\bigg\}\\*
&\qquad-\frac{2|\calX|}{n^2}\label{eqn:dmsach}.
\end{align}

We first consider Case (i) where $R(P_X,D_1)<R_1^*<\rvR(R_1^*,D_1,D_2|P_X)$ and $R_1^*+R_2^*=\rvR(R_1^*,D_1,D_2|P_X)$. Using the weak law of large numbers in Theorem \ref{wlln}, we obtain
\begin{align}
\Pr\bigg\{\frac{1}{n}\sum_{i\in[n]}\jmath(X_i,D_1|P_X)\leq R_1^*+\frac{L_1}{\sqrt{n}}+O\left(\eta_n\right)\bigg\}\to 1.
\end{align}
Invoking the  Berry-Esseen Theorem in Theorem \ref{berrytheorem}, we obtain
\begin{align} 
\nn&\Pr\Big\{\frac{1}{n}\sum_{i\in[n]}\jmath(X_i,R_1^*)\leq R_1^*+R_2^*+\xi^*\frac{L_1}{\sqrt{n}}+\frac{L_2}{\sqrt{n}}+O(\eta_n)\Big\}\\*
&\geq 1-\rmQ\left(\frac{\xi^*L_1+L_2+O(\sqrt{n}\eta_n)}{\sqrt{\mathrm{V}(R_1^*,D_1,D_2|P_X)}}\right)-\frac{6\mathrm{T}(R_1^*,D_1,D_2|P_X)}{\sqrt{n}\mathrm{V}^{3/2}(R_1^*,D_1,D_2|P_X)}\label{berryesseen},
\end{align}
where  $\mathrm{T}(R_1^*,D_1,D_2|P_X)$ is the third absolute moment of \\$\jmath(X,R_1^*,D_1,D_2|P_X)$, which is finite for a DMS. 
Hence,
\begin{align}
\rmP_{\rme,n}(D_1,D_2)\nn&\leq \rmQ\left(\frac{\xi^*L_1+L_2+O(\sqrt{n}\eta_n)}{\sqrt{\mathrm{V}(R_1^*,D_1,D_2|P_X)}}\right)+\frac{6\mathrm{T}(R_1^*,D_1,D_2|P_X)}{\sqrt{n}\mathrm{V}^{3/2}(R_1^*,D_1,D_2|P_X)}\\*
&\qquad+\frac{2|\calX|}{n^2}\label{phitoq}.
\end{align}
Hence, if $(L_1,L_2)$ satisfies 
\begin{align}
\xi^*L_1+L_2\geq \sqrt{\mathrm{V}(R_1^*,D_1,D_2|P_X)}\rm\rmQ^{-1}(\varepsilon),
\end{align}
then $\limsup_{n\to\infty}\rmP_{\rme,n}(D_1,D_2)\leq \varepsilon$. The proof of Case (ii) is omitted since it is similar to Case (i). 

The most interesting case is Case (iii) where $R_1^*=R(P_X,D_1)$ and $R_1^*+R_2^*=\rvR(R_1^*,D_1,D_2|P_X)$. If $\mathbf{V}(R_1^*,D_1,D_2|P_X)$ is positive definite we invoke the multi-variate Berry-Esseen Theorem in Theorem \ref{vectorBerry} to obtain 
\begin{align}
\nn&\rmP_{\rme,n}(D_1,D_2)\\*
\nn&\leq 1-\Psi\left(L_1+O\left(\eta_n\right),\xi^*L_1+L_2+O\left(\eta_n\right);\mathbf{0},\mathbf{V}(R_1^*,D_1,D_2|P_X)\right)\\*
&\quad+O\left(\frac{1}{\sqrt{n}}\right). \label{eqn:multi-be}
\end{align}  

Note that if $\mathbf{V}(R_1^*,D_1,D_2|P_X)$ is rank $1$, we can use the argument (projection onto a lower-dimensional subspace) in \cite[Proof of Theorem 6]{tan2014dispersions} to conclude that \eqref{eqn:multi-be} also holds. 
Now if we choose $(L_1,L_2)$ such that
\begin{align}
\Psi\left(L_1,\xi^*L_1+L_2;\mathbf{0},\mathbf{V}(R_1^*,D_1,D_2|P_X)\right)\geq 1-\varepsilon,
\end{align}
then $\limsup_{n\to\infty}\rmP_{\rme,n}(D_1,D_2)\leq \varepsilon$. The achievability proof is now completed.

\subsection{Converse}
We first prove a type-based strong converse. Define $\overline{d}_i:=\max_{x,y}d_1(x,\hatx_i)$ for each $i\in[2]$. Given a type $Q_X\in\calP_n(\calX)$, define
\begin{align}
g(Q_X)&:=\Pr\big\{d_1(X^n,\hatX_1^n)\leq D_1,~\mathrm{and}~d_2(X^n,\hatX_2^n)\leq D_2\,\big|\,X^n\in\calT_{Q_X}\big\}\label{conditionalnon-excessprob}.
\end{align}
\begin{lemma}
\label{mainresult_srtypestrongconverse_sr}
Fix $\alpha>0$ and a type $Q_X\in\calP_n(\calX)$. If the excess-distortion probability satisfies 
\begin{align}
 g(Q_X)\geq \exp(-n\alpha), \label{eqn:lb_type_str_conv}
\end{align}
then there exists a conditional distribution $Q_{\hatX_1\hatX_2|X}$ such that
\begin{align}
\log M_1 &\geq nI(Q_X,Q_{\hatX_1|X})-\vartheta_n,\\
\log (M_1M_2)&\geq nI(Q_X,Q_{\hatX_1\hatX_2|X})-\vartheta_n,
\end{align}
where $\vartheta_n:=|\calX|\log (n+1)+\log n+n\alpha$, 
and the expected distortions are bounded as 
\begin{align}
\mathbb{E}_{Q_X\times Q_{\hatX_1\hatX_2|X}}[d_1(X,\hatX_1)]&\leq
D_1+\frac{\overline{d}_1}{n}=:D_{1,n},\\*
\mathbb{E}_{Q_X\times Q_{\hatX_1\hatX_2|X}}[d_2(X,\hatX_2)]&\leq D_2+\frac{\overline{d}_2}{n}=:D_{2,n}.
\end{align}
\end{lemma}
The proof of Lemma \ref{mainresult_srtypestrongconverse_sr} is inspired by~\cite{wei2009strong}, which generalizes Lemma \ref{type:sc} for the rate-distortion problem and is available in \cite[Appendix E]{zhou2016second}.

Invoking Lemma~\ref{mainresult_srtypestrongconverse_sr} with $\alpha=\frac{\log n}{n}$, we can lower bound the excess-distortion probability for any $(n,M_1,M_2)$-code. Define $\beta_n=|\calX|\log (n+1)+2\log n$.
Define
\begin{align}
R_{1,n}'&:=\frac{1}{n}\log M_1+\beta_n,\\
R_{2,n}'&:=\frac{1}{n}\log(M_1M_2)+\beta_n-R_{1,n}'.
\end{align}

\begin{lemma}
\label{lbexcessp}
For any $(n,M_1,M_2)$-code, we have
\begin{align}
\nn&\rmP_{\rme,n}(D_1,D_2)\\*
\nn&\geq \Pr\Big\{R_{1,n}'<R(\hat{T}_{X^n},D_{1,n})~\mathrm{or}~R_{1,n}'+R_{2,n}'<\rvR(R_{1,n},D_{1,n},D_{2,n}|\hat{T}_{X^n})\Big\}\\*
&\qquad-\frac{1}{n}.
\end{align}
\end{lemma}

Choose $\log M_1=nR_1^*+L_1\sqrt{n}+\beta_n$ and $\log (M_1 M_2)=n(R_1^*+R_2^*)+L_2\sqrt{n}+\beta_n$. Recall the shorthand notation $\eta_n=\frac{\log n}{n}$. Now for $x^n$ such that $\hat{T}_{x^n}\in\calA_n(P_X)$, applying Taylor's expansion in a similar manner as~\eqref{taylor_1} and~\eqref{taylor_2}, invoking Lemma~\ref{lbexcessp} and noting that $\Pr\left\{\calF\cap \calG\right\}\geq \Pr\{\calF\}-\Pr\{\calG^{\mathrm{c}}\}$, we obtain 
\begin{align}
1-\rmP_{\rme,n}(D_1,D_2)
\nn&\leq \Pr\bigg\{\frac{1}{n}\sum_{i\in[n]}\jmath(X_i,D_1|P_X)\leq R_1^*+\frac{L_1}{\sqrt{n}}+O\left(\eta_n\right),\\*
\nn&\qquad\qquad \frac{1}{n}\sum_{i\in[n]}\jmath(X_i,R_1^*)\leq R_1^*+R_2^*+\xi^*\frac{L_1}{\sqrt{n}}+\frac{L_2}{\sqrt{n}}+O(\eta_n)\bigg\}\\*
&\qquad+\frac{1}{n}+\frac{2|\calX|}{n^2}\label{eqn:dmscon}.
\end{align}
Note that in \eqref{eqn:dmscon}, we Taylor expand $R(\hat{T}_{X^n},D_{1,n})$ around the source distribution $P_X$ and distortion level $D_1$.  We also Taylor expand the minimal sum rate function $\rvR(R_{1,n},D_1,D_2|\hat{T}_{X^n})$ at $(P_X,D_1,D_2)$. The residual terms when we Taylor expand with respect to the distortion levels are of the order $O(\frac{1}{n})$, which can be absorbed into $O(\eta_n)$. Furthermore, recall that we use $\jmath(x,R_1^*)$ and $\jmath(x,R_1^*,D_1,D_2|P_X)$ interchangeably.

The rest of converse proof can be done similarly as the achievability part in Section~\ref{secondach} by using the uni- or multi-variate  Berry-Esseen Theorem for Cases (i), (ii) and (iii).

\subsection{Proof of a Special Case}
\label{proofsrmain}
We now present a proof for the special case where the source-distortion measure triplet is successively refinable. Recall that for this case, $\xi^*=0$, $\nu_1^*=0$, and $\jmath(x,R_1^*,D_1,D_2|P_X)=\jmath(x,D_2|P_X)$ for $R(P_X,D_1)\leq R_1^*<R(P_X,D_2)$. For the achievability part, invoking~\eqref{eqn:dmsach}, we obtain
\begin{align}
1-\rmP_{\rme,n}(D_1,D_2)
\nn&\geq \Pr\bigg\{\frac{1}{n}\sum_{i\in[n]}\left(\jmath(X_i,D_1|P_X)- R_1^*\right)\leq \frac{L_1}{\sqrt{n}}+O\left(\eta_n\right),\\*
\nn&\qquad\qquad \frac{1}{n}\sum_{i\in[n]}\left(\jmath(X_i,D_2|P_X)- (R_1^*+R_2^*)\right)\leq \frac{L_2}{\sqrt{n}}+O(\eta_n)\bigg\}\\*
&\qquad-\frac{2|\calX|}{n^2}\label{eqn:dmsachspecial}.
\end{align}
According to the assumption in \eqref{eqn:ones} of Corollary~\ref{srmainresult_sr}, we have $\jmath(X_i,D_1|P_X)- R_1^*=\jmath(X_i,D_2|P_X)-(R_1^*+R_2^*)$. Given a random variable $X$ and two real numbers $a<b$, we obtain $\Pr\{X<a~\mathrm{and~}X<b\}=\Pr\{X<a\}$. Hence,
\begin{align}
\nn&1-\rmP_{\rme,n}(D_1,D_2)\\*
&\geq \Pr\bigg\{\frac{1}{n}\sum_{i\in[n]}\left(\jmath(X_i,D_1|P_X)- R_1^*\right)\leq \frac{\min\{L_1,L_2\}}{\sqrt{n}}+O\left(\eta_n\right)\bigg\}.
\end{align}
The rest of the proof is similar to Case (i) in Section~\ref{secondach}.

Using~\eqref{eqn:dmscon}, similarly to the achievability part, we complete the proof of converse part.

\subsection{Alternative Converse Proof}
\label{altproof}
We next present an alternative converse proof of Corollary~\ref{srmainresult_sr} using the finite blocklength converse bound in \cite[Lemma 15]{zhou2016second} that generalizes Theorem \ref{converse:fbl} for the rate-distortion problem.

\begin{lemma}
\label{tiltedconverse}
Given any $(\gamma_1,\gamma_2)\in\bbR_+^2$, any $(n,M_1,M_2)$-code for the successive refinement satisfies
\begin{align}
\rmP_{\rme,n}(D_1,D_2)
\nn&\geq \Pr\Big\{\sum_{i\in[n]}\jmath(X_i,D_1|P_X)\geq \log M_1+\gamma_1~\mathrm{or}\\*
\nn&\qquad \qquad\sum_{i\in[n]}\jmath(X_i,D_2|P_X)\geq \log(M_1M_2)+\gamma_2\Big\}\nn\\*
&\qquad-\exp(-n\gamma_1)-\exp(-n\gamma_2).
\end{align}
\end{lemma}
 
Choose  $\gamma_1=\gamma_2=\frac{\log n}{2n}$. Let $\log M_1=nR_1^*+L_1\sqrt{n}-\frac{1}{2}\log n$ and $\log(M_1M_2)=n(R_1^*+R_2^*)+L_2\sqrt{n}-\frac{1}{2}\log n$. Invoking Lemma~\ref{tiltedconverse}, we obtain
\begin{align}
1-\rmP_{\rme,n}(D_1,D_2)\nn&\leq \frac{2}{\sqrt{n}}+\Pr\Big\{\sum_{i\in[n]}\jmath(X_i,D_1|P_X)<nR_1^*+L_1\sqrt{n}\\*
&\qquad\mathrm{and}~\sum_{i\in[n]}\jmath(X_i,D_2|P_X)<n(R_1^*+R_2^*)+L_2\sqrt{n}\Big\}\label{ntiltedconverse}.
\end{align}
The rest of the proof is similar to the converse proof of Corollary~\ref{srmainresult_sr} in Section~\ref{proofsrmain}. We remark that this alternative converse proof also applies to continuous memoryless sources, such as a GMS under quadratic distortion measures and a Laplacian source with absolute distortion measures~\cite{zhong2006type}. 

A stronger non-asymptotic converse bound is provided in \cite[Corollary 2]{kostina2019sr}, which holds for any memoryless source and yields an alternative converse proof of Theorem \ref{mainresult_sr}. The same bound is also presented in Lemma \ref{lemma:oneshot4sr} in the next chapter, which is obtained as a special case of the non-asymptotic converse bound in Theorem \ref{oneshotconverse4fy} for the Fu-Yeung problem.

\chapter{Fu-Yeung Problem}
\label{chap:fu-yeung}

In this chapter, we study a special case of the multiple descriptions problem~\cite{wolf1980source,gamal1982achievable,zhang1987new,ahlswede1986multiple,zamir1999,Kramer2003,pradhan2008} with two encoders and three decoders proposed by Fu and Yeung~\cite{fu2002rate} and thus we term the problem as the Fu-Yeung problem. The Fu-Yeung problem generalizes the successive refinement problem by adding an additional decoder that aims to recover a deterministic function of the source sequence losslessly. The rate-distortion region was characterized by Fu and Yeung~\cite[Theorem 1]{fu2002rate}, which collects rate pairs to ensure reliable lossy compression at two decoders and reliable lossless data compression at the other decoder.  For this special case of multiple descriptions, the El Gamal-Cover inner bound~\cite{gamal1982achievable} was proved optimal.

Through the lens of the Fu-Yeung problem, this chapter reveals the tradeoff between encoders for simultaneous lossless and lossy compression. We will present a non-asymptotic converse bound and second-order asymptotics for the Fu-Yeung problem. Specifically, we first present properties of the minimal sum rate function given the rate of one encoder. Subsequently, we generalize the rate-distortions-tilted information for the successive refinement problem to the Fu-Yeung problem and present a non-asymptotic converse bound. This non-asymptotic bound, when specialized to the case where $|\calY|=1$, gives a stronger non-asymptotic converse bound for the successive refinement problem than Lemma \ref{tiltedconverse}. Finally, we present the second-order asymptotics for a DMS under bounded distortion measures and illustrate the results with numerical examples. This chapter is largely based on~\cite{zhou2017fy} and the second part of \cite{zhou2017non}.

\section{Problem Formulation and Asymptotic Result}

\subsection{Problem Formulation}
The setting for the Fu-Yeung problem is shown in Figure~\ref{systemmodelmd}. There are two encoders and three decoders. Each encoder $f_i,~i=1,2$ has access to the source sequence $X^n$ and compresses it into a message $S_i,~i=1,2$. Decoder $\phi_1$ aims to recover $X^n$ with distortion level $D_1$ using the encoded message $S_1$ from encoder $f_1$. Decoder $\phi_2$ aims to recover $X^n$ with distortion level $D_2$ using encoded messages $S_1$ and $S_2$. Decoder $\phi_3$ aims to recover $Y^n$, which is a symbolwise deterministic function  of the source sequence $X^n$. 
\begin{figure}[htbp]
\centering
\setlength{\unitlength}{0.5cm}
\scalebox{0.8}{
\begin{picture}(26,9)
\linethickness{1pt}
\put(2,5.5){\makebox{$X^n$}}
\put(6,1){\framebox(4,2)}

\put(6,7){\framebox(4,2)}
\put(7.7,1.8){\makebox{$f_2$}}
\put(7.7,7.8){\makebox{$f_1$}}
\put(1,5){\line(1,0){3.5}}
\put(4.5,5){\line(0,1){3}}
\put(4.5,8){\vector(1,0){1.5}}
\put(4.5,5){\line(0,-1){3}}
\put(4.5,2){\vector(1,0){1.5}}
\put(15,4){\framebox(4,2)}
\put(15,1){\framebox(4,2)}
\put(15,7){\framebox(4,2)}
\put(16.7,4.8){\makebox{$\phi_2$}}
\put(16.7,7.7){\makebox{$\phi_1$}}
\put(16.7,1.7){\makebox{$\phi_3$}}
\put(10,2){\vector(1,0){5}}
\put(12,2.5){\makebox(0,0){$S_2$}}
\put(10,8){\vector(1,0){5}}
\put(12,8.5){\makebox(0,0){$S_1$}}
\put(13,8){\line(0,-1){2.6}}
\put(13,5.4){\vector(1,0){2}}
\put(13,2){\line(0,1){2.6}}
\put(13,4.6){\vector(1,0){2}}
\put(19,2){\vector(1,0){4}}
\put(20,2.5){\makebox{$Y^n$}}
\put(19,8){\vector(1,0){4}}
\put(19.7,8.5){\makebox{$(\hatX_1^n,D_1)$}}
\put(19,5){\vector(1,0){4}}
\put(19.7,5.5){\makebox{$(\hatX_2^n,D_2)$}}
\end{picture}}
\caption{System model for the Fu-Yeung problem of multiple descriptions with one Semi-deterministic decoder~\cite{fu2002rate}.}
\label{systemmodelmd}
\end{figure}
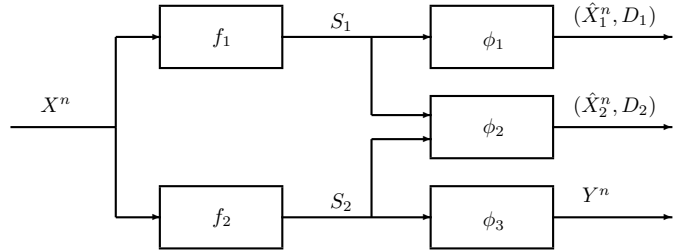

Consider a memoryless source $X^n$ generated i.i.d. from a probability mass function $P_X$ supported on a finite alphabet $\calX$. Let reproduction alphabets for decoders $\phi_1,\phi_2$ be $\hatcalX_1$ and $\hatcalX_2$ respectively. Fix a finite set $\calY$ and define a deterministic function $g:\calX \to \calY$. Let $Y_i=g(X_i),~i\in[1:n]$. Note that $P_Y$ is induced by the source distribution $P_X$ and the deterministic function $g$, i.e., for $y\in\calY$, $P_{Y}(y)=\sum_{x:g(x)=y}P_X(x)$. We assume that for each $y$, $P_Y(y)>0$. Decoder $\phi_3$ is required to recover $Y^n=g(X^n)=(g(X_1),\ldots,g(X_n))$ losslessly and the decoded sequence is denoted as $\hat{Y}^n$. We follow the definitions of codes and the rate-distortion region in~\cite{fu2002rate}.
\begin{definition}
\label{def:code4fy}
An $(n,M_1,M_2)$-code for the Fu-Yeung problem consists of two encoders:
\begin{align}
f_1:&\calX^n\to\calM_1=[M_1],\\
f_2:&\calX^n\to\calM_2=[M_2],
\end{align}
and three decoders:
\begin{align}
\phi_1&:\calM_1\to\mathcal{\hatX}_1^n,\\
\phi_2&:\calM_1\times\calM_2\to \mathcal{\hatX}_2^n,\\
\phi_3&:\calM_2\to \calY^n.
\end{align}
\end{definition}

Using the encoding and decoding functions, we have $\hatX_1^n=\phi_1(f_1(X^n))$, $\hatX_2^n=\phi_2(f_1(X^n),f_2(X^n))$ and $\hatY^n=\phi_3(f_2(X^n))$. Let $d_{\rmH}$ denote the Hamming distortion measure in \eqref{def:hamming} and let the average distortion between $y^n$ and its reproduced version $\haty^n$ be defined as $d_{\rmH}(Y^n,\hat{Y}^n):=\frac{1}{n}\sum_{i\in[n]}d_{\rmH}(Y_i,\hat{Y}_i)$. For each $i\in[2]$, let the distortion function $d_i:\calX\times\hatcalX_i\to[0,\infty)$ be a bounded distortion measure and let $d(x^n,\hatx_i^n)=\sum_{j\in[n]}d_i(x_j,\hatx_{j,i})$. The rate-distortion region for the Fu-Yeung problem is defined as follows.
\begin{definition}
A rate pair $(R_1,R_2)$ is said to be $(D_1,D_2)$-achievable for the Fu-Yeung problem if there exists a sequence of $(n,M_1,M_2)$-codes such that
\begin{align}
\limsup_{n\to\infty}\frac{\log M_i}{n}\leq R_i,~i=1,2,\label{def:ratecons4fy}
\end{align}
and
\begin{align}
\limsup_{n\to\infty} \bbE[d_i(X^n,\hatX_{i}^n)]&\leq D_i,~i=1,2,\\
\lim_{n\to\infty} \bbE[d_{\rmH}(Y^n,\hat{Y}^n)]&=0.
\end{align}
The closure of the set of all $(D_1,D_2)$-achievable rate pairs is called the first-order $(D_1,D_2)$-coding region and denoted as $\calR(D_1,D_2|P_X)$.
\end{definition}

\subsection{Rate-Distortion Region}
The first-order coding region $\calR(D_1,D_2|P_X)$ was characterized by Fu and Yeung in~\cite{fu2002rate} for DMS. In particular, Fu and Yeung~\cite{fu2002rate} showed that the El-Gamal-Cover inner bound~\cite{gamal1982achievable} for the multiple description coding problem is tight. 

To present the result, let $\calP(P_X,D_1,D_2)$ be the set of all pairs of conditional distributions $(P_{\hatX_1|X},P_{\hatX_2|X\hatX_1})\in\calP(\hatcalX_1|X)\times\calP(\hat\calX_2|\calX\hatcalX_1)$ such that $\bbE[d_1(X,\hatX_1)]\leq D_1$ and $\bbE[d_2(X,\hatX_2)]\leq D_2$. Given a pair of conditional distributions $(P_{\hatX_1|X},P_{\hatX_2|X\hatX_1})$, let $\calR(P_{\hatX_1|X},P_{\hatX_2|X\hatX_1})$ be the collection of rate pairs $(R_1,R_2)\in\bbR_+^2$ such that 
\begin{align}
R_1&\geq I(X;\hatX_1),\\
R_2&\geq H(Y),\\
R_1+R_2&\geq H(Y)+I(\hatX_1;Y)+I(X;\hatX_1,\hatX_2|Y).
\end{align}

\begin{theorem}
\label{mdfirst}
The rate-distortion region for the Fu-Yeung problem satisfies
\begin{align}
\calR(D_1,D_2|P_X)
&=\bigcup_{\substack{(P_{\hatX_1|X},P_{\hatX_2|X\hatX_1})\\\in \calP(P_X,D_1,D_2)}}\calR(P_{\hatX_1|X},P_{\hatX_2|X\hatX_1}).
\end{align}
\end{theorem}
When the $Y$ is a constant, i.e., $|\calY|=1$, the rate-distortion region in Theorem \ref{mdfirst} reduced to the rate-distortion region of the successive refinement problem. The rate-distortion function of the Kaspi problem can also be recovered from Theorem \ref{mdfirst} as the minimal rate $R_1$ by setting $R_2=H(Y)$ and choosing the source as $X=(S_1,S_2)$ and the side information as $Y=S_2$ for correlated discrete random variables $(S_1,S_2)$.

Although Theorem \ref{mdfirst} was derived under the average distortion criterion, the same rate-distortion region holds when one considers a vanishing  joint excess-distortion and error probability $\rmP_{\rme,n}(D_1,D_2)$ defined as follows:
\begin{align}
\rmP_{\rme,n}(D_1,D_2)
&:=\Pr\Big\{d_1(X^n,\hatX_1^n)>D_1~\mathrm{or}~d_2(X^n,\hatX_2^n)>D_2~\mathrm{or}~\hatY^n\neq Y^n\Big\}\label{def:error4fy}.
\end{align}
The reason is analogous to why Theorem \ref{lossy:Shannon} derived under the average distortion criterion still holds under the excess-distortion probability criterion for the rate-distortion problem.

\subsection{Boundary Rate Pairs}
We next discuss conditions for a rate pair $(R_1^*,R_2^*)$ to be on the boundary of the rate-distortion region $\calR(D_1,D_2|P_X)$, which enables our definition and analyses of second-order asymptotics.

Given any distributions $(P_{\hatX_1|X},P_{\hatX_2|X\hatX_1})$, let $P_{XY}$, $P_{X|Y}$, $P_{\hatX_1}$, $P_{Y\hatX_1}$, $P_{X\hatX_1}$, $P_{X\hatX_2}$, $P_{XY\hatX_1}$, $P_{\hatX_1|XY}$ and $P_{\hatX_2|Y\hatX_1}$ be induced by $P_X$, $P_{\hatX_1|X}$, $P_{\hatX_2|X\hatX_1}$ and the deterministic function $g:\calX\to\calY$. Recall the definition of $\calP(P_X,D_1,D_2)$ above Theorem \ref{mdfirst}. Given any rate $R_1$ of encoder $f_1$, define the following function
\begin{align}
\rvR(R_1,D_1,D_2|P_X)&:=\min_{\substack{(P_{\hatX_1|X},P_{\hatX_2|X\hatX_1})\\\in\calP(P_X,D_1,D_2):\\R_1\geq I(X;\hatX_1)}} I(\hatX_1;Y)+I(X;\hatX_1,\hatX_2|Y)\label{rvrmin}.
\end{align} 
It follows from the rate-distortion region in Theorem \ref{mdfirst} that given a rate $R_1$ of encoder $f_1$, the minimal achievable sum rate is \\$\rvR(R_1,D_1,D_2|P_X)+H(P_Y)$. Furthermore, the minimal achievable rate $R_1$ for encoder $f_1$ is the rate-distortion function $R(P_X,D_1)$~\cite{cover2012elements} and the minimal achievable rate $R_2$ for encoder $f_2$ is the entropy $H(P_Y)$. When $R_1=R(P_X,D_1)$, the minimal achievable rate $R_2$ is 
\begin{align}
\rvR_2^*(D_1,D_2|P_X)&:=H(P_Y)+\rvR(R(P_X,D_1),D_1,D_2|P_X)-R(P_X,D_1),
\end{align}
and when $R_2=H(P_Y)$, the minimal achievable rate $R_1$ is
\begin{align}
\rvR_1^*(D_1,D_2|P_X)&:=\min_{\substack{P_{\hatX_1|X},P_{\hatX_2|X\hatX_1}\\\in\calP(P_X,D_1,D_2)}} I(\hatX_1;Y)+I(X;\hatX_1,\hatX_2|Y)\label{def:rfyd1d2},
\end{align}
since $\rvR_1^*(D_1,D_2|P_X)$ is the solution to $R_1=\rvR(R_1,D_1,D_2|P_X)$. With these observations, we find all cases of boundary rate pairs and illustrate it in Figure \ref{rateregion}. Note that the Curve from case (ii) to Case (iv) is drawn as a line segment for ease of plot. In fact, it should be a convex curve.
\begin{figure}[t]
\centering
\begin{picture}(115, 135)
\setlength{\unitlength}{.47mm}
\put(30,50){\circle*{4}}
\put(17,52){\footnotesize (ii)}
\put(19,66){\footnotesize (i)}
\put(30,68){\circle*{4}}
\put(40,40){\circle*{4}}
\put(42,40){\footnotesize (iii)}
\put(50,30){\circle*{4}}
\put(51,24){\footnotesize (iv)}
\put(68,30){\circle*{4}}
\put(69,24){\footnotesize (v)}
\put(0, 10){\vector(1,0){110}}
\put(10, 0){\vector(0,1){110}}
\put(50, 30){\line(1, 0){55}}
\put(30, 50){\line(0,1){55}}
\put(50, 30){\line(-1, 1){20}}
\put(110,8){ $R_1$}
\put(-6, 105){$R_2$}
\multiput(10,50)(4,0){5}{\line(1,0){2}}
\put(-5, 50){$\overline{R_2}$}
\put(-10, 30){$\underline{R_2}$}
\multiput(10,30)(4,0){10}{\line(1,0){2}}
\multiput(50,10)(0,4){5}{\line(0,1){2}}
\multiput(30,10)(0,4){10}{\line(0,1){2}}
\put(25, 0){$\underline{R_1}$}
\put(50, 0){$\overline{R_1}$}
\end{picture}
\caption{Illustration of boundary rate pairs on the rate-distortion region of the Fu-Yeung problem, where $\underline{R_1}=R(P_X,D)$, $\overline{R_1}=\rvR_1^*(D_1,D_2|P_X)$, $\underline{R_2}=H(P_Y)$ and $\overline{R_2}=\rvR_2^*(D_1,D_2|P_X)$.
}
\label{rateregion}
\end{figure}
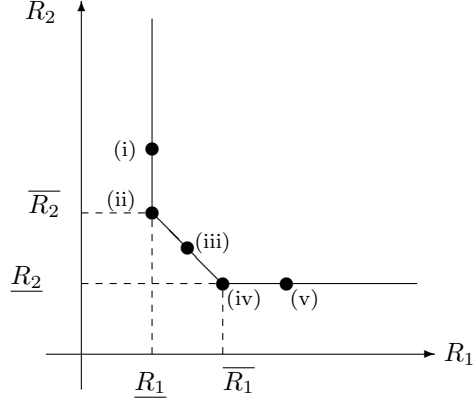

\section{Minimal Sum Rate Function and Its Properties}
\label{sec:mainresults4fy}

\subsection{Definitions}

Note that \eqref{rvrmin} is a convex optimization problem. Assume that $(R_1,D_1,D_2)$ is chosen such that $\rvR(R_1,D_1,D_2|P_X)$ is finite. Therefore, there exist test channels achieving $\rvR(R_1,D_1,D_2|P_X)$. Let $(\xi^*,\lambda_1^*,\lambda_2^*)$ be the optimal solutions to the dual problem of $\rvR(R_1,D_1,D_2|P_X)$, i.e.,
\begin{align}
\xi^*:=-\frac{\partial \rvR(R,D_1,D_2|P_X)}{\partial R}\Bigg|_{R=R_1}\label{def:s*},\\
\lambda_1^*:=-\frac{\partial \rvR(R_1,D,D_2|P_X)}{\partial D}\Bigg|_{D=D_1}\label{def:t1*},\\
\lambda_2^*:=-\frac{\partial \rvR(R_1,D_1,D|P_X)}{\partial D}\Bigg|_{D=D_2}\label{def:t2*}.
\end{align}

Given distributions $(Q_{\hatX_1},Q_{\hatX_2|Y\hatX_1})$ and $(x,y,\hatx_1)$, define the following two functions
\begin{align}
\nn&\beta_2(x,y,\hatx_1|Q_{\hatX_2|Y\hatX_1})\\*
&:=\Big\{\bbE_{Q_{\hatX_2|Y\hatX_1}}\Big[\exp(-\lambda_2^*d_2(x,\hatX_2))\big|Y=y,\hatX_1=\hatx_1\Big]\Big\}^{-1},\label{def:b2q}\\
\nn&\beta(x,y|Q_{\hatX_1},Q_{\hatX_2|Y\hatX_1})\\*
&:=\Bigg\{\bbE_{Q_{\hatX_1}}\Bigg[\exp\Bigg(-\frac{\lambda_1^*d_1(x,\hatX_1)}{1+\xi^*}-\frac{\log\beta_2(x,y,\hatX_1|Q_{\hatX_2|Y\hatX_1})}{1+\xi^*}\Bigg)\Bigg]\Bigg\}^{-1}\label{def:bq}.
\end{align}

\subsection{Properties}
We first present the properties of the optimal test channels that achieve \eqref{rvrmin}.
\begin{lemma}
\label{proptest4fy}
A pair of test channels $(P^*_{\hatX_1|X},P_{\hatX_2|X\hatX_1}^*)$ achieves \\$\rvR(R_1,D_1,D_2|P_X)$ if and only if
\begin{itemize}
\item For all $(x,y,\hatx_1,\hatx_2)$ such that $y=g(x)$,
\begin{align}
\nn&P_{\hatX_1|X}^*(\hatx_1|x)=\beta(x,y|P_{\hatX_1}^*,P_{\hatX_2|Y\hatX_1}^*)P^*_{\hatX_1}(\hatx_1)\\*
&\qquad\quad\times\exp\Bigg(-\frac{\lambda_1^*d_1(x,\hatx_1)+\log\beta_2(x,y,\hatx_1|P^*_{\hatX_2|Y\hatX_1})}{1+\xi^*}\Bigg)\label{optcond2},
\end{align}
\item For all $(x,y,\hatx_1,\hatx_2)$ such that $y=g(x)$ and $P_{\hatX_1|X}^*(\hatx_1|x)>0$
\begin{align}
\nn P_{\hatX_2|X\hatX_1}^*(\hatx_2|x,\hatx_1)&=\beta_2(x,y,\hatx_1|P_{\hatX_2|Y\hatX_1}^*)\\*
&\qquad\times P^*_{\hatX_2|Y\hatX_1}(\hatx_2|y,\hatx_1)\exp(-\lambda_2^*d_2(x,\hatx_2))\label{optcond1}.
\end{align}
\item For all $(x,\hatx_1,\hatx_2)$ such that $P_{\hatX_1|X}^*(\hatx_1|x)=0$, $P_{\hatX_2|X\hatX_1}^*(\cdot|x,\hatx_1)$ can be arbitrary distribution.
\end{itemize} 
Furthermore, if a pair of channels $(P^*_{\hatX_1|X},P_{\hatX_2|X\hatX_1}^*)$ achieves $\rvR(R_1,D_1,D_2)$, the following claims hold.
\begin{itemize}
\item The parametric representation of $\rvR(R_1,D_1,D_2|P_X)$ is
\begin{align}
\rvR(R_1,D_1,D_2|P_X)
\nn&=(1+\xi^*)\bbE_{P_{XY}}[\log\beta(X,Y|P_{\hatX_1}^*,P_{\hatX_2|Y\hatX_1^*})]\\*
&\qquad-\xi^*R_1-\lambda_1^*D_1-\lambda_2^*D_2\label{para4fy}.
\end{align}
\item For $(x,y,\hatx_1,\hatx_2)$ such that $y=g(x)$ and $P_{\hatX_1}^*(\hatx_1)P_{\hatX_2|Y\hatX_1}^*(\hatx_2|g(x),\hatx_1)>0$,
\begin{align}
\nn&(1+\xi^*)\log \beta(x,y|P_{\hatX_1}^*,P_{\hatX_2|Y\hatX_1}^*)\\*
&=(1+\xi^*)\log\frac{P_{\hatX_1|X}^*(\hatx_1|x)}{P_{\hatX_1}^*(\hatx_1)}+\log\frac{P_{\hatX_2|X\hatX_1}^*(\hatx_2|x,\hatx_1)}{P_{\hatX_2|Y\hatX_1}^*(\hatx_2|y,\hatx_1)}+\sum_{i\in[2]}\lambda_i^*d_i(x,\hatx_i)\label{expantypical}.
\end{align}
\end{itemize}
\end{lemma}
The proof of Lemma \ref{proptest4fy} is similar to \cite[Lemma 1.4]{csiszar1974}, \cite[Lemma 3]{watanabe2015second}, Lemma \ref{nule1} for the Kaspi problem and Lemma \ref{propertytilted_sr} for the successive refinement problem.

Similarly as \cite{watanabe2015second}, we can show that, for any pair of optimal test channels $(P_{\hatX_1|X}^*,P_{\hatX_2|X\hatX_1}^*)$, the value of $\beta(x,y|P_{\hatX_1}^*,P_{\hatX_2|Y\hatX_1}^*)$ and \\$\beta_2(x,y,\hatx_1|P_{\hatX_2|Y\hatX_1}^*)$ remain the same. From now on, fix a pair of test channels $(P^*_{\hatX_1|X},P_{\hatX_2|X\hatX_1}^*)$ such that that i) \eqref{optcond2}, \eqref{optcond1} hold; ii) for any $(y,\hatx_1)$ such that $P_{Y\hatX_1}^*(y,\hatx_1)=0$, the induced distribution defined as $P_{\hatX_2|Y\hatX_1}^*(\hatx_2|y,\hatx_1):=\sum_x P_X(x)\bbo(y=g(x))P_{\hatX_2|X\hatX_1}^*(\hatx_2|x,\hatx_1)$ satisfies
\begin{align}
\nn P_{\hatX_2|Y\hatX_1}^*
&=\argsup_{Q_{\hatX_2|Y\hatX_1}} \bbE_{P_{X|y}}\bigg[\beta(X,y)\beta_2^{-\frac{1}{1+\xi^*}}(X,y,\hatx_1|Q_{\hatX_2|Y\hatX_1})\\*
&\qquad\qquad\qquad\qquad\times\exp\Big(-\frac{\lambda_1^*}{1+\xi^*}d_1(X,\hatx_1)\Big)\bigg]\label{chooseopt}.
\end{align}
Note that the choice of $P_{\hatX_2|X\hatX_1}^*$ satisfying \eqref{chooseopt} is possible since the set $\{x:g(x)=y\}$ is disjoint for each $y\in\calY$.

For simplicity, given any $(x,y,\hatx_1)$, let
\begin{align}
\beta_2(x,y,\hatx_1)
&:=\beta_2(x,y,\hatx_1|P_{\hatX_2|Y\hatX_1}^*),\label{def:b2}\\
\beta(x,y)&:=\beta(x,y|P_{\hatX_1}^*,P_{\hatX_2|Y\hatX_1}^*)\label{def:b}\\
\imath_1(x,y,\hatx_1)&=\log \beta(x,y)-\frac{1}{1+\xi^*}\log\beta_2(x,y,\hatx_1)\label{def:i14fy}\\
\imath_2(x,y,\hatx_1)&:=\log \beta(x,y)+\frac{\xi^*}{1+\xi^*}\log \beta_2(x,y,\hatx_1)\label{def:i24fy}.
\end{align}
Furthermore, given any $\hatx_1$ and arbitrary conditional distribution $Q_{\hatX_2|Y\hatX_1}$, define
\begin{align}
w_1(\hatx_1)
&:=\bbE_{P_{XY}}\bigg[\exp\bigg(\imath_1(X,Y,\hatx_1)-\frac{\lambda_1^*d_1(X,\hatx_1)}{1+\xi^*}\bigg)\bigg]\label{def:w1},\\
w_2(\hatx_1,Q_{\hatX_2|Y\hatX_1})
\nn&:=\bbE_{P_{XY}\times Q_{\hatX_2|Y\hatX_1}}\bigg[\exp\big(\imath_2(X,Y,\hatx_1)-\frac{\lambda_1^*}{1+\xi^*}d_1(X,\hatx_1)\\*
&\qquad\qquad\qquad-\lambda_2^*d_2(X,\hatX_2)\big)\Big|\hatX_1=\hatx_1\bigg].\label{def:w2}
\end{align}

In the following, we present an important property of the quantities in \eqref{def:w1} and \eqref{def:w2}.
\begin{lemma}
\label{wleq1}
Given any $(P^*_{\hatX_1|XY},P_{\hatX_2|X\hatX_1}^*)$ satisfying \eqref{optcond2}, \eqref{optcond1}, and \eqref{chooseopt}, for any $\hatx_1\in\hatcalX_1$ and arbitrary distribution $Q_{\hatX_2|Y\hatX_1}$, 
\begin{align}
w_2(\hatx_1,Q_{\hatX_2|Y\hatX_1})\leq w_1(\hatx_1)\leq 1\label{w2leqw1le1}.
\end{align}
\end{lemma}
The proof of Lemma \ref{wleq1} is inspired by \cite[Lemma 5]{tuncel2003comp}, \cite[Theorem 2]{kostina2019sr} and omitted due to similarity to Lemma \ref{nule1} for the Kaspi problem. We remark that Lemmas \ref{proptest4fy} and \ref{wleq1} hold for any memoryless source, not restricted to a DMS. As we shall show, the result in Lemma \ref{wleq1} leads to a non-asymptotic converse bound for the Fu-Yeung problem.

\section{Rate-Distortions-Tilted Information Density}
Recall that $P_{XY}$ is induced by $P_X$ and the deterministic function $g:\calX\to\calY$.
\begin{definition}
For any $(x,y)\in\calX\times\calY$ such that $y=g(x)$, the rate-distortions-tilted information density for the Fu-Yeung problem is defined as
\begin{align}
\jmath(x,y|R_1,D_1,D_2,P_X)
&:=(1+\xi^*)\log \beta(x,y)-\xi^*R_1-\lambda_1^*D_1-\lambda_2^*D_2\label{def:j4fy},
\end{align}
where $\beta(\cdot)$ was defined in \eqref{def:b}
\end{definition}
The properties of $\jmath(x,y|R_1,D_1,D_2,P_X)$ follow from Lemma \ref{proptest4fy}. For example, it follows from \eqref{para4fy} that
\begin{align}
\rvR(R_1,D_1,D_2|P_X)&=\bbE_{P_{XY}}[\jmath(X,Y|R_1,D_1,D_2,P_X)]\\
&=\bbE_{P_X}[\jmath(X,g(X)|R_1,D_1,D_2,P_X)].
\end{align}
Let $\jmath(x,D_1|P_X)$ be the $D_1$-tilted information density in \eqref{def:dtilteddensity}, i.e.,
\begin{align}
\jmath(x,D_1|P_X)&:=-\log \Big(\sum_{\hatx_1}P_{\hatX_1}^*(\hatx_1)\exp(-\lambda^*(d_1(x,\hatx_1)-D_1))\Big)\label{djkostina},
\end{align}
where $P_{\hatX_1}^*$ is induced by the source distribution $P_X$ and the optimal test channel $P_{\hatX_1|X}^*$ for the rate-distortion function $R(P_X,D_1)$ (cf. \eqref{def:rd}) and $\lambda^*=-\frac{\partial R(P_X,D)}{\partial D}|_{D=D_1}$.

Furthermore, similarly to the proofs Lemma \ref{kaspifirstderive} for the Kaspi problem and Claim (iii) in Lemma \ref{propertytilted_sr},  we have the following lemma that further relates the rate-distortions-tilted information density with the derivative of the minimum sum rate function with respect to the distribution $P_X$ for the Fu-Yeung problem.
\begin{lemma}
\label{firstderive4fy}
Suppose that for all $Q_X$ in the neighborhood of $P_X$, $\mathrm{supp}(Q_{\hat{X}_1\hat{X}_2}^*)=\mathrm{supp}(P_{\hat{X}_1\hat{X}_2}^*)$. Then for any $a\in\mathrm{supp}(P_X)$,
\begin{align}
\frac{\partial \rvR(R_1,D_1,D_2|Q_X)}{\partial Q_X(a)}\Bigg|_{Q_X=P_X}&=\jmath(x,g(x)|R_1,D_1,D_2,P_X)-(1+s^*)\label{derivative4fy}.
\end{align}
\end{lemma}

\section{A Non-Asymptotic Converse Bound}
We next present a non-asymptotic converse bound for the Fu-Yeung problem. Given any $\gamma\in\bbR_+$, define the following three sets:
\begin{align}
\calA_1^n&:=\Big\{(x^n,y^n):\sum_{i\in[n]}\jmath(x_i,D_1|P_X)\geq \log M_1+n\gamma\Big\},\\
\calA_2^n&:=\Big\{(x^n,y^n):-\sum_{i\in[n]}\log P_Y(y_i)\geq \log M_2+n\gamma\Big\},\\
\calA_3^n\nn&:=\Big\{(x^n,y^n):\sum_{i\in[n]}\jmath(x_i,y_i|R_1,D_1,D_2,P_X)\geq \log M_1M_2\\*
&\qquad\qquad\qquad\qquad\qquad\qquad+\xi^*\log M_1+(1+\xi^*)n\gamma\Big\}.
\end{align}
\begin{lemma}
\label{oneshotconverse4fy}
Any $(n,M_1,M_2)$-code for the Fu-Yeung problem satisfies that for any $\gamma\geq 0$,
\begin{align}
\rmP_{\rme,n}(D_1,D_2)
&\geq \Pr\Big\{(X^n,Y^n)\in\bigcup_{i\in[3]}\calA_i^n\Big\}-4\exp(-\gamma).
\end{align}
\end{lemma}
We remark that Lemma \ref{wleq1} plays an important role in the proof of Lemma \ref{oneshotconverse4fy}. This can be made clear by the following definitions. Given $(x,y,\hatx_1,\hatx_2)$, using the definitions of $\imath_1(\cdot)$ in \eqref{def:i14fy} and $\imath_2(\cdot)$ in \eqref{def:i24fy}, we define
\begin{align}
\jmath_1(x,y,\hatx_1,D_1)&:=\imath_1(x,y,\hatx_1)-\frac{\lambda_1^*D_1}{1+\xi^*}\label{def:j14fy},\\
\jmath_2(x,y,\hatx_1,D_1,D_2)&:=\imath_2(x,y,\hatx_1)-\frac{\lambda_1^*D_1}{1+\xi^*}-\lambda_2^*D_2\label{def:j24fy}.
\end{align}
Using the definition of the rate-distortions-tilted information density in \eqref{def:j4fy}, we conclude that
\begin{align}
\jmath(x,y|R_1,D_1,D_2,P_X)=\xi^*\jmath_1(x,y,\hatx_1,D_1)+\jmath_2(x,y,\hatx_1,D_1,D_2)-\xi^*R_1\label{j=j1j24fy}.
\end{align}
In the proof of Lemma \ref{oneshotconverse4fy}, we make use of \eqref{j=j1j24fy} and the fact that $\Pr\{A+B\geq c+d\}\leq \Pr\{A\geq c\}+\Pr\{B\geq d\}$ for any variables $(A,B)$ and constants $(c,d)$.

Recall the setting of the Fu-Yeung problem in Figure \ref{systemmodelmd}. Note that when $Y$ is a constant, i.e. $|\calY|=1$, we recover the setting of the successive refinement problem~\cite{rimoldi1994}. Recall the definitions of an $(n,M_1,M_2)$-code for the successive refinement problem Definition \ref{code:sr}, the definition of the joint excess-distortion probability $\rmP_{\rme,n}^{\rm{SR}}(D_1,D_2)$ in \eqref{defexcessprob_sr}, the definition of the minimal sum rate $\rvR_{\rm{SR}}(R_1,D_1,D_2|P_X)$ in \eqref{minimumr2} and the definition of the rate-distortions tilted information density $\jmath_{\rm{SR}}(x|R_1,D_1,D_2,P_X)$ in \eqref{srtilted}. When $|\calY|=1$, it follows that
\begin{align}
\rvR_{\rm{SR}}(R_1,D_1,D_2|P_X)&=\rvR(R_1,D_1,D_2|P_X),\\
\jmath_{\rm{SR}}(x|R_1,D_1,D_2,P_X)&=\jmath(x,g(x)|R_1,D_1,D_2,P_X)\label{def:j4sr}.
\end{align}
We remark that although the definition of the rate-distortions-tilted information density for the successive refinement problem in the right hand side of \eqref{def:j4sr} appears different from \eqref{srtilted}, the two quantities share same properties (cf. \cite[Lemma 3]{zhou2016second}) and are thus essentially the same. Invoking Lemma \ref{oneshotconverse4fy} with $\calY=\{1\}$, we obtain the following non-asymptotic converse bound for the successive refinement problem.
\begin{lemma}
\label{lemma:oneshot4sr}
Any $(n,M_1,M_2)$-code for the successive refinement problem satisfies that for any $\gamma\geq 0$,
\begin{align}
\rmP_{\rme,n}^{\rm{SR}}(D_1,D_2)
\nn&\geq \Pr\Big\{\sum_{i\in[n]}\jmath(X_i,D_1|P_X)\geq \log M_1+n\gamma~\mathrm{or}\\
\nn&\qquad\quad\sum_{i\in[n]}\jmath_{\rm{SR}}(X_i|R_1,D_1,D_2,P_X)\geq \log M_1M_2\\*
&\qquad\qquad+\xi^*\log M_1+(1+\xi^*)n\gamma\Big\}-4\exp(-n\gamma)\label{converseoneshot4sr}.
\end{align}
\end{lemma}
Lemma \ref{lemma:oneshot4sr} was also derived by Kostina and Tuncel~\cite[Corollary 2]{kostina2019sr}. We remark that the non-asymptotic converse bound in \eqref{converseoneshot4sr} can be used to establish converse results for second-order asymptotics for any memoryless source, including the results in Theorem \ref{mainresult_sr} for a DMS. Invoking Lemma \ref{lemma:oneshot4sr}, for the successive refinement problem, we have the potential to establish tight second-order asymptotics for non-successively refinable continuous memoryless sources, e.g., a symmetric GMS under quadratic distortion measures~\cite{chowberger}.

\section{Second-Order Asymptotics}
\subsection{Preliminaries}
Let $\varepsilon\in(0,1)$ be fixed and let $(R_1^*,R_2^*)$ be a boundary rate pair on the rate-distortion region  $\calD(D_1,D_2|P_X)$ of the Fu-Yeung problem.
\begin{definition}
Given any $\varepsilon\in(0,1)$, a pair $(L_1,L_2)$ is said to be second-order $(R_1^*,R_2^*,D_1,D_2,\varepsilon)$-achievable for the Fu-Yeung problem if there exists a sequence of $(n,M_1,M_2)$-codes such that
\begin{align}
\limsup_{n\to\infty}\frac{\log M_i-nR_i}{\sqrt{n}}\leq L_i,~i=1,2,
\end{align}
and
\begin{align}
\limsup_{n\to\infty} \rmP_{\rme,n}(D_1,D_2)\leq \varepsilon.
\end{align}
The closure of the set of all second-order $(R_1^*,R_2^*,D_1,D_2,\varepsilon)$-achievable pairs is called the second-order $(R_1^*,R_2^*,D_1,D_2,\varepsilon)$ coding region and denoted as $\calL(R_1^*,R_2^*,D_1,D_2,\varepsilon)$.
\end{definition}

To present characterization of $\calL(R_1^*,R_2^*,D_1,D_2,\varepsilon)$, we need several definitions. Recall that $P_{XY}$ and $P_Y$ are induced by $P_X$ and the deterministic function $g:\calX\to\calY$ and the definition of the source dispersion function (cf. \eqref{def:sourcedispersion}), i.e.,
\begin{align}
\rmV(P_Y)
&=\sum_y P_Y(y)\big(-\log P_Y(y)-H(P_Y)\big)^2\\
&=\sum_x P_X(x)\big(-\log P_Y(g(x))-H(P_Y)\big)^2.
\end{align}
Recall that $\rmV(P_X,D_1)=\mathrm{Var}[\jmath(X,D_1|P_X)]$ is the distortion-dispersion function (cf. \eqref{def:disdispersion}). Let the rate-distortion-dispersion function be
\begin{align}
\rmV(R_1,D_1,D_2|P_X)&:=\mathrm{Var}\Big[\jmath(X,g(X)|R_1,D_1,D_2,P_X)-\log P_Y(Y)\Big].
\end{align}
Define two covariance matrices:
\begin{align}
\bV_1(R_1,D_1,D_2|P_X)
\nn&:=\mathrm{Cov}\Big([\jmath(X,g(X)|R_1,D_1,D_2,P_X)-\log P_Y(g(X))]^{\top},\\*
&\qquad\qquad\qquad\jmath(X,D_1|P_X)\Big),\\
\bV_2(R_1,D_1,D_2|P_X)
\nn&:=\mathrm{Cov}\Big([\jmath(X,g(X)|R_1,D_1,D_2,P_X)-\log P_Y(g(X)),\\*
&\qquad\qquad\qquad-\log P_Y(g(X))]^T\Big).
\end{align} 
Finally, recall that $\Psi(x_1,x_2;\bmu,\mathbf{\Sigma})$ is the bivariate generalization of the Gaussian cdf.

\subsection{Main Result and Discussions}
Suppose the following conditions hold:
\begin{enumerate}
\item \label{cond1}$(Q_X,D_1')\to R(Q_X,D_1')$ is twice differentiable in the neighborhood of $(P_X,D_1)$ and the derivatives are bounded;
\item $(Q_X,R_1',D_1',D_2')\to \rvR(R_1',D_1',D_2'|Q_X)$ is twice differentiable in the neighborhood of $(P_X,R_1,D_1,D_2)$ and the derivatives are bounded;
\item The functions $R(P_X,D_1)$, $R_1^*(D_1,D_2|P_X)$, $R_2^*(D_1,D_2|P_X)$ are positive and finite;
\item The dispersion functions $\rmV(P_X,D_1)$ and $\rmV(P_Y)$ are positive and the dispersion function $\rmV(R_1^*,D_1,D_2|P_X)$ is positive for \\$R(P_X,D_1)<R_1^*<\rvR_1^*(D_1,D_2|P_X)$;
\item  \label{condf} The covariance matrices $\bV_1(R(P_X,D_1),D_1,D_2|P_X)$ and \\ $\bV_2(\rvR_1^*(D_1,D_2|P_X),D_1,D_2|P_X)$ are positive semi-definite.
\end{enumerate}
Conditions (i) and (ii) concern the differentiability of rate-distortion functions and have been discussed in detail by Ingber and Kochman in \cite[Section III.A]{ingber2011dispersion}. Condition (iii) can easily verified by calculating the values of rate-distortion functions using convex optimization tools such as~\cite{boyd2004convex}. In order to verify conditions (iv) and (v), in general, one needs to develop specialized Blahut-Arimoto-type algorithms~\cite[Chapter~8]{csiszar2011information} to solve for the optimal test channels.

\begin{theorem}
\label{mddsecondregion}
Under conditions \eqref{cond1} to \eqref{condf}, depending on $(R_1^*,R_2^*)$, for any $\varepsilon\in(0,1)$, the second-order coding region satisfies
\begin{itemize}
\item Case (i): $R_1^*=R(P_X,D_1)$ and $R_2^*>\rvR_2^*(D_1,D_2|P_X)$
\begin{align}
\calL(R_1^*,R_2^*,D_1,D_2,\varepsilon)=\big\{(L_1,L_2):L_1\geq \sqrt{\rmV(P_X,D_1)}\rmQ^{-1}(\varepsilon)\big\}.
\end{align}
\item Case (ii): $R_1^*=R(P_X,D_1)$ and  $R_2^*=\rvR_2^*(D_1,D_2|P_X)$
\begin{align}
\nn&\calL(R_1^*,R_2^*,D_1,D_2,\varepsilon)=\big\{(L_1,L_2):\\*
&\qquad\Psi(L_1,(1+\xi^*)L_1+L_2;\bzero_2;\bV_1(R_1^*,D_1,D_2|P_X))\geq 1-\varepsilon\big\}.
\end{align}
\item Case (iii): $R(P_X,D_1)<R_1^*<\rvR_1^*(D_1,D_2|P_X)$ and\\ $R_2^*=\rvR^*(R_1^*,D_1,D_2|P_X)+H(P_Y)-R_1^*$,
\begin{align}
\nn&\calL(R_1^*,R_2^*,D_1,D_2,\varepsilon)=\big\{(L_1,L_2):\\*
&\quad(1+\xi^*)L_1+L_2\geq \sqrt{\rmV(R_1^*,D_1,D_2|P_X)}\rmQ^{-1}(\varepsilon)\big\}.
\end{align}
\item Case (iv) $R_1^*=\rvR_1^*(D_1,D_2|P_X)$ and $R_2^*=H(P_Y)$
\begin{align}
\nn&\calL(R_1^*,R_2^*,D_1,D_2,\varepsilon)=\big\{(L_1,L_2):\\*
&\quad\Psi(L_1+L_2,L_2;\bzero_2,\bV_2(R_1^*,D_1,D_2|P_X))\geq 1-\varepsilon\big\}.
\end{align}
\item Case (v) $R_1^*>\rvR_1^*(D_1,D_2|P_X)$ and $R_2^*=H(P_Y)$
\begin{align}
\calL(R_1^*,R_2^*,D_1,D_2,\varepsilon)=\big\{(L_1,L_2):L_2\geq \sqrt{\rmV(P_Y)}\rmQ^{-1}(\varepsilon)\big\}.
\end{align}
\end{itemize}
\end{theorem}
The proof of Theorem \ref{mddsecondregion} is provided in Section \ref{prove:mddsecondregion}. The achievability part follows by the method of types, where we first prove a type-covering lemma tailored to the Fu-Yeung problem, and subsequently apply Taylor expansions of the rate-distortion function and the minimal sum rate function of empirical distributions around the source distribution $P_X$, and finally apply the Berry-Esseen theorem for each case. The converse part follows by deriving a type-based strong converse analogously to the converse proof the successive refinement problem in Theorem \ref{mainresult_sr} and proceeding similarly to the achievability proof.

Since the successive refinement problem is special case of the Fu-Yeung problem when $Y=g(X)$ is a constant, the second-order asymptotics for the successive refinement problem for a DMS under bounded distortion measures in Theorem \ref{mainresult_sr} is recovered by cases (i)-(iii) in Theorem \ref{mddsecondregion} by noting that $R_2^*$ is used as the sum rate for the successive refinement problem.

\subsection{An Numerical Example}
\label{sec:caseiinum}
We consider the numerical example inspired by \cite{perron2006kaspi} and calculate the dispersion function for cases (iii) and (iv) in Theorem \ref{mddsecondregion}. Let $\calS_1=\{0,1\}$ and $\calS_2=\{0,1,\rme\}$. Let $S_1$ take values in $\calS_1$ with equal probability and let $P_{S_2|S_1}(s_2|s_1)=(1-p)\bbo(s_1=s_2)+p\bbo(s_2=\rme)$. 
Let the source be $X=(S_1,S_2)$ and the deterministic function be $Y=g(X)=g(S_1,S_2)=S_2$. Let $\hatX_1=\hatX_2=\{0,1\}$ and the distortion measures be $d_1(x,\hatx_1)=\bbo(s_1=\hatx_1)$ and $d_2(x,\hatx_2)=\bbo(s_2=\hatx_2)$. Choose $(p,D_1,D_2)$ such that $D_1\leq \frac{1}{2}$ and $D_1-\frac{1-p}{2}\leq D_2\leq p D_1$. For this case, using the definitions of $\xi^*$ in \eqref{def:s*}, $\lambda_1^*$ in \eqref{def:t1*} and $\lambda_2^*$ in \eqref{def:t2*}, we have
\begin{align}
\xi^*&=0,\\
\lambda_1^*&=\log\big((1-p)/(D_1-D_2)-1\big),\\
\lambda_2^*&=-\lambda_1^*+\log\big(p/D_2-1\big).
\end{align} 
Recall that $H_\rmb(\cdot)$ is the binary entropy function. Let
\begin{align}
\alpha_0&:=\log\big(2/(1+\exp(-\lambda_1^*))\big)-\lambda_1^*D_1-\lambda_2^*D_2,\\
\alpha&:=\log\big(2/(1+\exp(-\lambda_1^*-\lambda_2^*)\big)-\lambda_1^*D_1-\lambda_2^*D_2,
\end{align}
\begin{align}
g_1(p,D_1,D_2)
\nn&:=\log 2-(1-p)H_\rmb((D_1-D_2)/(1-p))\\*
&\qquad-pH_\rmb(D_2/p),\\
g_2(p,D_1,D_2)
\nn&:=p(1-p)\Big\{\log(1-D_2/p)\\*
&\qquad-\log(1-(D_1-D_2)/(1-p))\Big\}^2,\\
g_3(p,D_1,D_2)&:=(1-p)\alpha_0\log\frac{2}{1-p}+p\alpha\log\frac{1}{p}.
\end{align}
Then, it can be verified that
\begin{align}
H(P_Y)&=(1-p)\log 2+H_\rmb(p),\\
\rmV(P_Y)&=p(1-p)\Big(\log\frac{2p}{1-p}\Big)^2.
\end{align}
Thus,
\begin{align}
\rmV(R_1,D_1,D_2|P_X)
\nn&=2\Big(g_3(p,D_1,D_2)-H(P_Y)g_1(p,D_1,D_2)\Big)\\*
&\qquad+g_2(p,D_1,D_2)+\rmV(P_Y).
\end{align}

\section{Proof of Second-Order Asymptotics}
\label{prove:mddsecondregion}

\subsection{Achievability}
In this subsection, we first present a type covering lemma tailored to the Fu-Yeung problem, using which we derive an upper bound on the joint excess-distortion and error probability. Finally, invoking Taylor expansions and the Berry-Esseen Theorem, we derive an achievable second-order coding region.

Define 
\begin{align}
c_1&=|\calX|\cdot |\calY|\cdot |\hat{\calX_1}|+2,\\
c_2&=7|\calX|\cdot |\calY| \cdot |\hat{\calX}_1|\cdot|\hat{\calX_2}|+4.
\end{align}

We are now ready to present the type covering lemma.
\begin{lemma}
\label{mddtypecovering}
Consider any type $Q_X\in\calP_n(\calX)$. Let $Q_Y$ be induced by $Q_X$ and the deterministic function $g:\calX\to\calY$, and let $R_1\geq R(Q_X,D_1)$. The following conclusions hold.
\begin{enumerate}
\item There exists a set $\calB\in\calX_1^n$ such that for each $x^n\in\calT_{Q_X}$, \\$d_1(x^n,(z^n)^*)\leq D_1$ where $(z^n)^*:=\argmin_{z\in\calB} d_1(x^n,z)$.
\item Given $(z^n)^*$, there exists a set $\calB((z^n)^*)\in\calX_2^n$ such that
\begin{align}
\min_{\hatx_2^n\in\calB((z^n)^*)}d_2(x^n,\hatx_2^n)\leq D_2.
\end{align}
\item There exists a set $\calB_Y\in\hat{\calY}^n$ satisfying that $\frac{1}{n}\log |\calB_Y|\leq H(Q_Y)$ and there exists $\haty^n\in\calB_Y$ such that $\haty^n=g(x^n)$.
\item The sizes of sets $\calB$ and $\calB((z^n)^*)$ satisfy
\begin{align}
\frac{1}{n}\log |\calB|&\leq R_1+c_1\log (n+1)\\
\frac{1}{n}\log (|\calB|\cdot|\calB((z^n)^*)|)
&\leq \rvR(R_1,D_1,D_2|Q_X)+(c_1+c_2)\log(n+1).
\end{align}
\end{enumerate}

\end{lemma}
The proof of Lemma \ref{mddtypecovering} is similar to the type covering lemma for the successive refinement problem~\cite{no2016}.

Let 
\begin{align}
R_{1,n}:&=\frac{1}{n}\big(\log M_1 -(c_1+|\calX|)\log (n+1)\big)\label{def:r1n},\\
R_{2,n}:&=\frac{1}{n}\big(\log M_2-(c_2+|\calY|)\log (n+1)\big)\label{def:r2n}.
\end{align}

Invoking Lemma \ref{mddtypecovering}, we can upper bound the joint excess-distortion and error probability for an $(n,M_1,M_2)$-code. 

\begin{lemma}
\label{mdduppjoint}
There exists an $(n,M_1,M_2)$-code such that 
\begin{align}
\nn&\rmP_{\rme,n}(D_1,D_2)\\*
\nn&\leq \Pr\bigg\{R_{1,n}<R(\hatT_{X^n},D_1)\mathrm{~or}~R_{2,n}+\frac{c_2\log (n+1)}{n}<H(\hatT_{g(X^n)})\\
&\qquad\mathrm{or}~R_{1,n}+R_{2,n}<\rvR(R_{1,n},D_1,D_2|\hatT_{X^n})+H(\hatT_{g(X^n)})\bigg\}.
\end{align}
\end{lemma}
\begin{proof}
Set $(R_1,R_2)=(R_{1,n},R_{2,n})$. Consider the following coding scheme. Given a source $x^n$, the encoder $f_2$ calculates its type $\hatT_{x^n}$. Then, the encoder $f_2$ obtain $y^n$ using the deterministic function $y_i=g(x_i)$ and its type $\hatT_{y^n}$. Now encoder $f_2$ calculates $R(\hatT_{x^n},D_1)$ and $\rvR(R_{1,n},D_1,D_2|\hatT_{x^n})$. If $\log M_1<nR(\hatT_{x^n},D_1)+(c_1+|\calX|)\log (n+1)$ or  $\log M_2<nH(\hatT_{y^n})+|\calY|\log (n+1)$ or $\log M_1M_2<n\rvR(R_{1,n},D_1,D_2|\hatT_{x^n})+nH(\hatT_{y^n})+(c_1+c_2+|\calX|+|\calY|)\log (n+1)$, then the system declares an error.  Otherwise, the encoder $f_1$ sends the type of $x^n$ with at most $|\calX|\log (n+1)$ nats and the encoder $f_2$ sends the type of $y^n$ using at most $|\calY|\log (n+1)$ nats. Furthermore, the encoder $f_2$ sends the index of $y^n=g(x^n)$ in the type class $\calT_{\hatT_{y^n}}$.
Now, choose $\calB\in\calX_1^n$ in Lemma \ref{mddtypecovering} and let $(z^n)^*=\argmin_{z\in\calB}d_1(x^n,z)$. Given $(z^n)^*$, choose $\calB((z^n)^*)$ in Lemma \ref{mddtypecovering} and let $z_2^*=\argmin_{z_2\in\calB((z^n)^*)}d_2(x^n,z_2)$. Finally, we use the encoder $f_1$ to send the index of $z_1^*$ and use either $f_1$ or $f_2$ to send out the index of $z_2^*$. Invoking Lemma \ref{mddtypecovering}, we conclude that no error will be made if $\log M_1\geq nR(\hatT_{x^n},D_1)+(c_1+|\calX|)\log (n+1)$, $\log M_2\geq nH(\hatT_{y^n})+|\calY|\log (n+1)$ and $\log M_1M_2\geq n\rvR(R_{1,n},D_1,D_2|\hatT_{x^n})+nH(\hatT_{y^n})+(c_1+c_2+|\calX|+|\calY|)\log (n+1)$. The proof is now complete.
\end{proof}

Recall that $(R_1^*,R_2^*)$ is a boundary rate-pair on the rate-distortion region of the Fu-Yeung problem. Choose $(M_1,M_2)$ such that
\begin{align}
\log M_1&=nR_1^*+L_1\sqrt{n}+(c_1+|\calX|)\log (n+1),\\
\log M_2&=nR_2^*+L_2\sqrt{n}+(c_2+|\calY|)\log (n+1).
\end{align}
It follows from \eqref{def:r1n} and \eqref{def:r2n} that
\begin{align}
R_{i,n}&=R_i^*+\frac{L_i}{\sqrt{n}},~i=1,2.
\end{align}
Recall the definition of the typical set $\calA_n(P_X)$ in \eqref{def:typicalset}. The result in \eqref{upp:atypical} states that
\begin{align}
\Pr\Big\{\hatT_{X^n}\notin \calA_n(P_X)\Big\}\leq \frac{2|\calX|}{n^2}.
\end{align}
Recall that $P_Y$ is induced by the source distribution $P_X$ and the deterministic function $g:\calX\to\calY$. Thus, given any $x^n$, for each $y\in\calY$,
\begin{align}
\hatT_{y^n}(y)-P_Y(y)&=
\hatT_{g(x^n)}(y)-P_Y(y)\\
&=\sum_{x:g(x)=y} \Big(\hatT_{x^n}(x)-P_X(x)\Big).
\end{align}
Thus, if $\hatT_{X^n}\in\calA_n(P_X)$, 
\begin{align}
\|\hatT_{Y^n}-P_Y\|_{\infty}&\leq |\calX|\sqrt{\frac{\log n}{n}}.
\end{align}

For $x^n$ such that $\hatT_{x^n}\in\calA_n(P_X)$, applying Taylor's expansions and noting that $y^n=g(x^n)$, we obtain
\begin{align}
H(\hatT_{g(x^n)})
&=H(\hatT_{y^n})\\
&=H(P_Y)+\sum_{y} \Big(\hatT_{y^n}(y)-P_Y(y)\Big)(-\log P_Y(y))+O\Big(\|\hatT_{y^n}-P_Y\|^2\Big)\\
&=\sum_{y} -\hatT_{y^n}(y)\log P_Y(y)+O\Bigg(\frac{\log n}{n}\Bigg)\\
&=\frac{1}{n}\sum_{i\in[n]} -\log P_Y(y_i)+O\Bigg(\frac{\log n}{n}\Bigg)\label{fy:taylor1},
\end{align}
and
\begin{align}
\nn &\rvR(R_{1,n},D_1,D_2|\hatT_{x^n})\\
\nn&=\rvR(R_1^*,D_1,D_2|P_X)-s^*\frac{L_1}{\sqrt{n}}+O(|R_{1,n}-R_1^*|^2)\\*
&\qquad+\sum_x \Big(\hatT_{x^n}-P_X(x)\Big)\jmath(x,g(x)|R_1^*,D_1,D_2,P_X)+O\Big(\|\hatT_{x^n}-P_X\|^2\Big),\label{usederivative4fy}\\
&=\frac{1}{n}\sum_{i\in[n]} \jmath(x_i,g(x_i)|R_1^*,D_1,D_2,P_X) -\frac{s^*L_1}{\sqrt{n}}+O\Big(\frac{\log n}{n}\Big)\label{fy:taylor2},
\end{align}
where \eqref{usederivative4fy} follows from Lemma \ref{firstderive4fy}. Furthermore, for $x^n$ such that $\hatT_{x^n}\in\calA_n(P_X)$, it follows from \eqref{useprops}that
\begin{align}
R(\hatT_{x^n},D)
&=\frac{1}{n}\sum_{i\in[n]}\jmath(x_i|D,P_X)+O\Bigg(\frac{\log n}{n}\Bigg)\label{fy:taylor3}.
\end{align}

Recall that $\xi_n=\frac{\log n}{n}$. Therefore, invoking Lemma \ref{mdduppjoint}, we obtain 
\begin{align}
\nn&\rmP_{\rme,n}(D_1,D_2)\\*
\nn&\leq \Pr\bigg\{R_{1,n}<R(\hatT_{X^n},D_1)~\mathrm{or}~R_{2,n}+\frac{c_2\log (n+1)}{n}< H(\hatT_{g(X^n)})\\
\nn&\qquad \qquad \mathrm{or}~R_{1,n}+R_{2,n}<H(\hatT_{g(X^n)})+\rvR(R_{1,n},D_1,D_2|\hatT_{X^n})\\*
&\qquad\qquad\mathrm{and}~\hatT_{X^n}\in\calA_n(P_{XY})\bigg\}+\Pr\Big\{\hatT_{X^n}\notin \calA_n(P_{XY})\Big\}\\
\nn&\leq \Pr\bigg\{R_1+\frac{L_1}{\sqrt{n}}<\frac{1}{n}\sum_{i\in[n]} \jmath(X_i,D_1|P_X)+O(\xi_n)~\mathrm{or}~R_2<\frac{1}{n}\sum_{i\in[n]} \log \frac{1}{P_Y(Y_i)}+O(\xi_n)\\*
\nn&\qquad \qquad \mathrm{or}~R_1+R_2+\frac{(1+s^*)L_1+L_2}{n}\\*
\nn&\qquad\qquad\quad<\frac{1}{n}\sum_{i\in[n]}\Big(\jmath(X_i,g(X_i)|R_1,D_1,D_2,P_X)-\log P_Y(Y_i)\Big)+O(\xi_n)\bigg\}\\*
&\qquad+\frac{2|\calX|}{n^2}\label{upperbd1}.
\end{align}

Subsequently, we upper bound \eqref{upperbd1} for different cases of boundary rate pairs $(R_1^*,R_2^*)$ in Theorem \ref{mddsecondregion}. For simplicity, let $\jmath(x|R_1^*)$ denote $\jmath(x^*,g(x^*)|R_1^*,D_1,D_2,P_X)$ for each $x\in\calX$.

\begin{itemize}
\item Case (i) $R_1^*=R(P_X,D_1)$ and $R_2^*>\rvR_2^*(D_1,D_2|P_X)$

In this case $R_2^*>H(P_Y)$. Thus, it follows from the weak law of large numbers in Theorem \ref {wlln} that
\begin{align}
\kappa_{1,n}:=\Pr\Big\{R_2^*<\frac{1}{n}\sum_{i\in[n]} \log \frac{1}{P_Y(Y_i)}+O(\xi_n)\Big\}\to 0
\end{align}
and
\begin{align}
\nn\kappa_{2,n}&:=\Pr\Big\{R_1^*+R_2^*+\frac{(1+s^*)L_1+L_2}{n}\\*
\nn&\qquad\qquad<\frac{1}{n}\sum_{i\in[n]}\Big(\jmath(X_i|R_1^*)-\log P_Y(Y_i)\Big)+O(\xi_n)\Big\}\\
&\to 0.
\end{align}
It follows from \eqref{upperbd1} that
\begin{align}
\nn&\rmP_{\rme,n}(D_1,D_2)\\*
\nn&\leq \Pr\Big\{R_1^*+\frac{L_1}{\sqrt{n}}<\frac{1}{n}\sum_{i\in[n]} \jmath(X_i,D_1|P_X)+O(\xi_n)\Big\}\\*
&\qquad+\frac{2|\calX|}{n^2}+\kappa_{1,n}+\kappa_{2,n}\label{step1casei}\\
&\leq \rmQ\Bigg(\frac{L_1+O(\sqrt{n}\xi_n)}{\sqrt{\rmV(P_X,D_1)}}\Bigg)+\frac{6\rmT(P_X,D_1)}{\sqrt{n}\rmV(P_X,D_1)}+\frac{2|\calX|}{n^2}+\kappa_{1,n}+\kappa_{2,n}\label{laststepcasei},
\end{align}
where $\rmT(P_X,D_1)$ is the third absolute moment of $\jmath(X,D_1|P_X)$ (which is finite for a DMS) and \eqref{laststepcasei} follows by applying the Berry-Esseen theorem to the first term in \eqref{step1casei}. If we choose $(L_1,L_2)$ such that
\begin{align}
L_1\geq \sqrt{\rmV(P_X,D_1)}\rmQ^{-1}(\varepsilon),
\end{align}
then $\limsup_{n\to\infty} \rmP_{\rme,n}(D_1,D_2)\leq \varepsilon$ as desired.

\item Case (ii) $R_1^*=R(P_X,D_1)$ and $R_2^*=\rvR_2^*(D_1,D_2|P_X)$

In this case, $R_2^*>H(P_Y)$ still holds. Hence, invoking \eqref{upperbd1}, we obtain
\begin{align}
\nn&1-\rmP_{\rme,n}(D_1,D_2)\\*
\nn&\geq \Pr\bigg\{R_1^*+\frac{L_1}{\sqrt{n}}\geq \frac{1}{n}\sum_{i\in[n]} \jmath(X_i,D_1|P_X)+O(\xi_n),\\
\nn&\qquad\qquad R_1^*+R_2^*+\frac{(1+s^*)L_1+L_2}{\sqrt{n}}\geq \\*
\nn&\qquad\qquad\qquad\qquad\frac{1}{n}\sum_{i\in[n]} \Big(-\log P_Y(Y_i)+\jmath(X_i|R_1^*)+O(\xi_n)\bigg\}\\*
&\qquad-\frac{2|\calX|}{n^2}-\kappa_{1,n}\\
\nn&\geq 1-\Psi\big(L_1+O(\xi_n),(1+s^*)L_1+L_2+O(\xi_n);\bzero_2;\bV_1(R_1^*,D_1,D_2|P_X)\Big)\\*
&\qquad-\frac{2|\calX|}{n^2}-\kappa_{1,n}+O\Big(\frac{1}{\sqrt{n}}\Big).
\end{align}
Hence, if we choose $(L_1,L_2)$ such that 
\begin{align}
\Psi\big(L_1,(1+s^*)L_1+L_2;\bzero_2;\bV_1(R_1^*,D_1,D_2|P_X)\big)\geq 1-\varepsilon,
\end{align}
$\limsup_{n\to\infty} \rmP_{\rme,n}(D_1,D_2)\leq \varepsilon$.

\item Case (iii) $R(P_X,D_1)<R_1^*<\rvR_1^*(D_1,D_2|P_X)$, and \\$R_2^*=\rvR_1^*(D_1,D_2|P_X)+H(P_Y)-R_1^*$

In this case, $R_2^*>H(P_Y)$ holds again. The analysis is similar to Case (i). It can be verified that if we choose $(L_1,L_2)$ such that 
\begin{align}
(1+s^*)L_1+L_2\geq \sqrt{\rmV(R_1^*,D_1,D_2|P_X)}\rmQ^{-1}(\varepsilon),
\end{align}
$\limsup_{n\to\infty} \rmP_{\rme,n}(D_1,D_2)\leq \varepsilon$.

\item Case (iv) $R_1^*=\rvR_1^*(D_1,D_2|P_X)$ and $R_2^*=H(P_Y)$

The analysis is similar to Case (ii). It can be verified that if 
\begin{align}
\Psi((1+s^*)L_1+L_2,L_2;\bzero_2;\bV_2(R_1V_2^*,D_1,D_2|P_X))\geq 1-\varepsilon,
\end{align}
$\limsup_{n\to\infty}\rmP_{\rme,n}(D_1,D_2)\leq \varepsilon$.

\item Case (v) $R_1>\rvR_1^*(D_1,D_2|P_X)$ and $R_2^*=H(P_Y)$

The analysis is similar to Case (i). It can be verified that if
\begin{align}
L_2\geq \sqrt{\rmV(P_Y)}\rmQ^{-1}(\varepsilon),
\end{align}
we have $\limsup_{n\to\infty}\rmP_{\rme,n}(D_1,D_2)\leq \varepsilon$.
\end{itemize}
The achievability proof of Theorem \ref{mddsecondregion} is now completed.

\subsection{Converse}
The following type-based strong converse lemma is critical in the converse proof.
\begin{lemma}
\label{mddstrongconverse}
Fix $c>0$ and a type $Q_X\in\calP_n(P_X)$. For any $(n,M_1,M_2)$-code such that
\begin{align}
\nn&\Pr\Big\{d_1(X^n,\hatX_1^n)\leq D_1,~d_2(X^n,\hatX_2^n)\leq D_2,~\hatY^n=Y^n|X^n\in\calT_{Q_X}\Big\}\\*
&\geq \exp(-nc)\label{mddassump},
\end{align}
there exists a conditional distribution $Q_{\hatX_1\hatX_2|X}$ such that
\begin{align}
\frac{1}{n}\log M_1&\geq I(Q_X,Q_{\hatX_1|X})-\xi_{1,n},\\
\frac{1}{n}\log M_2&\geq H(Q_Y)-\xi_{2,n},\\
\frac{1}{n}\log M_1M_2\nn&\geq H(Q_Y)+I(Q_Y,Q_{\hatX_1|Y})\\*
&\qquad+I(Q_{X|Y},Q_{\hatX_1\hatX_2|XY}|Q_Y)-\xi_{1,n}-\xi_{2,n},
\end{align}
where 
\begin{align}
\xi_{1,n}&=\frac{|\calX|\log (n+1)+\log n+nc}{n},\\
\xi_{2,n}&=2\xi_{1,n}+\frac{2(\log n+nc)+|\calX|\cdot|\calY|\log (n+1)}{n}+\frac{\log |\calY|+h_b(1/n)}{n}.
\end{align}
and $Q_{X|Y}$, $Q_{\hatX_1|Y}$, $Q_{\hatX_1\hatX_2|XY}$ are induced by $Q_X$, $Q_{\hatX_1\hatX_2|X}$ and the deterministic function $y=g(x)$.

Furthermore, the expected distortions are bounded as
\begin{align}
\mathbb{E}_{Q_X\times Q_{\hatX_1\hatX_2|X}}[d_1(X,\hatX_1)]&\leq D_1+\frac{\overline{d}_1}{n}:=D_{1,n},\\
\mathbb{E}_{Q_X\times Q_{\hatX_1\hatX_2|X}}[d_2(X,\hatX_2)]&\leq D_2+\frac{\overline{d}_2}{n}:=D_{2,n}.
\end{align}
\end{lemma}
The proof of Lemma \ref{mddstrongconverse} is similar to Lemma \ref{type:sc} for the rate-distortion problem and Lemma \ref{mainresult_srtypestrongconverse_sr} for the successive refinement problem. The main technique is the perturbation approach by Gu and Effros~\cite{wei2009strong} and the generalization with method of types~\cite{watanabe2015second}.

Let $c=\frac{\log n}{n}$, then we have
\begin{align}
\xi_{1,n}&=\frac{|\calX|\log (n+1)+2\log n}{n},\\
\xi_{2,n}&=\frac{8\log n+(|\calX|\cdot|\calY|+2|\calX|)\log (n+1)}{n}+\frac{\log |\calY|+h_b(1/n)}{n}.
\end{align}
Define
\begin{align}
R_{i,n}=\frac{1}{n}(\log M_i+n\xi_{i,n}),~i\in[2].
\end{align}

Invoking Lemma \ref{mddstrongconverse}, we can prove the following lower bound on the joint excess-distortion and error probability for any $(n,M_1,M_2)$-code.
\begin{lemma}
\label{mddlbjoint}
Any $(n,M_1,M_2)$-code satisfies that
\begin{align}
\nn&\rmP_{\rme,n}(D_1,D_2)\\*
\nn&\geq \Pr\Big\{R_{1,n}<R(\hatT_{X^n},D_{1,n})~\mathrm{or}~R_{2,n}<H(\hatT_{g(X^n)})~\mathrm{or}\\
&\qquad R_{1,n}+R_{2,n}<\rvR(R_{1,n},D_{1,n},D_{2,n}|\hatT_{X^n})+H(\hatT_{g(X^n)})\Big\}.
\end{align}
\end{lemma}

The rest of the converse proof is omitted since it is analogous to the achievability proof where we use Taylor expansions similarly to \eqref{fy:taylor1} to \eqref{fy:taylor3} and apply (multi-variate) Berry-Esseen theorems for each case of boundary rate pairs $(R_1^*,R_2^*)$ in Theorem \ref{mddsecondregion}.


\chapter{Gray-Wyner Problem}
\label{chap:gw}

This chapter studies the lossy Gray-Wyner problem where three encoders cooperatively compress two correlated source sequences so that each of the two decoders could recover a source sequence reliably in a lossy manner. The lossy Gray-Wyner problem is a paradigm of the multiterminal lossy source coding problem where there exist multiple source sequences, multiple encoders and multiple decoders. The problem significantly generalizes the rate-distortion problem by introducing one more source sequence, two more encoders and one more decoder. 

The rate-distortion region for the problem was derived by Gray and Wyner~\cite{gray1974source} and this is why the problem is so named. An auxiliary random variable is needed to characterize the rate-distortion region of the lossy Gray-Wyner problem, which makes it significantly different from all problems discussed in previous chapters. The second-order asymptotics for the lossless version of the Gray-Wyner problem was derived by Watanabe~\cite{watanabe2015second}. This chapter presents the generalization of \cite{watanabe2015second} to the lossy case, analogously to the generalization of second-order asymptotics from lossless source coding in Chapter \ref{chap:lossless} (cf. \cite{strassen1962asymptotische,hayashi2008source}) to the rate-distortion problem in Chapter \ref{chap:rd} (cf. \cite{kostina2012fixed,ingber2011dispersion}). 

The Gray-Wyner problem is interesting beyond data compression. In the Gray-Wyner problem, there is an encoder who transmits messages to both decoders and its rate is known as the common rate. Given rates of the other two encoders, the minimal common rate equals a measure of common information of two correlated random variables~\cite{viswanatha2014}. Leveraging results on lossy common information by~Viswanatha, Akyol, and  Rose \cite{viswanatha2014} and considering rate triples on the Pangloss plane where the sum rate is constrained, the second-order asymptotic result is simplified and numerically illustrated. This chapter is largely based on \cite{zhou2015second}.

\section{Problem Formulation and Asymptotic Result}
\label{existingresult}
\subsection{Problem Formulation}
The lossy Gray-Wyner source coding problem \cite{gray1974source} is shown in Figure \ref{systemmodel:gw}. There are three encoders and two decoders. Encoder $f_i$ has access to a source sequence pair $(X^n,Y^n)$ and compresses it into a message $S_i$. Decoder $\phi_1$ aims to recover source sequence $X^n$ under fidelity criterion $d_1$ and distortion level $D_1$ with the encoded message $S_0$ from encoder $f_0$ and $S_1$ from encoder $f_1$. Similarly, the decoder $\phi_2$ aims to recover $Y^n$ with messages $S_0$ and $S_2$. 
We consider a correlated memoryless source $(X^n,Y^n)$ generated i.i.d. from a joint distribution $P_{XY}$ defined on a finite alphabet $\calX\times\calY$. 

\begin{figure}[htbp]
\centering
\setlength{\unitlength}{0.5cm}
\scalebox{0.8}{
\begin{picture}(26,9)
\linethickness{1pt}
\put(1,5.5){\makebox{$(X^n,Y^n)$}}
\put(6,1){\framebox(4,2)}
\put(6,4){\framebox(4,2)}
\put(6,7){\framebox(4,2)}
\put(7.7,1.8){\makebox{$f_2$}}
\put(7.7,4.8){\makebox{$f_0$}}
\put(7.7,7.8){\makebox{$f_1$}}
\put(1,5){\vector(1,0){5}}
\put(4.5,5){\line(0,1){3}}
\put(4.5,8){\vector(1,0){1.5}}
\put(4.5,5){\line(0,-1){3}}
\put(4.5,2){\vector(1,0){1.5}}
\put(14,1){\framebox(4,2)}
\put(14,7){\framebox(4,2)}
\put(15.7,7.7){\makebox{$\phi_1$}}
\put(15.7,1.7){\makebox{$\phi_2$}}
\put(10,2){\vector(1,0){4}}
\put(12,2.5){\makebox(0,0){$S_2$}}
\put(10,5){\line(1,0){6}}
\put(12,5.5){\makebox(0,0){$S_0$}}
\put(16,5){\vector(0,1){2}}
\put(16,5){\vector(0,-1){2}}
\put(10,8){\vector(1,0){4}}
\put(12,8.5){\makebox(0,0){$S_1$}}
\put(18,2){\vector(1,0){4}}
\put(18.7,2.5){\makebox{$(\hat{Y}^n,D_2$)}}
\put(18,8){\vector(1,0){4}}
\put(18.7,8.5){\makebox{$(\hat{X}^n,D_1)$}}
\end{picture}}
\caption{System model for the lossy Gray-Wyner source coding problem~\cite{gray1974source}.}
\label{systemmodel:gw}
\end{figure}
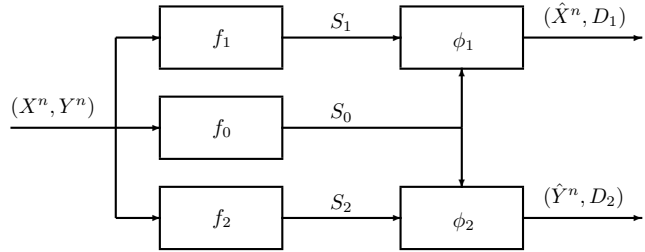

\begin{definition}
An $(n,M_0,M_1,M_2)$-code for lossy Gray-Wyner source coding consists of three encoders:
\begin{align}
f_0:\calX^n\times\calY^n\to\calM_0:=[M_0],\\
f_1:\calX^n\times\calY^n\to\calM_1:=[M_1],\\
f_2:\calX^n\times\calY^n\to\calM_2:=[M_2],
\end{align}
and two decoders:
\begin{align}
\phi_1:\calM_0\times\calM_1\to\hat{\calX}^n,\\
\phi_2:\calM_0\times\calM_2\to\hat{\calY}^n.
\end{align}
\end{definition}
Let $d_1:\calX\times\hat{\calX}\to[0,\infty)$ and $d_2:\calY\times\hat{\calY}\to[0,\infty)$ be two bounded distortion measures, Let $\overline{d}_1:=\max_{x,\hat{x}}d_1(x,\hat{x})$ and $\underline{d}_1:=\min_{x,\hat{x}:d_1(x,\hat{x})>0}d_1(x,\hat{x})$ denote the maximal and minimal distortion, respectively. Similarly, we define $\overline{d}_2$ and $\underline{d}_2$. Furthermore, let the average distortion between $x^n$ and $\hat{x}^n$ be defined as $d_1(x^n,\hat{x}^n):=\frac{1}{n}\sum_{i=1}^nd_1(x_i,\hat{x}_i)$ and the average distortion $d_2(y^n,\hat{y}^n)$ be defined similarly.

\subsection{Rate-Distortion Region}

The rate-distortion region of the lossy Gray-Wyner problem is defined as follows.
\begin{definition}
\label{deffirst}
A rate triplet $(R_0,R_1,R_2)$ is said to be $(D_1,D_2)$-achievable if there exists a sequence of $(n,M_0,M_1,M_2)$-codes such that
\begin{align}
\limsup_{n\to\infty}\frac{1}{n}\log M_0\leq R_0,\\*
\limsup_{n\to\infty}\frac{1}{n}\log M_1\leq R_1,\\*
\limsup_{n\to\infty}\frac{1}{n}\log M_2\leq R_2,
\end{align}
and
\begin{align}
\limsup_{n\to\infty} \mathbb{E}\left[d_1(X^n,\hat{X}^n)\right]\leq D_1,\\
\limsup_{n\to\infty} \mathbb{E}\left[d_2(Y^n,\hat{Y}^n)\right]\leq D_2.
\end{align}
The closure of the set of all $(D_1,D_2)$-achievable rate triplets is the $(D_1,D_2)$-optimal rate region and denoted as $\calR(D_1,D_2|P_{XY})$.
\end{definition}

Gray and Wyner characterized the $(D_1,D_2)$-achievable rate region in \cite{gray1974source}. Let $\calP(P_{XY})$ be the set of all joint distributions $P_{XYW}\in\calP(\calX\times\calY\times\calW)$ such that the $\calX\times\calY$-marginal of $P_{XYW}$ is the source distribution $P_{XY}$ and $|\calW|\leq |\calX||\calY|+2$. Denote the $\calX\times\calW$ marginal distribution as $P_{XW}$ and the $\calY\times\calW$ marginal distribution as $P_{YW}$.
\begin{theorem}
\label{gwregion}
The $(D_1,D_2)$-achievable rate region for lossy Gray-Wyner source coding is
\begin{align}
\nn&\calR(D_1,D_2|P_{XY})\\*
\nn&=\bigcup_{P_{XYW}\in\calP(P_{XY})}\Big\{(R_0,R_1,R_2):R_0\geq I(X,Y;W)\\*
&\qquad\qquad R_1\geq R_{X|W}(P_{XW},D_1), R_2\geq R_{Y|W}(P_{YW},D_2)\Big\},
\end{align}
where $R_{X|W}(P_{XW},D_1)$ and $R_{Y|W}(P_{YW},D_2)$ are conditional rate-distortion functions \cite[pp.~275, chapter 11]{el2011network}, i.e.,
\begin{align}
  R_{X|W}(P_{XW},D_1)&=\min_{P_{\hat{X}|XW}:\mathbb{E}[d_1(X,\hat{X})]\leq D_1}   I(X;\hat{X}|W)\label{gw:condrdf},
\end{align}
and $R_{Y|W}(P_{YW},D_2)$ is defined similarly.
\end{theorem}
Similarly to the rate-distortion and the Kaspi problems, the rate-distortion region in Theorem \ref{gwregion} still hold under the vanishing joint excess-distortion probability criterion, i.e., when $\lim_{n\to\infty}\rmP_{\rme,n}(D_1,D_2)=0$, where
 \begin{align}
\label{defexcessprob}
\rmP_{\rme,n}(D_1,D_2):=\Pr\Big\{d_1(X^n,\hat{X}^n)>D_1~\mathrm{or}~d_2(Y^n,\hat{Y}^n)> D_2\Big\}.
\end{align}

An equivalent form of the first-order coding region for Gray-Wyner problem was given in \cite[Exercise 14.9]{el2011network} and states that
\begin{align}
\calR(D_1,D_2|P_{XY})
\nn&=\bigcup_{\substack{P_{W|XY},P_{\hatX_1|XW},P_{\hatY|YW}:\\\mathbb{E}[d_1(X,\hatX)]\leq D_1,~\mathbb{E}[d_2(Y,\hatY)]\leq D_2}}\Big\{(R_0,R_1,R_2):R_0\geq I(X,Y;W)\\*
& \qquad\qquad\qquad R_1\geq I(X;\hatX|W), R_2\geq I(Y;\hatY|W)\Big\}, \label{gwregion2}
\end{align}

Given any rates $(R_1,R_2)$, let the minimal common rate be defined as
\begin{align}
\rvR_0(R_1,R_2,D_1,D_2|P_{XY})&:=\min\{R_0:(R_0,R_1,R_2)\in\calR(D_1,D_2|P_{XY})\}\label{minkey0}\\
&=\min_{\substack{P_{XYW}\in\in\calP(P_{XY}):\\R_1\geq R_{X|W}(P_{XW},D_1)\\R_2\geq R_{Y|W}(P_{YW},D_2)\\}}I(X,Y;W)\label{minkey}\\
&=\min_{\substack{P_{W|XW}P_{\hatX_1|XW}P_{\hatY|YW}:\\\mathbb{E}[d_1(X,\hatX)]\leq D_1,~\mathbb{E}[d_2(Y,\hatY)]\leq D_2\\ I(X;\hatX|W)\leq R_1,~I(Y;\hatY|W)\leq R_2}}          I(X,Y;W)\label{minkey2},
\end{align}
where \eqref{minkey} follows from Theorem \ref{gwregion} and \eqref{minkey2} follows from \eqref{gwregion2}. Given distortion levels $(D_1,D_2)$, a rate triple $(R_0^*,R_1^*,R_2^*)$ lies on the boundary of the rate-distortion region if and only if $R_0^*=\rvR_0(R_0^*,R_1^*,D_1,D_2)$, which is of interest in the study of second-order asymptotics.

\section{Rates-Distortions-Tilted Information Density}
Analogously to the derivation of second-order asymptotics for the rate-distortion problem, the definition of a tilted information density is critical and is usually related to the rate-distortion function (region). For the lossy Gray-Wyner problem, a slight obstacle is encountered on whether to define the rates-distortions-tilted information density using the formula of the minimal common rate in \eqref{minkey} that follows from Theorem \ref{gwregion} or the formula in \eqref{minkey2} that follows from the equivalent form of the rate-distortion region in \eqref{gwregion2}. This chapter shows that the latter is more amenable since it does not involve optimization in the conditional rate-distortion function in \eqref{gw:condrdf}.  

We now introduce the rates-distortions-tilted information density for the lossy Gray-Wyner problem. 
Since $\calR(D_1,D_2|P_{XY})$ is a convex set \cite{gray1974source}, the minimization in \eqref{minkey} is attained when $R_1=R_{X|W}(P_{XW},D_1)$ and $R_2=R_{Y|W}(P_{YW},D_2)$ for some optimal test channel $P_{W|XY}$ unless $\rvR_0(R_1,R_2,D_1,D_2|P_{XY})=0$ or $\infty$. To avoid degenerate cases, assume that $\rvR_0(R_1,R_2,D_1,D_2|P_{XY})>0$ is finite and $\calR(D_1,D_2|P_{XY})$ is smooth at a boundary rate triplet $(R_0^*,R_1^*,R_2^*)$ of our interest, i.e.,
\begin{align}
\xi_i^*:=-\frac{\partial \rvR_0(R_1,R_2,D_1,D_2|P_{XY})}{\partial R_i}\bigg|_{(R_1,R_2)=(R_1^*,R_2^*)},\label{definelambdai}
\\
\lambda_i^*:=-\frac{\partial \rvR_0(R_1,R_2,D_1',D_2'|P_{XY})}{\partial D_i'}\bigg|_{(D_1',D_2')=(D_1,D_2)}\label{definegammai},
\end{align}
are well-defined for $i\in[2]$\footnote{Due to these regularity conditions, our result in Chapter 2 does not hold for some singular points (e.g., where the derivatives do not exist) of the rate-distortion region, as in the lossless case by Watanabe in \cite{watanabe2015second}.}. Note that $\xi_i^*,~\lambda_i^*\geq 0$ since \\$\rvR_0(R_1,R_2,D_1,D_2|P_{XY})$ is non-increasing in $(R_1,R_2,D_1,D_2)$. Assume that all derivatives $(\xi_1^*,\xi_2^*,\lambda_1^*,\lambda_2^*)$ are strictly positive, which holds for all rate triplets $(R_0^*,R_1^*,R_2^*)$ such that $R_0^*=\rvR_0(R_1^*,R_2^*,D_1,D_2|P_{XY})$ is positive and finite.

Let $P_{W|XY}^*P_{\hatX|XW}^*P_{\hatY|YW}^*$ a tuple of optimal test channels\footnote{The following tilted information density is still well-defined even if the optimal test channel is not unique due to similar arguments as \cite[Lemma 2]{watanabe2015second}} that achieves $\rvR_0(R_1^*,R_2^*,D_1,D_2|P_{XY})$ in \eqref{minkey2}. Let $P_{\hatX|W}^*, P_{\hatY|W}^*, P_W^*$ be the induced (conditional) distributions. Given any $(x,y,w)\in\calX\times\calY\times\calW$, define the following two conditional distortion-tilted information densities
\begin{align}
&\jmath(x,D_1|w)
:=\log \frac{1}{\sum_{\hatx} P_{\hatX|W}^*(\hatx|w)\exp\Big(\frac{\lambda_1^*}{\xi_1^*}(D_1-d_1(x,\hatx))\Big)},\label{def:j1x}\\*
&\jmath(y,D_2|w)
:=\log \frac{1}{\sum_{\haty} P_{\hatY|W}^*(\haty|w)\exp\Big(\frac{\lambda_2^*}{\xi_2^*}(D_2-d_2(y,\haty))\Big)}\label{def:j2y}.
\end{align}

The rates-distortions-tilted information density for the lossy Gray-Wyner problem is defined as follows.
\begin{definition}
For a boundary rate triplet $(R_0^*,R_1^*,R_2^*)$, given any $(D_1,D_2)$, the rates-distortions-tilted information density for lossy Gray-Wyner source coding is defined as
\begin{align}
\jmath(x,y|R_1^*,R_2^*,D_1,D_2)
\nn&:=-\log \bigg(\sum_{w}P_W^*(w)\exp\Big(\xi_1^*(R_1^*-\jmath(x,D_1|w))\\*
&\qquad\qquad\qquad+\xi_2^*(R_2^*-\jmath(y,D_2|w))\Big)\bigg)\label{def:gwtilt}.
\end{align}
\end{definition}

Recall that there are two equivalent characterizations of the Gray-Wyner region, one defined in terms of conditional rate-distortion functions in Theorem \ref{gwregion} and the other defined solely in terms of (conditional) mutual information quantities in \eqref{gwregion2}. For the lossless Gray-Wyner problem~\cite{watanabe2015second}, the two regions are exactly the same. The tilted information densities derived based on these two regions are subtly different. We find that the tilted information density derived from the second region in \eqref{gwregion2} is more amenable to subsequent second-order analyses on the Pangloss plane (Lemma \ref{panglosstilted}). Thus the ``correct'' non-asymptotic fundamental quantity for the lossy Gray-Wyner problem is the rates-distortions-tilted information density in \eqref{def:gwtilt}.

The rates-distortions-tilted information density for lossy Gray-Wyner source coding has the following properties.
\begin{lemma}
\label{propertytilted}
The following properties hold.
\begin{enumerate}
\item The minimal common rate function equals the following expectation of the rate-distortions-tilted information density, i.e.,
\begin{align}
\rvR_0(R_1^*,R_2^*,D_1,D_2|P_{XY})&=\mathbb{E}_{P_{XY}}\left[\jmath_{XY}(X,Y|R_1^*,R_2^*,D_1,D_2,P_{XY})\right],
\end{align}
\item For $(w,\hatx,\haty)$ such that $P_W^*(w)P_{\hatX|W}^*(\hatx|w)P_{\hatY|W}^*(\haty|w)>0$,
\begin{align}
\nn&\jmath_{XY}(x,y|R_1^*,R_2^*,D_1,D_2,P_{XY})\\*
\nn&=\log \frac{P_{W|XY}^*(w|x,y)}{P_W^*(w)}+\xi_1^*\log\frac{P_{\hatX|XW}^*(\hatx|x,w)}{P_{\hatX|W}^*(\hatx|w)}+\xi_2^*\log\frac{P_{\hatY|YW}^*(\haty|y,w)}{P_{\hatY|W}^*(\haty|w)}\\*
&\qquad-\xi_1^*R_1^*-\xi_2^*R_2^*+\lambda_1^*(d_1(x,\hatx)-D_1)+\lambda_2^*(d_2(y,\haty)-D_2).
\end{align}  
\end{enumerate}
\end{lemma}
Lemma \ref{propertytilted} generalizes \cite[Lemma 1]{watanabe2015second} for the lossless Gray-Wyner problem and \cite[Lemma 1.4]{csiszar1974} for the rate-distortion problem.

In the following lemma, we relate the derivative of the minimum common rate function with the rates-distortions-tilted information density.
Recall that given a joint probability distribution $P_{XY}\in\calP(\calX\times\calY)$,  $m=|\supp(P_{XY})|$ and $\Gamma(P_{XY})$ be the sorted distribution such that for each $i\in[m]$, $\Gamma_i(P_{XY})=P_{XY}(x_i,y_i)$ is the $i$-th largest value of $\{P_{XY}(x,y):~(x,y)\in\calX\times\calY\}$. For any $Q_{XY}$, let $Q_{W|XY}^*Q_{\hatX|XW}^*Q_{\hatY|YW}^*$ be the optimal test channel for $\rvR_0(R_1^*,R_2^*,D_1,D_2|\Gamma(Q_{XY}))$ in~\eqref{minkey2}. Let $Q_W^*,Q_{\hatX|W}^*,Q_{\hatY|W}^*$ be the corresponding induced distributions.
\begin{lemma} 
\label{linkrjxy}
Suppose that for all $Q_{XY}$ in some neighborhood of $P_{XY}$, $\mathrm{supp}(Q_W^*)\subset\mathrm{supp}(P_W^*)$, $\mathrm{supp}(Q_{\hatX|W}^*)\subset\mathrm{supp}(P_{\hatX|W}^*)$ and $\mathrm{supp}(Q_{\hatY|W}^*)\subset\mathrm{supp}(P_{\hatX|W}^*)$. Then for $i\in[1:m-1]$,
\begin{align}
\nn&\frac{\partial \rvR_0(R_1^*,R_2^*,D_1,D_2|\Gamma(Q_{XY}))}{\partial \Gamma_i(Q_{XY})}\bigg|_{Q_{XY}=P_{XY}}\\*
&=\jmath_{XY}(x_i,y_i|R_1^*,R_2^*,D_1,D_2,\Gamma(P_{XY}))-\jmath_{XY}(x_m,y_m|R_1^*,R_2^*,D_1,D_2,\Gamma(P_{XY})).\label{eqn:derivate_rd}
\end{align}

\end{lemma}
Lemma \ref{linkrjxy} generalizes \cite[Lemma 3]{watanabe2015second} for the lossless Gray-Wyner problem and~\cite[Theorem 2.2]{kostina2013lossy} for the rate-distortion problem.

\section{Second-Order Asymptotics}
\subsection{Result}
Let $(R_0^*,R_1^*,R_2^*)$ be a boundary rate triplet on the rate-distortion region of the lossy Gray-Wyner problem.
\begin{definition}
\label{defsecond}
Given any $\varepsilon\in(0,1)$, a triplet $(L_0,L_1,L_2)$ is said to be second-order $(R_0^*,R_1^*,R_2^*,D_1,D_2,\varepsilon)$-achievable if there exists a sequence of $(n,M_0,M_1,M_2)$-codes such that
\begin{align}
\limsup_{n\to\infty}\frac{1}{\sqrt{n}}\left(\log M_0-nR_0\right)\leq L_0,\\
\limsup_{n\to\infty}\frac{1}{\sqrt{n}}\left(\log M_1-nR_1\right)\leq L_1,\\
\limsup_{n\to\infty}\frac{1}{\sqrt{n}}\left(\log M_2-nR_2\right)\leq L_2,
\end{align}
and
\begin{align}
\limsup_{n\to\infty}\rmP_{\rme,n}(D_1,D_2)\leq \varepsilon.
\end{align}
The closure of the set of all second-order $(R_0^*,R_1^*,R_2^*,D_1,D_2,\varepsilon)$-achievable triplets is called the second-order coding region and denoted as\\ $\calL(R_0^*,R_1^*,R_2^*,D_1,D_2,\varepsilon)$.
\end{definition}
Note that in Definition~\ref{deffirst} of the rate-distortion region, the expected distortion measure was considered, whereas in Definition~\ref{defsecond} , the excess-distortion probability is considered. This is consistent with other lossy source coding problems studied in previous chapters and the joint-excess-distortion probability allows us to derive second-order asymptotics that provides deeper understanding of the tradeoff among encoders beyond the rate-distortion region.

Let the rates-distortions-dispersion function be 
\begin{align}
\rmV(R_1^*,R_2^*,D_1,D_2|P_{XY})&:=\mathrm{Var}\left[\jmath_{XY}(X,Y|R_1^*,R_2^*,D_1,D_2,P_{XY})\right].
\end{align}
For any boundary rate triplet $(R_0^*,R_1^*,R_2^*)\in\calR(D_1,D_2|P_{XY})$, we impose the following conditions:
\begin{enumerate}
\item \label{cond1:gw} $R_0^*=\rvR_0(R_1^*,R_2^*,D_1,D_2|P_{XY})$ is positive and finite;
\item For $i\in[2]$, the derivatives $\xi_i$ in \eqref{definelambdai} and $\lambda_i^*$ in \eqref{definegammai} are well-defined and positive;
\item \label{cond2} $(R_1, R_2, Q_{XY})\mapsto \rvR_0(R_1,R_2,D_1,D_2|Q_{XY})$ is twice differentiable in the neighborhood of $(R_1^*,R_2^*,P_{XY})$ and the derivatives are bounded;
\item The dispersion function $\rmV(R_1^*,R_2^*,D_1,D_2|P_{XY})$ is finite.
\end{enumerate}

\begin{theorem}
\label{mainresult}
Under conditions (\ref{cond1:gw}) to (\ref{cond2}), given any $\varepsilon\in(0,1)$, the second-order coding region satisfies
\begin{align}
\calL(R_0^*,R_1^*,R_2^*,D_1,D_2,\varepsilon)
\nn&=\Big\{(L_0,L_1,L_2):L_0+\xi_1^*L_1+\xi_2^*L_2\\*
&\qquad\quad\geq \sqrt{\rmV(R_1^*,R_2^*,D_1,D_2|P_{XY})}\mathrm{Q}^{-1}(\varepsilon)\Big\}.
\end{align}
\end{theorem}
Theorem \ref{mainresult} is proved in Section \ref{secondorderproof}. In the achievability proofs, we derive a type covering lemma (cf. Lemma \ref{achievable}) designed specifically for the lossy Gray-Wyner source coding problem. While the proof of this type covering lemma itself hinges on various other works, e.g.,~\cite{no2016, Marton74,watanabe2015second}, piecing the ingredients together and ensuring that the resultant asymptotic results are tight is non-trivial. One of the main challenges here in proving the  type covering lemma is the requirement to establish the uniform continuity of the conditional rate-distortion function in {\em both} the source  distribution and distortion level. The converse proof is done similarly to the successive refinement or the Fu-Yeung problem where we first derive a type-based strong converse, then use Taylor expansions of the minimal common rate function of empirical distributions and finally apply the Berry-Esseen theorem (cf. Theorem \ref{berrytheorem}).

\subsection{Specialization to the Pangloss Plane} \label{sec:pang}
In general, it is not easy to calculate $\calL(R_0^*,R_1^*,R_2^*,D_1,D_2,\varepsilon)$. Here we consider calculating $\calL(R_0^*,R_1^*,R_2^*,D_1,D_2,\varepsilon)$ for a rate triplet $(R_0^*,R_1^*,R_2^*)$ on the Pangloss plane \cite{gray1974source}. It is shown in Theorem 6 in \cite{gray1974source} that $(R_0,R_1,R_2)$ is $(D_1,D_2)$-achievable if
\begin{align}
R_0+R_1+R_2&\geq R(P_{XY},D_1,D_2),\label{panglossbd}\\
R_0+R_1&\geq R(P_X,D_1),\\
R_0+R_2&\geq R(P_Y,D_2),
\end{align}
where $R(P_X,D_1)$, $R(P_Y,D_2)$ are rate-distortion functions (cf. \eqref{def:rd}) and $R(P_{XY},D_1,D_2)$ is the following joint rate-distortion function
\begin{align}
R(P_{XY},D_1,D_2)&:=\min_{P_{\hat{X}\hat{Y}|XY}:\mathbb{E}[d_1(X,\hat{X})]\leq D_1,~\mathbb{E}[d_2(Y,\hat{Y})]\leq D_2} I(X,Y;\hatX,\hatY).
\end{align}
The set of $(D_1,D_2)$-achievable rate triplets $(R_0,R_1,R_2)$ satisfying $R_0+R_1+R_2=R(P_{XY},D_1,D_2)$ is called the Pangloss plane, denoted as $\calR_{\mathrm{pgp}}(D_1,D_2|P_{XY})$, i.e.,
\begin{align}
\calR_{\mathrm{pgp}}(D_1,D_2|P_{XY})
\nn&:=\Big\{(R_0,R_1,R_2):(R_0,R_1,R_2)\in\calR(D_1,D_2|P_{XY})\\*
&\qquad\qquad R_0+R_1+R_2=R(P_{XY},D_1,D_2)\Big\}.
\end{align}
Let $P_{\hat{X}\hat{Y}|XY}^*$ be an optimal conditional distribution that achieves \\$R(P_{XY},D_1,D_2)$. Let $P_{\hat{X}\hat{Y}}^*$ be induced by $P_{\hat{X}\hat{Y}|XY}^*$ and $P_{XY}$. Define the following distortions-tilted information density:
\begin{align}
\nn&\imath_{XY}(x,y|D_1,D_2,P_{XY})\\*
&:=-\log \mathbb{E}_{P_{\hat{X}\hat{Y}}^*}\bigg[\exp\Big(\nu_1^*(D_1-d_1(x,\hat{X}))+\nu_2^*(D_2-d_2(y,\hat{Y}))\Big)\bigg]\label{defjrdtilted},
\end{align}
where
\begin{align}
\nu_1^*:&=-\frac{\partial R(P_{XY},D,D_2)}{\partial D}\bigg|_{D=D_1},\label{defjrdnu1}\\
\nu_2^*:&=-\frac{\partial R(P_{XY},D_1,D)}{\partial D}\bigg|_{D=D_2}\label{defjrdnu2}.
\end{align}
\begin{lemma}
\label{propertyjointrd}
The properties of $\imath_{XY}(\cdot|D_1,D_2,P_{XY})$ include
\begin{itemize}
\item The joint rate-distortion function is the expectation of the joint tilted information density, i.e.,
\begin{align}
R(P_{XY},D_1,D_2)=\mathbb{E}_{P_{XY}}\left[\imath_{XY}(X,Y|D_1,D_2,P_{XY})\right].
\end{align}
\item For each $(\hat{x},\hat{y})\in\mathrm{supp}(P_{\hatX\hatY})^*$, 
\begin{align}
&\imath_{XY}(x,y|D_1,D_2,P_{XY})
\nn=\log \frac{P_{\hat{X}\hat{Y}|XY}^*(\hat{x},\hat{y}|x,y)}{P_{\hat{X}\hat{Y}}^*(\hat{x},\hat{y})}\\*
&\qquad\qquad+\nu_1^*(d_1(x,\hat{x})-D_1)+\nu_2^*(d_2(y,\hat{y})-D_2).
\end{align}

\end{itemize}
\end{lemma}
Lemma \ref{propertyjointrd} can be proved similarly to \cite[Lemma 1]{watanabe2015second} for the lossless Gray-Wyner problem and \cite[Lemma 1.4]{csiszar1974} for the rate-distortion problem. By considering a fixed rate triplet on the Pangloss plane, we can relate $\jmath_{XY}(x,y|R_1^*,R_2^*,D_1,D_2,P_{XY})$ to $\imath_{XY}(x,y|D_1,D_2,P_{XY})$.

\begin{lemma}
\label{panglosstilted}
When a boundary rate-triple lies in the Pangloss plane, i.e.,~$(R_0^*,R_1^*,R_2^*)\in\calR_{\mathrm{pgp}}(D_1,D_2|P_{XY})$ and the common rate $R_0^*>0$,
\begin{align}
\jmath_{XY}(x,y|R_1^*,R_2^*,D_1,D_2,P_{XY})&=\imath_{XY}(x,y|D_1,D_2,P_{XY})-R_1^*-R_2^*\label{needproof}.
\end{align}
\end{lemma}
The proof of Lemma \ref{panglosstilted} invokes Lemma \ref{propertytilted}. Besides, we use an idea from \cite{viswanatha2014} in which it was shown that the following Markov chains hold for the optimal test channels $P_{W|XY}^*$ achieving $\calR(R_1^*,R_2^*,D_1,D_2|P_{XY})$ and  $P_{\hat{X}|XW}^*$ as well as $P_{\hat{Y}|YW}^*$ achieving conditional rate-distortion functions $R_{X|W}(P_{XW}^*,D_1)$ and $R_{Y|W}(P_{YW}^*,D_2)$: 
\begin{align}
\hat{X}&\to W\to \hat{Y} \\* 
(X,Y)&\to (\hat{X},\hat{Y})\to W\\* 
\hat{X}&\to(X,Y,W)\to \hat{Y}\\*  
\hat{X}&\to(X,W)\to Y \\* 
\hat{Y}&\to (Y,W) \to X.
\end{align}
 Invoking Lemma \ref{panglosstilted}, for a rate triplet $(R_0^*,R_1^*,R_2^*)$ on the Pangloss plane, the expression of the second-order coding region is simplified as follows.
\begin{proposition}
\label{proppangloss}
When $(R_0^*,R_1^*,R_2^*)\in\calR_{\mathrm{pgp}}(D_1,D_2|P_{XY})$ and the conditions in Theorem \ref{mainresult} are satisfied, we have
\begin{align}
\calL(R_0^*,R_1^*,R_2^*,D_1,D_2,\varepsilon)
\nn&=\Big\{(L_0,L_1,L_2):L_0+L_1+L_2\\*
&\qquad\quad\geq \sqrt{\rmV(R_1^*,R_2^*,D_1,D_2|P_{XY})}\mathrm{Q}^{-1}(\varepsilon)\Big\},
\end{align}
where the rate-dispersion function~\cite{kostina2012fixed} is 
\begin{align}
\rmV(R_1^*,R_2^*,D_1,D_2|P_{XY})
&=\mathrm{Var}[\jmath_{XY}(X,Y|R_1^*,R_2^*,D_1,D_2,P_{XY})]\\
&=\mathrm{Var}[\imath_{XY}(X,Y|D_1,D_2,P_{XY})].
\end{align}
\end{proposition}

\subsection{A Numerical Example for the Pangloss Plane}
Consider a doubly symmetric binary source (DSBS), where $\calX=\calY=\{0,1\}$, $P_{XY}(0,0)=P_{XY}(1,1)=\frac{1-p}{2}$ and $P_{XY}(0,1)=P_{XY}(1,0)=\frac{p}{2}$ for $p\in[0,\frac{1}{2}]$. We consider $\hat{\calX}=\hat{\calY}=\{0,1\}$ and Hamming distortion for both sources, i.e., $d_1(x,\hat{x})=\bbo(x=\hat{x})$ and $d_2(y,\hat{y})=\bbo(y=\hat{y})$. Furthermore, let $R_1=R_2=R$ and $D_1=D_2=D$. 
Recall that $H_\rmb(\delta)=-\delta\log(\delta)-(1-\delta)\log(1-\delta)$ is the binary entropy function. Define $f(x):=-x\log x$.  Let $p_1:=\frac{1}{2}-\frac{1}{2}\sqrt{1-2p}$. It follows from \cite[Exercise 2.7.2]{berger1971rate} that
\begin{align}
\nn&R(P_{XY},D,D)\\*
&=\left\{
\begin{array}{lr}
1+H_\rmb(p)-2H_\rmb(D)&0\leq D\leq p_1,\\
f(1-p)-\frac{1}{2}\left(f(2D-p)+f(2(1-D)-p)\right) &p_1\leq D\leq \frac{1}{2}.
\end{array}\right.
\end{align}
It was shown in \cite[Example 2.5(A)]{gray1974source} that for $0 \leq D\leq \Delta\leq p_1$, if $R_0=R(P_{XY},\Delta,\Delta)$, $R_1=R_2=H_\rmb(\Delta)-H_\rmb(D)$, then $(R_0,R_1,R_2)\in\calR_{\mathrm{pgp}}(D,D|P_{XY})$. When $D\leq p_1$, the joint $(D,D)$-tilted information density satisfies
\begin{align}
\imath_{XY}(0,0|D,D,P_{XY})&=\imath_{XY}(1,1|D,D,P_{XY})\\*
&=\log\frac{1}{(2p-1)D-(2p-1)D^2+\frac{1}{2}(1-p)}-2H_\rmb(D),\\
\imath_{XY}(0,1|D,D,P_{XY})&=\imath_{XY}(1,0|D,D,P_{XY})\\*
&=\log\frac{1}{(2p-1)D^2-(2p-1)D+\frac{1}{2}p}-2H_\rmb(D).
\end{align}
Hence, the joint dispersion function satisfies
\begin{align}
\nn&\mathrm{Var}[\imath_{XY}(X,Y|D,D,P_{XY})]\\*
&=\sum_{x,y}P_{XY}(x,y)
\left(\imath_{XY}(x,y|D,D,P_{XY})-R(P_{XY},D,D)\right)^2\\*
&\nn=(1-p)\left(\log\frac{1}{(2p-1)D-(2p-1)D^2+\frac{1}{2}(1-p)}-1-H_\rmb(p)\right)^2\\*
&\qquad +
p\left(\log\frac{1}{(2p-1)D^2-(2p-1)D+\frac{1}{2}p}-1-H_\rmb(p)\right)^2
\label{egvaluev}.
\end{align}

\section{Proof of Second-Order Asymptotics} 
\label{secondorderproof}
\subsection{Achievability}
We first prove that for any given joint type $Q_{XY} \in\calP_n(\calX\times\calY)$, there exists an $(n,M_0,M_1,M_2)$-code such that the excess-distortion probability is mainly due to the incorrect decoding of side information $W$. To do so, we present a novel type covering lemma for the lossy Gray-Wyner problem. Using this result, we then prove an upper bound of the excess-distortion probability for the $(n,M_0,M_1,M_2)$-code. Finally, we establish the achievable second-order coding region by estimating this probability.

Define four constants
\begin{align}
c_0&=\left(3|\calX||\calY||\calW|+4\right),\\
c_0'&=c_0+|\calX||\calY|\label{defc0p},\\ 
c_1&=\left(\frac{11\overline{d}_1}{\underline{d}_1}|\calX||\calY||\calW|+3|\calX||\calW||\hat{\calX}|+5\right)\label{defc1},\\
c_2&=\left(\frac{11\overline{d}_2}{\underline{d}_2}|\calX||\calY||\calW|+3|\calY||\calW||\hat{\calY}|+5\right)\label{defc2}.
\end{align}
The following type covering lemma is critical for second-order analysis for the lossy Gray-Wyner problem. 
\begin{lemma}
\label{achievable}
Let $n$ satisfy $(n+1)^4>n\log|\calX||\calY|$, $\log n\geq \frac{|\calX||\calW||\hat{\calX}|\log|\calX|\overline{d}_1}{D_1}$, $\log n\geq \frac{|\calY||\calW||\hat{\calY}|\log|\calY|\overline{d}_2}{D_2}$, and $\log n\geq \log\frac{|\hat{\calX}|}{|\calY|}$. Given a joint type $Q_{XY}\in\calP_n(\calX\times\calY)$, for any rate pair $(R_1,R_2)\in\bbR_{++}^2$ such that \\$\rvR_0(R_1,R_2,D_1,D_2|Q_{XY})$ is achievable by some test channel, there exists a conditional type $Q_{W|XY}\in\calV_{n}(\calW,Q_{XY})$ such that the following holds:
\begin{itemize}
\item There exists a set $\calC_n\subset\calT_{Q_{W}}$ ($Q_W$ is induced by $Q_{XY}$ and $Q_{W|XY}$) such that
\begin{itemize}
\item For any $(x^n,y^n)\in\calT_{Q_{XY}}$, there exists a $w^n\in\calC_n$ whose joint type with $(x^n,y^n)$ is $Q_{XYW}$, i.e., $(x^n,y^n,w^n)\in\calT_{Q_{XYW}}$.
\item  The size of $\calC_n$ is upper bounded by
\begin{align}
\frac{1}{n}\log |\calC_n|\leq \rvR_0(R_1,R_2,D_1,D_2|Q_{XY})+c_0\frac{\log(n+1)}{n}.
\end{align}
\end{itemize}
\item 
 For each $w^n\in\calT_{Q_{W|XY}}(x^n,y^n)$, there exist sets $\calB_{\hat{X}}(w^n)\in\hat{\calX}^n$ and $\calB_{\hat{Y}}(w^n)\in\hat{\calY}^n$ satisfying
\begin{itemize}
\item For each $(x^n,y^n)\in\calT_{Q_{XY|W}}(w^n)$, there exists $\hat{x}^n\in\calB_{\hat{X}}(w^n)$ and $\hat{y}^n\in\calB_{\hat{Y}}(w^n)$ such that
$d_1(x^n,\hat{x}^n)\leq D_1$ and $d_2(y^n,\hat{y}^n)\leq D_2$,
\item The sizes of $\calB_{\hat{X}}(w^n)$ and $\calB_{\hat{Y}}(w^n) $ are upper bounded as
\begin{align}
\frac{1}{n}\log |\calB_{\hat{X}}(w^n)|&\leq R_1+c_1\frac{\log n}{n},\\*
\frac{1}{n}\log |\calB_{\hat{Y}}(w^n)|&\leq R_2+c_2\frac{\log n}{n}.
\end{align}
\end{itemize}

\end{itemize} 
\end{lemma}
Lemma \ref{achievable} is proved by combining a few ideas from the literature: a type covering lemma for the conditional rate-distortion problem (modified from Lemma 4.1 in \cite{csiszar2011information} for the standard rate-distortion problem and Lemma 8 in \cite{no2016} for the successive refinement problem),  a type covering lemma for the common side information for the Gray-Wyner problem (Lemma 4 in \cite{watanabe2015second}) and finally, a uniform continuity lemma for the conditional rate-distortion function (modified from \cite{no2016,palaiyanur2008uniform}). The proof of Lemma \ref{achievable} adopts similar ideas as the proof of the first-order coding region~\cite{gray1974source} and is available in~\cite[Appendix F]{zhou2015second}. The main idea is that we first send the common information via the common link carrying $S_0$ and then we consider two conditional rate-distortion problems on the two private links carrying $S_1,S_2$ using the common information as the side information.

Invoking Lemma \ref{achievable}, we show that there exists an $(n,M_0,M_1,M_2)$-code whose excess-distortion probability can be upper bounded as follows. Recall the definitions of $c_0'$ in \eqref{defc0p}, $c_1$ in \eqref{defc1} and $c_2$ in \eqref{defc2}. Define three rates
\begin{align}
R_{0,n}&=\frac{1}{n}\log M_0-c_0'\frac{\log (n+1)}{n},\\*
R_{1,n}&=\frac{1}{n}\log M_1-c_1\frac{\log n}{n},\\*
R_{2,n}&=\frac{1}{n}\log M_2-c_2\frac{\log n}{n}.
\end{align}
\begin{lemma}
\label{upperboundexcessp}
There exists an $(n,M_0.M_1,M_2)$-code such that
\begin{align}
\rmP_{\rme,n}(D_1,D_2)\leq \Pr\left\{R_{0,n}<\rvR_0(R_{1,n},R_{2,n},D_1,D_2|\hat{T}_{X^nY^n})\right\}. \label{eqn:ed_bd}
\end{align}
\end{lemma}
The proof of Lemma \ref{upperboundexcessp} is similar to \cite[Lemma 5]{watanabe2015second} and available in~\cite[Appendix J]{zhou2015second}.

Recall the definition of the typical set $\calA_n(P_{XY})$ in \eqref{def:calan} and the result in \eqref{atypicalprob} that
\begin{align}
\Pr\left\{\hat{T}_{X^nY^n}\notin \calA_{n}(P_{XY})\right\}\leq \frac{2|\calX||\calY|}{n^2}.
\end{align}
For a rate triplet $(R_0^*,R_1^*,R_2^*)$ satisfying conditions in Theorem \ref{mainresult}, let 
\begin{align}
\frac{1}{n}\log M_0&=\rvR_0(R_1^*,R_2^*,D_1,D_2|P_{XY})+\frac{L_0}{\sqrt{n}}+c_0'\frac{\log (n+1)}{n},\\*
\frac{1}{n}\log M_1&=R_1^*+\frac{L_1}{\sqrt{n}}+c_1\frac{\log n}{n},\\*
\frac{1}{n}\log M_2&=R_2^*+\frac{L_2}{\sqrt{n}}+c_2\frac{\log n}{n}.
\end{align}
It follows that
\begin{align}
R_{i,n}=R_{i}^*+\frac{L_i}{\sqrt{n}},~i=0,1,2.
\end{align}

In subsequent analyses, for ease of notation, we use $\jmath_{XY}(X_i,Y_i)$ to denote $\jmath_{XY}(X_i,Y_i|R_1^*,R_2^*,D_1,D_2,P_{XY})$. From the conditions in Theorem \ref{mainresult}, the second derivatives of the minimal sum rate function $\rvR_0(R_1,R_2,D_1,D_2|P_{XY})$ with respect to $(R_1,R_2,P_{XY})$ are bounded around a neighborhood of $(R_1^*,R_2^*,P_{XY})$. Hence, for any $\hat{T}_{x^ny^n}\in\calA_{n}(P_{XY})$, for large $n$, invoking Lemma \ref{linkrjxy} and applying Taylor's expansion for $\rvR_0(R_{1,n},R_{2,n},D_1,D_2|\hat{T}_{x^ny^n})$, we obtain:
\begin{align}
\nn &\rvR_0(R_{1,n},R_{2,n},D_1,D_2|\hat{T}_{x^ny^n})\\*
\nn&=\rvR_0(R_1^*,R_2^*,D_1,D_2|P_{XY})-\xi_1^*\frac{L_1}{\sqrt{n}}-\xi_2^*\frac{L_2}{\sqrt{n}}\\*
\nn&\qquad+\sum_{i=1}^m \left(\lambda_i(\hatT_{x^ny^n})-\lambda_i(P_{XY})\right)\Big(\jmath(x_i,y_i)-\jmath(x_m,y_m)\Big)\\*
&\qquad+O\left(\|\lambda(\hatT_{x^ny^n})-\Gamma(P_{XY})\|^2\right)+O\left((R_{1,n}-R_1^*)^2+(R_{2,n}-R_2^*)^2\right)\\
&=\nn\rvR_0(R_1^*,R_2^*,D_1,D_2|P_{XY})-\xi_1^*\frac{L_1}{\sqrt{n}}-\xi_2^*\frac{L_2}{\sqrt{n}}+O\left(\frac{\log n}{n}\right)\\*
&\qquad+\sum_{x,y}\left(\hat{T}_{x^ny^n}(x,y)-P_{XY}(x,y)\right)\jmath_{XY}(x,y)\\
&\leq \sum_{x,y}Q_{XY}(x,y)\jmath_{XY}(x,y)-\xi_1^*\frac{L_1}{\sqrt{n}}-\xi_2^*\frac{L_2}{\sqrt{n}}+O\left(\frac{\log n}{n}\right)\label{minr0}\\*
&=\frac{1}{n}\sum_{i=1}^n\jmath_{XY}(x_i,y_i)-\xi_1^*\frac{L_1}{\sqrt{n}}-\xi_2^*\frac{L_2}{\sqrt{n}}+O\left(\frac{\log n}{n}\right)\label{taylorfisrtt},
\end{align}
where \eqref{minr0} follows from Lemma \ref{propertytilted} and the definition of the typical set $\calA_{n}(P_{XY})$ in \eqref{def:calan}.

Define $\eta_n=\frac{\log n}{n}$. Invoking Lemma \ref{upperboundexcessp}, we can upper bound the excess-distortion probability 
as follows:
\begin{align}
&\nn \rmP_{\rme,n}(D_1,D_2)\\*
&\leq \Pr\left\{R_{0,n}<\rvR_0(R_{1,n},R_{2,n},D_1,D_2|\hat{T}_{X^nY^n})\right\}\\
\nn&\leq \Pr\left\{\hat{T}_{X^nY^n}\in\calA_{n}(P_{XY}), R_{0,n}<\rvR_0(R_{1,n},R_{2,n},D_1,D_2|\hat{T}_{X^nY^n})\right\}\\*
&\qquad+\Pr\left\{\hat{T}_{X^nY^n}\notin\calA_{n}(P_{XY})\right\}\\
\nn&\leq \Pr\bigg\{R_{0,n}<\frac{1}{n}\sum_{i=1}^n \jmath_{XY}(X_i,Y_i)-\xi_{1}^*\frac{L_1}{\sqrt{n}}-\xi_{2}^*\frac{L_2}{\sqrt{n}}+O(\eta_n)\bigg\}\\*
&\qquad+\frac{2|\calX||\calY|}{n^2}\\
\nn&=\Pr\bigg\{\frac{L_0}{\sqrt{n}}+\xi_{1}^*
\frac{L_1}{\sqrt{n}}+\xi_{2}^*\frac{L_2}{\sqrt{n}}+O(\eta_n)<\frac{1}{n}\sum_{i=1}^n \Big(\jmath_{XY}(X_i,Y_i)\\*
&\qquad\qquad\qquad-\rvR_0(R_1^*,R_2^*,D_1,D_2|P_{XY})\Big)\bigg\}+\frac{2|\calX||\calY|}{n^2}\\
\nn&\leq \rmQ\left(\frac{L_0+\xi_1^*L_1+\xi_2^*L_2+O\left(\sqrt{n}\eta_n\right)}{\sqrt{\rmV(R_1^*,R_2^*,D_1,D_2|P_{XY})}}\right)+\frac{6\mathrm{T}(R_1^*,R_2^*,D_1,D_2)}{\sqrt{n}\rmV^{3/2}(R_1^*,R_2^*,D_1,D_2)}\\*
&\qquad+\frac{2|\calX||\calY|}{n^2},\label{berryessen}
\end{align}
where \eqref{berryessen} follows from the Berry-Esseen Theorem and $\mathrm{T}(R_1^*,R_2^*,D_1,D_2)$ is third absolute moment of the rates-distortions-tilted information density  $\jmath_{XY}(X,Y)$, which is finite for a DMS from the conditions in Theorem \ref{mainresult}. Therefore, if $(L_0,L_1,L_2)$ satisfies
\begin{align}
L_0+\xi_1^*L_1+\xi_2^*L_2 \geq \sqrt{\rmV(R_1^*,R_2^*,D_1,D_2|P_{XY})}\mathrm{Q}^{-1}(\varepsilon),
\end{align}
then $\limsup_{n\to\infty}\rmP_{\rme,n}(D_1,D_2)\leq \varepsilon$.

\subsection{Converse}
We follow the method of types, similarly to the proof of the lossless case in \cite{watanabe2015second} and to the converse proof of the successive refinement and Fu-Yeung problem in previous chapters. We first establish a type-based strong converse and use it to derive a lower bound on excess-distortion probability $\rmP_{\rme,n}(D_1,D_2)$. Finally, we use a Taylor expansion and apply the Berry-Esseen Theorem to obtain an outer region expressed essentially using $\rmV(R_1^*,R_2^*,D_1,D_2|P_{XY})$.
 
We now consider an $(n,M_0,M_1,M_2)$-code for the correlated source $(X^n,Y^n)$ with joint distribution $U_{\calT_{Q_{XY}}}(x^n,y^n)=|\calT_{Q_{XY}}|^{-1}$, the uniform distribution over the type class $\calT_{Q_{XY}}$.
\begin{lemma}
\label{typestrongconverse}
If the non-excess-distortion probability satisfies
\begin{align}
\label{assumption}
\nn&\Pr\left\{d_1(X^n,\hat{X}^n)\leq D_1, d_2(Y^n,\hat{Y}^n)\leq D_2|(X^n,Y^n)\in\calT_{Q_{XY}}\right\}\\*
&\geq \exp(-n\alpha)
\end{align}
for some positive number $\alpha$, then for $n$ large enough such that $\log n\geq\max\{\overline{d}_1,\overline{d}_2\}\log|\calX|$, there exists a conditional distribution $Q_{W|XY}$ with $|\calW|\leq |\calX||\calY|+2$ such that
\begin{align}
\frac{1}{n}\log M_0&\geq I(X,Y;W)-\frac{\Big(|\calX||\calY|+1\Big)\log(n+1)}{n}-\alpha,\\
\frac{1}{n}\log M_1&\geq R_{X|W}(Q_{XW},D_1)-\frac{\log n}{n},\\
\frac{1}{n}\log M_2&\geq R_{Y|W}(Q_{YW},D_2)-\frac{\log n}{n}.
\end{align}
where $(X,Y,W)\sim Q_{XY}\times Q_{W|XY}$.
\end{lemma}
The proof of Lemma \ref{typestrongconverse} is similar to the lossless Gray-Wyner problem~\cite[Lemma 6]{watanabe2015second} but we need to also combine this with the (weak) converse proof for lossy Gray-Wyner problem under the expected distortion criterion in \cite{gray1974source}. Readers could refer to \cite[Appendix K]{zhou2015second} for details.

We then prove a lower bound on the excess-distortion probability $\rmP_{\rme,n}(D_1,D_2)$
in \eqref{defexcessprob}. 
Define the constant $c=\frac{|\calX||\calY|+2}{n}$ and the three quantities
\begin{align}
R_{0,n}&:=\frac{1}{n}\log M_0+c\frac{\log (n+1)}{n},\label{eq1}\\
R_{1,n}&:=\frac{1}{n}\log M_1+\frac{\log n}{n},\label{eq2}\\
R_{2,n}&:=\frac{1}{n}\log M_2+\frac{\log n}{n}.\label{eq3}
\end{align}
\begin{lemma}
\label{lowerboundexcessp}
Consider any $n\in\bbN$ such that $\log n\geq\max\{\overline{d}_1,\overline{d}_2\}\log|\calX|$. Any $(n,M_0,M_1,M_2)$-code satisfies
\begin{align}
\rmP_{\rme,n}(D_1,D_2)
&\geq \Pr\left\{R_{0,n}< \rvR_0(R_{1,n},R_{2,n},D_1,D_2|\hat{T}_{X^nY^n})\right\}-\frac{1}{n}.
\end{align} 
\end{lemma}
The proof of Lemma \ref{lowerboundexcessp} is similar to \cite[Lemma 7]{watanabe2015second} and available in \cite[Appendix L]{zhou2015second}.

Choose $(M_0,M_1,M_2)$ such that
\begin{align}
\frac{1}{n}\log M_0&=R_0^*+\frac{L_0}{\sqrt{n}}-c\frac{\log (n+1)}{n},\\
\frac{1}{n}\log M_1&=R_1^*+\frac{L_1}{\sqrt{n}}-\frac{\log n}{n},\\
\frac{1}{n}\log M_2&=R_2^*+\frac{L_2}{\sqrt{n}}-\frac{\log n}{n}.
\end{align}

Hence, according to \eqref{eq1} to \eqref{eq3} in Lemma \ref{lowerboundexcessp}, for $i\in[0:2]$,
\begin{align}
R_{i,n}&=R_i^*+\frac{L_i}{\sqrt{n}}.
\end{align} 

Recall that we use $\jmath_{XY}(X_i,Y_i)$ to denote $\jmath_{XY}(X_i,Y_i|R_1^*,R_2^*,D_1,D_2,P_{XY})$.
Invoking Lemma \ref{lowerboundexcessp}, similarly to the achievability proof,
\begin{align}
\nn&\rmP_{\rme,n}(D_1,D_2)\\*
&\geq \Pr\left\{R_{0,n}<\rvR_0(R_{1,n},R_{2,n},D_1,D_2|\hat{T}_{X^nY^n})\right\}-\frac{1}{n}\\
\nn&\geq \Pr\left\{R_{0,n}<\rvR_0(R_{1,n},R_{2,n},D_1,D_2|\hat{T}_{X^nY^n}),~\hat{T}_{X^nY^n}\in \calA_{n}(P_{XY})\right\}\\*
&\qquad-\frac{1}{n}\\
\nn&\geq\Pr\bigg\{R_0^*+\frac{L_0}{\sqrt{n}}<\frac{1}{n}\sum_{i=1}^n\jmath_{XY}(X_i,Y_i)-\xi_1^*\frac{L_1}{\sqrt{n}}-\xi_2^*\frac{L_2}{\sqrt{n}}+O(\eta_n)\\*
&\qquad\qquad\qquad\mathrm{~and~}\hat{T}_{X^nY^n}\in \calA_{n}(P_{XY})\bigg\}-\frac{1}{n}\\ 
\nn&\geq \Pr\bigg\{R_0^*+\frac{L_0}{\sqrt{n}}<\frac{1}{n}\sum_{i=1}^n\jmath_{XY}(X_i,Y_i)-\xi_1^*\frac{L_1}{\sqrt{n}}-\xi_2^*\frac{L_2}{\sqrt{n}}+O(\eta_n)\bigg\}\\*
&\qquad-\Pr\left\{\hat{T}_{X^nY^n}\notin \calA_n(P_{XY})\right\}-\frac{1}{n} \label{eqn:rev_union_bd}\\
\nn&=\Pr\bigg\{\frac{L_0}{\sqrt{n}}+\xi_1^*\frac{L_1}{\sqrt{n}}+\xi_2^*\frac{L_2}{\sqrt{n}}+O(\eta_n)<\frac{1}{n}\sum_{i=1}^n\jmath_{XY}(X_i,Y_i)\\*
&\qquad\qquad-\rvR_0(R_1^*,R_2^*,D_1,D_2|P_{XY})\bigg\}-\frac{2|\calX||\calY|}{n^2}-\frac{1}{n}\\
\nn&\geq \rmQ\left(\frac{L_0+\xi_1^*L_1+\xi_2^*L_2+O\left(\sqrt{n}\eta_n\right)}{\sqrt{\rmV(R_1^*,R_2^*,D_1,D_2|P_{XY})}}\right)-\frac{6\mathrm{T}(R_1^*,R_2^*,D_1,D_2)}{\sqrt{n}\rmV^{3/2}(R_1^*,R_2^*,D_1,D_2)}\\*
&\qquad\qquad-\frac{2|\calX||\calY|}{n^2}-\frac{1}{n},
\end{align} 
where \eqref{eqn:rev_union_bd} follows from the fact that $\Pr\{\calE\cap\calF\}\ge\Pr\{\calE\}-\Pr\{\calF^c\}$. 
Hence, if $(L_0,L_1,L_2)$ satisfies
\begin{align}
L_0+\xi_1^*L_1+\xi_2^*L_2< \sqrt{\rmV(R_1^*,R_2^*,D_1,D_2|P_{XY})}\mathrm{Q}^{-1}(\varepsilon),
\end{align} 
then $\liminf_{n\to\infty}\rmP_{\rme,n}(D_1,D_2)>\varepsilon$. Therefore, for sufficiently large $n$, any second-order $(R_0^*,R_1^*,R_2^*,D_1,D_2,\varepsilon)$-achievable triplet $(L_0,L_1,L_2)$ must satisfy
\begin{align}
\label{outerregion}
L_0+\xi_1^*L_1+\xi_2^*L_2\geq \sqrt{\rmV(R_1^*,R_2^*,D_1,D_2|P_{XY})}\mathrm{Q}^{-1}(\varepsilon).
\end{align}


\chapter{Reflections, Other Results and Future Directions}
\label{chap:future}

\section{Reflections}

In this monograph, we reviewed recent advances in the second-order asymptotics for lossy source coding, which provides approximation to the finite blocklength performance of optimal codes. The monograph is divided into three parts: Part I, consisting of two chapters,  introduces the basics; Part II, consisting of six chapters, is concerned with the point-to-point setting; and Part III, consisting of four chapters, deals with multiterminal settings.

Specifically, in Chapter \ref{chap:intro}, we introduced the notation and critical mathematical background. In Chapter \ref{chap:lossless}, we illustrated non-asymptotic and second-order asymptotic analyses via lossless source coding. Subsequently, in Chapter \ref{chap:rd} of Part II, we presented the generalization of the results from lossless source coding to the rate-distortion problem of lossy source coding, highlighted the role of the distortion-tilted information density and introduced two proof sketches. One proof method to yield second-order asymptotics is applying the Berry-Esseen theorem to carefully derive non-asymptotic achievability and converse bounds, where the achievability part uses random coding and minimal distortion encoding while the converse part relies on the properties of the distortion-tilted information density. Although this method is simple and elegant, it is not always possible to derive the desired non-asymptotic bounds for multiterminal lossy source coding problems. Thus, we also introduced another proof technique using the method of types, where the achievability part uses the type covering lemma tailored to the rate-distortion problem and the converse part depends on a type-based strong converse analysis. The first proof sketch using the non-asymptotic bounds usually applies to any memoryless source while the method of types is valid only for a DMS. In the rest of Part II, the results and proofs for the rate-distortion problem are generalized to account for noisy sources, noisy channels, mismatched compression, sources with memory and variable length compression in Chapters \ref{chap:noisy} to \ref{chap:variable}.

In Part III, the two proof methods for the rate-distortion problem are generalized in combination to derive non-asymptotic and second-order asymptotic bounds for four multiterminal lossy source coding problems in the increasingly complicated order: the Kaspi problem in Chapter \ref{chap:kaspi}; the successive refinement problem in Chapter \ref{chap:sr}; the Fu-Yeung problem in Chapter \ref{chap:fu-yeung}; and the Gray-Wyner problem in Chapter \ref{chap:gw}. For the Kaspi problem, we introduced the distortions-tilted information density, illustrated the role of side information and showed that the conditional rate-distortion problem is a special case of the Kaspi problem. For the successive refinement problem, we defined a rate-distortions-tilted information density, showed its connection to the minimal sum rate subject to the rate of one encoder, demonstrated the tradeoff between second-order coding rates of two encoders, and validated the joint excess-distortion probability as the ``correct'' performance criterion. For the Fu-Yeung problem, we presented a non-asymptotic converse bound which yielded tight second-order converse result when specializing to the successive refinement problem and presented tight second-order asymptotics for simultaneous lossless and lossy compression. Finally, for the Gray-Wyner problem in which an auxiliary random variable is required in the characterization of the rate-distortion region, we presented a second-order asymptotic result, where the achievability part follows by deriving a type covering lemma tailored to the problem which uses the continuity of conditional rate-distortion function with respect to the distortion level and the distributions.

\section{Other Results}
This monograph mainly focused on fixed-length compression of a DMS under bounded distortion measures with the excess-distortion probability as the performance criterion. For a GMS under quadratic distortion measures, the second-order asymptotics for the rate-distortion problem was derived by Ingber and Kochman~\cite[Theorem 2]{ingber2011dispersion} and by Kostina and Verd\'u~\cite[Theorem 40]{kostina2012fixed}, and the second-order asymptotics for the successive refinement problem was derived by No, Ingber and Weissman~\cite[Theorem 7]{no2016} and by Zhou, Tan and Motani~\cite[Theorem 20]{zhou2016second}, and the second-order asymptotics for a Laplacian source under the magnitude-error distortion measure could be derived using the type-covering lemma in~\cite{zhong2006type} for the achievability result and using the non-asymptotic converse bound in \cite[Corollary 2]{kostina2019sr}. When the distortion measure is the logarithm loss, the non-asymptotic analysis for the rate-distortion and the multiple descriptions problem was derived by Shkel and Verd\'u~\cite{shkel2018tit} and the successive refinement problem was studied by No~\cite{no2019entropy}. When the excess-distortion probability is replaced by the average distortion, a non-asymptotic analysis of the rate-distortion problem was done by Moulin~\cite{moulin2017isit} and by Elkayam and Feder~\cite{Nir2020isit}.

Besides second-order asymptotics, the large and moderate deviations asymptotic analyses also provide deeper understanding beyond Shannon theory analyses, as illustrated in Fig. \ref{illus:refined4lossless} for lossless source coding. For simplicity, we call the rate-distortion function or the rate-distortion region the Shannon limit. Large deviations, also known as the error exponent analysis, focuses on deriving the exponential decay rate of excess-distortion probabilities for rates beyond the Shannon limit in lossy source coding problems.  For the rate-distortion problem, the error exponent was derived by Marton for a DMS~\cite{Marton74}, by Ihara and Kubo~\cite{ihara2000error} for a GMS under the quadratic distortion measure and by Zhong, Alajaji and Campbell~\cite{zhong2006type} for a Laplacian memoryless source under the magnitude-error distortion measure. For the successive refinement problem with a DMS, the error exponent region was derived by Tuncel and Rose~\cite{tuncel2003} under the separate excess-distortion probabilities criterion and by Kanlis and Narayan under the joint excess-distortion probability criterion~\cite{kanlis1996error}. For a DMS, the error exponent (region) for the Kaspi problem was derived in~\cite[Theorem 7]{zhou2017non}, for the Fu-Yeung problem was derived in~\cite[Theorem 16]{zhou2017non} and for the Gray-Wyner problem was derived in~\cite[Theorem 12]{zhou2015second}.

Moderate deviations asymptotics~\cite{chen2007redundancy,he2009redundancy,altugwagner2014} compromises between large deviations and second-order asymptotics by deriving the subexponential decay rates, also known as the moderate deviations constants, of excess-distortion probabilities while allowing rates to approach the Shannon limit. The moderate deviations constant for the rate-distortion problem was derived by Tan~\cite{tan2012moderate} for a DMS. For the successive refinement problem, the moderate deviations constants were derived by Zhou, Tan and Motaini for both a DMS and a GMS~\cite[Theorems 6 and 15]{zhou2016second}. For a DMS, the moderate deviations asymptotics was derived for the Kaspi problem in~\cite[Theorem 8]{zhou2017non}, for the Fu-Yeung problem was derived in~\cite[Theorem 17]{zhou2017non} and for the Gray-Wyner problem was derived in~\cite[Theorem 13]{zhou2015second}.

\section{Future Directions}
We briefly discuss possible future research directions for lossy source coding beyond the results covered in this monograph.

\subsection{Higher-Order Asymptotics}
For the rate-distortion problem and its five generalizations in Part II of this monograph, we present a second-order asymptotic approximation to the finite blocklength performance. It was recently shown by Yavas, Kostina and Effros~\cite{yavas2022isit} that for channel coding, the third-order asymptotic approximation in the moderate deviations regime could provide a rather accurate approximation to the performance of an optimal code for blocklengths as small as $n=100$ with error probabilities as small as $10^{-10}$. This high-order approximation is of great interest for beyond 5G communication networks where enhanced ultra-reliable and low-latency communication is required. However, to the best of our knowledge, in general, no tight third-order asymptotic results have been established for the rate-distortion problem. It would be worthwhile to derive higher-order asymptotic results to complement the second-order asymptotics for the problems presented in this monograph.

\subsection{Multiterminal Compression of a GMS}
Although the second-order asymptotics results of the rate-distortion and the successive refinement problems have been established for a GMS under quadratic measures, the second-order asymptotics of many other multiterminal lossy source coding for a GMS is generally unknown. For the Kaspi problem, the non-asymptotic converse bound in Chapter 3 is valid for a GMS, but the achievability analysis is non-trivial despite the rate-distortion function was derived by Perron, Diggavi and Telatar~\cite{perron2005kaspi}. For the multiple descriptions problem~\cite{wolf1980source}, the rate-distortion region for a GMS was derived by Ozarow~\cite{ozarow1980source}. Both achievability and converse analyses of non-asymptotic and second-order asymptotic bounds require novel ideas. For the Gray-Wyner problem, although the rate-distortion region is known~\cite{gray1974source}, the exact formula for a GMS remains open and the second-order asymptotics are challenging.

\subsection{Mismatched Multiterminal Compression}
Most contents in this monograph concerned matched compression, where the distribution of the source sequence is assumed perfectly known. Such an assumption is invalid in practice because one is not able to know the exact distribution of a source to be compressed. Thus, it is important to use mismatched coding schemes ignorant of the exact source distribution to compress any memoryless sources. In Chapter \ref{chap:mismatch}, we presented the second-order asymptotics by Zhou, Tan and Motani~\cite{zhou2017refined}, who analyzed the mismatched compression scheme proposed by Lapidoth~\cite[Theorem 3]{lapidoth1997}, where the minimum Euclidean distance encoding with the i.i.d. Gaussian codebook is used to compress an arbitrary memoryless source.  However, the non-asymptotic and second-order asymptotic analysis for more complicated multiterminal lossy source coding remains largely unexplored. Some attempts have been made very recently by Wu, Bai and Zhou~\cite{wu2022ISIT,wu2021} in the achievability analysis of the successive refinement problem.

\subsection{Variable-Length Multiterminal Compression}
This monograph focused on fixed-length lossy source coding. Motivated by the need to reduce the codeword length of frequently appeared symbols, fixed-to-variable length (FVL) source coding has also been widely studied for the point-to-point case~\cite{han2000tit,Koga2005tit,kostina2015tit,saito2017isit,saitoisit2018,sakai2021tit}. In particular, Kostina and Verd\'u derived the second-order asymptotics for average codeword length of the FVL rate-distortion problem subject to a non-vanishing excess-distortion probability, which was presented in Chapter \ref{chap:variable}. Saito, Yagi and Matsushima~\cite{saito2017isit,saitoisit2018} studied the FVL rate-distortion problem under constraints on both the excess-distortion probability and the excess-length probability. However, no results have been established for FVL multiterminal lossy source coding. It would be worthwhile to derive non-asymptotic and second-order asymptotic bounds on the average codeword length for a multiterminal lossy source coding problem such as the successive refinement problem.

\subsection{Decoder Side Information Problems}
Although we have presented results for several multiterminal lossy source coding problems, many more remain open, such as the Wyner-Ziv problem~\cite{wyner1976rate}, the Kaspi-Heegard-Berger problem~\cite[Theorem 2]{kaspi1994},~\cite{heegard1985}, and the Berger-Tung problem~\cite{berger1978multiterminal}. A common feature of these problems is that in the asymptotic rate-distortion function (region), there exists an auxiliary random variable that forms a Markov chain with the source sequences and/or the side information. For the Wyner-Ziv problem, some attempts in characterizing the second-order asymptotics have been made in the achievability part by Watanabe, Kuzuoka and Tan~\cite{watanabe2015} and by Yassaee, Aref and Gohari~\cite{yassaee2013technique} and the converse part by Oohama~\cite{oohama2016wynerziv}. However, the achievability and converse bounds do not match even in the sign of the second-order term. Novel ideas and mathematical tools are required to establish second-order asymptotics.


\subsection{Rate-Distortion-Perception Tradeoff}
As evidenced in many applications of image compression, optimal schemes achieving the rate-distortion function lead to low performance due to the ignorance of the distribution of the reproduced sequences. The perceptual quality of an image is shown to be determined by the distribution of the reproduced sequences. However, this information is omitted in the design of codes described in this monograph. To solve this problem, recent studies on rate-distortion-perception tradeoff~\cite{theis2021coding,blau2019rethinking} revisit the rate-distortion problem by constraining that the distribution of the output of the decoder is either identical or approximately identical to the distribution of the source sequence. All these results are asymptotic Shannon theoretical analysis on the rate-distortion function for the simple point-to-point case. It would be of interest to conduct a non-asymptotic and second-order asymptotic analysis of the rate-distortion-perception problem and also generalize it to more complicated multiterminal lossy source coding problems.


\bibliographystyle{unsrt}
\bibliography{IEEEfull_lin}

\end{document}